\documentclass[12pt,times,oneside]{PhDThesisPSnPDF}

\usepackage{geometry}
\ifsetCustomFont
  \RequirePackage{helvet}
\fi

\usepackage{algorithm,algpseudocode}
\algnewcommand\algorithmicinput{\textbf{INPUT:}}
\algnewcommand\INPUT{\item[\algorithmicinput]}
\algnewcommand{\algorithmicoutput}{\textbf{OUTPUT:}}
\algnewcommand\OUTPUT{\item[\algorithmicoutput]}
\algrenewcommand{\algorithmiccomment}[1]{\hskip3em$\%$ #1}

\RequirePackage[labelsep=space,tableposition=top]{caption}
\renewcommand{\figurename}{Fig.} 

\usepackage{caption}
\usepackage{subcaption}
\usepackage{booktabs} 
\usepackage{slashbox,multirow}
\usepackage{amsfonts}
\usepackage{amsmath}
\usepackage{amssymb}

\newtheorem{theorem}{Theorem}[section]
\newtheorem{lemma}[theorem]{Lemma}
\newtheorem{proposition}[theorem]{Proposition}

\newenvironment{proof}[1][Proof]{\begin{trivlist}
\item[\hskip \labelsep {\bfseries #1}]}{\end{trivlist}}

\newcommand{\qed}{\nobreak \ifvmode \relax \else
      \ifdim\lastskip<1.5em \hskip-\lastskip
      \hskip1.5em plus0em minus0.5em \fi \nobreak
      \vrule height0.75em width0.5em depth0.25em\fi}
			
\ifuseCustomBib
   \RequirePackage[square, sort, numbers, authoryear]{natbib} 
\fi

\title{Wireless Security with Beamforming Technique}

\author{Yuanrui Zhang M.Sc.}

\dept{School of Electronics, Electrical Engineering and Computer Science}

\university{Queen's University Belfast}
\crest{\includegraphics[width=0.25\textwidth]{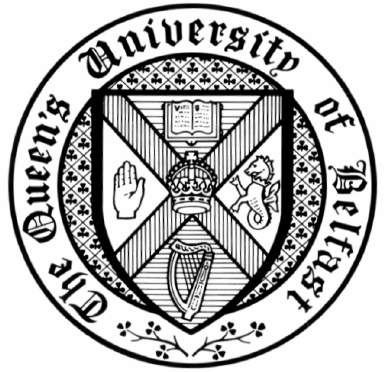}}

\degreetitle{Doctor of Philosophy}

\degreedate{June 2016} 

\subject{LaTeX} \keywords{{LaTeX} {PhD Thesis} {School of EEECS} {Queen's University Belfast}}

\usepackage{nomencl}
\makenomenclature

\begin{document}
\frontmatter
\begin{titlepage}
  \maketitle
\end{titlepage}

\begin{abstract}

This thesis focuses on the wireless security in the physical layer with beamforming technique.
As the wireless communications grow more important, a higher level of security is more demanding as well.
New techniques have been proposed in complement to the existing encryption-based methods in the communications protocols.
One of the emerging areas is the security enhancement in the physical layer, which exploits the intrinsic properties of the wireless medium.
Beamforming, which has been proved to have many advantages, such as increasing data rates and reducing interference, can also be applied to enhance the wireless security.

One of the most common threats, i.e., passive eavesdropping, is addressed in this thesis.
Passive eavesdroppers bring significant challenges to the system design, because the knowledge of their locations and channel condition is generally difficult to acquire.
To reduce the risk of leaking information to the eavesdroppers, the physical region where the transmission is exposed to eavesdropping has been studied.
In this thesis, the exposure region based beamforming technique is proposed to combat the threat from the randomly located eavesdroppers in a fading channel.

A stochastic geometry tool is used to model the distribution of the passive eavesdroppers.
In this system model, the large-scale path loss and a general Rician fading channel model are considered.
The exposure region is established to describe the secrecy performance of the system, based on which the probability that the secrecy outage event occurs is defined to evaluate the security level of the exposure region.

The antenna array is one of the most important factors that affect the secrecy performance of the exposure region based beamforming technique.
The potential of using different array geometry and array configuration to improve security is explored.
In this thesis, two common arrays, i.e., linear and circular arrays, are considered.
Analytic expressions for general array geometry as well as for the linear and circular arrays are derived.
In addition, numerical results are used to analyze the behaviors of the antenna array towards security.
Based on the empirical results, numerical optimization algorithms are developed to exploit the array configuration to enhance the system security level.

In complement to the theoretical analysis, experiments are carried out to study the performance of the beamformer with linear and circular arrays in practice.
Especially, the impact of the mutual coupling on the security performance is investigated.
To this end, numerical simulation results as well experimental results are used to study the behaviors of the linear and circular arrays towards security.

\end{abstract}

\tableofcontents
\listoffigures
\listoftables
\printnomenclature 

\mainmatter

\chapter{Introduction}
\label{chp1}

\section{Background}
\label{chp1:bkgnd}


As wireless technologies become more reliable, more efficient and more convenient, wireless communications have gradually become an integral part of human activities over the past few decades, from everyday life to industrial productions.
New concepts and ideas have revolutionized the way of communication, e.g., from single-antenna system to massive multiple-input-multiple-output (MIMO) system; 
and the frontier of their applications has been extended to the latest technologies, e.g., `Internet of things', which is core for both 5G and `Industry 4.0'.


\nomenclature{MIMO}{multiple-input-multiple-output}

According to one of Cisco's latest reports\,\cite{index20152020global}, there will be an eightfold increase of the mobile data traffic globally in the next five years.
With the evidently increasing data traffic, the need of the wireless security grows more demanding.
Opposite to the benefits that the wireless transmission brings, the broadcast nature makes it more vulnerable to adversarial behaviors than wired transmission.
Shannon's perfect secrecy was first established in 1949\,\cite{shannon1949communication} and later in 1975, Wyner conceived the wiretap channel which in theory makes the perfect secrecy achievable\,\cite{wyner1975wire}.
However, due to the lack of practical codes for the wiretap channel and the assumption that an adversary should suffer more noise than the legitimate user, the information-theoretic security was soon overtaken by encryption techniques\,\cite{salomaa2013public} in practice.
Since then, most standard security solutions rely on the authentication and encryption schemes in the communications protocols, e.g., the IEEE 802.11. 

Encryption techniques essentially rely on the computational complexity and have not been proved strictly secure from the information-theoretic perspective\,\cite{massey1988introduction}.
Furthermore, there are many limitations when they are used to cope with various threats that exist in wireless networks\,\cite{al2006ieee}\,\cite{Dhiman2014}.  
For example, the techniques used in the 802.11 protocols, such as wireless encryption protocol (WEP) and Wi-Fi protected access (WPA), are vulnerable to practical attacks\,\cite{tews2009practical}\,\cite{shiu2011physical}. 


\nomenclature{WEP}{wireless encryption protocol}
\nomenclature{WPA}{Wi-Fi protected access}

Because of the limitations of the encryption techniques in the high layers above the physical layer, different solutions from the physical layer have been proposed.
The primary difference between the information-theoretic security and the encryption techniques is that the former limits the amount of information that can be obtained at the bit level by the adversarial receiver, whereas the latter makes it computationally hard to decipher\,\cite{mukherjee2010principles}.
Various physical layer techniques that exploit the inherent randomness of noise and wireless channels are developed towards the security aspect in order to meet the challenges that are raised by the boom of wireless communications\,\cite{bloch2011physical,mukherjee2010principles,liu2010securing,zhou2013physical}.

Security methods from the physical layer can co-exist with the existing encryption based schemes.
The security enhancement 
from the physical layer can provide many advantages compared to the conventional encryption techniques\,\cite{bloch2008wireless}.
The amount of information that can be obtained by the adversary can be precisely measured based on the channel quality; 
 it is not subject to the growing computational resources that constantly threaten the computational models based on the encryption techniques, e.g., via brute force attack.
In theory, it is possible to approach perfect secrecy using suitably long codes; 
  it is realized for quantum key distribution in practice.
In addition, it does not need complex protocols for key distribution and management.
Thus, not much change to the existing system architecture for communication is required to achieve security in the physical layer.

Beamforming has been proved to exploit the wireless channel to achieve better quality-of-service in terms of bit rate and error performance in the physical layer\,\cite{mietzner2009multiple} and has a wide range of applications in wireless communications, e.g., MIMO relay networks\,\cite{vouyioukas2013survey}.
To approach Shannon's perfect secrecy in the wiretap channel, the legitimate user's channel should have some kind of advantage over the adversary's channel. 
While the omni-directional antenna cannot provide such an  advantage, beamforming, one of the prominent multiple-antenna techniques, has the ability to create better channel for the legitimate user\,\cite{mukherjee2011robust,mukherjee2009utility}.

\nomenclature{Eve(s)}{eavesdropper(s)}

There has been work that exploits beamforming to achieve security from the information-theoretic perspective.
For example, the simplest case would be to use beamforming to increase the signal power at the legitimate user's direction while suppressing the signal power at other directions.
There are also examples to use beamforming to generate interference specifically towards adversary users\,\cite{negi2005secret,goel2008guaranteeing}.
Therefore, it is desirable to achieve reliability, efficiency and security from the beamforming techniques which have the potential to provide an all-in-one solution.

\section{Motivation}
\label{chp1:motive}


The broadcast nature of wireless transmission makes it vulnerable to many attacks, e.g., eavesdropping.
Beamforming technique has the ability to control the direction of transmission.
In fact, beamforming is essentially a spatial filter that focuses energy at a certain direction while suppresses energy at some other directions\,\cite{van1988beamforming}.
It exploits the spatial domain other than the time and frequency domains compared to the single-antenna techniques (e.g., channel coding).

This thesis mainly investigates the problem caused by a particular adversarial behavior, i.e., passive eavesdropping, which already existed before the time of wireless communications and threatens the secrecy and privacy of wireless transmissions\,\cite{pathan2006security}.
The classical model consists of Alice, Bob and Eve, where Alice is the transmitter, Bob is the legitimate user and Eve is the eavesdropper.
Alice sends a message to Bob in the presence of Eve(s). 
However, Eve can easily intercept the message sent to Bob, if she is within the coverage range.
Thus, it is suitable to use beamforming to provide the advantage that is required for Bob over Eve in Wyner's wiretap channel model, in order to achieve the information-theoretic security.

Besides the intentionally imposed difference between Bob's and Eve's channels, e.g., via beamforming, users' locations provide a certain level of distinction of the related channels.
Thus, in this thesis, the geometric locations of Bob and Eve are also considered.
Generally speaking, the closer a user to the transmitter is, the stronger the received signal is, thus the better the user's channel is.
Location was ignored in the beginning of information-theoretic security research, partly attributed to the fact that the Eve's location is often random and unknown to Alice.
With the aid of stochastic geometry theory, the distribution of the random users' locations can be modeled, e.g., via a Poisson point process (PPP)\,\cite{haenggi2009stochastic,chiu2013stochastic}.
In addition, the fact that the localization granularity keeps improving drives the related research in many contexts\,\cite{liu2008lke,yan2014optimal,yan2014signal}.
Therefore, the awareness of user's location promotes the utilization of location towards wireless security.
For example, `ArrayTrack'\,\cite{xiong2013arraytrack} that improves the granularity is also studied in the context of enhancing security\,\cite{xiong2013securearray}.

\nomenclature{PPP}{Poisson point process}

The spatial-filtering ability of beamforming is suitable for distinguishing the locations that are secure or insecure for the transmission to Bob.
In this thesis, beamforming is used to form physical region in terms of information-theoretic security.
There has been related work that attempts to create physical region to combat the randomness of both Eves' location and the wireless channel, 
however, one important factor, i.e., the antenna array itself, is overlooked.

Since beamforming is performed via antenna arrays, its security performance relies on the array configuration.
Naturally, the physical region created by using beamforming is highly related to the array and can be altered via changing the array configuration.
However, the impact of the array configuration is rarely studied.
One related work that explores the security performance for the geometric distribution of Eve\,\cite{yan2014secrecy} overlooked the importance of the array configuration.
In another work\,\cite{mehmood2015secure}, the array pattern is synthesized towards the security performance metric, but not related to the physical region.
Therefore, this thesis investigates the possibility of leveraging the array configuration to improve the wireless security. 


\section{Contributions}
\label{chp1:contri}

The objective of this thesis is to enhance the security level of the wireless transmission from Alice to Bob in the presence of the PPP distributed Eves from the spatial aspect via beamforming technique.
The major challenge is that Eves' channel state information (CSI) and locations are random and unknown to Alice.
Thus, the randomness comes both from the PPP distribution and the small-scale fading.
The key contribution of this thesis is to propose an exposure region (ER) based beamforming technique with both information-theoretic and numerical analysis to meet the previous challenge from a new angle, i.e., the antenna array that shapes the ER, which serves as the bridge between the information-theoretic security and a more controllable security-related physical region. 
In this thesis, to the author's best knowledge, the existing work related to wireless security based on the physical region is summarized for the first time.
Moreover, to examine the practicality of the ER-based beamforming technique, a transmit beamformer is built on a hardware platform to investigate its secrecy performance. 
In the following, a summary of key contributions is provided.


\nomenclature{CSI}{channel state information}
\nomenclature{ER}{exposure region}

\begin{itemize}
	\item The concept of the ER is established based on information-theoretic secrecy parameter, i.e., the secrecy outage, based on which the spatial secrecy outage probability (SSOP) is defined the performance metric and its accurate expression is derived  for a general fading channel and an arbitrary antenna array.
	\item The upper bound of the SSOP is derived to facilitate analytical analysis.
	Particularly, the analytic expressions for the uniform linear array (ULA) and the uniform circular array (UCA) are obtained and analyzed to show how the upper bound changes with different array configurations, which is used to predict the properties of the SSOP.
	\item With analytical and numerical analysis, the properties of the SSOP for the ULA and the UCA are compared with respect to various parameters.
	As the conclusion, the UCA is more suitable to develop optimization algorithms that minimize the SSOP.
	\item Based on the empirical results, two numerical optimization algorithms are developed for the adjustable and fixed transmit power scenarios to minimize the SSOP.
	One algorithm produces the optimum radius of the UCA.
The other one called the configurable beamforming technique leverages different array configurations to achieve the minimum SSOP according to Bob's dynamic location.
The algorithms can be generalized and thus are applicable to a wide range of parameters.
	\item A practical issue, i.e., the mutual coupling, is examined with the aid of the wireless open access research platform (WARP) and the numerical electromagnetics code (NEC).
A practical beamformer is built on WARP.
The impact of the mutual coupling to ULA and UCA is compared and the implication to the ER-based beamforming technique and the numerical optimization algorithms is investigated.
\end{itemize}

Part of the contributions regarding to the study of security behaviors with respect to the array parameters and the numerical optimization algorithms on the ULA is published in\,\cite{mypaper}.
Part of the contributions regarding to building a practical beamformer on WARP and the study of the mutual coupling for the ULA is published in\,\cite{mypaper2}. 
Part of the contributions regarding to the derivation of the analytic expression of the SSOP on the UCA is under view in\,\cite{mypaper3}.
Part of the contributions regarding to the concept of the ER and the SSOP and the analysis on the ULA is under view in\,\cite{mypaper4}.
Part of the contributions regarding to the optimization algorithm on the UCA is in preparation in\,\cite{mypaper5}.

\nomenclature{SSOP}{spatial secrecy outage probability}
\nomenclature{ULA}{uniform linear array}

\nomenclature{NEC}{numerical electromagnetics code}
\nomenclature{UCA}{uniform circular array}
\nomenclature{WARP}{wireless open access research platform}


\section{Thesis Outline}
\label{chp1_outline}

The rest of the thesis is organized as follows.
In Chapter\,2, a collection of several topics are presented as research background and preliminaries for the thesis.
First, an overview of the wireless security in the physical layer is given and some fundamental concepts in the area of information-theoretic security are introduced.
Subsequently, the key literature review about the security methods related to the physical region is provided with a short summary in the end.
Afterwards, the fundamental concepts, i.e., the array steering vector and the array factor, are introduced and the impact of the mutual coupling is briefly explained.
Next, the wireless channel used in the thesis is given.
In the end, an introduction for WARP and NEC is provided.

In Chapter\,3, the system model that incorporates the geometric locations for the generalized Rician channel is introduced, and the ULA is chosen as an example to develop the concept of the ER, based on which the SSOP is then derived.
The analytic upper bound of the SSOP is obtained.
Then the SSOP and its upper bound are analyzed via analytic and numerical methods.
The analysis of the security performance regarding to the array parameters is first made for the deterministic channel, then generalized for the Rician fading channel.
In addition, the tightness of the upper bound is examined.

In Chapter\,4, with the aid of the general expressions in Chapter\,3, the SSOP and its upper bound for UCA are derived.
Then the security performance regarding to the array parameters for the UCA is studied and compared in parallel with the ULA for the deterministic channel.
Subsequently, the conclusions are extended to the Rician fading channel, including the comparison for the tightness of the upper bound for the ULA and the UCA.
The mutual coupling is investigated for both ULA and UCA via WARP experiments and NEC simulations.

In Chapter\,5, the system model for the UCA with adjustable array configuration is established with some basic concepts, e.g, the array mode and the coverage zone.
The optimization problem is formulated and the key parameters of the SSOP are jointly analyzed.
Based on the empirical results, two numerical optimization algorithms, which are applicable to the generalized Rician channel, are developed for the dynamic and fixed transmit power constraints.
The deterministic channel is used as an example to illustrate the numerical implementation of the algorithms.
The error analysis for the configurable beamforming technique is conducted.
Next, the analysis of the mutual coupling on the UCA with adjustable array configuration is conducted via NEC simulations,  and the impact of the mutual coupling on the two optimization algorithms is investigated.

In Chapter\,6, the summary of this thesis is provided and suggestions on this topic for future work are given.

\chapter{Literature Review and Research Background}
\label{chp2}

\section{Introduction}
\label{chp2:intro}

This chapter contains a collection of four topics that will be involved in this thesis.
Each section covers one topic.
For each topic, the fundamentals are introduced as preliminaries for this thesis.
A literature review is carried out for these topics and the key findings of the related work are presented.

The wireless security in the physical layer is a broad area.
Without distracting from the main topic of this thesis, a comprehensive overview with a few selected fields, such as fading channels and multiple-antenna techniques, is given to reveal the development of these fields.
To facilitate further understanding, some basic concepts are presented.
Then, the related work to the wireless security from the physical region perspective is surveyed.
To the author's best knowledge, this is the first time that this area is comprehensively reviewed from two aspects, i.e., the information-theoretical aspect and the physical space security aspect.

The antenna array and the wireless channels are studied, which serves as the preliminaries for the system models in this thesis.
Then, some entry-level introduction is provided for the experiment and simulation tools that are used in this thesis, in order to help readers understand the set-ups and results in this thesis.

This chapter is organized as follows. 
In Section\,\ref{chp2:PhySec}, an overview for the wireless security in the physical layer is provided and some basic concepts are introduced; the related work to the physical region is surveyed with a brief summary in the end.
In Section\,\ref{chp2:antennas}, the fundamental concepts used in the field of antenna array are presented and the mutual coupling is introduced.
In Section\,\ref{chp2:channel}, the background of the wireless channel is given and the channel model used in this thesis is explained.
In Section\,\ref{chp2:mc}, an introduction for the experiment and simulation tools is provided.
In Section\,\ref{chp2:conclu}, the conclusions of this chapter are given.

\section{Wireless Security in the Physical Layer}
\label{chp2:PhySec}
\subsection{Overview}
\label{chp2:PhySec:zmkcs}

When the wireless security is discussed for the physical layer, 
what most people refer to is the information-theoretic security.
The beginning is when the concept of Shannon's perfect secrecy was conceived in 1949\,\cite{shannon1949communication}.
After some initial developments, Wyner established the wiretap channel model and showed the possibility to approach Shannon's perfect secrecy on the condition that Eve's channel must be weaker than Bob's channel\,\cite{wyner1975wire}.

Wyner's wiretap channel model has laid the foundation for much follow-up work.
A few years later, the wiretap channel model was extended to non-degraded discrete memoryless broadcast channels\,\cite{csiszar1978broadcast}. 
The secrecy capacity of the Gaussian wiretap channel is characterized in\,\cite{leung1978gaussian}, based on which a substantial body of work is developed.
In this work, both Bob's and Eve's channels are additive white Gaussian noise (AWGN) channels.
The \textit{secrecy capacity}, which is the maximum transmission rate at which Eve cannot decode any information, is calculated by the difference between Bob's and Eve's channel capacities.

\nomenclature{AWGN}{additive white Gaussian noise}

Several fading channels are considered based on the Gaussian wiretap channel, which leads to the usage of outage probability in the security performance metrics.
In\,\cite{barros2006secrecy,bloch2008wireless}, quasi-static fading is studied under the assumption that Eve's CSI is not available at Alice.
In the absence of Eve's CSI, the outage formulation is adopted to evaluate the secrecy performance and the \textit{secrecy outage probability} (SOP) is defined by the probability that the secrecy capacity is below a target secrecy rate.
The SOP is a useful measure for delay-limited applications.
On the other hand, for delay-tolerant applications, the ergodic secrecy capacity can be used to measure the performance of the secure communications. 
In such cases, opportunistic exploitation of the time intervals that Bob has a better channel even allows secure transmission when Eve's channel is on average better\,\cite{hong2013enhancing}.
In\,\cite{li2009secrecy,liang2008secure,gopala2008secrecy,khisti2008secure}, the secrecy capacity for ergodic fading models provided with different levels of CSI are studied for optimal power and rate allocation.
In addition, the secrecy capacity for block fading channels is studied in\,\cite{gopala2008secrecy}.

\nomenclature{SOP}{secrecy outage probability}

Various types of systems that exploit multiple antennas leverage the available spatial dimension to enhance the secrecy capabilities.
In\,\cite{li2007secret}, the achievable secrecy rate for MIMO system is studied and the analytic solution to the optimal input structure is derived for a degraded system, i.e., the multiple-input-single-output (MISO) system.
In\,\cite{shafiee2007achievable}, the achievable secrecy rate is considered for MISO system with Gaussian channel inputs and it shows that the optimal transmission strategy is beamforming.
The work is extended to the MISO system with multiple Eves, i.e., multi-input, multi-output, multi-eavesdropper (MISOME) channel and beamforming is proved to the capacity-achieving solution\,\cite{khisti2010secure}.
Furthermore, the work is extended to other multiple-antenna systems\,\cite{parada2005secrecy,khisti2010secure2,oggier2011secrecy, mukherjee2011robust}, where single-input-multiple-output (SIMO), MIMOME and multi-user MIMO are investigated, respectively.

\nomenclature{MISO}{multiple-input-single-output}
\nomenclature{MISOME}{multi-input, multi-output, multi-eavesdropper}
\nomenclature{SIMO}{single-input-multiple-output}

The study of the wiretap channel goes further in a wider range as the wireless communications techniques evolve.
After some pioneering work of the relay channel model\,\cite{van1968transmission,cover1979capacity}, cooperative systems regained much attention as distributed antenna array in the wireless network\,\cite{wang2010cooperative}.
An explicit inner bound of the capacity region for the confidential message transmission is derived in\,\cite{oohama2001coding}.
In\,\cite{he2010cooperation},  an achievable secrecy rate is studied for the general untrusted relay channel under two different scenarios.
In\,\cite{chen2015physical}, the secrecy performance of the full-duplex relay network is investigated.
Further on, the distributed network secrecy is conceived in the multilevel network that contains scattered sensors and monitors in hierarchical architecture\,\cite{lee2013distributed}; and in\,\cite{win2014cognitive} a framework with interference engineering strategies is designed and analyzed for the cognitive networks with secrecy.

As a promising technique for 5G, massive MIMO is already being considered for security.
As the number of antenna elements increases towards infinity, the spectrum and power efficiency grows rapidly\,\cite{ngo2013energy} and it produces very sharp beams and low sidelobes\,\cite{alrabadi2013beamforming}.
The advantages provided by the massive MIMO could potentially benefit the security performance.
In\,\cite{dean2013physical}, a low-complexity physical-layer cryptography based on the massive MIMO channel is developed, where  digital signatures based on location or time is created.
In\,\cite{chen2014secrecy}, explicit expressions of secrecy outage capacities for the massive MIMO channel are derived for different relay strategies.

Based on the aforementioned results in different systems with the wiretap channel model, there are two ways to increase the secrecy capacity, i.e., either by improving Bob's channel or deteriorating Eve's channel.
Motivated by the results from the information-theoretic security, signal processing techniques are developed to enlarge the difference between Bob's and Eve's channels\,\cite{hong2013enhancing}.
Previously, the literature review does not differentiate the signal processing and information-theoretic treatments\,\cite{mukherjee2010principles}.
In the following, it will be focused on the techniques from the signal processing perspective, which helps construct effective wiretap channels that allow the adoption of high-rate wiretap codes\,\cite{hong2013enhancing}.

Secrecy beamforming and precoding schemes are exploited to enhance the Bob's channel, e.g.,\,\cite{khisti2010secure,khisti2010secure2,li2011optimal,li2011multicast}.
Generally speaking, beamforming refers to the transmissions where only one data stream is sent via multiple antennas, while precoding generally means the simultaneous transmission of multiple data streams via multiple antennas. 
The key idea of both schemes is to transmit signal at directions in the spatial dimension that generates the best quality of reception for Bob while reduce the quality of reception for Eve if possible.
This thesis focuses on the beamforming techniques for the single data stream transmission as a starting point, which is also reasonable in certain scenarios, e.g., securing a transmission from an access point (AP) with an antenna array to a user with a single antenna.
In\,\cite{shafiee2007achievable}, it has been shown that beamforming is the optimum transmission strategy.
With beamforming, it often yields simpler designs.
For example, the MISO channel can be simplified into a single-input-single-output (SISO) channel.
It can also be extended into multicast scenario\,\cite{li2011multicast}.


\nomenclature{SISO}{single-input-single-output}

Artificial noise (AN) or jamming can be used on top of beamforming to further deteriorate Eve's channel\,\cite{negi2005secret,goel2008guaranteeing,mukherjee2009utility,zhou2010secure}.
It is especially useful when Eve's CSI is not known or partial known, in which case it is difficult to exploit beamforming to suppress Eve's signal quality.
The key idea is to superimpose AN to the information-bearing signal to increase the inference at Eve while Bob's reception is not or little affected, because the AN is added in the null space of Bob's channel\,\cite{goel2008guaranteeing}.
The concepts of beamforming and AN are also carried over to relay systems\,\cite{zhang2010collaborative,huang2011cooperative,jeong2012joint}.

\nomenclature{AN}{artificial noise}

So far, the reviewed work mostly refers to the information-theoretic security based on Wyner's work with a focus on the fading channel and multiple-antenna systems. 
There is, however, another branch of information-theoretic security based on the secret key that is extracted from the physical channels and is defined as the key-based security.
The key-based security is different from the encryption techniques in the way that the key is generated and shared by Alice and Bob in the physical layer and its security performance can be measured precisely by the secrecy capacity\,\cite{bloch2011physical}.
For convenience, the work based on Wyner's wiretap channel model is often referred to as the `keyless' security.

The key-based security originated in Maurer's work in 1993\,\cite{maurer1993secret} which guarantees secrecy even when Eve observes a better channel than Bob.
The key lies in the joint development of a secret key by Alice and Bob over public channel. 
The advantage of the key-based security over Wyner's wiretap model is that there is no restrictions on Eve's channel and it is simpler to design.
However, the key generation is often limited by the physical channel between Alice and Bob. 
The discussion of the key-based security is beyond the scope of this thesis and more details are referred to in\,\cite{mukherjee2010principles,rawat2013security,wang2015survey,zeng2015physical,zhangkey} and the references therein.

In addition to the information-theoretic security, there are other techniques employed in the physical layer for the purpose of achieving security for the systems with randomly located Eves.
The key idea is to exploit the spatial domain to isolate a physical region via directional antenna or smart antenna, i.e., beamforming, to limit Eve's access to the transmission between Alice and Bob.
In this chapter, these work is uniformly referred to as `physical space security' that is first mentioned in\,\cite{4595864}.

To help better understand the relationships between the aforementioned concepts and also give a high-level overview for the rest of this section, a structured diagram is presented in Fig.\,\ref{fig:chp2_structure}.
In Section\,\ref{chp1:bkgnd}, wireless security is discussed from two aspects, the encryption techniques in the higher layer and the emerging area of the information-theoretic security in the physical layer, the latter of which is the starting point of this thesis.
In this subsection, the overview of the information-theoretic security is presented with the focus on the keyless security based on Wyner's wiretap channel model, for which the information-theoretic and the signal processing aspects are both discussed.
In Section\,\ref{chp2:PhySec:mnyo} and\,\ref{chp2:PhySec:bmdl}, the physical region related work is discussed from two different aspects, one based on the information-theoretic parameters and the other based on the conventional performance metrics.
The work in this thesis is based on the physical region related work and provides the information-theoretic analysis for the created physical region, especially from the array configuration perspective, which is missing in the existing work.

\begin{figure}[t]
\centering
\includegraphics[scale=1]{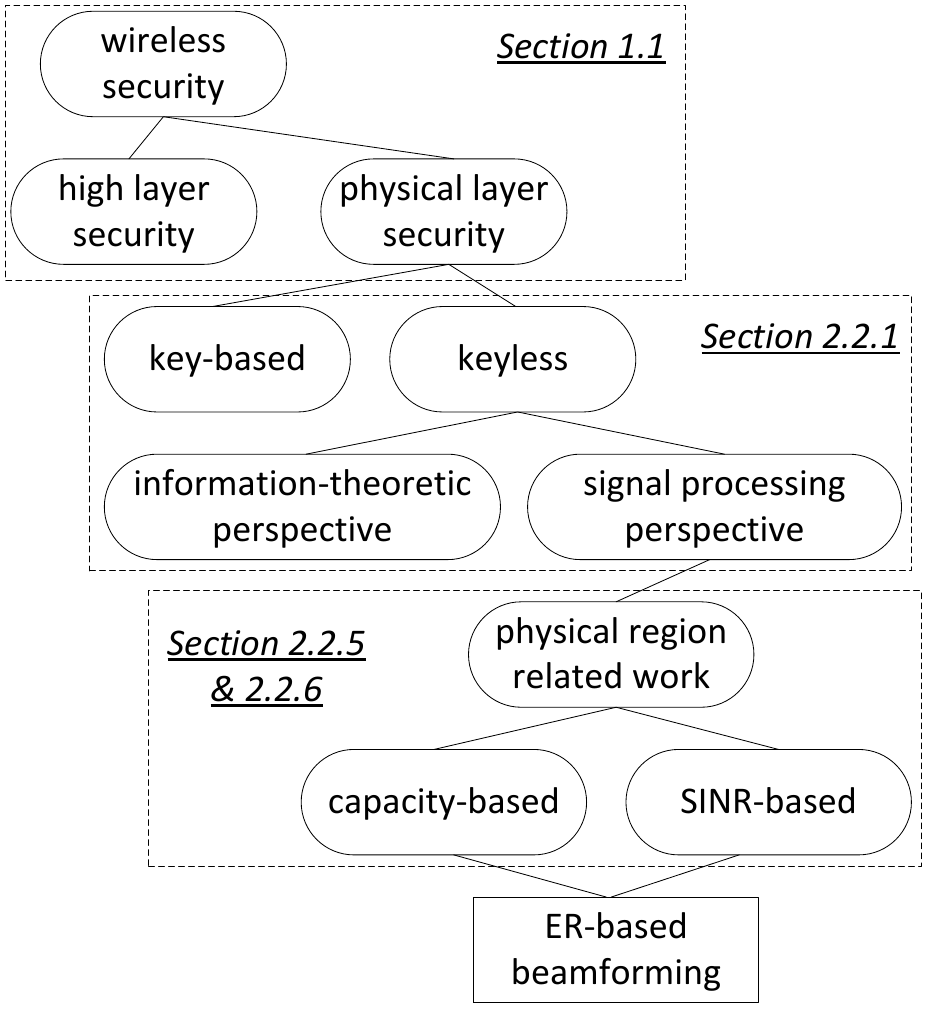}
\caption{Diagram of the reviewed work}
\label{fig:chp2_structure}
\end{figure}

\subsection{Shannon's Perfect Secrecy}
\label{chp2:PhySec:wopvm}

The fundamental principle of secure transmission was formalized by Shannon in\,\cite{shannon1949communication}.
That is, the intended receivers should recover the transmitted message without errors, while other users should acquire no information.
Fig.\,\ref{fig:chp2_shannonperfectsecrecy} illustrates Shannon's system for secrecy.
The figures and notations in this section and the next section are referenced from\,\cite{bloch2011physical}.

\begin{figure}[t]
\centering
\includegraphics[scale=1]{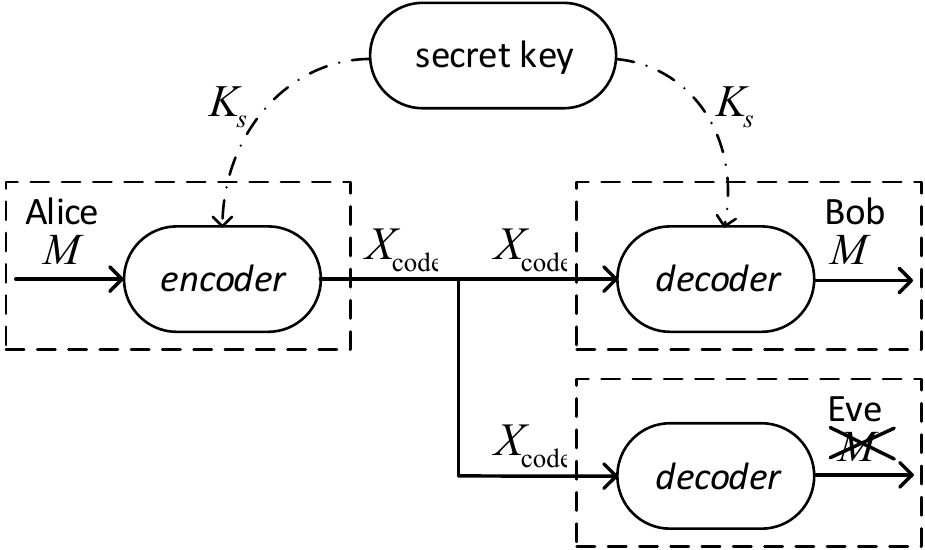}
\caption{System model for Shannon's perfect secrecy}
\label{fig:chp2_shannonperfectsecrecy}
\end{figure}

There are three parties in the system, i.e., Alice, Bob and Eve.
Alice attempts to transmit a message, denoted by $M$, to Bob in the presence of Eve.
In this system, both Bob's and Eve's channels are error-free and there is no restriction on Eve's computational power, which corresponds to the worst-case scenario for the encryption methods.

%
\nomenclature{$_i$,$_j$,$_n$,$_l$}{general index}
\nomenclature{$m,n$}{general integer}
\nomenclature{$c_0$, $c_1$}{deterministic constant}

\nomenclature{$M$}{message}

In order to achieve secrecy, Bob must gain some sort of advantage over Eve.
In this case, Alice encodes the message $M$ with a secret key $K_s$ into a codeword, denoted by $X_{code}$.
The key $K_s$ is shared by Alice and Bob, but is not known by Eve.
The encoder could be some complex computing functions or simply a XOR operator, i.e, $X_{code}=M\oplus K_s$, which is known as the one-time pad\,\cite{vernam1919secret}.
With $K_s$, Bob can recover $M$ without any error, while Eve cannot get any useful information other than some random guess.

\nomenclature{$K_s$}{secret key}
\nomenclature{$X_{code}$}{codeword}

From the information-theoretic perspective, the message $M$ and the codeword $X_{code}$ are random variables.
The entropy of a random variable indicates the amount of information that this variable has or the uncertainty level of this variable\,\cite{cover2012elements}.
The secrecy is measured by the conditional entropy of $M$ given $X_{code}$, which is also known as Eve's equivocation. 
Denoted by $\mathbb{H}(M|X_{code})$, it measures the uncertainty of $M$ at Eve based upon the correct reception of $X_{code}$.
Perfect secrecy can be achieved if Eve's equivocation equals to the a-priori uncertainty of $M$, i.e., $\mathbb{H}(M|X_{code})=\mathbb{H}(M)$.
In other words, $X_{code}$ and $M$ are statistically independent.
Since there is no correlation between $X_{code}$ and $M$, Eve cannot acquire any information about $M$ from $X_{code}$.

\nomenclature{$\mathbb{H}(M)$}{entropy of the random variable $M$}
\nomenclature{$\mathbb{H}(M|X_{code})$}{conditional entropy of $M$ given $X_{code}$}

To achieve the aforementioned condition for perfect secrecy, it is shown that the uncertainty of $K_s$ must be at least the same as $M$, i.e., $\mathbb{H}(K_s)\geq \mathbb{H}(M)$, which means the random secret key must have at least the same length as the message\,\cite{hellman1977extension}.
However, this raises questions in key distribution and management.

\subsection{Secrecy Capacity}
\label{chp2:PhySec:vks}

While Shannon's system relies on the secret key to create the advantage for Bob over Eve, Wyner's work in\,\cite{wyner1975wire} leverages the imperfections of the channel instead of using the secret key, provided that Bob's channel is better than Eve's.
Wyner's channel model is shown in Fig.\,\ref{fig:chp2_wiretap}.

\begin{figure}
\centering
\includegraphics[scale=1]{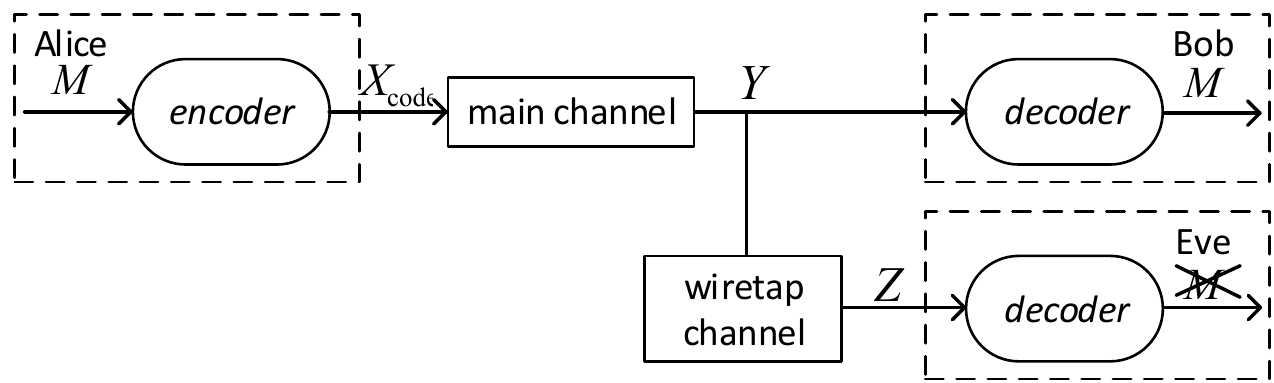}
\caption{System model for the wiretap channel model}
\label{fig:chp2_wiretap}
\end{figure}

The encoder generates codeword $X_{code}$ with block length $n$, which is the input of the main channel between Alice and Bob.
The output of the main channel, denoted by $Y$, is the input of the decoder at Bob and at the same time is the input of the wiretap channel.
The output of the wiretap channel, denoted by $Z$, is the observation of $X_{code}$ at Eve.
Both the main channel and the wiretap channel are noisy channels.
Eve's channel is a probabilistically degraded version of Bob's channel.

\nomenclature{$\mathbb{H}(M)$}{entropy of the random variable $M$}
\nomenclature{$\mathbb{H}(M|X_{code})$}{conditional entropy of $M$ given $X_{code}$}

Instead of achieving $\mathbb{H}(M|Z)=\mathbb{H}(M)$, Wyner relaxed this secrecy condition into that the equivocation rate $\frac{1}{n}\mathbb{H}(M|Z)$ is arbitrarily close to the entropy rate $\frac{1}{n}\mathbb{H}(M)$ for sufficiently large $n$, i.e., 
\begin{align}\label{eq:chp2_wyner}
	\frac{1}{n}\mathbb{I}(M;Z)=\frac{1}{n}\mathbb{H}(M)-\frac{1}{n}\mathbb{H}(M|Z)\leq \varepsilon,
\end{align}
where $\varepsilon$ is an arbitrary small value.
With this relaxed constraint, Wyner proved the existence of such codes that asymptotically guarantee the secrecy against Eve and at the same time a positive transmission rate for Bob's reliable transmission.
The maximum achievable transmission rate under these premises is the secrecy capacity.
It is worth noticing that the maximum achievable transmission rate of the main channel is regardless of the secrecy constraint.
The wiretap channel induces maximum equivocation at Eve.


Wyner's wiretap channel model was later generalized for the broadcast channel with two receivers\,\cite{csiszar1978broadcast} and the Gaussian channel\,\cite{leung1978gaussian} which lies the foundation for many wireless channels.
In\,\cite{csiszar1978broadcast}, a single-input two-output channel with private messages is considered. 
The goal is to design a encoder that a common message can be decoded by Bob and Eve while the private message is only decoded by Bob.
There exists a rate triple, \{private message rate, equivocation rate at Eve, common message rate\} for secrecy if the the private message rate is equal to the equivocation rate.
For the special case when there is common message transmitted, the secrecy capacity, denoted by $C_s$, can be defined by the maximum achievable private message rate.
Further, it can be expressed by
\begin{align}
	C_s= \max_{V\to X_{code}\to Y\,Z} I(V;Y)-I(V;Z),
\end{align}
where $V$ is an auxiliary input variable and $V\to X_{code}\to Y\,Z$ denotes the Markov relationship. 
For the degraded Gaussian wiretap channel, let the channel capacity of Bob and Eve be denoted by $C_B$ and $C_E$, respectively.
The secrecy capacity $C_s$ can be expressed by
\begin{align}
	C_s = (C_B -C_E)^+,
\end{align}
where $(x)^+$ takes the larger value between $x$ and $0$.
When $C_B>C_E$, there is a positive secrecy capacity.
When $C_B\leq C_E$, the secrecy capacity is zero.
For complex Gaussian channel via which complex-valued signals are transmitted, the real and imaginary parts of the additive noise are jointly Gaussian random variables.
The channel capacity $C$ can be calculated by\,\cite{tse2005fundamentals}
\begin{align}\label{eq:chp2_capacity}
	C=\log(1+\gamma),
\end{align}
where $\gamma$ is the signal-to-noise ratio (SNR). 
Notice that the channel capacity is only achieved when the channel input is a Gaussian random variable with zero mean.
The discussion of secrecy capacity derivation is out of the scope of this thesis. 
More details are available in\,\cite{bloch2011physical}.

\nomenclature{SNR}{signal-to-noise ratio}
\nomenclature{$C_s$}{secrecy capacity}
\nomenclature{$C$}{channel capacity}
\nomenclature{$C_B$}{channel capacity of Bob}
\nomenclature{$C_E$}{channel capacity of Eve}
\nomenclature{$(x)^+$}{the larger value between $x$ and $0$}
\nomenclature{$\gamma$}{SNR}

\subsection{Secrecy Outage Probability}
\label{chp2:PhySec:vsp}

For a non-fading channel, the secrecy capacity solely relies on the Bob's and Eve's received SNR, whereas in fading channels, it also depends on the random channel coefficient that is subject to a certain distribution, e.g., a Rayleigh distribution.
Therefore, the secrecy capacity becomes a random variable that is subject to certain fading distribution.
For a quasi-static fading channels, the channel capacity can be calculated using (\ref{eq:chp2_capacity}) with random $\gamma$ that is subject to certain fading distribution\,\cite{tse2005fundamentals}.

Analogy to the conventional outage metric, the outage formulation can be applied to the random secrecy capacity in fading channels.
In\,\cite{barros2006secrecy}, the SOP is defined by the probability that the instantaneous secrecy capacity is less than certain target secrecy rate $R_s>0$.
Denoted by $p_{out}(R_s)$, the SOP can be calculated by
\begin{align}\label{eq:chp2_SOP}
	p_{out}(R_s)=\text{Prob}\{C_s<R_s\}.
\end{align}

\nomenclature{$R_s$}{target secrecy rate}
\nomenclature{$p_{out}$}{SOP}

$p_{out}(R_s)$ indicates the percentage of fading realizations where the wiretap channel model can sustain target secrecy rate $R_s$.
Moreover, it is a useful performance metric when Eve's CSI is not known to Alice.
Notice that $R_s$ is an arbitrarily chosen value for certain system.
The meaning of $R_s$ is that Alice assumes that Eve's channel capacity is $C_E'=C_B-R_s$.
If the actual channel capacity of Eve is less than Alice's assumed capacity, i.e., $C_E<C_E'$, then it can be derived that $C_s>R_s$. In this case, the wiretap codes with transmission rate higher than $C_E'$ can guarantee the perfect secrecy.
Otherwise, if Eve's channel is better than Alice's assumption, i.e., $C_E>C_E'$, then $C_s<R_s$.
In this case, the wiretap codes with transmission rate higher than $C_E'$ is at the risk of leaking information to Eve and the information-theoretic security is compromised. 

\nomenclature{$C_E'$}{the assumed Eve's channel capacity by Alice}

The SOP is particular useful when Eve's instantaneous CSI is not known by Alice, which is usually the case because Eves could be passive and do not easily give away their CSI.
On the other hand, Bob's CSI can be assumed to be available by Alice.
In this case, the SOP can be calculated if the distribution of Eve's fading channel is known, which can be used as performance measure for secure communications.

The work in\,\cite{zhou2011rethinking} puts forth an alternative formulation other than (\ref{eq:chp2_SOP}), which distinguishes the difference between insecure transmission and unreliable transmission.  
For example, when $C_B<R_s$ (which implies that $C_s<R_s$), Alice knows that Bob's channel cannot support the secrecy rate, then suspends the transmission, which is not a failure in achieving perfect secrecy.

To explicitly measure the probability that a transmission fails to achieve perfect secrecy, two rates are employed, i.e., the rate of the transmitted codewords $R_B$ and the rate of the confidential information $R_s$.
The rate difference $R_B-R_s$ is the cost for secure communication against passive Eves.
When $C_B>R_B$, Bob can successfully recover the transmitted message.
However, if $C_E>R_B-R_s$, the secure transmission fails.
Thus, the SOP can be defined as the conditional probability\,\cite{zhou2011rethinking},
\begin{align}\label{eq:chp2_sop1}
	p_{so}=\text{Prob}\{C_E>R_B-R_s|\text{message transmission}\}.
\end{align}
The condition of message transmission is designable according to different targets.
It can be set to $C_B>R_s$ (e.g., in\,\cite{yan2014secrecy}) or $C_B>R_B$ to guarantee Bob's correct reception, or even to maximize the throughput of the secure transmission\,\cite{zhou2011rethinking}.
With Bob's instantaneous CSI, Alice will transmit to Bob when the message transmission is guaranteed; otherwise, Alice will stop the transmission.

\nomenclature{$R_B$}{the rate of the transmitted codewords}
\nomenclature{$p_{so}$}{alternative formulation of SOP}

When $R_B$ and $R_s$ are fixed system parameters, $p_{so}$ is independent of the condition of message transmission\,\cite{zhou2011rethinking}.
Thus, $p_{so}$ reduces to 
\begin{align}\label{eq:chp2_sop2}
	p_{so}=\text{Prob}\{C_E>R_B-R_s\}.
\end{align}
In this case, $C_B$ and $C_E$ can be separately studied according to the fading distribution.

\subsection{Physical Region in Information-Theoretic Security}
\label{chp2:PhySec:mnyo}

As mentioned in Section\,\ref{chp1:motive}, the location of a user plays a key role in the user's channel capacity, thus directly affects the difference between Bob's and Eve's channel capacities.
Therefore, the large-scale path loss, which mainly relies on the users' locations, should be incorporated in practical scenarios, no matter it is Gaussian channel or fading channel, or it is a single-antenna system or multiple-antenna system.
This section reviews the work that considers users' locations, which is normally studied in relation to some sort of physical regions.
An example of the physical region is illustrated in Fig.\,\ref{fig:chp2_demon} where the AP that is equipped with an antenna array performs beamforming to secure a region surrounding Bob and limits Eve's access to the legitimate transmission.
The common property shared by these papers is that their physical regions are based on the information-theoretic security parameters.
In comparison to the physical space security that will be introduced in the next section, it is called capacity-based approach for convenience.

\begin{figure}
\centering
\includegraphics[scale=1]{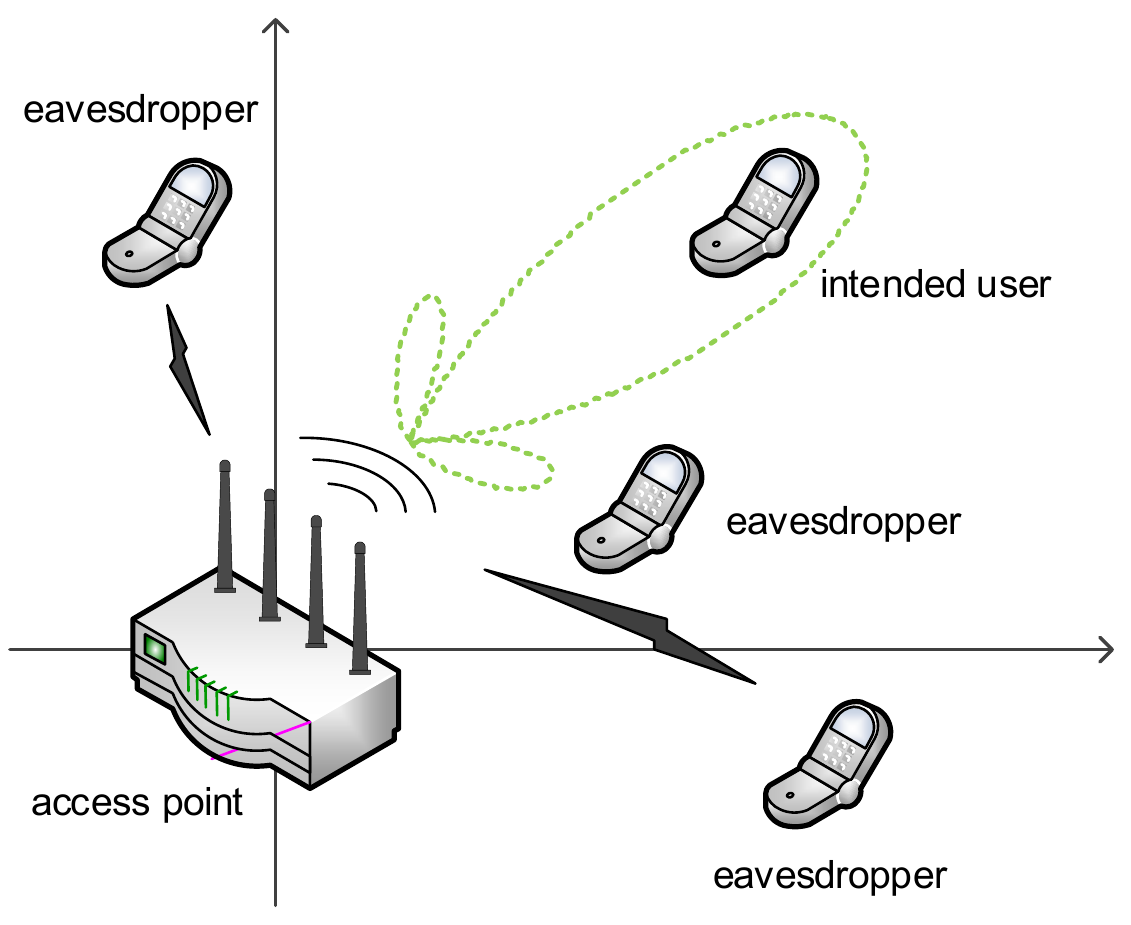}
\caption{An example of secure physical region in a Wi-Fi network with one intended user and multiple eavesdroppers.}
\label{fig:chp2_demon}
\end{figure}

There are versatile approaches from the location or physical region perspective.
This section provides a high-level overview in terms of system model, performance metric and signal processing technique. 
Most of the related work\,\cite{vilela2011wireless,li2012secure,li2013security,li2014secure,wang2015jamming} exploits the advantage of beamforming and AN/jamming, while other work\,\cite{zheng2014transmission,yan2014secrecy,yan2014line} purely investigates  the security performance of beamforming techniques.
Not surprisingly, there is also work that explores the possibility of achieving security in relay systems\,\cite{marina2010characterization,sarma2013joint}.

A physical region is usually defined for the reason that Eve's CSI or location is random and unknown to Alice.
Different kinds of regions are defined for different beamforming and AN techniques, provided with certain level of Bob's CSI or location.
The most common definition is based on the SOP. 
In\,\cite{li2013security,wang2015jamming,marina2010characterization,sarma2013joint}, the system performance metric is some sort of insecure region where the secrecy goal is compromised.
For example, the compromised secrecy region (CSR) in\,\cite{li2013security} is the region where the SOP is above a certain threshold.
In\,\cite{wang2015jamming}, the secrecy outage region (SOR) is used to define the region where Eve causes the secrecy capacity below a target rate, i.e., secrecy outage happens.
The special case for such a definition is when the target secrecy rate is set to zero, which is used to defined the vulnerability region (VR) in relay systems\,\cite{marina2010characterization, sarma2013joint}.

When fading is not considered, e.g., spatial diversity and time-diversity are used to counteract channel fading, the secrecy outage is caused solely by Eve's random location.
Notice in such case, the meaning of SOP is extended from the small-scale fading to the large-scale path loss, i.e., Eve's random location. 
Such an example is the insecure region in\,\cite{li2014secure}.

Opposite to the insecure region, the secure region is also used.
In\,\cite{li2012secure},  the outage secrecy region (OSR) is defined by the region where Eve causes the SOP that is below a threshold.
A similar definition is the secure region in\,\cite{zheng2014transmission} where for any Bob within the SOP is below an arbitrarily small value.
In addition, the jamming coverage is defined by the region where the SOP is reduced by the utilization of jamming in a quasi-static fading channel\,\cite{vilela2011wireless}.

\nomenclature{CSR}{compromised secrecy region}
\nomenclature{SOR}{secrecy outage region}
\nomenclature{VR}{vulnerability region}
\nomenclature{OSR}{outage secrecy region}

The goal of these papers is to either minimize the insecure region or maximize the secure region or jamming coverage.
In\,\cite{vilela2011wireless}, a legitimate transmission from Alice to Bob is aided by friendly jamming.
The secrecy performance of various jamming strategies given different levels of CSI is evaluated and the optimal jamming configuration is studied regardless of Eve's location.
In particular, it is shown that a single jammer is not sufficient to maximize the jamming coverage and efficiency simultaneously.
As an extension for\,\cite{vilela2011wireless}, the cooperative jamming system is developed to optimize the location and the power allocation for the jammer in\,\cite{li2013security}.

Aided with multiple antennas, the authors in\,\cite{li2012secure} develop a novel technique to generate AN at Bob, when Bob has stronger ability than Alice, e.g., more antennas.
This method is robust in the sense that no feedback of Bob's CSI is needed by Alice and there is no restriction of Eve's number of antennas.
However, it is shown that the area in the vicinity of Bob is well protected whereas the area surrounding Alice is still vulnerable.
\,\cite{li2014secure} extended the work in\,\cite{li2012secure} by generating AN from both Alice and Bob to impair Eve's channel with an optimum power allocation strategy to minimize the size of the insecure region.
In\,\cite{wang2015jamming}, different jamming strategies as well as the optimal power allocation between the information-bearing signal and the AN are investigated in a massive MIMO system with and without the information of possible locations of Eves.

Besides AN/jamming, multiple antennas are used for beamforming in\,\cite{zheng2014transmission}.
Two scenarios, i.e., non-colluding and colluding Eves, are investigated with the integral expression of SOP and the closed-form upper bound.
Based on these, the secure region is derived to guide Bob's location in presence of randomly located Eves.
In addition, the parameters that impact the secure region is analyzed.
As for the distributed antenna array, a simple cooperative system with a single relay is proposed in\,\cite{marina2010characterization}, where the VR is studied for different synchronization and interference models.
\,\cite{sarma2013joint} continues the work in a multi-hop relay system.
However, both papers use the Gaussian channel as a starting point.

There is other work that is closely related to, but not directly based on the physical region\,\cite{yan2014secrecy,yan2014line}.
Although there is no concise geometric model, the performance of these methods is evaluated based on geometric locations.
In\,\cite{yan2014secrecy}, a scenario where Eves' locations follow the PPP distribution is considered.
Alice is aware of Bob's location and only the distribution of Eves' locations, but not Eves' CSI.
The closed-form expression of SOP is derived for Rician fading channel where a line-of-sight (LOS) component exists.
It is shown that beamforming towards Bob's location is the optimal strategy that minimizes the SOP.
In\,\cite{yan2014line}, a threat model that describes possible locations for Eve, e.g., an annulus threat model with a uniform distribution of Eve, is used to quantify the SOP for beamforming towards known Bob's location with multiple antennas.
However, the work is limited in free-space scenario.

\nomenclature{LOS}{line-of-sight}

In most reviewed work, there is no closed-form formulation for these physical regions, and only numerical approximations or results are used.
Except that in\,\cite{yan2014line,zheng2014transmission,wang2015jamming}, the SOR is analytically derived and a new outage probability is defined based on the SOR.
In\,\cite{yan2014line}, the analytic expressions are given for free-space scenario without considering fading channel.
The Rayleigh fading that generates simple expressions is considered in\,\cite{zheng2014transmission}.
However, it is not very practical to obtain Bob's location or CSI without the LOS component.
In\,\cite{wang2015jamming}, the Rician fading channel is used, but the fading effect is completely averaged out for very large number of antennas in massive MIMO system and is treated as constant.

While a single Eve is at present in the network in most reviewed work, multiple Eves are considered in\,\cite{zheng2014transmission,yan2014secrecy,wang2015jamming}.
In particular, the PPP is exploited to study the distribution of unknown Eves' locations in\,\cite{zheng2014transmission,yan2014secrecy}.
It is worth noticing that almost all the reviewed work does not take the antenna array's configuration into consideration to optimize the physical region.
The only work that considers some aspect of the array configuration does not explicitly have analytic expressions for the array configuration\,\cite{yan2014secrecy}.

\subsection{Physical Space Security}
\label{chp2:PhySec:bmdl}

While the reviewed work in Section\,\ref{chp2:PhySec:mnyo} is based on the information-theoretic parameters, there is another branch of work from the signal processing perspective that are based on the traditional performance metrics, e.g., the bit error rate (BER) or signal-to-interference-plus-noise ratio (SINR).
In comparison to the capacity-based approaches, the work that are reviewed in this section is also referred to as the SINR-based approach for convenience.

\nomenclature{BER}{bit error rate}
\nomenclature{SINR}{signal-to-interference-plus-noise ratio}

The principle of the SINR-based approaches is to limit the knowledge of the existence of the message to Eve\cite{4595864,5357443}.  
To this end, various techniques are developed to confine the effective communications into certain physical region, e.g., by designing transmission schemes that restrict the BER or SINR at Eve below certain thresholds.
The SINR-based approach and the capacity-based approach share the common ground in the sense that the difference of Bob's and Eve's channel should be enlarged to improve the security performance of the system.

In fact, the boundary between the capacity-based approaches and the SINR-based approaches is not so strict.
From the theoretical perspective, the channel capacity is determined by the SINR for most channels.
For example, the SINR of Bob and Eve is used to define the VR where Eve's channel capacity is larger than Bob's channel capacity, i.e.,  zero secrecy capacity\,\cite{sarma2013joint,sarma2015optimal}.
As well pointed out in\,\cite{hong2013enhancing}, while the SINR-based approaches do not guarantee perfect secrecy in the information-theoretic sense, they achieve a practical notion of secrecy in the way that discriminates the performance among Bob and Eve, and are useful in some applications.
In addition, they can often simplify the system design\,\cite{mukherjee2010principles}.
For example, a SINR-based power allocation and scheduling technique provides a simple solution for a multi-hop wireless network, because finding the secrecy capacity for some complex systems is a hard problem\,\cite{sarma2015optimal}.

Despite the difference between the SINR-based and the capacity-based approaches in terms of the performance metric, there is another important difference, that is, the SINR-based approaches have a strong background from the smart antennas, i.e., beamforming and direction-of-arrival (DoA) estimation\,\cite{gross2005smart}.
The benefits that are brought by beamforming and DoA estimation, i.e., focusing or suppressing energy at certain directions and direction-finding, had been applied to improve security in the physical layer in the early 2000s\,\cite{sun2003improving}, when the information-theoretic security still waited for its reemergence.
In fact, one of the early attempts even employed directional antenna on both the transmitter and the receiver to reduce the signal coverage region\,\cite{1606699}.

\nomenclature{DoA}{direction-of-arrival}

A minor distinction of the SINR-based approaches from the capacity-based approaches is that most work takes application for the wireless local area network (WLAN), such as 802.11.
Thus, the AP acts as Alice and the downlink transmission from the AP to Bob is to be protected in the presence of Eves.
Nevertheless, the developed techniques are also applicable to other wireless networks.

\nomenclature{AP}{access point}
\nomenclature{WLAN}{wireless local area network}

Intuitively, the ability of beamforming (or directional antennas) can be constructively exploited to enhance the signal strength at Bob's direction, while suppressing the signal strength at other directions, especially at Eve's direction if Eve's location is known to the AP.
Therefore, the physical region can be created by the AP(s) that is(are) equipped with directional antenna or antenna array, and Eve's access to the signal is limited if Eve is not inside such a  physical region.
In the absence of Eve's location or CSI, the created physical region should be minimized so that the possibility of Eve being within this region is minimized and the system security level is enhanced.

A common set-up is to use multiple APs, each of which is equipped with multiple antennas, to jointly created a small region\,\cite{1400008,4595864,5357443,sheth2009geo,sattari2009secure,6618765}.
While one AP can only limit the physical region to a certain extent, multiple APs can further reduce this region by creating a smaller joint region.
The idea is conceived in\,\cite{1400008}.
To achieve this goal, a single packet is divided into fragments, each of which is separately transmitted by one AP in the network in a time-division manner.
To form the region, the AP needs to adjust its transmit power according to the user's location.
Only the users in the joint region can access the whole packet.

Although the multiple-AP technique brings some challenges to the practical design, such as synchronization of multiple APs and other protocol modifications\,\cite{1400008}, this idea is further developed by the authors in\,\cite{4595864,5357443}.
To achieve a higher level of security, secret sharing is used\,\cite{shamir1979share}. 
All fragments of the packet are encrypted in a way that the whole packet can be decrypted only if all fragments are correctly received.
In the same work, the term `physical space security' is coined.
In\,\cite{4595864,5357443}, the authors for the first time defined the ER as performance metric, which refers to the area within which Eve(s) can access and decode the signals being transmitted.
Note that the ER here is not defined based on the information-theoretic parameters. 
Without Eve's location or CSI, the ER is to be minimized in order to improve the system security level.

In the multiple-AP system, each AP can be assigned to different tasks, i.e., beamforming or jamming, depending on the transmission strategy.
Besides the secret sharing strategy, two other strategies are proposed to reduce the ER\,\cite{4595864,5357443}, which uses jamming signal or signals from multiple sources to cause more interference to reduce Eve's quality of reception.
By controlling the direction of a jamming signal or multiple-source signals, Bob's reception is not affected.
A similar idea of jamming can be found in\,\cite{kim2012carving}, where jammers use an omni-directional antenna to forge a walled wireless coverage, which is a secure Wi-Fi zone.
Through adjusting locations and the transmit power of jammers,
the forged secure zone matches well with the prediction model against the leakage to other zones.

The idea behind the multiple-AP systems is to confine the signal transmission in a controlled region.
The motivation is that the WLAN usually operates inside a physical perimeter, e.g, an office floor, and the security threat can be reduced by imposing physical boundaries to the boundless radio transmission through manipulation of the properties of signal propagation.
Such idea is emphasized in\,\cite{sheth2009geo,sattari2009secure,tiwari2008wireless}.
In\,\cite{sheth2009geo}, multiple APs jointly perform beamforming with the transmit power control to isolate a physical region.
The approach is similar to the secret-sharing strategy used in\,\cite{4595864,5357443}, but with an improvement on the joint optimization of beam patterns for all APs.
The experiment results in several indoor scenarios show that 
different shapes and sizes from 5 feet $\times$ 5 feet to 25 feet $\times$ 20 feet can be isolated by three such APs.

The work in\,\cite{sattari2009secure,tiwari2008wireless} takes different routes to achieve the confinement of the radio propagation.
In\,\cite{sattari2009secure}, all the users inside a certain perimeter communicate with the base station through some intermediate nodes.
Multiple nodes are deployed alongside the physical perimeter to detect the users within and manage the access control to the base station.
Once the users are detected and recognized as a legitimate user, they are granted access to the base station.
A similar idea is found in\,\cite{tiwari2008wireless} where the Radio Frequency Sentry Devices (RFSD) are deployed on the perimeter of a confined region.
The RFSD performs a `cloaking' function, which consists of two stages, i.e., detection of the signal transmission from the users inside the perimeter via DoA estimation and transmission of an altered signal with approximately the same transmit power.
Eves outside the perimeter receive an superposition of the original signal and the altered signal from the RFSD, thus cannot decode the message correctly.
Both methods in\,\cite{sattari2009secure,tiwari2008wireless} serve the purpose of confining a local transmission inside a predefined region.

\nomenclature{RFSD}{radio frequency sentry device}

So far, multiple-AP systems are mainly used to create the  physical region.
In the following, the work that focuses on improving the performance of a single antenna array with more advanced techniques is presented.
In \cite{anand2012strobe} the authors proposed a cross-layer design called the `simultaneous transmission with orthogonally blinded  eavesdroppers' (STROBE) to reduce Eves' signal quality.
The multiple antennas, such as in 802.11n and 802.11ac standards, are designed to simultaneously transmit multiple data streams using zero-forcing beamforming.
STROBE exploits the capability of this multiple-antenna technique to insert orthogonal interference that are transmitted simultaneously with the intended data stream, so that potential Eves cannot decode correctly while Bob is remain unaffected by the interference. 
Multipath creates advantages for Bob in the STROBE system and the indoor experimental results show that a difference of 15\,dB of the SINR between Bob and Eves can be consistently served. 
The work in\,\cite{6502515} designs a type of smart antenna that has two synthesized radiation patterns that can alternatively transmit in a time-division manner.
The transmitted packet is divided into two parts, each of which is transmitted via one synthesized pattern.
Both patterns are slightly away from Bob's direction, but have an overlap at Bob's direction.
By fast switching between the two patterns, an artificial fading effect is created for users that are not within the overlapped region, thus reduces the signal quality of unintended users, while Bob is little affected.
The overlap region can then be minimized to enhance the security.

\nomenclature{STROBE}{simultaneous transmission with orthogonally blinded  eavesdroppers}

There are some versatile approaches that are combined with other techniques, e.g., joint design with encryption methods.  
In\,\cite{matoba2012novel}, the distributed nodes that are equipped a single antenna in the same network cooperatively create the ER in a similar way to\,\cite{kim2012carving}, except that instead of jamming, the neighbor nodes transmit side information which can be used to encrypt the transmission between Alice and Bob.
Decryption is only possible with both the encrypted message and the side information, i.e., the receiver needs to be in the overlapped region from Alice and the helping nodes.
The overlapped region where decryption can be done is called the ER, which is to be reduced by dynamic selection of the helping nodes according to their locations.
In\,\cite{6618765}, a hybrid cross-layer protocol that combines the network security protocol with the exploitation of the secret-sharing scheme is designed as an extension to\,\cite{4595864,5357443}.
The combination of the public key encryption and the ER reduction restricts the access to the legitimate transmission even when Eve is located inside the ER.

The SINR-based approaches reviewed in this section do not have closed-form expressions for the created physical regions, e.g., the ER.
Since they are not based on the information-theoretic parameters, the information-theoretic analysis is absent.
However, these arguments do not dismiss the usefulness of the SINR-based approaches.
On the contrary, many work resort to experimental results to prove the effectiveness of these approaches\,\cite{1606699,4595864,5357443,sheth2009geo,anand2012strobe,kim2012carving}.
After all, the beamforming and jamming technique that are exploited here are not essentially different from the capacity-based approaches.
Thus, the basic principle can be carried over to other applications.
It is worth noticing that while the linear array is used in most approaches, the circular array is chosen in\,\cite{1400008,sheth2009geo,6502515} where the fine-grained region is shaped.

\section{Antenna Array Fundamentals}
\label{chp2:antennas}

Antenna arrays are used in many areas, such as land-mobile, indoor-radio, and satellite-based systems\,\cite{godara1997applications1}, and are used for a wide range of purposes, e.g., achieving security as mentioned in Section\,\ref{chp2:PhySec}.
From the smart antennas perspective, beamforming is the signal processing algorithm that performs on the antenna array and makes the array `smart'\,\cite{gross2005smart}.
Besides beamforming, the other main function of the smart antennas is DoA estimation\,\cite{godara1997application2}.
This section introduces some fundamental concepts. 
More details are provided in\,\cite{godara1997applications1,gross2005smart,godara1997application2,adaptivearraysystems} and the references therein.

\subsection{Uniform Linear Array}
\label{chp2:antennas:ULA}

The antenna array consists of multiple antennas that are deployed in a certain geometry.
The array geometry refers to the positions of the antenna elements that make up the array.
The most common array geometry is the ULA where all elements are in a line with equal spacing.
The number of elements and the spacing of the ULA are denoted by $N$ and $\Delta d$, respectively.
An example of ULA is shown in Fig.\,\ref{fig:chp2_ULA}.
To study the behavior of the array, the antenna element is usually assumed as an omni-directional antenna with spherical radiation pattern.

\nomenclature{$N$}{(active) number of antenna elements in the array}
\nomenclature{$\Delta d$}{element spacing in the array}

\begin{figure}
\centering
\includegraphics[scale=1]{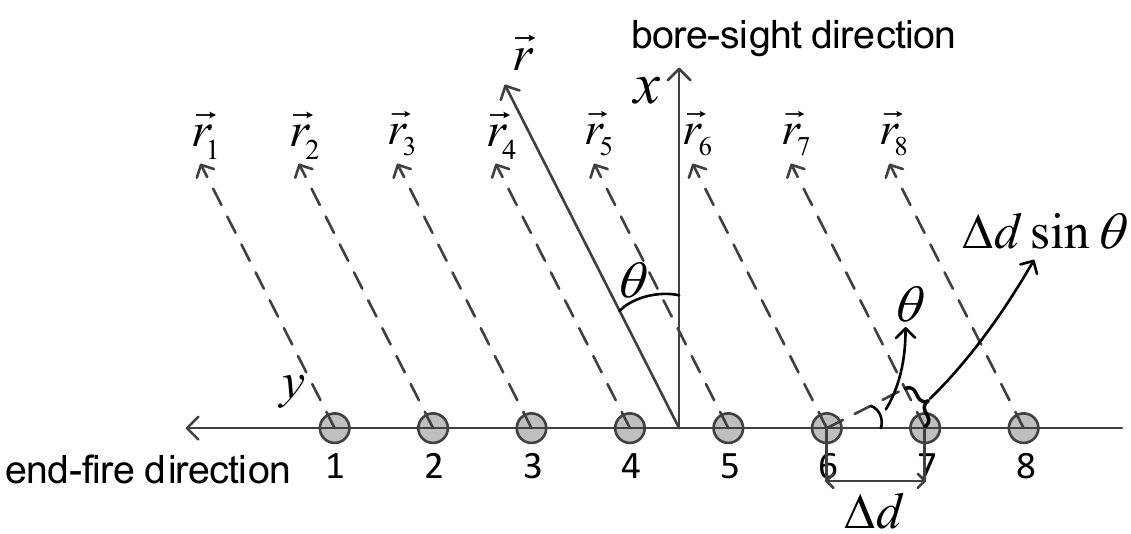}
\caption{Linear array of 8 elements}
\label{fig:chp2_ULA}
\end{figure}

As shown in Fig.\,\ref{fig:chp2_ULA}, the bore-sight direction of the ULA is the norm direction; the end-fire direction is parallel to the array.
The direction-of-emission (DoE), denoted by $\theta_{\text{doe}}$, is the angle at which the ULA concentrates energy, and is usually defined in relation to the bore-sight direction. 
For ease of mathematical derivation, the ULA shown in Fig.\,\ref{fig:chp2_ULA} is put along $y$-axis with its center at the origin point and the first element at the positive $y$-axis.
In this way, the bore-sight of the array is on the $x$-axis and the angles related to the array is the same angle in the polar coordinates.

\nomenclature{DoE}{direction-of-emission}
\nomenclature{$\theta_{\text{doe}}$}{DoE angle}

The signals transmitted from different antenna elements interfere with each other in space.
The overall signal at a certain point is the superposition of all signals with different amplitudes and phases.
The array factor, denoted by $G$, indicates the complex gain of the signal at a certain angle $\theta$.
As shown Fig.\,\ref{fig:chp2_ULA}, assume that the vector from the center of the array to the target user is $\overrightarrow{r}$ and the vector from the $i$-th element to the target user is $\overrightarrow{r}_i$, $i=1,...,N$.

\nomenclature{$G(\cdot)$}{array factor}
\nomenclature{$\theta$}{user's angle}
\nomenclature{$\overrightarrow{r}$}{vector from the center of the array to the target user}
\nomenclature{$\overrightarrow{r}_i$}{vector from the $i$-th element in the array to the target user}
\nomenclature{$|\cdot|$}{$L^2$-norm of a vector or magnitude of a scalar}

Strictly speaking, $\overrightarrow{r}$ and $\overrightarrow{r}_i$, $i=1,...,N$ should point at the same position.
When $|\overrightarrow{r}|\gg \Delta d$, i.e., the distance of the target user is far greater than the size of the array, the far-field condition is fulfilled.
In this case,  $|\overrightarrow{r}|\approx|\overrightarrow{r}_i|$, and $\overrightarrow{r}$ and $\overrightarrow{r}_i$ are assumed to be parallel.

To calculate $G$, first consider a simple case where the array does not focus energy at any particular angle and all elements transmit with the same amplitude and phase.
The signal from the right element always arrives sooner at the target user than the signal from the left element for $\theta\in[0,\frac{\pi}{2}]$, which leads to certain advance in phase.
The phase difference can be calculated by $2\pi\frac{\Delta d\sin\theta}{\lambda}$ for angle $\theta$.
Take the 1st element as the reference point for phase zero.
Then the relative phase shift between the $i$-th element and the 1st element, denoted by $\phi_i(\theta)$, is 
\begin{align}
	\phi_i(\theta)=k\Delta d(i-1)\sin\theta,
\end{align}
where $k=\frac{2\pi}{\lambda}$ is the wave number.
The array steering vector, denoted by $\mathbf{s}(\theta)$, is defined based on the relative phase shifts of the signals from all elements at angle $\theta$,
\begin{align}\label{eq:chp2_steeringvector}
	\mathbf{s}(\theta)=[e^{-j\phi_1(\theta)},...,e^{-j\phi_i(\theta)},...,e^{-j\phi_N(\theta)}]^T.
\end{align}
In this case, the superposition of all signals is then 
\begin{align}
	G(\theta)=\sum_{i=1}^N e^{-j\phi_i(\theta)}.
\end{align}

\nomenclature{$\lambda$}{wavelength}
\nomenclature{$\phi_i(\cdot)$}{relative phase shift on the $i$-th element}
\nomenclature{$k$}{wave number}
\nomenclature{$\mathbf{s}(\cdot)$}{array steering vector}
\nomenclature{$j$}{imaginary unit}

The beamforming weight vector, denoted by $\mathbf{w}$, is used to precode the transmitted signal.
The vector $\mathbf{w}$ is a complex vector, thus the signal from each element is weighted by a complex number.
To concentrate energy at $\theta_{\text{doe}}$, all signals should arrive at angle $\theta_{\text{doe}}$ at the same time, which requires phase alignment.
To correct for the different phase shift $\phi_i(\theta_{\text{doe}})$, $\mathbf{w}$ is set by 
\begin{align}\label{eq:ch2_BF_Weights}
	\mathbf{w}=\frac{\mathbf{s}(\theta_{\text{doe}})}{\sqrt{N}},
\end{align}
where $\mathbf{s}(\theta_{\text{doe}})$ is the array steering vector at $\theta_{\text{doe}}$ and $\sqrt{N}$ is the normalization factor that keeps unit transmit power.
In this case, $G$ can be calculated by
\begin{align}\label{eq:chp2_TXgain}
	G(\theta,\theta_{\text{doe}})=\mathbf{w}^H\mathbf{s}(\theta)=\frac{1}{\sqrt{N}} \sum_{i=1}^N e^{j[\phi_i(\theta_{\text{doe}})-\phi_i(\theta)]}.
\end{align}
Because $G$ is determined by two angles, i.e, $\theta$ and $\theta_{\text{doe}}$, the notation of $G(\theta,\theta_{\text{doe}})$ is used in this thesis.
Physically, it means the complex gain of the signal at angle $\theta$ when the DoE angle is $\theta_{\text{doe}}$.
Notice that (\ref{eq:chp2_steeringvector})-(\ref{eq:chp2_TXgain}) are the general expressions which are valid for any array geometry.
For ULA with $N$ elements and $\Delta d$ spacing, $G(\theta,\theta_{\text{doe}})$ is obtained by
\begin{align}
G(\theta,\theta_{\text{doe}}) &=\frac{1}{\sqrt{N}}\sum_{i=1}^N e^{jk\Delta d(\sin\theta_{\text{doe}}-\sin\theta)(i-1)} \label{eq:chp2_AF_ULA} \\
 &= \frac{1}{\sqrt{N}}\frac{1-e^{jNk\Delta d(\sin\theta_{\text{doe}}-\sin\theta)}}{1-e^{jk\Delta d(\sin\theta_{\text{doe}}-\sin\theta)}}. \label{eq:chp2_AF_ULA2}
\end{align}
Some examples of the array patterns for $G(\theta,\theta_{\text{doe}})$ is shown in Section\,\ref{chp3:analysis:jfeow}.

\nomenclature{$\mathbf{w}$}{beamforming weight vector}
\nomenclature{$^H$}{Hermitian transpose}

%

\subsection{Uniform Circular Array}
\label{chp2:antennas:UCA}

Although the ULA is very common in practice, there are occasions where a ULA is not appropriate.
Other array geometries, e.g., the UCA, can be used.
Examples of UCA have been shown in the literature for creating physical regions for wireless security\,\cite{1400008,sheth2009geo,6502515}.

For the UCA, all the $N$ elements are equally allocated with spacing $\Delta d$ on a circle with radius $R$.
$\theta_{\text{doe}}$ is the angle between the target user and the first element in the array.
In this thesis, only 2D-plane is considered.
An example of a UCA with $8$ elements is shown in Fig.\,\ref{fig:chp2_UCA}.
For the ease of mathematical derivation, the center of the UCA is at the origin point and the first element is put on the positive $x$-axis.
In this way, $\theta_{\text{doe}}$ is the angle between the target user and the positive $x$-axis.
In addition, the phase angle $\psi_i$ for the $i$-th element is 
\begin{align}
	\psi_i=2\pi(i-1)/N.
\end{align}

\nomenclature{$R$}{radius of the UCA}
\nomenclature{$\psi_i$}{phase angle of the $i$-th element in the UCA}

\begin{figure}
\centering
\includegraphics[scale=1]{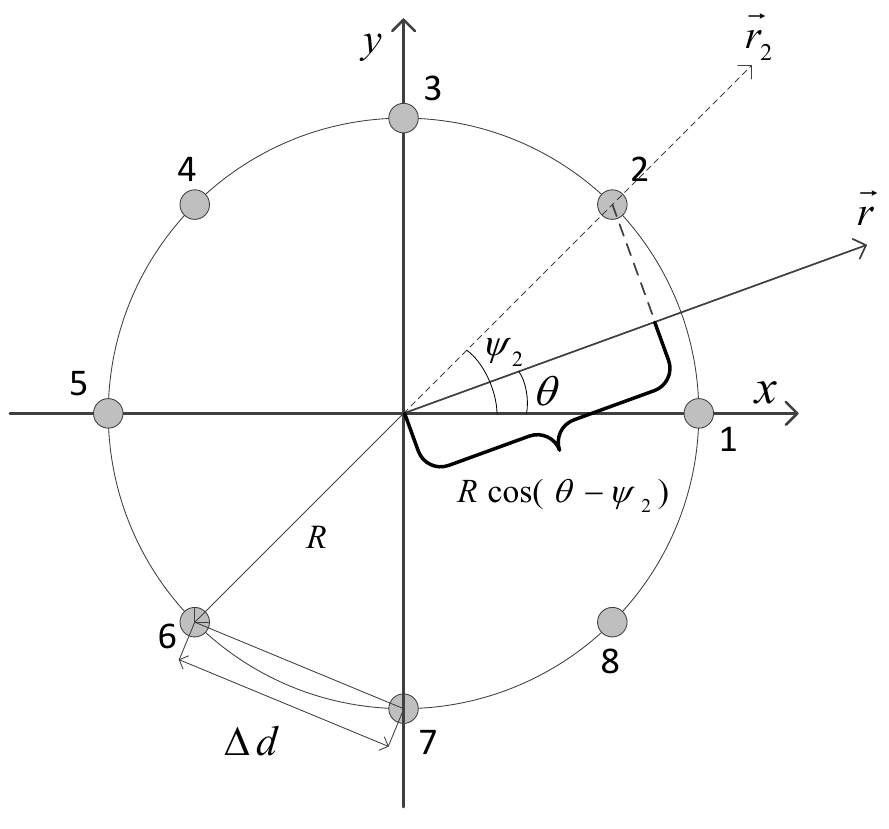}
\caption{Circular array of 8 elements}
\label{fig:chp2_UCA}
\end{figure}

$G(\theta,\theta_{\text{doe}})$ is determined by the array geometry.
The overall signal at angle $\theta$ is the superposition of the signals transmitted from all elements in the UCA.
For the UCA, $|\overrightarrow{r}|\gg R$ is assumed for the far-field condition.
The same as ULA, assume $|\overrightarrow{r}|\approx|\overrightarrow{r}_i|$, and $\overrightarrow{r}$ and $\overrightarrow{r}_i$ are assumed to be parallel.

To calculate $G(\theta,\theta_{\text{doe}})$, $\overrightarrow{e}$ and $\overrightarrow{e}_i$ are used to denote the unit vectors in the direction of $\overrightarrow{r}$ and $\overrightarrow{r}_i$, respectively.
\begin{align}
	&\overrightarrow{e}=\cos(\theta)\overrightarrow{e}_x+\sin(\theta)\overrightarrow{e}_y, \\
	&\overrightarrow{e}_i=\cos(\psi_i)\overrightarrow{e}_x+\sin(\psi_i)\overrightarrow{e}_y,
\end{align}
where $\overrightarrow{e}_x$ and $\overrightarrow{e}_y$ are the unit vector on the $x$-axis and $y$-axis, respectively.
The distance $|\overrightarrow{r}_i|$ is less than the distance $|\overrightarrow{r}|$ by the scalar projection of $\overrightarrow{r}_i$ into $\overrightarrow{r}$,
\begin{align}
	\overrightarrow{r}_i=\overrightarrow{r}-R\overrightarrow{r}_i\cdot\overrightarrow{r}.
\end{align}
An example is shown by $\overrightarrow{r}$ and $\overrightarrow{r}_2$ for the 2nd element in Fig.\,\ref{fig:chp2_UCA}.
It can be calculated that
\begin{align}
	\overrightarrow{r}_i\cdot\overrightarrow{r}=\cos\theta\cos\psi_i+\sin\theta\sin\psi_i=\cos(\theta-\psi_i).
\end{align}
Thus, the relative phase shift for the $i$-th element is
\begin{align}
	\phi_i(\theta)=kR\cos(\theta-\psi_i).
\end{align}
According to (\ref{eq:chp2_steeringvector})-(\ref{eq:chp2_TXgain}), $G(\theta,\theta_{\text{doe}})$ for the UCA when transmitting towards $\theta_{\text{doe}}$ is 
\begin{align} \label{eq:chp2_AF_UCA}
G(\theta,\theta_{\text{doe}}) =\frac{1}{\sqrt{N}}\sum_{i=1}^N e^{jkR[\cos(\theta_{\text{doe}}-\psi_i)-\cos(\theta-\psi_i)]}.
\end{align}

\nomenclature{$\overrightarrow{e}$}{unit vector}

The maximum gain, denoted by $G_{\text{max}}$, is obtained at the DoE angle $\theta_{\text{doe}}$, which can be calculated by
\begin{align}\label{eq:chp2_max_gain}
	G_{\text{max}}=\max_{\theta}\{G(\theta,\theta_{\text{doe}})\}
	=G(\theta_{\text{doe}},\theta_{\text{doe}})
	=\sqrt{N}.
\end{align}
It is worth noticing that $G_{\text{max}}$ only depends on the number of elements $N$ and is regardless to the array geometry.

\nomenclature{$G_{\text{max}}$}{maximum gain}
\nomenclature{$_{\text{max}}$}{maximum value}

The ULA is a one-dimensional array, while the UCA is a planar array in 2D space.
In addition to the ULA and the UCA, there are other array geometries of planar arrays  for different purposes.
In\,\cite{sanudin2012semi,yuan2012direction,heidenreich2012joint}, semi-circular, triangular and rectangular arrays are used for the DoA estimation.
In\,\cite{zaman2013application,gazzah2013optimizing,biao2009doa}, the L-shape, V-shape and Y-shape arrays are exploited to address issues in the DoA estimation, e.g., pair matching and estimation failure.
The basic concepts, e.g., the array steering vector and the array factor still apply for these array geometries.

\subsection{Mutual Coupling}
\label{chp2:antennas:novwep}

In this thesis, besides theoretical analysis, a practical issue, i.e., mutual coupling, which is inherent in antenna arrays is investigated.
The mutual coupling between two nearby antennas is caused by the energy absorption of one antenna from another antenna which either radiates or receives.
The nearby antenna absorbs part of the energy that is supposed to either radiate away from or be received by the other antenna. 
When two antennas are close together, their transmitted/received energy is highly correlated, which degrades the antenna efficiency in both radiation and reception modes.

An example of a 2-antenna system is shown in Fig.\,\ref{fig:chp2_mutualcoupling} to illustrate the impact of mutual coupling.
The radiation and reception modes have the same principle.
Thus, in this example, the radiation mode is studied.
Antenna 1 is excited by a voltage source $v_{g1}$ with source
internal impedance $Z_{g1}$.
The current and voltage on antenna 1 are denoted by $i_1$ and $v_1$, respectively. 
The radiated field from antenna 1 is intercepted by antenna 2.  The current and voltage that are induced on antenna 2 are denoted by $i_2$ and $v_2$, respectively.
The radiated field from antenna 2 again affects $i_1$ and $v_1$, which changes the radiation pattern of antenna 1.

\begin{figure}
\centering
\includegraphics[scale=1]{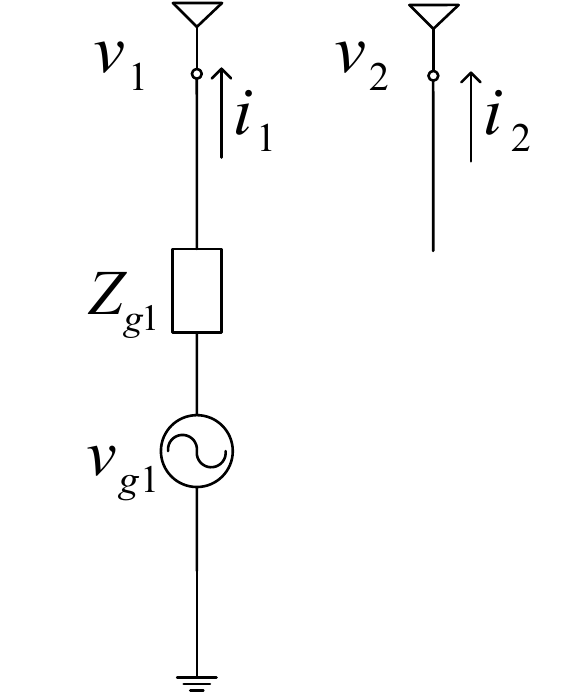}
\caption{An example of a 2-antenna system with antenna 1 excited by a voltage source }
\label{fig:chp2_mutualcoupling}
\end{figure}

Mutual coupling is also known as active element pattern\,\cite{310010} and it is always associated with multiple antenna techniques\,\cite{friedlander1991direction,dai2014recursive}.
$G(\theta,\theta_B)$ is subject to mutual coupling, because it is calculated based on the assumption of an omnidirectional antenna element, the pattern of which is distorted by the mutual coupling.


\section{Channel Models}
\label{chp2:channel}

The security performance of the various techniques reviewed in Section\,\ref{chp2:PhySec} is determined by the wireless channel through which the signal propagates. 
Wireless communication techniques are developed to take advantage of the wireless channel and mitigate the impairments brought by wireless propagation.
In this section, a brief introduction of wireless channel models is given and  some basic aspects that are involved in this thesis, i.e., large-scale path loss and small-scale fading as well as the MISO channel, are covered.
More details are available in\,\cite{goldsmith2005wireless,4460436,cho2010mimo} and the references therein.
In this thesis, the indoor channel models, e.g., TGn channel\,\cite{erceg2004tgn}, are focused on.

\subsection{Large-Scale Path Loss}
\label{chp2:channel:iewurwo}

The basic channel model is the free-space path loss (FSPL) channel model when the signal propagates in the free-space.
The channel causes attenuation in the amplitude, i.e., path loss.
Let $PL$ denote the path loss, which is usually measured in dB scale.
\begin{align}
	PL=10\log_{10}\frac{P_t}{P_r},
\end{align}
where $P_t$ and $P_r$ are the transmit and receive power, respectively, and the unit transmit and receive gains are assumed.
In the free space, $PL$ depends on the signal frequency and the distance that the signal travels.
\begin{align}\label{eq:chp2_FSPL}
	PL=10\log_{10}\Big(\frac{4\pi d}{\lambda}\Big)^2,
\end{align}
where $d$ is the distance and $\lambda=2\pi/f_0$ is the wavelength, where $f_0$ is the carrier frequency.

\nomenclature{FSPL}{free-space path loss}
\nomenclature{$PL$}{path loss}
\nomenclature{$P_t$}{transmit power}
\nomenclature{$P_r$}{receive power}
\nomenclature{$d$}{distance between transmitter and receiver or user's distance in polar coordinates}
\nomenclature{$d_0$}{breakpoint distance}
\nomenclature{$\beta$}{path loss factor}

In a more realistic environment, it is difficult to obtain an accurate model that characterizes the path loss.
A simplified model, i.e., the large-scale path loss model, is used to refer to the average loss in the signal strength over distance.
Let $d_0$ denote the breakpoint distance.
In the close range (i.e., $d\leq d_0$), the channel can be still  assumed to be the FSPL model.
When $d>d_0$, the path loss is
\begin{align}\label{eq:chp2_PL}
	PL=10\log_{10}\Big(\frac{4\pi d_0}{\lambda}\Big)^2+10\log_{10}\Big(\frac{d}{d_0}\Big)^{\beta},
\end{align}
where $\beta$ is the path loss factor and its typical value is from 2 to 6\,\cite{goldsmith2005wireless}.
When $\beta=2$, it reduces to the FSPL channel model.
For realistic channels, the signal attenuates quicker over distance than the free-space environment.

When there are objects that block the signal path or there are some changes in the reflecting surfaces and scatters, the path loss varies randomly for a given distance, which is referred to as shadowing.
A common model is the lognormal shadowing model, which includes a combination of a large number of random variations, and thus is characterized by a decibel (dB) Gaussian random variable $\chi$.
Combined with shadowing, $PL$ can be expressed by
\begin{align}\label{eq:chp2_PL_shadowing}
	PL=10\log_{10}\Big(\frac{4\pi d_0}{\lambda}\Big)^2+10\log_{10}\Big(\frac{d}{d_0}\Big)^{\beta}+\chi.
\end{align}

\nomenclature{$\chi$}{Gaussian random variable for shadowing effect}
\nomenclature{dB}{decibel}

\subsection{Small-Scale Fading}
\label{chp2:channel:nvwo}

While the path loss models refer to the the signal variation over a large distance, the small-scale fading effects are caused by changes over a small distance.
The small-scale fading is caused by multiple versions of the signal when it takes different paths to arrive at the receive antenna;
the multiple versions are combined either constructively or destructively, which causes severe changes in the signal.
Besides the multipath, another reason that causes the small-scale fading is motion.
The movements of the transmitter, receiver or the surrounding objects change the channel characteristics.

The impulse response of a multipath channel is commonly comprised of a discrete number of taps (hence, it is called tapped delay line model). 
Let $h(t,\tau)$ denote the impulse response,
\begin{align}
	h(t,\tau)=\sum_{i} \alpha_i(t,\tau)\delta[\tau-\tau_i(t)],
\end{align}
where $\alpha_i(t,\tau)$ and $\tau_i(t)$ are the complex channel gain and the tap delay for the $i$-th path.
Note that $\alpha_i(t,\tau)$ and $\tau_i(t)$ change with time.

\nomenclature{$h(t,\tau)$}{channel impulse response}
\nomenclature{$t$}{time}
\nomenclature{$\tau$}{time delay}
\nomenclature{$\alpha$}{complex channel gain}

WLAN packets are designed to have short time durations, which is illustrated by an example in\,\cite{4460436}.
The human walking speed is very low (e.g. 1\,m/s) for a typical indoor environment, which leads to a large coherence time (e.g., 70\,ms) compared to the WLAN packet duration (e.g., less than 1\,ms).
Thus, the channel can be regarded as a quasi-static fading channel, 
and the time variance can be suppressed.
Thus, $h(t,\tau)$ reduces to
\begin{align}
	h(\tau)=\sum_{i} \alpha_i(\tau)\delta(\tau-\tau_i).
\end{align}
The multiple versions of packet arrive at the receiver with different delay $\tau_i$, which could causes inter-symbol interference (ISI) among sequential packets.
It is shown by another example in\,\cite{4460436} that the frequency-selective fading could be regarded as a flat fading channel if the symbol duration is designed to be much longer than the root mean square (RMS) delay spread.
This is possible with the orthogonal frequency division multiplexing (OFDM) technique that is incorporated in the 802.11 protocols.
In this thesis, a quasi-static fading channel with a single tap is used and the channel gain $h$ is a random variable subject to certain fading distribution.

\nomenclature{ISI}{inter-symbol interference}
\nomenclature{OFDM}{orthogonal frequency division multiplexing}
\nomenclature{RMS}{root mean square}

There are two commonly used fading channels, Rician fading and Rayleigh fading channels.
When there exists a LOS, the channel is subject to Rician fading;
when there is no dominant path, the channel is called a non-line-of-sight (NLOS) channel and is subject to Rayleigh fading.
Let $X$ be a complex Gaussian random variable with zero mean and variance $2\sigma^2$, i.e., $X\sim\mathcal{CN}(0,2\sigma^2)$.
For the Rayleigh fading channel, the channel gain $h$ can be represented by
\begin{align}
	h=X=X_{Re}+jX_{Im},
\end{align}
where $X_{Re}$ and $X_{Im}$ are the real and imaginary part of $X$, and $X_{Re},X_{Im}\sim\mathcal{N}(0,\sigma^2)$.
The magnitude of $h$, i.e., $|h|$, is a Rayleigh random variable with probability density function (PDF)
\begin{align}
	f_{|h|}(x)=\frac{x}{\sigma^2}e^{-\frac{x^2}{2\sigma^2}}.
\end{align}
When there is a LOS, the Rician channel can be represented by
\begin{align}
	h=\nu+X,
\end{align}
where $\nu$ represents the LOS component.
The magnitude $|h|$ is a Rician random variable with PDF
\begin{align}\label{eq:chp2_rician_pdf}
	f_{|h|}(x)=\frac{x}{\sigma^2}e^{-\frac{x^2+\nu^2}{2\sigma^2}}I_0(\frac{\nu}{\sigma^2}x),
\end{align}
where $I_0(\cdot)$ is the modified Bessel function of the first kind with order zero.
Conventionally, the Rician $K$-factor is used to denote the power ratio of the LOS and NLOS component,
\begin{align}
	K=\frac{\nu^2}{2\sigma^2}.
\end{align}
Alternatively, the complex channel gain for Rician channel can be written in the form of $K$-factor.
\begin{align}
	h=\sqrt{\frac{K}{K+1}}e^{j\phi}+\sqrt{\frac{1}{K+1}}g,
\end{align}
where $\phi$ is the phase component of the LOS path and $g$ is complex Gaussian random variable with unit variance, i.e., $g\sim\mathcal{CN}(0,1)$.
The LOS component $\sqrt{\frac{K}{K+1}}e^{j\phi}$ is deterministic and the NLOS component is $\sqrt{\frac{1}{K+1}}g$.
The total power of $h$ is normalized to one.

\nomenclature{NLOS}{non-line-of-sight}
\nomenclature{PDF}{probability density function}
\nomenclature{$X$}{complex Gaussian random variable with zero mean and variance $2\sigma^2$}
\nomenclature{$\sigma$}{standard deviation for distribution}
\nomenclature{$_{Re}$}{real part of complex variable}
\nomenclature{$_{Im}$}{imaginary part of complex variable}
\nomenclature{$\mathcal{CN}$}{circularly-symmetric complex Gaussian distribution}
\nomenclature{$\mathcal{N}$}{Normal distribution}
\nomenclature{$I_0(\cdot)$}{modified Bessel function of the first kind with order zero}
\nomenclature{$K$}{Rician $K$ factor}
\nomenclature{$g$}{complex Gaussian random variable with zero mean and unit variance}
\nomenclature{$f_{X}(x)$}{PDF of random variable $X$}

The Rician fading channel model is a generalized model.
Note that when $K=0$, the Rician channel degrades into the Rayleigh channel.
When $K$ approaches infinity, the fading channel becomes deterministic.
In this thesis, the generalized Rician fading channel model is considered.

\subsection{MISO Channel Model}
\label{chp2:channel:nvuewbaav}


For the MISO channel,  the signals for the array elements on the LOS path will experience the phase differences between them, which can be captured by $\mathbf{s}(\theta)$, where $\theta$ is the user's angle.
For the generalized Rician channel model, the LOS component should encompass $\mathbf{s}(\theta)$, while the channel for each antenna element experiences the independent and identically distributed (i.i.d.) Rician fading.
Let $\mathbf{h}$ denote the channel gain vector between the multi-antenna transmitter and the receiver.
$\mathbf{h}$ can be written as 
\begin{align}
	\mathbf{h}=\sqrt{\frac{K}{K+1}}\mathbf{s}(\theta)+\sqrt{\frac{1}{K+1}}\mathbf{g},
\end{align}
where $\mathbf{s}(\theta)$ is the LOS component and $\mathbf{g}=[g_1,g_2,...,g_N]^T$ is the NLOS component; the entry $g_i$ is i.i.d. circularly-symmetric complex Gaussian random variable with zero mean and unit variance, i.e., $g_i\sim\mathcal{CN}(0,1)$.
Note that in practice there is spatial correlation in the channels between different antennas.
The spatial correlation decreases in a rich multipath propagation environment or when the spacing among the antenna elements increases.
In this thesis, the impact of the spatial correlation is not considered.

\nomenclature{i.i.d.}{independent and identically distributed}
\nomenclature{$\mathbf{h}$}{channel gain vector}
\nomenclature{$\mathbf{g}$}{complex Gaussian random vector}

In this thesis, the channel with both the large-scale path loss and the small-scale fading for a certain environment is considered.
Thus, the breakpoint distance can be assumed constant and the random shadowing can be ignored.
In some literature\,\cite{zheng2014transmission,wang2015jamming}, the large-scale path loss is simply represented by $d^{-\frac{\beta}{2}}$ and the constant components in (\ref{eq:chp2_PL}) is omitted.
Therefore, the channel gain vector combining the large-scale path loss and the small-scale fading can be expressed by
\begin{align}\label{eq:chp2_ch_gain_vec}
	\mathbf{h}=\frac{1}{\sqrt{d^{\beta}}}\Big[\sqrt{\frac{K}{K+1}}\mathbf{s}(\theta)+\sqrt{\frac{1}{K+1}}\mathbf{g}\Big].
\end{align}
This expression has been used in\,\cite{yan2014secrecy,wang2015jamming}.
Notice that $d$ is assumed to be sufficiently large so that the far-field assumption mentioned in Section\,\ref{chp2:antennas} is fulfilled.

\section{Experiment and Simulation Tools}
\label{chp2:mc}

In this thesis, the array factor is measured in real experiments on WARP as well as in numerical simulations by NEC.
WARP is a soft-defined radio platform which enables transmission/reception and processing of signals in the physical layer\,\cite{warpProject,xiong2010secureangle,xiong2013securearray,xiong2013arraytrack}. 
The results from NEC simulations are well accepted in the literature~\cite{dandekar2000effect}.
In this section, a brief introduction to WARP and NEC tool is given.

\subsection{WARP Hardware}
\label{appdx:warpnec:wioo}

WARP version 3 board integrates a Virtex-6 FPGA with peripheral functional modules, which is shown in Fig.\,\ref{fig:appdx_WARP_warp_v3_kit_sm}.
For example, the clocking module generates the reference frequency for up/down conversion in the transceiver as well as the sampling frequency for the AD/DA conversion in the baseband processing.
The radio-frequency (RF) module, which mainly consists of the AD/DA chips and the transceiver chips, transforms the sampled data into a radio signal for transmission, and also captures the received radio signal and stores the sampled data.
Just to name a few, there are also memory module, Ethernet module, power module and so on.

\nomenclature{RF}{radio-frequency}

\begin{figure}
\centering
\includegraphics[width=0.8\textwidth]{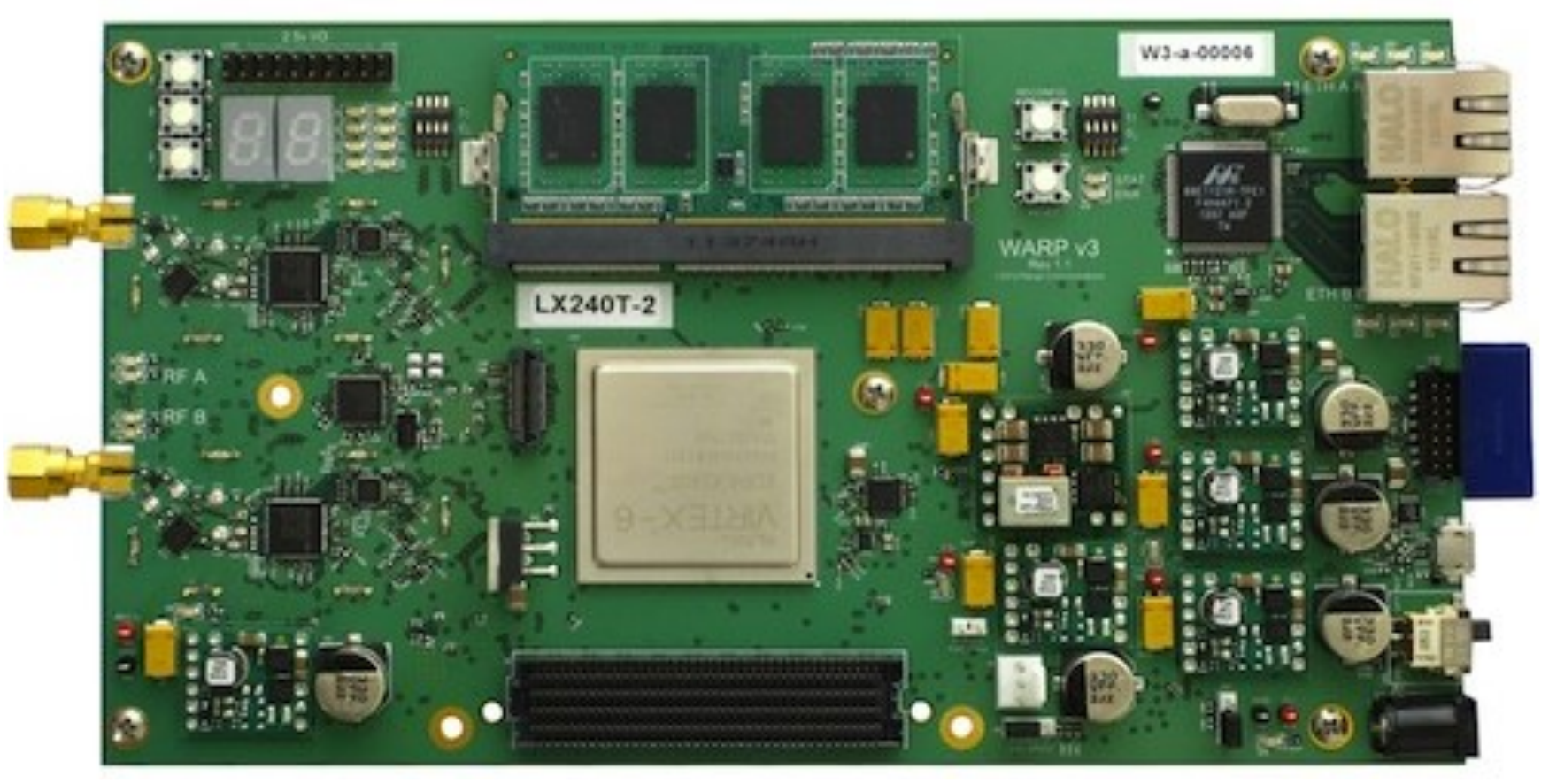}
\caption{WARP v3 board}
\label{fig:appdx_WARP_warp_v3_kit_sm}
\end{figure}

The FPGA executes commands to control the peripherals.
Each WARP board is a node, like a PC/laptop in the network.
To make the node work, there are some customized hardware designs (e.g., set of commands, memory allocations and etc.) that are loaded to the FPGA via various methods, such as JTAG and SD card, when the FPGA is powered on.

In this thesis, the WARPLab design is used, which allows physical layer prototyping for single and multi-antenna transmit and receive nodes.
Each WARPLab node (for short `node' hereinafter) is connected in a local network via Ethernet switch and cables, together with a PC/laptop. 
A typical topology is shown in Fig.\,\ref{fig:appdx_WARP_warplab}.

\begin{figure}
\centering
\includegraphics[width=0.8\textwidth]{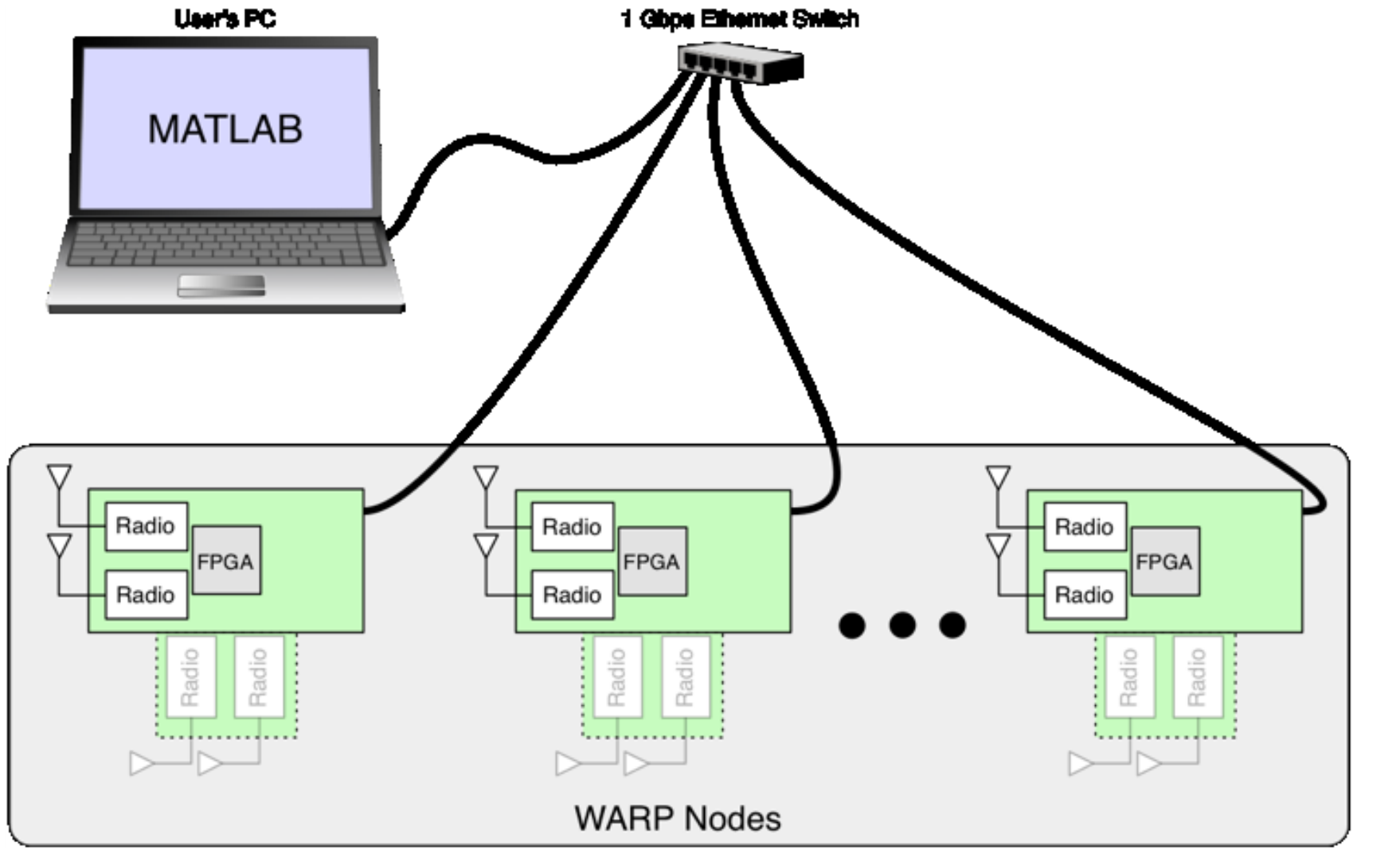}
\caption{WARPLab network}
\label{fig:appdx_WARP_warplab}
\end{figure}

In Fig.\,\ref{fig:appdx_WARP_warplab}, there is one laptop installed with MATLAB and several nodes, which are connected via the switch.
In the WARPLab hardware design, there are communications protocols (TCP/UDP) that enable the exchange of packets between MATLAB and the nodes.
The packet flow between MATLAB and the nodes is focused on to give a brief description of how a communications system is realized.
The details of the WARPLab design are beyond the scope of this thesis.

A packet can be either a command packet from MATLAB to the nodes, e.g., transmit/receive, or a data packet between MATLAB and the nodes, e.g., samples of BPSK modulated symbols.
As will be seen, these packets are transmitted between MATLAB and the RF module on the node, which constitutes a communications system.

Before introducing the communications system, first, the transceiver's structure is given in Fig.\,\ref{fig:appdx_WARP_transceiver}.
There are two or four RF interfaces on each node. 
Each RF interface is half-duplex and can be set to transmit or receive mode per request.
On the right side of block diagram, there are three sample buffers which store data samples that are sent from the MATLAB or captured from the RF band.
On the left side, there are baseband (BB) and RF amplifiers as well as the up/down-convertors.
On the left side, the RF chain is connected to the antenna.

\nomenclature{BB}{baseband}

\begin{figure}
\centering
\includegraphics[scale=1.2]{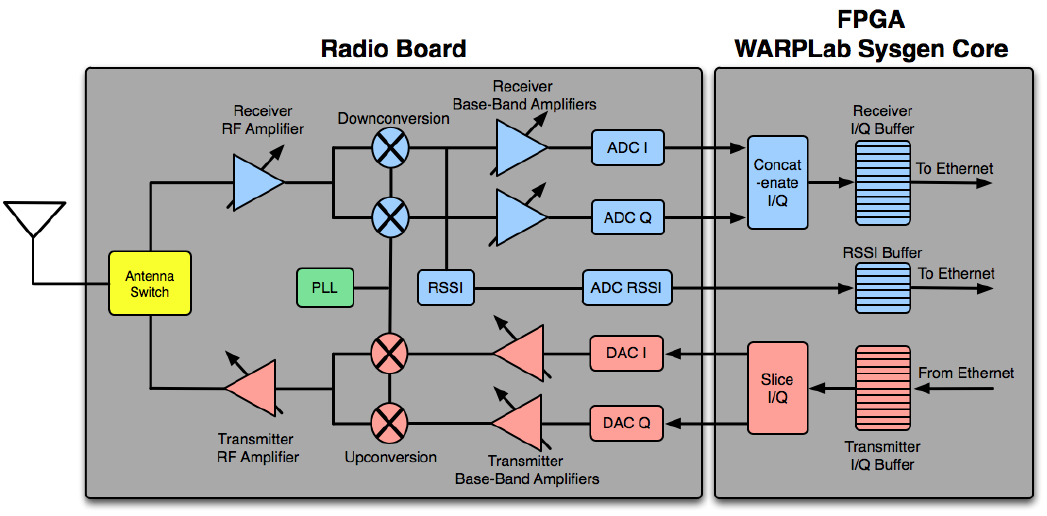}
\caption{Transceiver of one RF interface}
\label{fig:appdx_WARP_transceiver}
\end{figure}

In the transmit mode (red color, in the bottom), the data packet is delivered from the MATLAB via switch to the buffer.
The packet is normally on a intermediate frequency (IF) band that contains both real and imaginary parts.
Then it is split into I/Q branches that are passed to the DA convertor and the up-converter before they are physically sent via an antenna.
During this process, there are two amplifiers, i.e., BB and RF amplifiers, each has discrete power levels.
In the receive mode (blue color on the top), the process is just in the reverse order. 
The signal is captured by the antenna and is down-converted and sampled.
The I/Q branches are stored as complex numbers in the buffer and later will be sent to MATLAB for post-processing.
In addition to the captured data, the received signal strength indicator (RSSI) is measured in hardware and stored in a separate buffer.
Notice that the reference and sampling frequencies can be either generated by a local oscillator or obtained from an external source, which corresponds to the phase-lock loop (PLL) unit in Fig.\,\ref{fig:appdx_WARP_transceiver}.

\nomenclature{IF}{intermediate frequency}
\nomenclature{RSSI}{received signal strength indicator}
\nomenclature{PLL}{phase-lock loop}

\subsection{Communications System on WARPLab}
\label{appdx:warpnec:rqe}

A basic SISO system can be formed by a transmit node and a receive node as shown in Fig.\,\ref{fig:appdx_WARP_warplab}. 
First, a packet is generated in MATLAB and delivered to the transmit buffer. 
The structure of a data packet is shown in Fig.\,\ref{fig:appdx_WARP_packet}.
The packet consists of a preamble and a payload. 
The preamble is used for sample synchronization, channel estimation and etc; the payload stores the messages or commands.
Normally, the length of the packet does not exceed the buffer's size, e.g., $2^{15}$ samples.
The complex (or real) samples in the packet are then transformed into an analog signal before up-conversion to 2.4\,GHz or 5\,GHz. 
The transmitted signal travels via the wireless (or wired) channel and reaches the receive antenna, where it is captured, down-converted, sampled and stored in the receive buffer and waits to be delivered to MATLAB.

\begin{figure}
\centering
\includegraphics[scale=1]{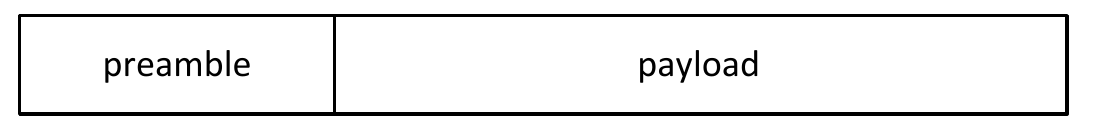}
\caption{Data packet}
\label{fig:appdx_WARP_packet}
\end{figure}

The process that a packet is generated from MATLAB and passed down to the transmit node, then transmitted over the air and captured by the receive node, and finally delivered back to MATLAB is how a basic communications system works on WARPLab.
The packet is both generated and post-processed in MATLAB, which grants WARPLab users the freedom to construct their systems at will.
A simple example is shown in Fig.\,\ref{fig:appdx_WARP_commsyst} to illustrate the design of a communications system.
The blocks in the dashed circle represent  MATLAB processing.

\begin{figure}
\centering
\includegraphics[scale=1]{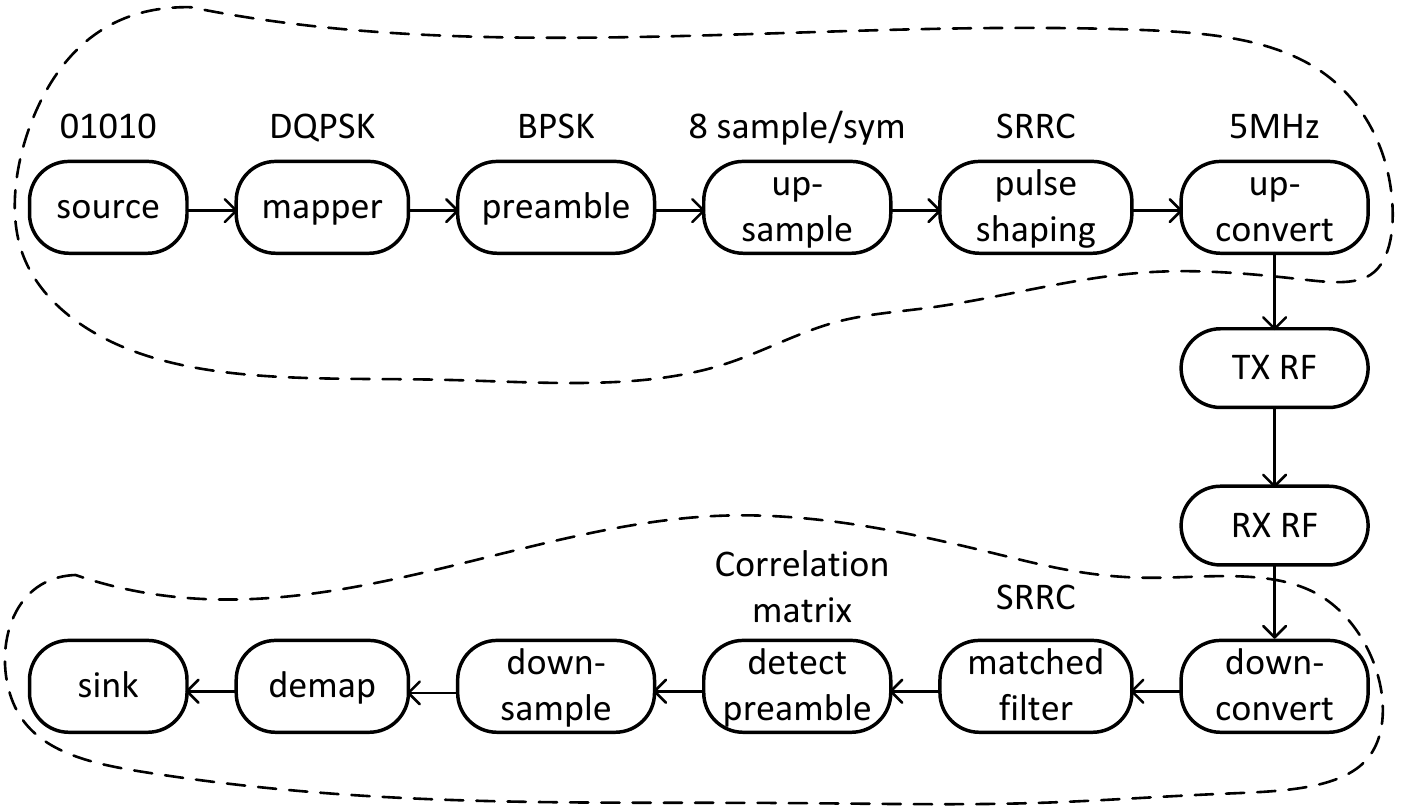}
\caption{An example of SISO system on WARPLab}
\label{fig:appdx_WARP_commsyst}
\end{figure}

In Fig.\,\ref{fig:appdx_WARP_commsyst}, an example of SISO system in WARPLab is shown. 
Data bits are randomly generated from the source then mapped to a DQPSK symbol, which constitutes the payload.
A BPSK modulated preamble is added to the payload, as shown in Fig.\,\ref{fig:appdx_WARP_packet}. 
The packet is then up-sampled and passed a square-root raised cosine (SRRC) filter.
The packet is up-converted to 5\,MHz before being sent to the transmit buffer.
Up till this point, everything takes place in MATLAB.
After the packet is sent to the transmit buffer, it will be sent and received on chosen RF frequency and is finally stored in the receiver buffer.
The samples in the receiver buffer is packed and sent to MATLAB for filtering, de-mapping and so on.

\nomenclature{SRRC}{square-root raised cosine}


The MISO system can be built in the same way as the SISO system.
The major difference is that multiple RF interfaces are needed at the transmitter side, in order to send the packets via different wireless/wired channels.
A generalized system model is shown in Fig.\,\ref{fig:appdx_WARP_commsyst2}, where $N$ RF interfaces are at the transmit node(s).
Each RF interface is a half-duplex transceiver shown in Fig.\,\ref{fig:appdx_WARP_transceiver}. 
The blocks inside the dashed boxes are MATLAB processing.

\begin{figure}
\centering
\includegraphics[scale=1]{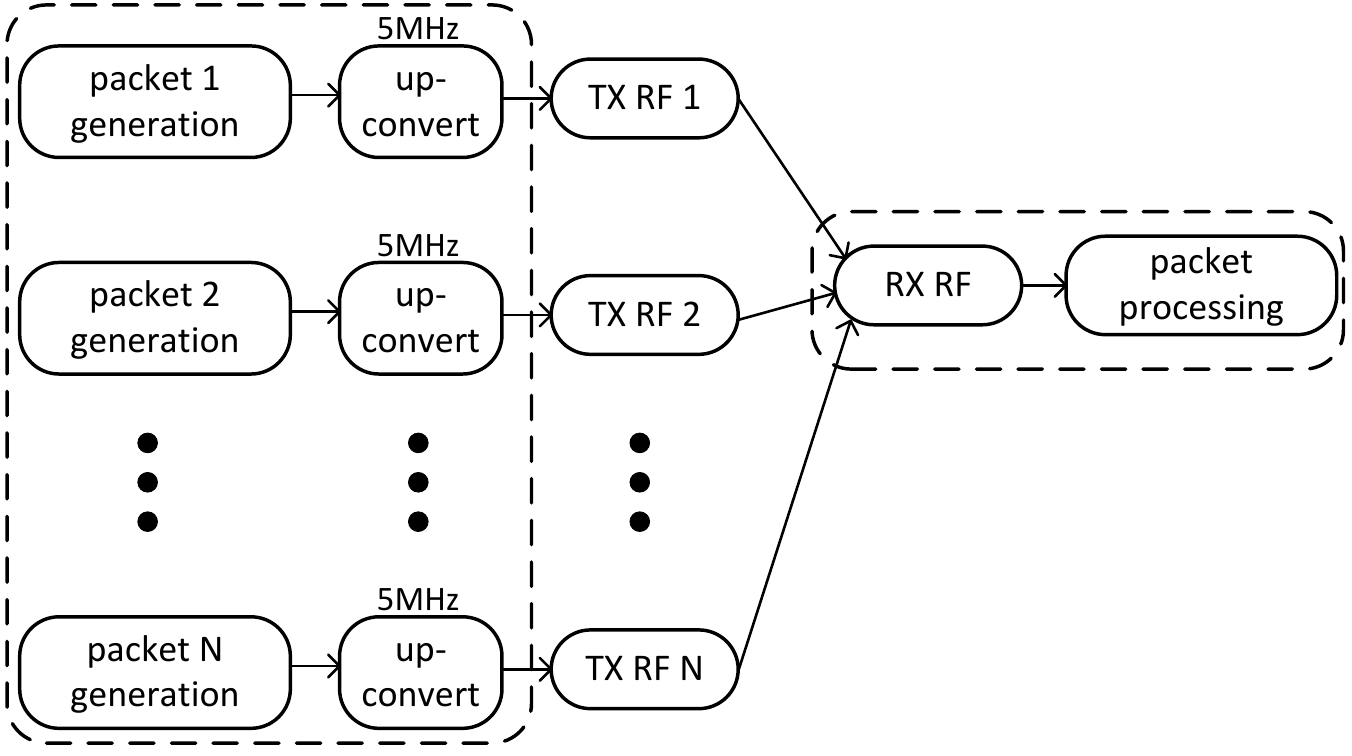}
\caption{An example of MISO system on WARPLab}
\label{fig:appdx_WARP_commsyst2}
\end{figure}

In MATLAB, $N$ packets are generated for $N$ buffers on WARP nodes. 
For the purpose of building a transmit beamformer, each packet is assigned to a different complex weight, i.e., amplitude and phase.
The radio waves that emit from different antennas superimpose over the air and create patterns according to the design.


\subsection{NEC}
\label{appdx:warpnec:vbjhbie}

NEC provides us a numerical method to calculate patterns that have the mutual coupling effect.
It was originally developed at the Lawrence Livermore Laboratory for wire antennas~\cite{burke1981numerical}.
Different versions of NEC are freely distributed in the internet.
In this thesis, the 4NEC2 tool is used, which is based on NEC, to create antenna array and generate array patterns.

In NEC, a wire antenna is decomposed into multiple thin, perfectly conducting wire segments, each of which can be excited with different amplitudes and phases, which corresponds to $\mathbf{w}$. 
Multiple wire antennas can be allocated by an arbitrary array geometry.
Then, for any array geometry and any $\mathbf{w}$, $G(\theta,\theta_B)$ can be simulated based on a numerical solution of electromagnetic field integrals on such segments using the method of moments\,\cite{burke1981numerical}.

An input file manages all the information needed for simulations, e.g., the configuration of the wire segments, the array geometry and $\mathbf{w}$.
In this thesis, the half-wavelength dipole, which is a copper wire of 0.001\,m radius, is used as the array element.
Each wire is divided into 9 segments, of which the middle segment is excited by a voltage source with a frequency of 2414\,MHz, which corresponds the Wi-Fi channel 14. 
Different patterns are produced by imposing $\mathbf{w}$ to the voltage sources.
The dipoles are placed along z-axis in 3-D Cartesian coordinates, so the azimuth pattern of each dipole is circular if there is no mutual coupling. 
All elements are located in the x-y plane according to Fig.\,\ref{fig:chp2_ULA} and\,\ref{fig:chp2_UCA}.

$G(\theta,\theta_B)$ that includes the mutual coupling can be simulated for a specific input file.
Compared to WARP experiments, the mutual coupling effect simulated in the NEC tool does not incorporate numerous imperfections in practice, such as technical malfunctions on fabrication, carrier frequency/phase offsets, inaccurate alignment of array geometry and so on. 
The advantage of using NEC results is that they are more accurate than the WARP results that need calibration before measurements, as will be discussed in Section\,\ref{chp4:sec5:pzoxsp}.

\section{Conclusions}
\label{chp2:conclu}

In this chapter, the literature for the wireless security in the physical layer is reviewed with the focus on the physical region related work.
The problem is to enhance the security in the physical layer against the randomly distributed passive Eves.
There are two different routes to solve this problem, each of which incorporates various techniques.
One starts from information-theoretic security concepts and analyzes the secrecy provided by creating the physical region via beamfomring, jamming, and etc.
The other route focuses on the creation of the physical region from a more practical point of view.
Both routes share the common ground, that is, the physical region that is vulnerable to Eves (or secure from Eves) should be minimized (or maximized), which serves as the basic principle of this thesis.

The array factors for both ULA and UCA are introduced. 
Although there is a wide range of array geometries, e.g., L-shape and V-shape arrays, the commonly used ULA and UCA are examined in this thesis.
The fundamental concepts introduced in this thesis and the methodologies used based on the array factor are applicable to other arrays.

In this thesis, the antenna array is mounted on the AP and beamforming is exploited to create the physical region.
Thus, the MISO channel model is introduced in this chapter.
The generalized Rician fading channel model that incorporates the large-scale path loss and the small-scale fading is used.

Finally, the WARP hardware is introduced to illustrate how to build a MISO system using WARPLab design in order to develop the beamformer.
The NEC tool is introduced to build the ULA and the UCA and measure their array patterns.

\chapter{Spatial Secrecy Outage Probability and Analysis for Uniform Circular Arrays}
\label{chp3}

\section{Introduction}
\label{chp3:intro}

In this chapter, the security performance of the ER-based beamforming with the ULA is investigated.
As discussed in Chapter\,2, the issue caused by passive Eves can be addressed by creating a physical region using the beamforming technique.
Now the potential of using antenna array to enhance the security is explored from the physical region perspective in different channel conditions.

Previous research based on the information-theoretic parameters has focused on the information-theoretic secrecy with less attention to the spatial area in the physical environment.
However, there are many applications that require security inside an enclosed area, such as different zones in an exhibition hall or different assembly lines in a factory.
It is desirable that signals can be confined in a limited physical region, which is defined as the ER.
In this thesis, the ER is defined based on the information-theoretic secrecy parameter.
The challenges for reducing the ER towards high wireless security is the small-scale fading effects and the randomly located Eves.

In this chapter, ER-based beamforming is proposed and the spatial security performance is evaluated for the generalized Rician channel.
To properly investigate the security performance, first the system with geometric locations is defined, which enables  exploration from the spatial aspect.
As for the system performance metric, the SSOP, which has its roots in the information-theoretic security, is defined based on the ER to describe the spatial security performance.
To facilitate the analysis of the SSOP, its analytic upper bound is derived.
Combined with numerical results, the SSOP and its upper bound for the ULA are analyzed with respect to different array parameters in different channel conditions.

This chapter is organized as follows. 
In Section\,\ref{chp3:syst}, the system model that incorporates the geometric locations in the generalized Rician channel is introduced.
In Section\,\ref{chp3:metric}, the concept of the ER is introduced. Then the SSOP and its upper bound are derived and analyzed.
In Section\,\ref{chp3:analysis}, the SSOP and its upper bound are analyzed for the simple path loss channel model.
In Section\,\ref{chp3:result}, the numerical results are discussed and the properties of the SSOP and its upper bound are investigated for the generalized Rician channel.
In Section\,\ref{chp3:concl}, the conclusions of this chapter are given.

\section{System and Channel Models}
\label{chp3:syst}
\subsection{System Model with Geometric Locations}
\label{chp3:syst:model}

Consider a dense wireless communications system where the AP communicates to Bob in presence of a large number of Eves, as shown in Fig.\,\ref{fig:chp3_demon}. 
The AP is equipped with an antenna array, e.g., the ULA, which has $N$ antenna elements with spacing $\Delta d$; 
while Bob and Eves have a single antenna. 
So, it is a MISO system from the AP to either Bob or Eve.
For convenience, both Bob and Eves are simply referred to as a `general user' or a `user', unless otherwise stated.

\begin{figure}[t]
\centering
\includegraphics[scale=1]{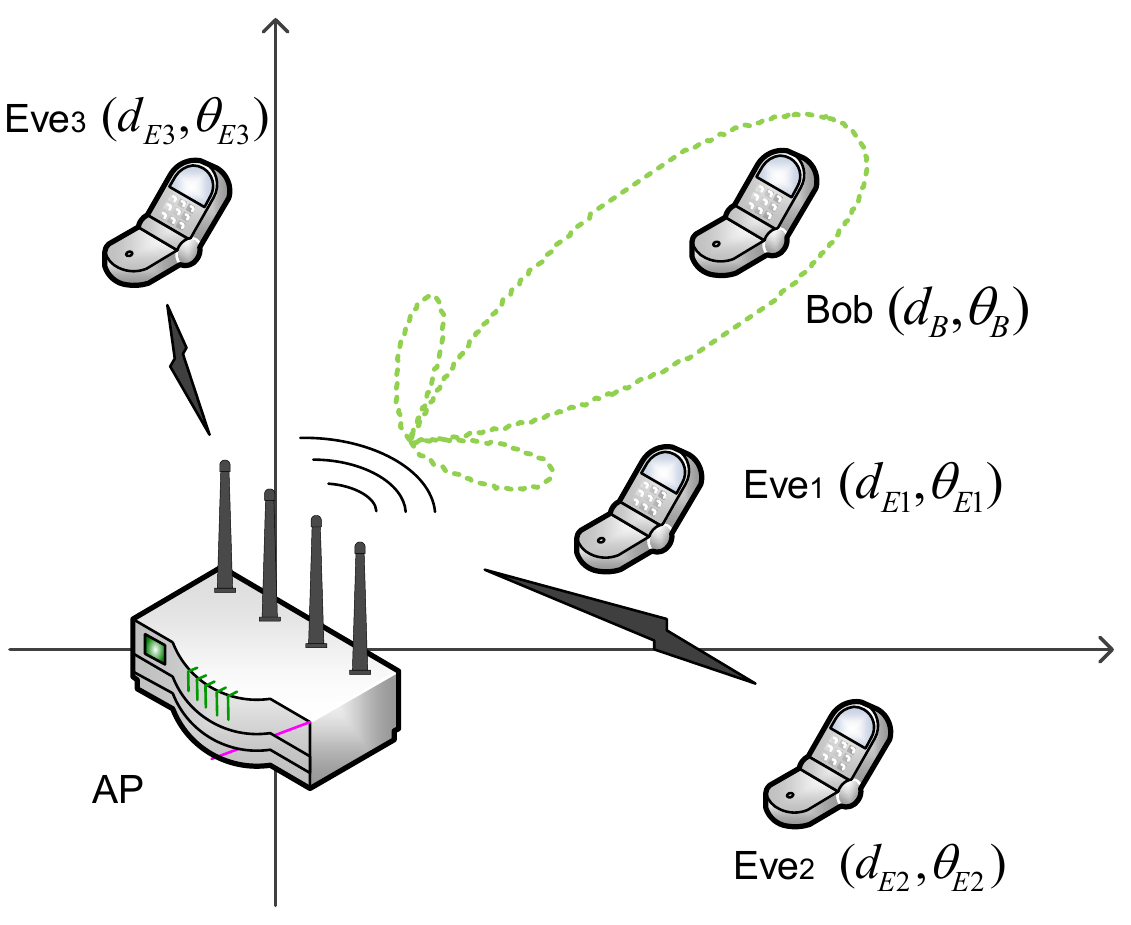}
\caption{An example of a dense communications system with one AP, Bob and several Eves}
\label{fig:chp3_demon}
\end{figure}

Without loss of generality, the AP is located at the origin point in polar coordinates, as shown in Fig.\,\ref{fig:chp3_demon}. 
Assume that the users are distributed by a homogeneous PPP $\Phi_e$ with density $\lambda_e$\,\cite{ghogho2011physical}.
The user's coordinates are denoted by $z=(d,\theta)$.
The subscripts `$_B$' and `$_E$' are used for Bob and Eves hereinafter.
Thus, Bob's coordinates are denoted by $z_B=(d_B,\theta_B)$; the $i$-th Eve's coordinates are $z_{Ei}=(d_{Ei},\theta_{Ei}), i\in\mathbb{N}^+$.

\nomenclature{$\Phi_e$}{homogeneous PPP}
\nomenclature{$\lambda_e$}{density of PPP}
\nomenclature{$z$}{user's coordinates}
\nomenclature{$_B$}{Bob}
\nomenclature{$_E$}{Eve}

Consider a time-division duplex (TDD) system. 
Assume that the AP could estimate Bob's CSI, which can be further used to estimate Bob's coordinates $(d_B,\theta_B)$.
For example, the AP sends a probe message, then the CSI can be estimated via a separate feedback channel.
Alternatively, the user could send a probe message and the AP listens to and estimates the channel.
Unlike the availability of Bob's CSI, the AP cannot obtain Eves' CSI, because the passive Eves only listen to the channel between the AP and Bob; thus their CSI is difficult to detect.

\nomenclature{TDD}{time-division duplex}


Let $x$ be the modulated symbol with unit power, $\mathbb{E}[|x|^2]=1$, and $P_t$ be its transmit power. 
The transmitted vector, denoted by $\mathbf{u}$, is given by
\begin{align}\label{eq:chp3_TX_Signal}
	\mathbf{u}=\sqrt{P_t}\mathbf{w}^*x,
\end{align}
where $\mathbf{w}$ is the beamforming weight vector.
For given $\theta_B$, the AP designs $\mathbf{w}$ associated with $\theta_B$.
To maximize Bob's received signal power, the AP steers the mainbeam of array pattern towards $\theta_B$, i.e., $\theta_{\text{doe}}=\theta_B$.
In this case, $\mathbf{w}$ is obtained by substituting $\theta_{\text{doe}}=\theta_B$ into (\ref{eq:ch2_BF_Weights}),
\begin{align}\label{eq:ch3_BF_Weights}
	\mathbf{w}=\frac{\mathbf{s}(\theta_B)}{\sqrt{N}}.
\end{align}

\nomenclature{$x$}{modulated symbol with unit power}
\nomenclature{$\mathbb{E}[\cdot]$}{mean}
\nomenclature{$\mathbf{u}$}{transmitted vector}

Consider the channel model that incorporates both the large-scale path loss and the Rician small-scale fading.
In this thesis, the channel gain vector in (\ref{eq:chp2_ch_gain_vec}) is used.
The large-scale path loss is determined by $d$, and $\mathbf{s}(\theta)$ is related to $\theta$. 
To emphasize the location-based channel, $\mathbf{h}(z)$ is used to denote the MISO channel gain vector between the AP and the user at $z=(d,\theta)$,
\begin{align}\label{eq:chp3_CH_Gain}
	\mathbf{h}(z)=\frac{1}{\sqrt{d^{\beta}}}\big(\sqrt{\frac{K}{K+1}}\mathbf{s}(\theta)+\sqrt{\frac{1}{K+1}}\mathbf{g}\big),
\end{align}
where $\beta$ is the path loss factor and the Rician fading has Rician $K$ factor.
In this thesis, the term `the generalized Rician channel' is used to incorporate both the large-scale path loss and the small-scale fading.

According to (\ref{eq:chp3_TX_Signal})-(\ref{eq:chp3_CH_Gain}), the received signal of the user at $z$ can be obtained by
\begin{align}\label{eq:chp3_RX_Signal}
	r(z)&=\mathbf{h}^T(z)\mathbf{u}+n_W \nonumber \\
	    &=\sqrt{\frac{P_t}{d^{\beta}}}\big(\sqrt{\frac{K}{K+1}}\mathbf{s}^T(\theta)+\sqrt{\frac{1}{K+1}}\mathbf{g}^T\big)\frac{\mathbf{s}^*(\theta_B)}{\sqrt{N}}x+n_W \nonumber \\			
			&=\sqrt{\frac{P_t}{d^{\beta}}}\tilde{h}x+n_W,
\end{align}
where $n_W$ is the AWGN with zero mean and variance $\sigma_n^2$ and $\tilde{h}$ is the equivalent channel factor, which is given by
\begin{align}\label{eq:chp3_h_tilde_Ri}
	\tilde{h}&=\big(\sqrt{\frac{K}{K+1}}\mathbf{s}^T(\theta)+\sqrt{\frac{1}{K+1}}\mathbf{g}^T\big)\frac{\mathbf{s}^*(\theta_B)}{\sqrt{N}} \nonumber \\
	         &=\sqrt{\frac{K}{K+1}}\frac{\mathbf{s}^T(\theta)\mathbf{s}^*(\theta_B)}{\sqrt{N}}+\sqrt{\frac{1}{K+1}}\frac{\mathbf{g}^T\mathbf{s}^*(\theta_B)}{\sqrt{N}} \nonumber \\	    
			     &=\sqrt{\frac{K}{K+1}}G(\theta,\theta_B)+\sqrt{\frac{1}{K+1}}\frac{\mathbf{s}^H(\theta_B)\mathbf{g}}{\sqrt{N}},
\end{align}
where $G(\theta,\theta_B)$ is the array factor when $\theta_{\text{doe}}=\theta_B$.
For the ULA, $G(\theta,\theta_B)$ is obtained by substituting $\theta_{\text{doe}}=\theta_B$ in (\ref{eq:chp2_AF_ULA}) and (\ref{eq:chp2_AF_ULA2}).
\begin{align}
G(\theta,\theta_B) &=\frac{1}{\sqrt{N}}\sum_{i=1}^N e^{jk\Delta d(\sin\theta_B-\sin\theta)(i-1)} \label{eq:chp3_AF_ULA} \\
 &= \frac{1}{\sqrt{N}}\frac{1-e^{jNk\Delta d(\sin\theta_B-\sin\theta)}}{1-e^{jk\Delta d(\sin\theta_B-\sin\theta)}} \label{eq:chp3_AF_ULA2},
\end{align}
where $k=2\pi/\lambda$.
For the 2.4\,GHz Wi-Fi signals, $\lambda=\frac{3\times10^8\text{\,m/s}}{2.4\times10^9\text{\,Hz}}=0.125$\,m.

It is reflection symmetry for the ULA, as shown in Fig.\,\ref{fig:chp3_ULA_symmetric}.
Thus, the following proposition can be deduced.
\begin{proposition}\label{prop:chp3_theta_B_range}
Because of the symmetric geometry, the array patterns for $G(\theta,\theta_\text{doe})$ at $\theta_{\text{doe}}=\pm(\theta_B\pm\pi)$ are of the same shape and are symmetric to each other.
In other words, it suffices to study $G(\theta,\theta_B)$ only in the range of $\theta_B\in[0,\frac{\pi}{2}]$.
\end{proposition}

\begin{figure}[t]
\centering
\includegraphics[scale=1]{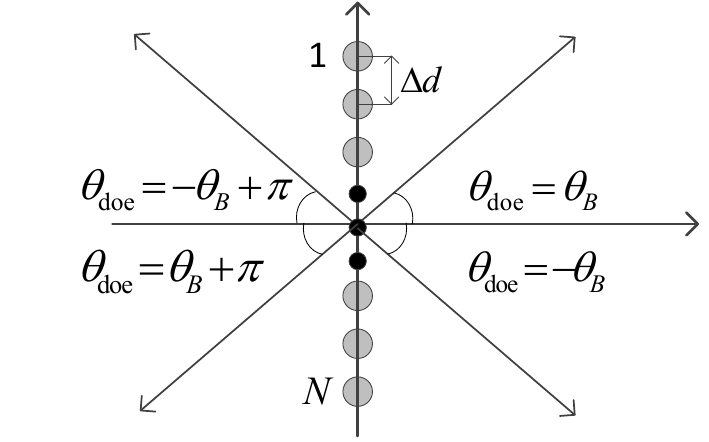}
\caption{An example of ULA with $N$ elements}
\label{fig:chp3_ULA_symmetric}
\end{figure}

\nomenclature{$r$}{received signal}
\nomenclature{$^T$}{transpose}
\nomenclature{$n_W$}{AWGN with zero mean and variance $\sigma_n^2$}
\nomenclature{$\tilde{h}$}{equivalent channel factor}

\subsection{Channel Model and Distribution}
\label{chp3:syst:AF}

Since $\tilde{h}$ is a random variable that depends on $\mathbf{g}$, it is of interest to know the distribution of the amplitude $|\tilde{h}|$.
In (\ref{eq:chp3_h_tilde_Ri}), $\mathbf{s}^H(\theta_B)\mathbf{g}$ is a circularly-symmetric complex Gaussian random variable, because $\mathbf{s}(\theta_B)$ is deterministic for certain Bob's location and is independent to $\mathbf{g}$. 
Similarly, $G(\theta,\theta_B)$ is deterministic as well.
Therefore, $|\tilde{h}|$ is a Rician random variable.
Let $\tilde{h}_{LOS}$ and $\tilde{h}_{NLOS}$ denote the LOS and NLOS components of $\tilde{h}$, respectively.
According to (\ref{eq:chp3_h_tilde_Ri}), they can be written as
\begin{align}
	\tilde{h}_{LOS}&=\sqrt{\frac{K}{K+1}}G(\theta,\theta_B), \\
	\tilde{h}_{NLOS}&=\sqrt{\frac{1}{K+1}}\frac{\mathbf{s}^H(\theta_B)\mathbf{g}}{\sqrt{N}}.
\end{align}
The power of the LOS component $\tilde{\nu}^2$ and the average power of the NLOS component $2\tilde{\sigma}^2$ are
\begin{align}
	\tilde{\nu}^2&=|\tilde{h}_{LOS}|^2= \frac{K}{K+1}G^2(\theta,\theta_B), \\
	2\tilde{\sigma}^2&=\mathbb{E}[|\tilde{h}_{NLOS}|^2]= \frac{1}{N(K+1)}\mathbb{E}[|\mathbf{s}^H(\theta_B)\mathbf{g}|^2]=\frac{1}{K+1}.
\end{align}
The pdf of $|\tilde{h}|$ can be obtained according to (\ref{eq:chp2_rician_pdf}),
\begin{align}\label{eq:chp3_h_tilde_pdf}
	f_{|\tilde{h}|}(x)=\frac{x}{\tilde{\sigma}^2}e^{-\frac{x^2+\tilde{\nu}^2}{2\tilde{\sigma}^2}}I_0(\frac{\tilde{\nu}}{\tilde{\sigma}^2}x),
\end{align}
The Rician $K$ factor for $|\tilde{h}|$ is $\frac{\tilde{\nu}^2}{2\tilde{\sigma}^2}=KG^2(\theta,\theta_B)$ and the total power is $\tilde{\nu}^2+2\tilde{\sigma}^2=\frac{KG^2(\theta,\theta_B)+1}{K+1}$.
Both depend on $G(\theta,\theta_B)$ in (\ref{eq:chp2_AF_ULA}) and (\ref{eq:chp2_AF_ULA2}).

\nomenclature{$_{LOS}$}{line-of-sight component for Rician random variable}
\nomenclature{$_{NLOS}$}{non-line-of-sight component for Rician random variable}

For the case when the channel only contains the LOS path, the Rician fading channel becomes deterministic. 
In (\ref{eq:chp3_h_tilde_Ri}), let $K$ approach the infinity, the equivalent channel $\tilde{h}$ is reduced to
\begin{align}\label{eq:chp3_h_tilde_De}
	\tilde{h}=G(\theta,\theta_B).
\end{align}
For the worst case when the channel does not contain the LOS path, the Rician channel becomes the Rayleigh channel.

Accordingly, the received signal power, denoted by $P_r(z)$, can be computed from (\ref{eq:chp3_RX_Signal}),
\begin{align}\label{eq:chp3_receivedpower}
	P_r(z) =\frac{P_t}{d^{\beta}}|\tilde{h}|^2.
\end{align}
Then, the SNR $\gamma(z)$ can be written as
\begin{align}\label{eq:chp3_SNR}
	\gamma(z)=\frac{P_r(z)}{\sigma_n^2}=\frac{P_t}{\sigma_n^2d^{\beta}}|\tilde{h}|^2.
\end{align}
The channel capacity of the general user located at $z$ can be calculated by
\begin{align}\label{eq:chp3_channelcapacity}
	C(z)=\log_2 [1+\gamma(z)]=\log_2 \Big[1+\frac{P_r(z)}{\sigma_n^2}\Big]=\log_2 \Big[1+\frac{P_t}{\sigma_n^2d^{\beta}}|\tilde{h}|^2\Big].
\end{align}
For convenience, let $C_B=C(z_B)$ and $C_{Ei}=C(z_{Ei})$ denote the channel capacities of Bob and the $i$-th Eve hereinafter.
From (\ref{eq:chp3_receivedpower}) to (\ref{eq:chp3_channelcapacity}), it can be seen that the randomness of $C_{Ei}$ comes from the random location $z_{Ei}$ and the small-scale fading.
Specially, due to the fact that $\tilde{h}$ in (\ref{eq:chp3_h_tilde_Ri}) depends on $G(\theta,\theta_B)$, $G(\theta,\theta_B)$ can be improved by properly designing. 

Given $\theta_B$, the equivalent channel factor for Bob $\tilde{h}_B$ can be calculated by substituting (\ref{eq:chp2_max_gain}) into (\ref{eq:chp3_h_tilde_Ri}),
\begin{align}\label{eq:chp3_h_tilde_Bob}
	\tilde{h}_B=\sqrt{\frac{KN}{K+1}}+\sqrt{\frac{1}{K+1}}\frac{\mathbf{s}^H(\theta_B)\mathbf{g}}{\sqrt{N}}.
\end{align}
$\tilde{h}_B$ has the maximum total power, i.e., $\mathbb{E}[|\tilde{h}_B|^2]=\frac{KN+1}{K+1}$.
Thus, $P_{rB}$ (or $C_B$) is also the maximum value in the range $\theta\in[0,2\pi]$.

\section{SSOP and Upper Bound Derivation}
\label{chp3:metric}

The goal is to provide a reliable and secure transmission for Bob in presence of randomly located passive Eves, by using the transmit beamforming with antenna arrays.
Due to the randomness of the $i$-th Eve's capacity $C_{Ei}$, due to either random location or small-scale fading, it is likely that one or more Eves have a higher channel capacity than the known channel capacity $C_B$.
In this context, while Bob has a reliable transmission from the AP, Eve can correctly receive the whole or part of the signal.
That is referred to as the secrecy outage event.

Because the randomness of $C_{Ei}$ is partially from Eves' random locations, it is natural to exploit the geometric region in the description of the above secrecy outage event.
To this end, \textit{the exposure region (ER) is introduced.
Secrecy outage occurs whenever any Eve randomly appears within the ER.} 
Then, the SSOP is derived based on the ER, and its analytic upper bound is derived, in order to measure the system security level.

\subsection{Exposure Region}
\label{chp3:metric:ER}

The secrecy outage formulation in\,\cite{zhou2011rethinking}, as shown by (\ref{eq:chp2_sop1}) in Section\,\ref{chp2:PhySec:vsp}, is adopted in this thesis.
A secrecy outage event occurs when the perfect secrecy is compromised based upon a reliable transmission for Bob.
Let $R_B$ and $R_s$ be the rate of the transmitted codewords and the rate of the confidential information, respectively.
For fixed $R_B$ and $R_s$, the formulation in (\ref{eq:chp2_sop2}) should be used.
A reliable transmission to Bob can be guaranteed when $C_B\geq R_B$.
Eves' CSI is unknown to the AP and is independent from $C_B$.
The value $C_{Ei}$ can be so large that $C_{Ei}>R_B-R_s$. 
In this case, the secrecy is breached and the secrecy outage event occurs.

The SOP in (\ref{eq:chp2_sop1}) cannot characterize the secrecy outage event for the PPP distributed Eves.
To solve this problem, the ER, denoted by $\Theta$, is defined by the geometric region where Bob faces the secrecy outage event.
Particularly, the geometric region $\Theta$ where $C(z)>R_B-R_s, \exists z=(d,\theta)\in \Theta$, is considered.
The $i$-th Eve will cause secrecy outage, if and only if $z_{Ei}$ randomly appears in $\Theta$.  
Accordingly, $\Theta$ can be represented by
\begin{align}\label{eq:chp3_ERdef}
	\Theta=\{z:\;C(z)>R_B-R_s\}.
\end{align}
In the same time, $C_B\geq R_B$ need to be guaranteed. 

\nomenclature{$\Theta$}{ER}

Substitute (\ref{eq:chp3_channelcapacity}) into (\ref{eq:chp3_ERdef}) and rearrange $(d,\theta)$. $\Theta$ can be transformed into
\begin{align}\label{eq:chp3_erdefinition2}
	\Theta=\{z:\;d<D(\theta)\},
\end{align}
where 
\begin{align}\label{eq:chp3_ERboundary}
	D(\theta)=\Big[{\frac{P_t|\tilde{h}|^2}{\sigma_n^2(2^{R_B-R_S}-1)}}\Big]^{\frac{1}{\beta}}.
\end{align}
$D(\theta)$ is a function regarding to $\theta$ and is regardless of $d$.
In (\ref{eq:chp3_ERboundary}), $\tilde{h}$ can be replaced by that for any channel model. 

\nomenclature{$D(\theta)$}{contour of $\Theta$}

Notice that $D(\theta)$ determines the contour of $\Theta$, which gives a clear geometric meaning.
In polar coordinates, $\Theta$ corresponds an enclosed area that is bounded by $D(\theta)$ and contains the origin point (i.e., the AP).

The shape of $D(\theta)$ (i.e., $\Theta$) is mainly determined by $\tilde{h}$.
Because the equivalent channel factor $\tilde{h}$ is random, $D(\theta)$ is a random variable which is subject to the small-scale fading.
Thus, $\Theta$ is a dynamic region with a shifting boundary.
When the channel is deterministic, $D(\theta)$ is also deterministic.

According to $\tilde{h}$ in (\ref{eq:chp3_h_tilde_Ri}), $D(\theta)$ is highly related to the array factor $G(\theta,\theta_B)$.
Especially when it is the deterministic channel in (\ref{eq:chp3_h_tilde_De}), $\Theta$ is affected by the array pattern.
An example of $\Theta$ is depicted in Fig.\,\ref{fig:chp3_ER_illustration}, where $D(\theta)$ is indicated by the $(R_B-R_s)$-curve.
A reliable transmission is guaranteed for Bob, if Bob is inside the $R_B$-curve. 
\textit{Secrecy outage} occurs if the randomly located Eve is inside $D(\theta)$.
If Eve is located inside the $R_B$-curve, i.e., $C_{Ei}>R_B$, the secrecy capacity is zero.
For random fading channels, the boundary of $\Theta$, i.e., $D(\theta)$, also changes. 
So there will be some variations based on pattern shown in Fig.\,\ref{fig:chp3_ER_illustration}.

\begin{figure}
\centering
\includegraphics[scale=1]{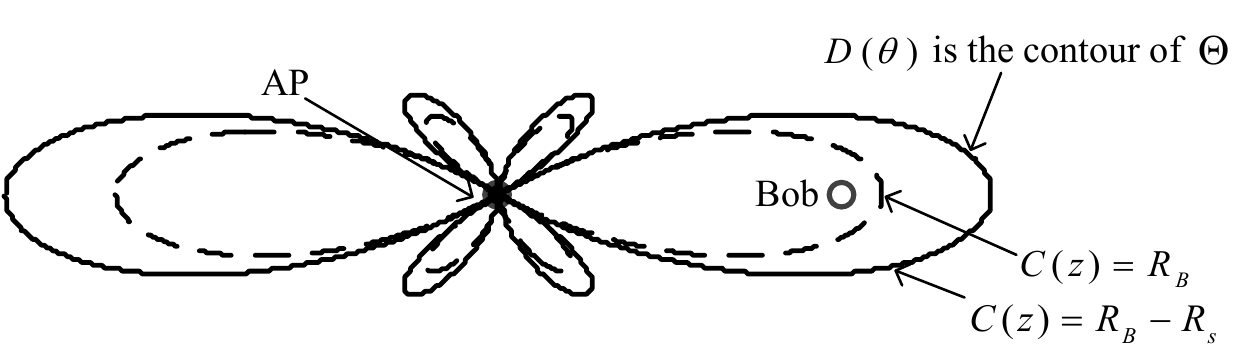}
\caption{Illustration of $\Theta$}
\label{fig:chp3_ER_illustration}
\end{figure}

Now a new expression for the size of $\Theta$ is formulated as one key parameter metric.
To this end, focus only on Eves that are randomly located inside $\Theta$. This is because the perfect secrecy is compromised only when $z_{Ei}\in\Theta$. 
Intuitively, the smaller area of $\Theta$ is, the smaller the number of Eves that are statistically located in $\Theta$. 
This leads to lower secrecy outage based on (\ref{eq:chp3_ERdef}). 
The size of $\Theta$ is denoted by $A$.
In polar coordinates, $A$ can be formulated using (\ref{eq:chp3_ERboundary}),
\begin{align}\label{eq:chp3_A1}
A=\frac{1}{2}\int_0^{2\pi}D^2(\theta)\,\mathrm{d}\theta
 =\frac{1}{2}\int_0^{2\pi}\Big[{\frac{P_t|\tilde{h}|^2}{\sigma_n^2(2^{R_B-R_S}-1)}}\Big]^{\frac{2}{\beta}}\,\mathrm{d}\theta.
\end{align}
$A$ is measured in\,m$^2$ and depends on $\tilde{h}$.


\nomenclature{$A$}{size of $\Theta$}

\subsection{Spatial Secrecy Outage Probability}
\label{chp3:metric:ssop}

Using (\ref{eq:chp3_A1}), the probability of secrecy outage event associated to $\Theta$ can be quantitatively measured.
Any Eve located inside $\Theta$ causes secrecy outage and this is referred to as spatial secrecy outage (SSO) event.
Thus, the probability of the SSO event is equal to the probability that any Eve is inside $\Theta$.

\nomenclature{SSO}{spatial secrecy outage}

Denoted by $p$, the \textit{spatial secrecy outage probability} (SSOP) is defined by the probability that 
any Eve located inside $\Theta$.
$p$ measures the security level of the system in presence of randomly located Eves.
The meaning of the term `spatial' is two fold: it differentiates from the conventional secrecy formulation which does not have a  dynamic geometric implication; it also emphasizes the fact that the secrecy outage is caused by spatially distributed Eves within the dynamic region $\Theta$.

\nomenclature{$p$}{SSOP}

For the PPP-distributed Eves, the probability that $m$ Eves are located inside $\Theta$ (with size $A$) is given by
\begin{align}\label{eq:chp3_PPP}
	\text{Prob}\{m\;\text{Eves in}\;\Theta\}=\frac{(\lambda_eA)^m}{m!}e^{-\lambda_eA},
\end{align}
The SSOP $p$ can be formulated by referring to the `no secrecy outage' event that no Eves are located inside $\Theta$.
Using (\ref{eq:chp3_PPP}), it can be derived that
\begin{align}\label{eq:chp3_SSOP}
	p=1-\text{Prob}\{0\;\text{Eve in}\;\Theta\}=1-e^{-\lambda_eA}.
\end{align}
The smaller the SSOP $p$ is, the less probable the secrecy outage occurs and the more secure the transmission to Bob is.

The secrecy outage formulation in (\ref{eq:chp3_SSOP}) is from the geometric aspect which allows us to compute the probability that no Eve is inside the ER, whereas the geometric concept is weakened in some research and the average secrecy outage is calculated over all PPP distributed Eves' locations\,\cite{zheng2014transmission,zheng2015multi}.

Notice in (\ref{eq:chp3_SSOP}) that there is a positive correlation between $p$ and $A$, which verifies the intuition of `the smaller $\Theta$ is, the lower the secrecy outage' in Section\,\ref{chp3:metric:ER}.
For certain environment, it is reasonable to assume that $\lambda_e$ is a constant. Thus, $p$ solely depends on $A$.

$p$ depends on the equivalent channel factor $\tilde{h}$ via $A$.
Due to the fact that $\tilde{h}$ represents random channel fading, it is more interesting to study the expectation of $p$, which reflects the averaged behavior over a short time period.
The averaged $p$, denoted by $\bar{p}$, can be calculated by
\begin{align}\label{eq:chp3_meanSSOP_Ri_0}
	\bar{p}&=\mathbb{E}_{|\tilde{h}|}[p]=1-\mathbb{E}_{|\tilde{h}|}[e^{-\lambda_eA}].
\end{align}

\nomenclature{$\bar{p}$}{averaged spatial secrecy outage probability for fading channel}

\begin{theorem}\label{th:chp3_ssop}
$\bar{p}$ in (\ref{eq:chp3_meanSSOP_Ri_0}) can be computed by
\begin{align}\label{eq:chp3_meanSSOP_Ri_2}
	\bar{p}&=1-\int_{-\infty}^{\infty}\int_{-\infty}^{\infty} \text{exp}\Big\{-\frac{\lambda_e}{2}c_0^{\frac{2}{\beta}}\int_0^{2\pi}\Big[\frac{KG^2(\theta,\theta_B)}{K+1} \nonumber \\
	& +\frac{x^2+y^2}{K+1}+\frac{2\sqrt{K}G(\theta,\theta_B)}{K+1}x\Big]^{\frac{2}{\beta}}\,\mathrm{d}\theta\Big\} \frac{e^{-(x^2+y^2)}}{\pi} \,\mathrm{d}x\,\mathrm{d}y,
\end{align}
where $\lambda_e$ is the density of Eves, $c_0=\frac{P_t}{\sigma_n^2(2^{R_B-R_S}-1)}$, which is deterministic, $\beta$ is the path loss factor, $K$ is the Rician $K$ factor, $G(\theta,\theta_B)$ is the array factor when the DoE angle is Bob's angle $\theta_B$.
\end{theorem}

\begin{lemma}\label{le:chp3_h_tilde_squre}
$|\tilde{h}|^2$ can be decomposed by
\begin{align}\label{eq:chp3_h_tilde_square}
	|\tilde{h}|^2 =\frac{KG^2(\theta,\theta_B)}{K+1}+\frac{1}{K+1}g_{Re}^2+\frac{1}{K+1}g_{Im}^2+\frac{2\sqrt{K}G(\theta,\theta_B)}{K+1}g_{Re},
\end{align}
where $g_{Re}$ and $g_{Im}$ are the real and imaginary parts of a complex Gaussian random variable $g\sim{CN}(0,1)$.
So, $g_{Re}$ and $g_{Im}$ are jointly normal distributed variables, i.e., $g_{Re}, g_{Im}\sim N(0,\frac{1}{2})$.
\end{lemma}
The following proof of Theorem\,\ref{th:chp3_ssop} requires Lemma\,\ref{le:chp3_h_tilde_squre}, the proof of which is given in Appendix\,\ref{appdx:bessel:nmve}.
\begin{proof}

First, substituting $c_0$ into (\ref{eq:chp3_A1}), $A$ in (\ref{eq:chp3_meanSSOP_Ri_0}) can be simplified into
\begin{align}\label{eq:chp3_A2}
A=\frac{1}{2}\int_0^{2\pi}(c_0|\tilde{h}|^2)^{\frac{2}{\beta}}\,\mathrm{d}\theta.
\end{align}

$\bar{p}$ in (\ref{eq:chp3_meanSSOP_Ri_0}) cannot be directly computed using the pdf of $|\tilde{h}|$ in (\ref{eq:chp3_h_tilde_pdf}) and (\ref{eq:chp3_A2}),
\begin{align}\label{eq:chp3_meanSSOP_Ri_1}
	\bar{p}&=1-\int_0^{\infty} e^{-\lambda_eA} f_{|\tilde{h}|}(x) \,\mathrm{d}x \nonumber \\
	&=1-\int_0^{\infty} \text{exp}\Big[-\frac{\lambda_e}{2}\int_0^{2\pi}(c_0x^2)^{\frac{2}{\beta}}\,\mathrm{d}\theta\Big] 
	\frac{x}{\tilde{\sigma}^2}e^{-\frac{x^2+\tilde{\nu}^2}{2\tilde{\sigma}^2}}I_0(\frac{\tilde{\nu}}{\tilde{\sigma}^2}x),
\end{align}
where $\tilde{\nu}^2=\frac{KG^2(\theta,\theta_B)}{K+1}$, which depends on $\theta$. 
It means that at different angle $\theta$, $\tilde{h}$ follows the Rician fading with different Rician $K$-factor and total power.

Instead, $\tilde{h}$ for each user at different location $(d,\theta)$ experiences an i.i.d. circularly-symmetric complex Gaussian distribution.
According to Lemma\,\ref{le:chp3_h_tilde_squre}, $\bar{p}$ can be calculated by
\begin{align}\label{eq:chp3_meanSSOP_Ri_temp}
	\bar{p}&=\mathbb{E}_{g_{Re},g_{Im}}[p]=1-\int_{-\infty}^{\infty}\int_{-\infty}^{\infty} \text{exp}\Big\{-\frac{\lambda_e}{2}c_0^{\frac{2}{\beta}}\int_0^{2\pi}\Big[\frac{KG^2(\theta,\theta_B)}{K+1} \nonumber \\
	& +\frac{1}{K+1}x^2+\frac{1}{K+1}y^2+\frac{2\sqrt{K}G(\theta,\theta_B)}{K+1}x\Big]^{\frac{2}{\beta}}\,\mathrm{d}\theta\Big\} f_{g_{Re}}(x)f_{g_{Im}}(y) \,\mathrm{d}x\,\mathrm{d}y.
\end{align}
For the normal distribution,
\begin{align}
	 f_{g_{Re}}(x)&=\frac{1}{\sqrt{\pi}}e^{-x^2} \label{eq:chp3_normal_pdf_x},\\
	 f_{g_{Im}}(y)&=\frac{1}{\sqrt{\pi}}e^{-y^2} \label{eq:chp3_normal_pdf_y}.
\end{align}
(\ref{eq:chp3_meanSSOP_Ri_2}) can be obtained by substituting (\ref{eq:chp3_normal_pdf_x}) and (\ref{eq:chp3_normal_pdf_y}) into (\ref{eq:chp3_meanSSOP_Ri_temp}).
Thus, the proof is completed.
\end{proof}

\begin{theorem}\label{th:chp3_p_De}
For the special case when $K\to\infty$, i.e., the channel is deterministic, $\bar{p}$ in (\ref{eq:chp3_meanSSOP_Ri_2}) can be simplified into
\begin{align}\label{eq:chp3_SSOP_De}
	\bar{p}=1-\text{exp}\Big\{-\frac{\lambda_e}{2}c_0^{\frac{2}{\beta}}\int_0^{2\pi}[G^2(\theta,\theta_B)]^{\frac{2}{\beta}}\,\mathrm{d}\theta\Big\}.
\end{align}
\end{theorem}

\begin{proof}
$\lim_{K\to\infty}\frac{K}{K+1}=1$, $\lim_{K\to\infty}\frac{1}{K+1}=0$ and $\lim_{K\to\infty}\frac{\sqrt{K}}{K+1}=0$.
For simplicity, assume that $x$, $y$ take value from a finite range $[-Q,Q]$, where $Q$ a sufficiently large positive real number that satisfies $\int_{-Q}^Q f_{g_{Re}}(x)\,\mathrm{d}x\approx 1$.
For example, when $Q=3$, $\int_{-Q}^Q f_{g_{Re}}(x)\,\mathrm{d}x=0.9999779$.
Thus, when $K\to\infty$, $\bar{p}$ can be approximated by
\begin{align}
	\lim_{K\to\infty}\bar{p}&\approx1-\int_{-Q}^{Q}\int_{-Q}^{Q} \text{exp}\Big\{-\frac{\lambda_e}{2}c_0^{\frac{2}{\beta}}\int_0^{2\pi}[G^2(\theta,\theta_B)]^{\frac{2}{\beta}}\,\mathrm{d}\theta\Big\} \frac{e^{-(x^2+y^2)}}{\pi} \,\mathrm{d}x\,\mathrm{d}y \nonumber \\
	&=1-\text{exp}\Big\{-\frac{\lambda_e}{2}c_0^{\frac{2}{\beta}}\int_0^{2\pi}[G^2(\theta,\theta_B)]^{\frac{2}{\beta}}\,\mathrm{d}\theta\Big\}\int_{-Q}^{Q}\int_{-Q}^{Q}  \frac{e^{-(x^2+y^2)}}{\pi} \,\mathrm{d}x\,\mathrm{d}y \nonumber \\
	&=1-\text{exp}\Big\{-\frac{\lambda_e}{2}c_0^{\frac{2}{\beta}}\int_0^{2\pi}[G^2(\theta,\theta_B)]^{\frac{2}{\beta}}\,\mathrm{d}\theta\Big\}.
\end{align}
\end{proof}


\nomenclature{$Q$}{sufficiently large positive value for certain distribution}

\begin{theorem}\label{th:chp3_p_Ra}
For the special case when $K=0$, i.e., the Rayleigh channel, $\bar{p}$ in (\ref{eq:chp3_meanSSOP_Ri_2}) can be simplified into
\begin{align}\label{eq:chp3_meanSSOP_Ra}
	\lim_{K\to 0}\bar{p}&\approx 1-\int_{-Q}^{Q}\int_{-Q}^{Q} \text{exp}\Big\{-\frac{\lambda_e}{2}c_0^{\frac{2}{\beta}}\int_0^{2\pi}(x^2+y^2)^{\frac{2}{\beta}}\,\mathrm{d}\theta\Big\} \frac{e^{-(x^2+y^2)}}{\pi} \,\mathrm{d}x\,\mathrm{d}y \nonumber \\
	&=1-\int_{-Q}^{Q}\int_{-Q}^{Q} \text{exp}\Big\{-\lambda_e\pi c_0^{\frac{2}{\beta}}(x^2+y^2)^{\frac{2}{\beta}}\Big\} \frac{e^{-(x^2+y^2)}}{\pi} \,\mathrm{d}x\,\mathrm{d}y.
\end{align}
\end{theorem}

For the deterministic channel, $\bar{p}$ is mainly decided by $G(\theta,\theta_B)$.
For the Rayleigh channel,  $\bar{p}$ in (\ref{eq:chp3_meanSSOP_Ra}) does not contain $G(\theta,\theta_B)$, because there is no LOS component in Rayleigh fading channel.

Note that $\bar{p}$ can be simulated by Monte-Carlo method based on (\ref{eq:chp3_meanSSOP_Ri_0}) and numerically calculated based on (\ref{eq:chp3_meanSSOP_Ri_2}).
However, these methods do not provide any analytic insights; thus the impact of certain parameters, such as $N$, $\theta_B$, $K$ and $\beta$, cannot be understood.
In order to investigate the impact of these parameters, in the next section, the analytic expression for the upper bound of $\bar{p}$ will be derived to facilitate further theoretical analysis.

\subsection{Upper Bound of SSOP}
\label{chp3:metric:bounds}

To get the analytic expression for the upper bound, consider two major obstacles to obtain the analytic expression for $\bar{p}$.
First, let $X_{\theta}=c_0|\tilde{h}|^2$.
Then (\ref{eq:chp3_A2}) can be rewritten as
\begin{align}\label{eq:chp3_A3}
A=\frac{1}{2}\int_0^{2\pi}X_{\theta}^{\frac{2}{\beta}}\,\mathrm{d}\theta.
\end{align}
$X_{\theta}$ contains the array factor $G(\theta,\theta_B)$, which makes the integral very difficult to solve when $\beta>2$.
This is true for both deterministic and fading channels.
The other obstacle is that $\mathbb{E}[e^{-\lambda_eA}]$ in (\ref{eq:chp3_meanSSOP_Ri_0}) is difficult to obtain due to the composite array factor and the Rician fading channels.

\nomenclature{$X_{\theta}$}{random variable depending on $\theta$}

The idea to overcome the aforementioned obstacles is that while the direct solution to $\mathbb{E}[e^{-\lambda_eA}]$ is difficult, it is possible to obtain the moments of $|\tilde{h}|$.
Using Jensen's Inequality, the upper bound for $\bar{p}$, denoted by $\bar{p}_{up}$, can be obtained via the moments of $|\tilde{h}|$.
The following two inequalities based on Jensen's Inequality are used.

\begin{lemma}\label{le:chp3_jensensinequality}
Because $e^{(\cdot)}$ is a convex function, according to Jensen's inequality, 
\begin{align}\label{eq:chp3_JI_1}
	\mathbb{E}[e^X]\geq e^{\mathbb{E}[X]},
\end{align}
where $X$ is a random variable.
The equality holds if and only if $X$ is a deterministic value.

\nomenclature{$\bar{p}_{up}$}{upper bound of $\bar{p}$}

On the other hand, $(\cdot)^{\frac{2}{\beta}}$ is a concave function, when $\beta>2$.
According to Jensen's inequality, 
\begin{align}\label{eq:chp3_JI_2}
  \mathbb{E}[X^{\frac{2}{\beta}}]\leq (\mathbb{E}[X])^{\frac{2}{\beta}}.
\end{align}
The equality holds when $\beta=2$ for any $X$.
\end{lemma}

\begin{theorem}
The upper bound $\bar{p}_{up}$ can be expressed by
\begin{align}\label{eq:chp3_meanSSOP_up_2}
	\bar{p}_{up}= 1-\text{exp}\Big\{-\lambda_e\pi\Big[\frac{c_0K}{2\pi(K+1)}\int_0^{2\pi}G^2(\theta,\theta_B)\,\mathrm{d}\theta+\frac{c_0}{K+1}\Big]^{\frac{2}{\beta}}\Big\}.
\end{align}
\end{theorem}

\begin{proof}
According to (\ref{eq:chp3_meanSSOP_Ri_0}) and (\ref{eq:chp3_JI_1}), it can be derived that
\begin{align}\label{eq:chp3_meanSSOP_up_inequality}
	 \bar{p}=1-\mathbb{E}_{|\tilde{h}|}[e^{-\lambda_eA}]\leq 1-e^{-\lambda_e\mathbb{E}_{|\tilde{h}|}[A]}.
\end{align}
Notice that $A$ depends on random variable $\tilde{h}$ and is not constant, except for $K=\infty$.
Thus, the equality holds only for the deterministic channel.

To solve (\ref{eq:chp3_meanSSOP_up_inequality}), assume that $\theta\sim \mathcal{U}(0,2\pi)$.
According to (\ref{eq:chp3_A3}), $A$ in (\ref{eq:chp3_meanSSOP_up_inequality}) can be converted into
\begin{align}\label{eq:chp3_meanSSOP_up_inequality2}
A=2\pi\frac{1}{2}\int_0^{2\pi}\frac{1}{2\pi}X_{\theta}^{\frac{2}{\beta}}\,\mathrm{d}\theta
 =\pi\mathbb{E}_{\theta}[X_{\theta}^{\frac{2}{\beta}}].
\end{align}
According to (\ref{eq:chp3_JI_2}), (\ref{eq:chp3_meanSSOP_up_inequality2}) can be bounded by
\begin{align}\label{eq:chp3_JI_theta}
A\leq \pi(\mathbb{E}_{\theta}[X_{\theta}])^{\frac{2}{\beta}}
 =\pi \Big(\int_0^{2\pi}\frac{1}{2\pi}X_{\theta}\,\mathrm{d}\theta \Big)^{\frac{2}{\beta}}.
\end{align}
In the inequality, the equality holds when $\beta=2$ for any $K$.

\nomenclature{$\mathcal{U}$}{uniform distribution}

According to (\ref{eq:chp3_meanSSOP_up_inequality}) and (\ref{eq:chp3_JI_theta}), it can be derived that
\begin{align}\label{eq:chp3_meanSSOP_up_inequality3}
	\mathbb{E}_{|\tilde{h}|}[A] 
	\leq \pi\mathbb{E}_{|\tilde{h}|}\Big[\Big(\int_0^{2\pi}\frac{1}{2\pi}X_{\theta}\,\mathrm{d}\theta\Big)^{\frac{2}{\beta}}\Big].
\end{align}
Then applying (\ref{eq:chp3_JI_2}) and (\ref{eq:chp3_meanSSOP_up_inequality3}), it can be derived that
\begin{align}\label{eq:chp3_inequality_jvje}
	\pi\mathbb{E}_{|\tilde{h}|}\Big[\Big(\int_0^{2\pi}\frac{1}{2\pi}X_{\theta}\,\mathrm{d}\theta\Big)^{\frac{2}{\beta}}\Big]
	\leq \pi\Big(\mathbb{E}_{|\tilde{h}|}\Big[\int_0^{2\pi}\frac{1}{2\pi}X_{\theta}\,\mathrm{d}\theta\Big]\Big)^{\frac{2}{\beta}}. 
\end{align}
Exchanging the integral and $\mathbb{E}_{|\tilde{h}|}$, then substituting $X_{\theta}=c_0|\tilde{h}|^2$, it can be derived that
\begin{align}\label{eq:chp3_meanA_up}
	\mathbb{E}_{|\tilde{h}|}[A] \leq \pi\Big(\frac{c_0}{2\pi}\int_0^{2\pi}\mathbb{E}_{|\tilde{h}|}[|\tilde{h}|^2]\,\mathrm{d}\theta\Big)^{\frac{2}{\beta}}.
\end{align}
Notice that when $\beta=2$, the equality holds.

Apply (\ref{eq:chp3_meanA_up}) to (\ref{eq:chp3_meanSSOP_up_inequality}) then obtain
\begin{align}
	\bar{p}
	\leq 1-e^{-\lambda_e\mathbb{E}_{|\tilde{h}|}[A]} 
	\leq 1-\text{exp}\Big[-\lambda_e\pi\Big(\frac{c_0}{2\pi}\int_0^{2\pi}\mathbb{E}_{|\tilde{h}|}[|\tilde{h}|^2]\,\mathrm{d}\theta\Big)^{\frac{2}{\beta}}\Big].
\end{align}
The upper bound $\bar{p}_{up}$ can be expressed by
\begin{align}\label{eq:chp3_meanSSOP_up_1}
	\bar{p}_{up}=1-\text{exp}\Big[-\lambda_e\pi\Big(\frac{c_0}{2\pi}\int_0^{2\pi}\mathbb{E}_{|\tilde{h}|}[|\tilde{h}|^2]\,\mathrm{d}\theta\Big)^{\frac{2}{\beta}}\Big].
\end{align}
As mentioned in Section\,\ref{chp3:syst:AF}, $|\tilde{h}|$ is a Rician random variable with total power $\frac{KG^2(\theta,\theta_B)+1}{K+1}$.
Thus, $\mathbb{E}_{|\tilde{h}|}[|\tilde{h}|^2]=\frac{KG^2(\theta,\theta_B)+1}{K+1}$.
Substituting the previous result into (\ref{eq:chp3_meanSSOP_up_1}), it can be derived that
\begin{align}
	\bar{p}_{up}&=1-\text{exp}\Big\{-\lambda_e\pi\Big[\frac{c_0}{2\pi}\int_0^{2\pi}\frac{KG^2(\theta,\theta_B)+1}{K+1}\,\mathrm{d}\theta\Big]^{\frac{2}{\beta}}\Big\} \nonumber \\
	&= 1-\text{exp}\Big\{-\lambda_e\pi\Big[\frac{c_0K}{2\pi(K+1)}\int_0^{2\pi}G^2(\theta,\theta_B)\,\mathrm{d}\theta+\frac{c_0}{K+1}\Big]^{\frac{2}{\beta}}\Big\}.
\end{align}
Thus, the proof is completed.
\end{proof}

\begin{proposition}\label{prop:chp3_SSOP_up_analysis}
For the determinist channel, (\ref{eq:chp3_meanSSOP_up_inequality}) reduces to an equality; thus, the upper bound is tighter when $\beta=2$ than that when $\beta>2$.
When $\beta=2$, (\ref{eq:chp3_meanA_up}) reduces to an equality;
thus, the upper bound is tighter for the deterministic channel than that for the fading channel.
Only for the deterministic channel when $\beta=2$, the equality holds for $\bar{p}_{up}=\bar{p}$.
However, for the fading channel when $\beta>2$, the tightness of the upper bound is not clear.
The numerical results of $\bar{p}_{up}$ for different $K$ and $\beta$ will be given in Section\,\ref{chp3:result:mnbv}.
\end{proposition}

\begin{theorem}\label{th:chp3_p_up_De}
For the special case when $K\to\infty$, $\bar{p}_{up}$ in (\ref{eq:chp3_meanSSOP_up_2}) can be simplified into
\begin{align}\label{eq:chp3_SSOP_De_up}
	\bar{p}_{up}=1-\text{exp}\Big\{-\lambda_e\pi \Big[\frac{c_0}{2\pi}\int_0^{2\pi}G^2(\theta,\theta_B)\,\mathrm{d}\theta \Big]^{\frac{2}{\beta}} \Big\}.
\end{align}
\end{theorem}
\begin{theorem}\label{th:chp3_p_up_Ra}
For the special case when $K=0$, $\bar{p}_{up}$ in (\ref{eq:chp3_meanSSOP_up_2}) can be simplified into
\begin{align}\label{eq:chp3_SSOP_Ra_up}
	\bar{p}_{up}=1-\text{exp}(-\lambda_e\pi c_0^{\frac{2}{\beta}}). 
\end{align}
\end{theorem}

Compare $\bar{p}_{up}$ in (\ref{eq:chp3_meanSSOP_up_2}) and (\ref{eq:chp3_SSOP_De_up}) to $\bar{p}$ in (\ref{eq:chp3_meanSSOP_Ri_2}) and (\ref{eq:chp3_SSOP_De}), it can be seen that
$(\cdot)^{\frac{2}{\beta}}$ is moved outside of the integral, which makes solving the integral in (\ref{eq:chp3_A2}) possible. 
Besides, the expectation takes place on $|\tilde{h}|$ directly, which is simpler with known pdf of $|\tilde{h}|$.
For the Rayleigh channel, $\bar{p}$ in (\ref{eq:chp3_meanSSOP_Ra}) and $\bar{p}_{up}$ in (\ref{eq:chp3_SSOP_Ra_up}) are both deterministic values.

So far, the two major obstacles mentioned in the beginning of this section are tackled.
The final step to obtain the analytic expression of $\bar{p}_{up}$ is to solve the integral $\int_0^{2\pi}G^2(\theta,\theta_B)$ in (\ref{eq:chp3_meanSSOP_up_2}).
Let $A_0$ denote the integral,
\begin{align}\label{eq:chp3_A_0}
	A_0=\int_0^{2\pi} G^2(\theta,\theta_B)\,\mathrm{d}\theta.
\end{align}
$A_0$ is actually the pattern area.

\nomenclature{$A_0$}{pattern area}

As (\ref{eq:chp3_A_0}) is a general expression, any type of array with an analytic expression or numerically measured pattern can be calculated using this equation.
For the ULA, the analytic expression of $A_0$ cannot be directly obtained by substituting $G(\theta,\theta_B)$ in (\ref{eq:chp2_AF_ULA2}).
Here, the analytic expression of $A_0$ is directly given.
\begin{theorem}\label{th:chp3_A_0L}
\begin{align}\label{eq:chp3_A_0L}
	A_0=2\pi+4\pi\sum_{n=1}^{N-1} \frac{N-n}{N}J_0(k\Delta dn)\cos(k\Delta dn\sin\theta_B).
\end{align}
$A_0$ is in the form of finite summation of weighted $J_0(x)$, where $J_0(x)$ is the Bessel function of the first kind with order zero.
\end{theorem}
The proof of Theorem\,\ref{th:chp3_A_0L} is in Appendix\,\ref{appdx:bessel:veaien}.

$A_0$ is determined by $N$, $\Delta d$ and $\theta_B$.
It is relatively easy to analytically analyze $A_0$, because $J_0(x)$ has a decreasing envelope with the maximum value $J_0(0)=1$ at $x=0$.
When $x$ increases to infinity, $J_0(x)$ approaches zero, which makes $J_0(x)$ negligible for certain values of $n$ and $\Delta d$.
Thus, appropriate approximations of $A_0$ can be found to analytically analyze $A_0$.

According to (\ref{eq:chp3_meanSSOP_up_2}) and (\ref{eq:chp3_A_0}), the general expression of $\bar{p}_{up}$ for any array type under the Rician channel can be written in the form of $A_0$,
\begin{align}\label{eq:chp3_meanSSOP_up_3}
	\bar{p}_{up}= 1-\text{exp}\Big\{-\lambda_e\pi\Big[\frac{c_0K}{2\pi(K+1)}A_0+\frac{c_0}{K+1}\Big]^{\frac{2}{\beta}}\Big\}.
\end{align}
Because $\bar{p}_{up}$ is positively correlated with $A_0$, the properties of $\bar{p}_{up}$ can be derived based on those of $A_0$, which serves as guidance for the analysis of $\bar{p}$.
Moreover, it is worth noticing that the approximated $\bar{p}_{up}$ based on the approximations of $A_0$ ceases to be the upper bound of $\bar{p}$, but rather approximations of $\bar{p}$.
In the rest of this chapter, the approximations of $\bar{p}_{up}$ (i.e., the approximations of $\bar{p}$) are used to provide theoretical analysis thanks to the tractable nature, and accurate numerical results of $\bar{p}_{up}$ and $\bar{p}$ are used to verify the theoretical analysis.

\nomenclature{$J_n(\cdot)$}{Bessel function of the first kind with order $n$}

\section{Impact of Array Parameters on SSOP}
\label{chp3:analysis}

It can be seen from (\ref{eq:chp3_meanSSOP_up_3}) that there is a positive correlation between $\bar{p}_{up}$ and $A_0$.
Thus, the analysis of $\bar{p}_{up}$ can be carried out by studying the behavior of $A_0$, which is determined by $N$ and $\theta_B$.
As stated in Proposition\,\ref{prop:chp3_SSOP_up_analysis}, $\bar{p}_{up}=\bar{p}$ for the deterministic channel when $\beta=2$. 
Thus, the investigation starts from this simple channel model.
When $K\to\infty$ and $\beta=2$, $\bar{p}_{up}$ in (\ref{eq:chp3_meanSSOP_up_3}) can be simplified into
\begin{align}\label{eq:chp3_p_De_beta_is_2}
	\bar{p}_{up}=1-\text{exp}(-\frac{\lambda_ec_0}{2}A_0 ).
\end{align}
First the behavior of $A_0$ for the ULA is analyzed with respect to $N$ and $\theta_B$ via both analytic and numerical methods.
More numerical results for $\bar{p}_{up}$ for other values of $\beta$ and $K$ will be shown in Section\,\ref{chp3:result:wier}.

\subsection{Impact of Array Parameters on Pattern Area}
\label{chp3:analysis:sve}

\subsubsection{Impact of Bob's Angle}
To analyze the impact of $\theta_B$ on $A_0$, an appropriate approximation of $A_0$ in (\ref{eq:chp3_A_0L}) is required.  
 As stated in Proposition\,\ref{prop:chp3_theta_B_range}, the range of $\theta_B\in[0,\frac{\pi}{2}]$ is chosen.

First, let $A_{0,n}$, $n=1,...,N-1$, denote the summation term in (\ref{eq:chp3_A_0L}),
\begin{align}\label{eq:chp3_A_0L_n}
	A_{0,n}&=4\pi\frac{N-n}{N}J_0(k\Delta dn)\cos(k\Delta dn\sin\theta_B).
\end{align}
Then, $A_0$ in (\ref{eq:chp3_A_0L}) can be written by
\begin{align}\label{eq:chp3_A_0L_2}
	A_0=2\pi+\sum_{n=1}^{N-1}A_{0,n}.
\end{align}

\nomenclature{$A_{0,n}$}{$n$-th summation term in $A_0$}

For the ULA, consider the half-wavelength spacing, i.e., $\Delta d=0.5\lambda$. 
This is because when $\Delta d$ is smaller than $0.5\lambda$, there will be very high mutual coupling that distorts the array pattern; 
when $\Delta d$ is larger than $0.5\lambda$, there will be large sidelobes and even multiple mainbeams\,\cite{adaptivearraysystems}.
Thus, the behaviors of $A_0$ are only considered in terms of $N$ and $\theta_B$ with $\Delta d=0.5\lambda$.

When $\Delta d=0.5\lambda$, $A_{0,n}$ in (\ref{eq:chp3_A_0L_n}) can be written as
\begin{align}\label{eq:chp3_A_0L_n_2}
	A_{0,n}=4\pi\frac{N-n}{N}J_0(n\pi)\cos(n\pi\sin\theta_B).
\end{align}
An example of ULA with $N=8$ and $\Delta d=0.5\lambda$ is shown in Fig.\,\ref{fig:chp3_besselj_ULA}.
In the upper plot, $J_0(n\pi)$ decreases as $n$. 
The lower plot shows the decreasing envelope of $A_{0,n}$, i.e., $\frac{N-n}{N}J_0(n\pi)$. 
When $n=1$, the value is the largest;
when $n=7$, the value is negligible.

\begin{figure}
\centering
\includegraphics[scale=0.9]{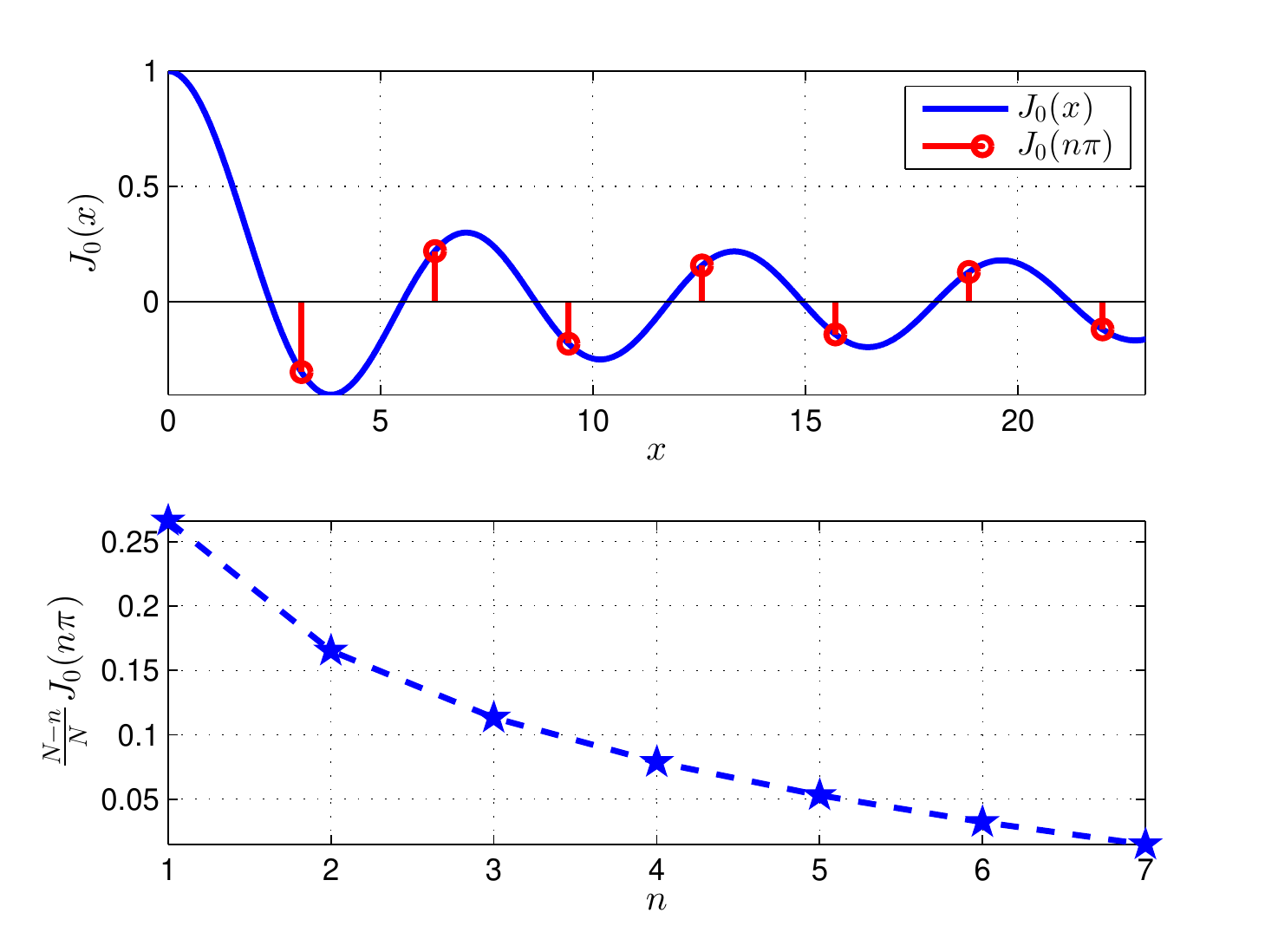}
\caption{$J_0(n\pi)$ and $\frac{N-n}{N}J_0(n\pi)$ for $n=1,2,...,N-1$. $N=8$}
\label{fig:chp3_besselj_ULA}
\end{figure}

As a result, for fixed $N$, $A_0$, the summation of $A_{0,n}$, can be approximated by the first few summation terms.
In the case when $\Delta d=0.5\lambda$, $A_{0,1}$ is very dominant and it suffices to approximate $A_0$ using only $A_{0,1}$, i.e., 
\begin{align}\label{eq:chp3_A_0L_approx}
	A_0\approx 2\pi+4\pi\frac{N-1}{N}J_0(\pi)\cos(\pi\sin\theta_B).
\end{align}

Thus, using (\ref{eq:chp3_p_De_beta_is_2}) and (\ref{eq:chp3_A_0L_approx}), $\bar{p}_{up}$ can be approximated by
\begin{align}\label{eq:chp3_p_De_approx_betais2}
	\bar{p}_{up}\approx 1-\text{exp}\Big\{-\frac{\lambda_ec_0}{2}\Big[2\pi+4\pi\frac{N-1}{N}J_0(\pi)\cos(\pi\sin\theta_B)\Big] \Big\}.
\end{align}

From (\ref{eq:chp3_A_0L_approx}) and (\ref{eq:chp3_p_De_approx_betais2}), it can be seen that for any fixed $N$, when $\theta_B$ increases from $0$ to $\frac{\pi}{2}$, $\pi\sin\theta_B$ increases from $0$ to $\pi$.
Then $\cos(\pi\sin\theta_B)$ decreases from $1$ to $-1$.
As shown in Fig.\,\ref{fig:chp3_besselj_ULA}, $J_0(\pi)<0$, leading to the approximations of $A_0$ and $\bar{p}_{up}$ being a monotonic increasing function in the range $\theta_B\in[0,\frac{\pi}{2}]$.

In Fig.\,\ref{fig:chp3_p_DoE_De_ULA}, $A_{0,n}$ and $\bar{p}_{up}$ versus $\theta_B$ are depicted for the ULA with $N=8$ and $\Delta d=0.5\lambda$.
In the left plot, $A_{0,1}$ has the largest variation from $\theta_B=0^{\circ}$ to $\theta_B=90^{\circ}$.
As $n$ increases, the variation becomes smaller.
This corresponds to the decreasing envelope shown in Fig.\,\ref{fig:chp3_besselj_ULA}.

\begin{figure}
\centering
\includegraphics[scale=0.9]{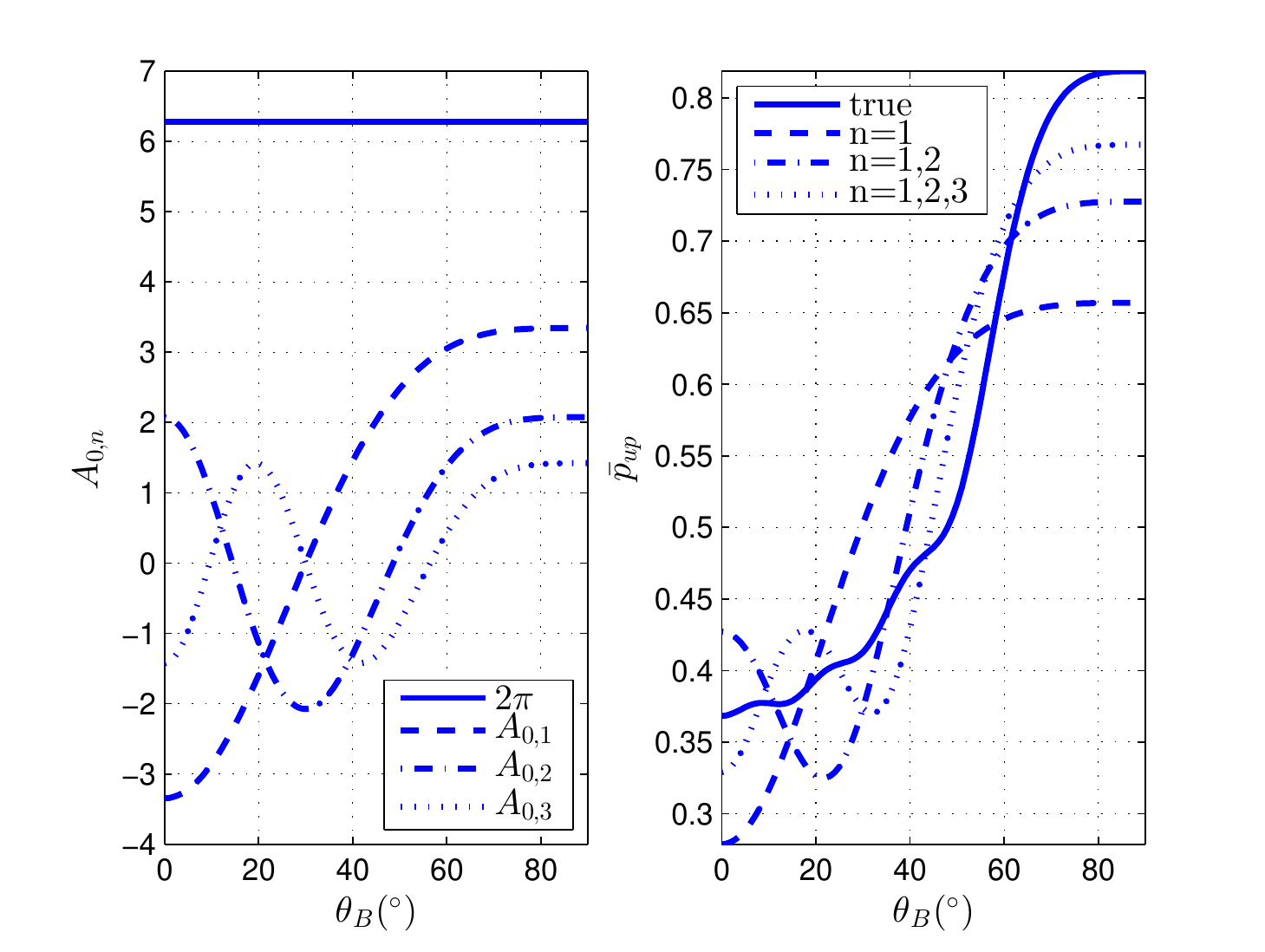}
\caption{Left plot: $A_{0,n}$ versus $\theta_B$. Right plot: true value and approximations of $\bar{p}_{up}$ versus $\theta_B$. $N=8$, $\Delta d=0.5\lambda$. $P_t/\sigma_n^2=40$\,dB, $R_B=3.4594$\,bps/Hz, $R_s=1$\,bps/Hz, $\lambda_e=1\times10^{-4}$}
\label{fig:chp3_p_DoE_De_ULA}
\end{figure}

In the right plot, the true value of $\bar{p}_{up}$ is shown as comparison to its different approximations.
When $n=1$, the approximated $\bar{p}_{up}$ in (\ref{eq:chp3_p_De_approx_betais2}) is comprised of the constant $2\pi$ and $A_{0,1}$;
when $n=1,2$, the approximated $\bar{p}_{up}$ in (\ref{eq:chp3_p_De_beta_is_2}) and (\ref{eq:chp3_A_0L_2}) is comprised of the constant $2\pi$, $A_{0,1}$ and $A_{0,2}$, and so forth.
It can be seen in Fig.\,\ref{fig:chp3_p_DoE_De_ULA} that when $n=1$, the approximation already captures the increasing trend of the true value.
With more $A_{0,n}$, the approximation becomes closer to the true value.
It is worth noticing from Fig.\,\ref{fig:chp3_p_DoE_De_ULA} that for $n>2$, $\cos(n\pi\sin\theta_B)$ (i.e., $A_{0,n}$) is not monotonic in the range $\theta_B\in[0,\frac{\pi}{2}]$.
However, when $n>2$, $\frac{N-n}{N}J_0(n\pi)$ is less dominant than $\frac{N-1}{N}J_0(n\pi)$ for $N=8$. Overall, the true value of $\bar{p}_{up}$ in general has a monotonic increasing relationship with $A_{0,1}$.

\subsubsection{Impact of Number of Elements}

While there exist simple approximations when $N$ is fixed, it is more complicated when $N$ changes.
Because when $N$ is fixed and only $\theta_B$ changes, there is a fixed envelope for $A_{0,n}$, $n=1,...N-1$.
When $N$ changes, the number of summation terms, i.e., $A_{0,n}$, as well as their envelopes also change.
Therefore, a different method is adopted to analyze how $A_0$ changes with $N$.

Let $\{q_n\}$ be a sequence, $n\in\mathbb{N}^+$,
\begin{align}
	q_n = J_0(k\Delta dn)\cos(k\Delta dn\sin\theta_B).
\end{align}
Notice that $\{q_n\}$ is an infinite sequence and is independent from $N$. 
$A_0$ in (\ref{eq:chp3_A_0L}) is the summation of the first $N-1$  terms of $\{q_n\}$ with weights and the constant $2\pi$,
\begin{align}\label{eq:chp3_A_0L_3}
	A_0=2\pi+4\pi\sum_{n=1}^{N-1}\frac{N-n}{N}q_n.
\end{align}
The weights are $\frac{N-n}{N}$ and $N$ can be any positive integer larger than 1.

\nomenclature{$\{q_n\}$}{series used for $A_0$}

Examples of $\{q_n\}$ when $\Delta d=0.5\lambda$ are shown in Fig.\,\ref{fig:chp3_besselj_ULA_2}.
There are three values of $\theta_B$.
For different value $\theta_B$, the behavior of $\{q_n\}$ differs greatly.
When $\theta_B=0^{\circ}$, $q_n=J_0(n\pi)$.
So, $\{q_n\}$ are discrete samples of $J_0(x)$.
When $\theta_B=30^{\circ}$, $q_n=J_0(n\pi)\cos(\frac{n\pi}{2})$, which is zero for odd $n$; and it is the samples of $(-1)^{n/2}J_0(x)$ for even $n$.
When $\theta_B=60^{\circ}$, $q_n=J_0(n\pi)\cos(\frac{n\sqrt{3}\pi}{2})$.

\begin{figure}
\centering
\includegraphics[scale=0.9]{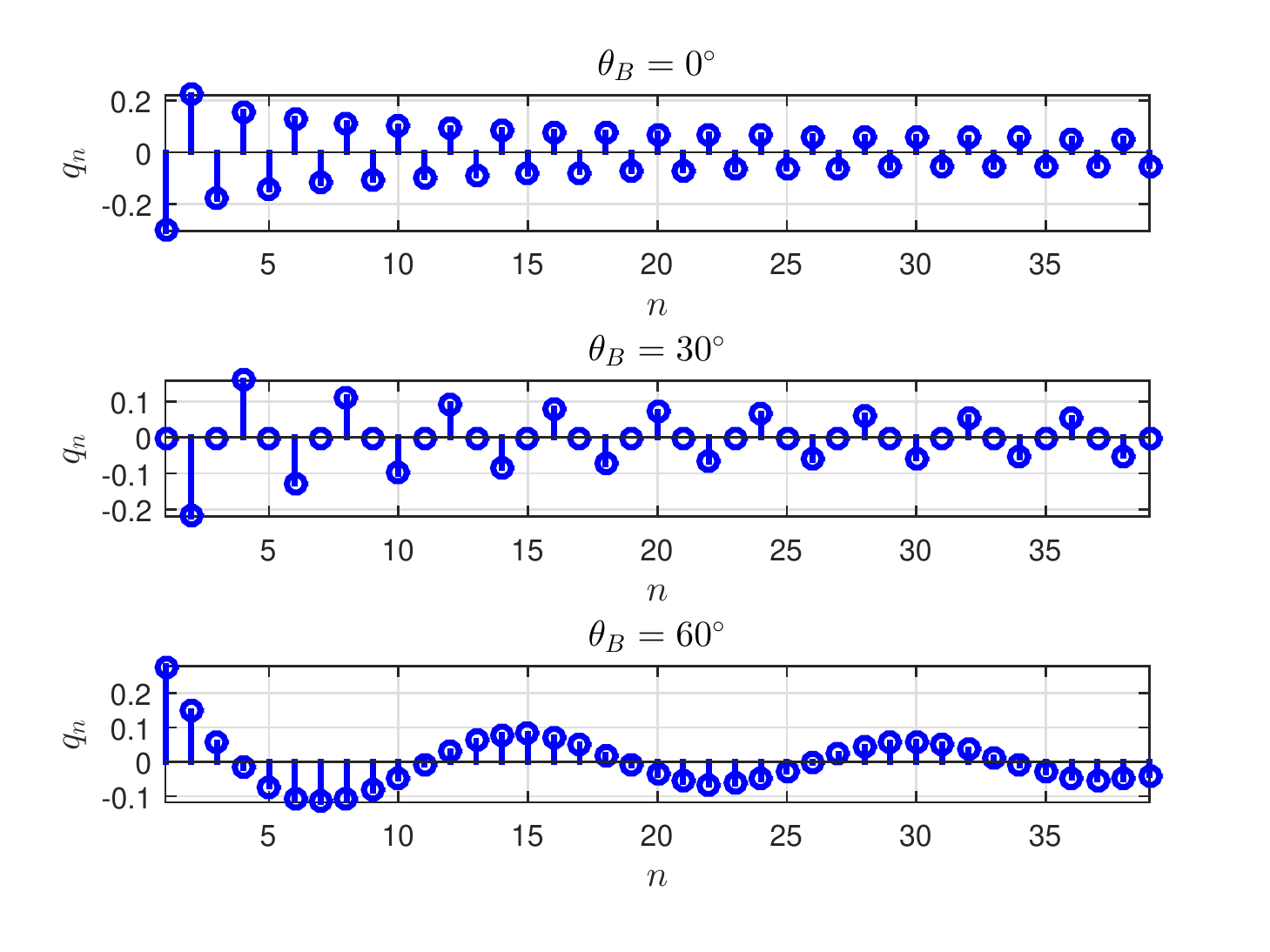}
\caption{$q_n$ for different $\theta_B$}
\label{fig:chp3_besselj_ULA_2}
\end{figure}

Take $\theta_B=0^{\circ}$ as an example.
When $N=2,3,4$, it can be derived that 
\begin{align}
	A_0&=2\pi+4\pi \frac{1}{2}q_1 =  4.3718, \\
	A_0&=2\pi+4\pi \frac{2}{3}q_1+4\pi \frac{1}{3}q_2 = 4.6575, \\
	A_0&=2\pi+4\pi \frac{3}{4}q_1+4\pi \frac{2}{4}q_2 +4\pi \frac{1}{4}q_3 = 4.2311,
\end{align}
where $q_1=-0.3042$, $q_2=0.2203$, $q_3=-0.1812$.
As $N$ increases, the weights for the first few terms becomes more significant, which indicates the changing envelope.
For example, the weight of $q_1$ increases from $\frac{1}{2}$ to $\frac{3}{4}$.
But, there is no clear increasing or decreasing relationship when $N$ increases from $2$ to $4$.
It depends on the specific values of $q_n$.

The conclusion can be generalized for any $\theta_B$.
When $N$ increases, although the number of summation terms increases with $N$, and the weights $\frac{N-n}{N}$ for smaller $n$ become more significant, the overall summation $A_0$ is still determined by the nature of $\{q_n\}$, which is in turn determined by $\theta_B$.
In Fig.\,\ref{fig:chp3_p_N_De_ULA}, the corresponding results for $\bar{p}_{up}$ versus $N$ are shown.
It can be seen that because of the difference in $\{q_n\}$, when $N$ increases, $\bar{p}$ changes differently for different $\theta_B$.

\begin{figure}
\centering
\includegraphics[scale=0.9]{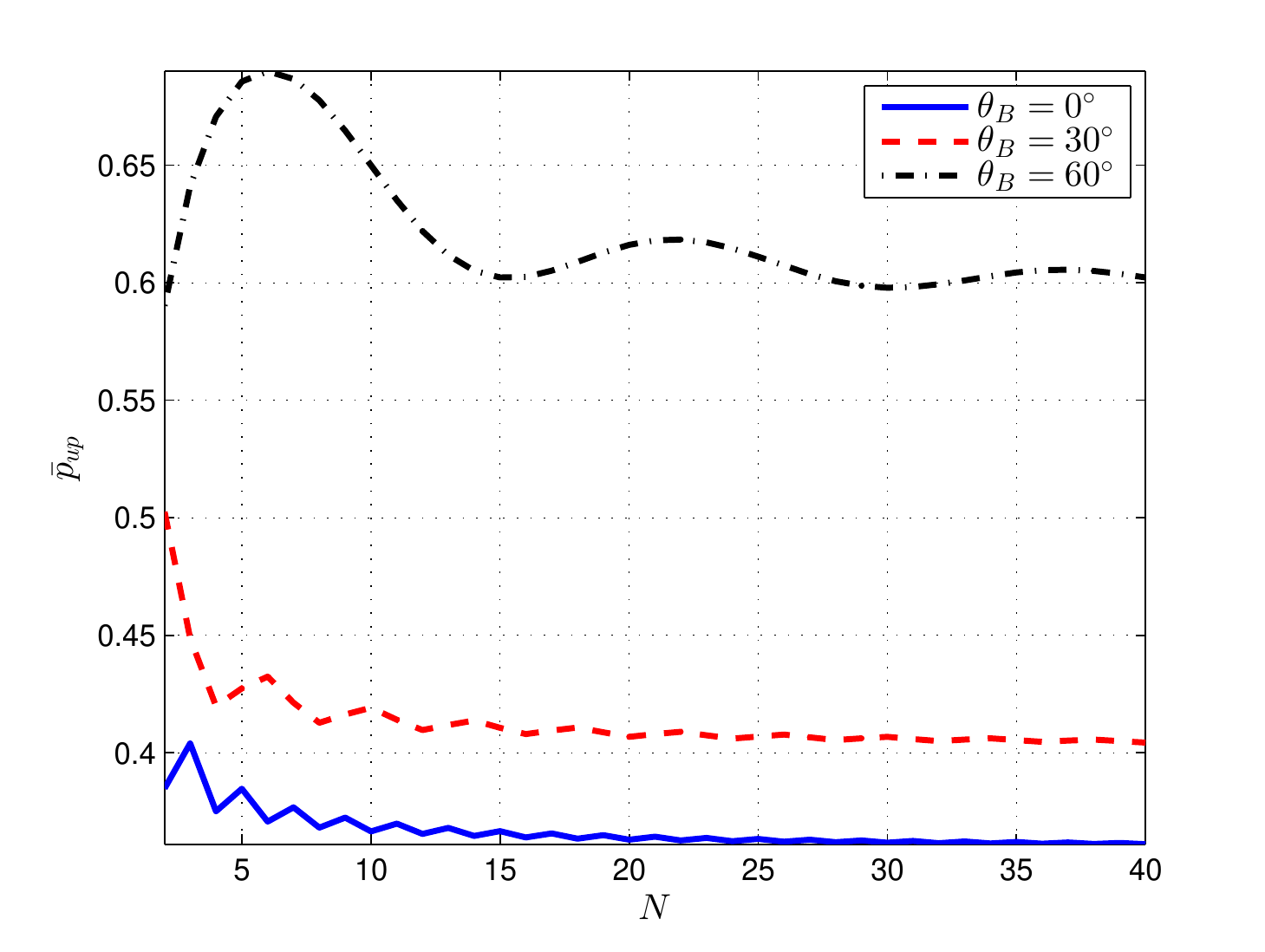}
\caption{$\bar{p}_{up}$ versus $N$ for different $\theta_B$. $\Delta d=0.5\lambda$, $P_t/\sigma_n^2=40$\,dB, $R_B=3.4594$\,bps/Hz, $R_s=1$\,bps/Hz, $\lambda_e=1\times10^{-4}$}
\label{fig:chp3_p_N_De_ULA}
\end{figure}

The common behavior shared by $\bar{p}$ for any $\theta_B$ is that when $N$ is sufficiently large, $\bar{p}$ approaches to a fixed value.
This can also be traced back to a property of $\{q_n\}$.
The sequence $\{q_n\}$ is comprised of weighted samples of $J_0(x)$, which approaches $0$ when $x$ goes to infinity.
As $N$ increases, the weights at the end of $\{q_n\}$ approach $0$ in addition to the vanishing tails of $J_0(x)$.
Thus, the total summation is more determined by the front terms of $\{q_n\}$.
In addition, it can be seen in Fig.\,\ref{fig:chp3_p_N_De_ULA} that $\bar{p}$ increases with $\theta_B$.

In summary, when $\Delta d=0.5\lambda$, $\bar{p}_{up}$ (i.e., $A_0$) in general increases with $\theta_B\in[0,\frac{\pi}{2}]$.
The relationship between $\bar{p}_{up}$ and $N$ is determined by $\theta_B$.
Nevertheless, for any $\theta_B$, $\bar{p}_{up}$ approaches to certain values (depending on $\theta_B$) when $N$ increases.

\subsection{Impact of Array Parameters on Array Pattern}
\label{chp3:analysis:jfeow}

In the previous section, $\bar{p}_{up}$ is analyzed via the pattern area $A_0$.
This is because $\bar{p}_{up}$ is based on $\Theta$, which is determined by $G(\theta,\theta_B)$.
Naturally, the area of $\Theta$ (i.e., $A$) is determined by the area of the array pattern, which is $A_0$.
To fully understand the properties of $\bar{p}_{up}$ in Section\,\ref{chp3:analysis:sve}, the array pattern of  the ULA is studied and $\bar{p}_{up}$ is analyzed from the spatial aspect with respect to $N$ and $\theta_B$.
Although the spatial perspective does not provide a very accurate analysis, it does helps better understanding of the properties of $\bar{p}_{up}$.

First, some array patterns with different $N$ and $\theta_B$ are shown in Fig.\,\ref{fig:chp3_patterns_ULA},  where $\Delta d=0.5\lambda$.
Some patterns have different $N$ with the same $\theta_B$; while others have different $\theta_B$ with the same $N$.

\begin{figure}
\centering
\includegraphics[scale=0.9]{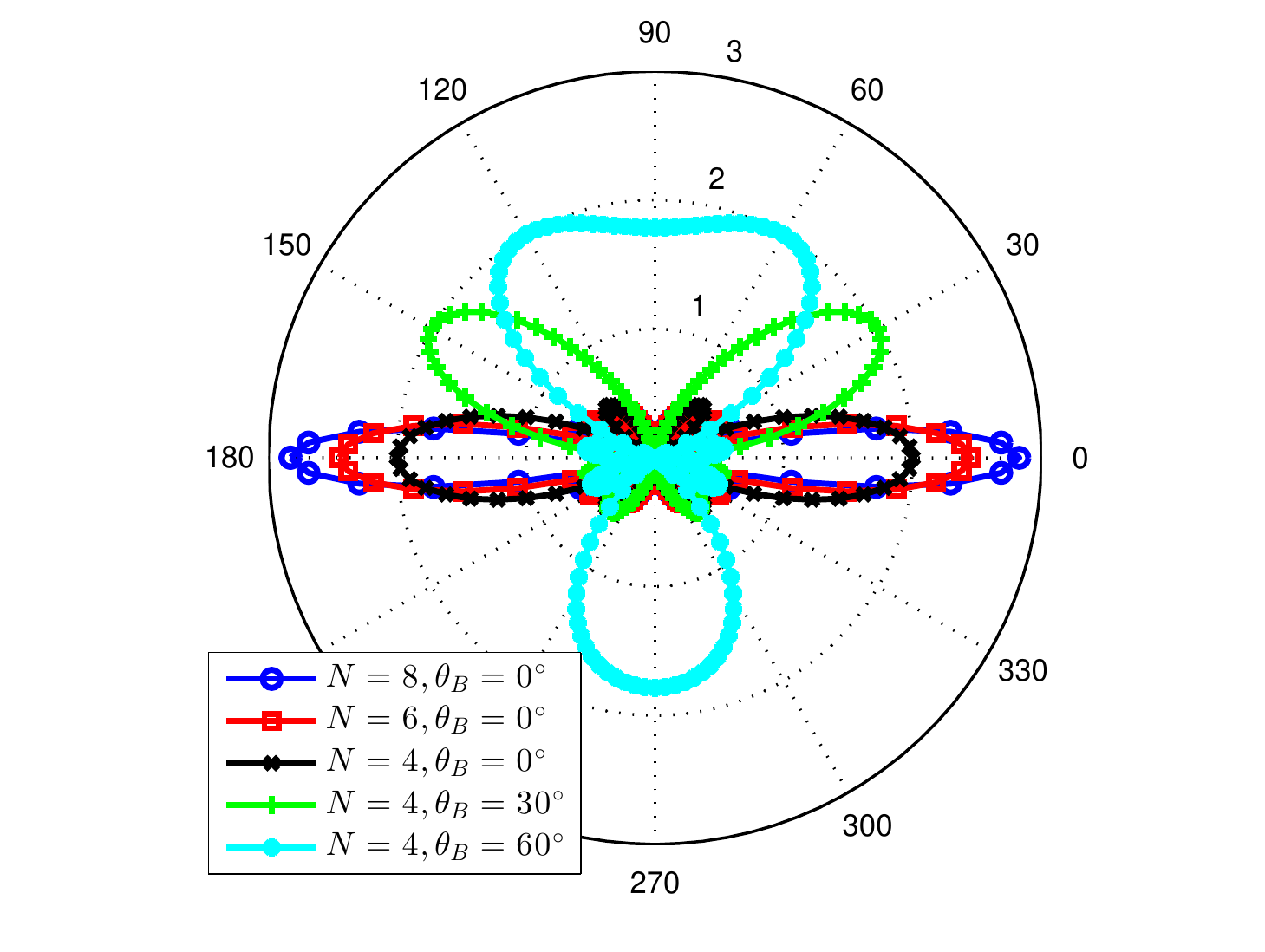}
\caption{Array patterns of ULA for different $N$ and $\theta_B$, $\Delta d=0.5\lambda$}
\label{fig:chp3_patterns_ULA}
\end{figure}

As can be seen in Fig.\,\ref{fig:chp3_patterns_ULA}, for the same $N=4$, the mainbeam becomes wider as $\theta_B$ increases from $0^{\circ}$ to $60^{\circ}$.
For $\theta_B=0^{\circ}$, the mainbeam gets narrower but longer when $N$ increases from $4$ to $8$.
According to (\ref{eq:chp2_max_gain}), $G_{\text{max}}=\sqrt{N}$.
Thus, the length of the mainbeam increases in order of $\sqrt{N}$.

The half-power beamwidth (HPBW) is used to measure the width of the mainbeam. 
Denoted by $\Delta \theta_{HP}$, it is the angular separation between the half power points of the mainbeam. 
Based on its definition, $G(\theta_B-\Delta \theta_{HP}/2,\theta_B)=G(\theta_B+\Delta \theta_{HP}/2,\theta_B)=\sqrt{\frac{N}{2}}$.
For $\theta_B\in[0,\frac{\pi}{2}]$, the HPBW for ULA can be
calculated by\,\cite{adaptivearraysystems}
\begin{align}\label{eq:chp3_HPBW_L}
	\Delta\theta_{HP}= 2\Big[\theta_B-\arcsin\Big(\sin\theta_B-\frac{2.782}{Nk\Delta d}\Big)\Big].
\end{align}

\nomenclature{HPBW}{half-power beamwidth}
\nomenclature{$\Delta \theta_{HP}$}{HPBW}

It can be seen that $\Delta\theta_{HP}$ is jointly determined by $N$ and $\theta_B$.
$\arcsin(\cdot)$ is a monotonically increasing function.
For certain $\theta_B$, when $N$ increases, $\frac{2.782}{Nk\Delta d}$ decreases, leading to the decrease of $\Delta\theta_{HP}$.

To study the relationship between $\Delta\theta_{HP}$ in (\ref{eq:chp3_HPBW_L}) and $\theta_B$.
The derivative of $\Delta\theta_{HP}$ is calculated, 
\begin{align}
	\frac{\partial}{\partial\theta_B}\Delta\theta_{HP}=2-\frac{2\cos\theta_B}{\sqrt{1-(\sin\theta_B-\frac{2.782}{Nk\Delta d})^2}}.
\end{align}

\begin{theorem}\label{th:chp3_HPBW_L}
When $\Delta d=0.5\lambda$,
\begin{align}
	\frac{\partial}{\partial\theta_B}\Delta\theta_{HP}=
	\begin{cases}
	<0 &  \theta_B\in[0,\arcsin\frac{1.391}{N\pi}) \\
	=0 &  \theta_B=\arcsin\frac{1.391}{N\pi} \\
	>0 &  \theta_B\in(\arcsin\frac{1.391}{N\pi},\frac{\pi}{2}]
	\end{cases}
\end{align}
\end{theorem}
The proof of Theorem\,\ref{th:chp3_HPBW_L} is in Appendix\,\ref{appdx:bessel:oitor}.
Theorem\,\ref{th:chp3_HPBW_L} suggests that the turning point is $\theta_B=\arcsin\frac{1.391}{N\pi}$.
$\Delta\theta_{HP}$ first decreases till the turning point, then increases till $\frac{\pi}{2}$.

Examples of $\Delta\theta_{HP}$ versus $\theta_B$ for different $N$ are shown in Fig.\,\ref{fig:chp3_HPBW_DoE_ULA}.
The turning point of $\Delta\theta_{HP}$ are at $6.36^{\circ}$, $4.23^{\circ}$ and $3.17^{\circ}$ for $N=4,6,8$, respectively, all of which are relatively small compared to the whole angle range.
Thus, $\Delta\theta_{HP}$ in general increases in the whole range $\theta_B\in[0,\frac{\pi}{2}]$.
In addition, given the same $\theta_B$, the larger $N$, the smaller $\Delta\theta_{HP}$ is.

\begin{figure}
\centering
\includegraphics[scale=0.9]{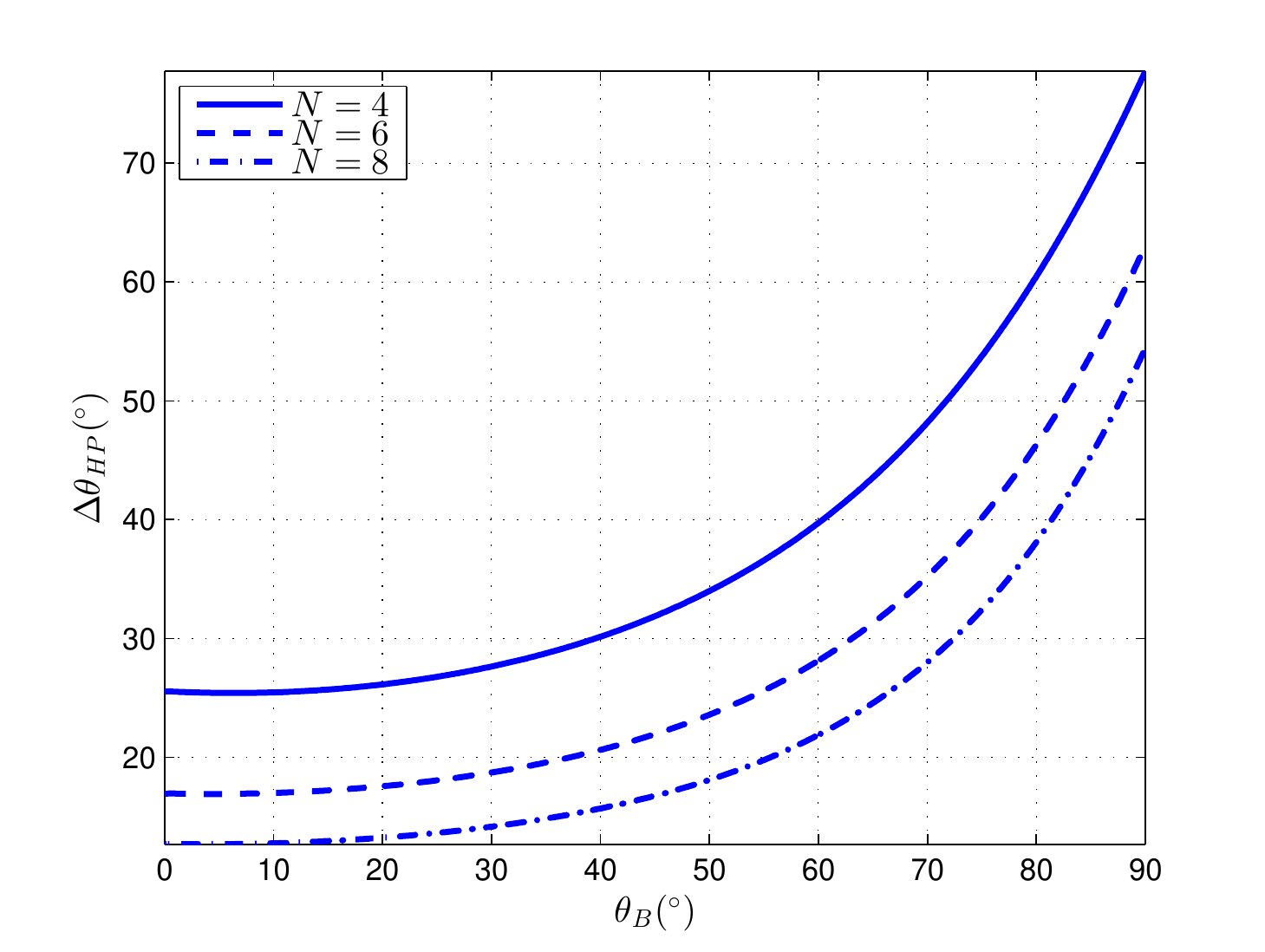}
\caption{$\Delta\theta_{HP}$ versus $\theta_B$ for different $N$, $\Delta d=0.5\lambda$}
\label{fig:chp3_HPBW_DoE_ULA}
\end{figure}

When $N$ is fixed, $G_{\text{max}}$ is the same for different $\theta_B$, which means the length of the mainbeam is fixed.
As $\theta_B$ increases from $0^{\circ}$ to $90^{\circ}$, the mainbeam in general becomes wider. 
Therefore, the mainbeam area becomes larger.
For $\Delta d=0.5\lambda$, the sidelobes are less dominant compared to the mainbeam.
Thus, the mainbeam contributes to the majority part of the pattern area, which explains why $A_0$ in general increases in the range $\theta_B\in[0,\frac{\pi}{2}]$.

When $\theta_B$ is fixed to a certain value, $G_{\text{max}}$ increases with $N$, which means the mainbeam becomes longer.
On the other hand, as $N$ increases, the mainbeam becomes narrower.
Thus, the area of the mainbeam could either increase or decrease  with $N$, resulting in a complex relationship between $A_0$ and $N$.

\section{Numerical Results for Generalized Rician Channel Model}
\label{chp3:result}
Both $\bar{p}$ and $\bar{p}_{up}$ are determined by two group of parameters: channel parameters ($K$ and $\beta$) and array parameters ($N$ and $\theta_B$).
In Section\,\ref{chp3:analysis}, the simple case that $\bar{p}=\bar{p}_{up}$ is analyzed with respect to the array parameters.
In this section, the numerical results are used to analyze the proprieties of $\bar{p}$ and $\bar{p}_{up}$ for the generalized Rician channel (i.e., any $K$ and $\beta$) with respect to the array parameters.
In addition, the tightness of the upper bound will be examined.

\subsection{SSOP and Its Upper Bound}
\label{chp3:result:wier}

In (\ref{eq:chp3_meanSSOP_up_3}), $\bar{p}_{up}$ is positively correlated with $\Big[\frac{c_0K}{2\pi(K+1)}A_0+\frac{c_0}{K+1}\Big]^{\frac{2}{\beta}}$.
For any fixed $\beta$ and $K$, $\bar{p}_{up}$ also has a positive relationship with $A_0$. 
Thus, the conclusions that are reached about $A_0$ regarding to the impact of $N$ and $\theta_B$ also apply to $\bar{p}_{up}$ of the generalized Rician channel with different $\beta$ and $K$.

For convenience, let $A_1$ denote $\frac{c_0K}{2\pi(K+1)}A_0+\frac{c_0}{K+1}$.
When $\beta$ increases from $2$ to $6$, $A_1^{\frac{2}{\beta}}$ decreases, because $A_1$ is generally larger than $1$.
It is also noticed that when $A_0=2\pi$, the $K$ factor disappears in the equation, i.e, $A_1=c_0$.
When $A_0<2\pi$, the larger $K$ is, the smaller $A_1$ (i.e., $\bar{p}_{up}$) is; 
when $A_0>2\pi$, the larger $K$ is, the larger $A_1$ (i.e., $\bar{p}_{up}$) is.

It can be seen that the properties of $\bar{p}_{up}$ with respect to $K$ and $\beta$ mainly relies on the value of $A_0$.
Thus, in Fig.\,\ref{fig:chp3_p_up_K_beta}, the examples of $\bar{p}_{up}$ for different values of $K$ and $\beta$ are given for three typical values of $A_0$, i.e., $4.1326$, $2\pi$ and $15.3761$, which corresponds to $\theta_B=0^{\circ}$, $48.35^{\circ}$ and $90^{\circ}$ when $N=8$.

\begin{figure}
\centering
\includegraphics[scale=0.9]{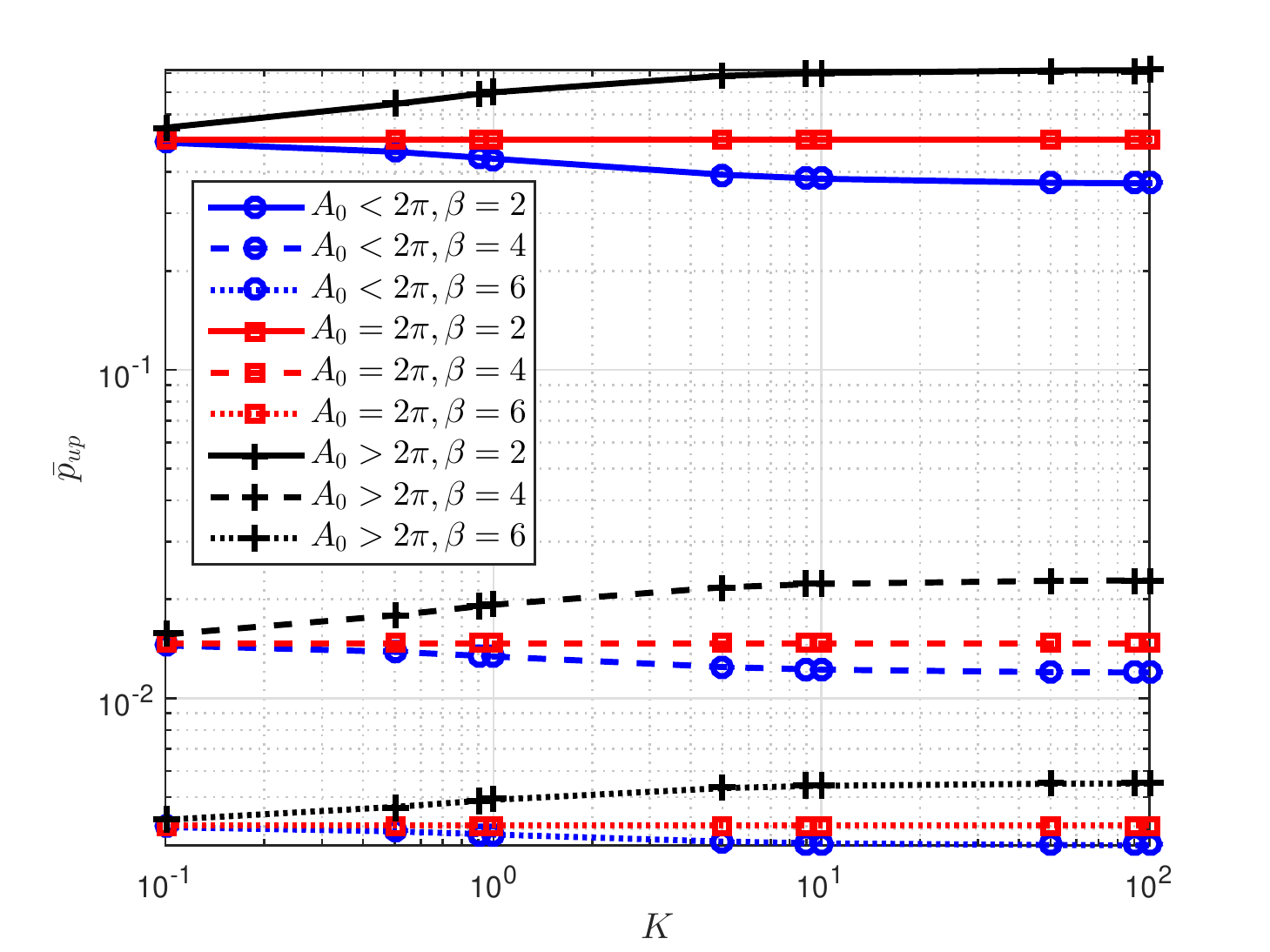}
\caption{$\bar{p}_{up}$ for different values of $A_0$, $K$ and $\beta$. $P_t/\sigma_n^2=40$\,dB, $R_B=3.4594$\,bps/Hz, $R_s=1$\,bps/Hz, $\lambda_e=1\times10^{-4}$}
\label{fig:chp3_p_up_K_beta}
\end{figure}

In Fig.\,\ref{fig:chp3_p_up_K_beta}, the logarithm scale is used to clearly show the ranges of $\bar{p}_{up}$ and $K$.
It can be seen that, when $\beta$ increases, $\bar{p}_{up}$ drops very quickly, because in this case, the constant $c_0$ is very large (i.e., $2.2222\times 10^3$).
For fixed $\beta$, $\bar{p}_{up}$ increases, stays unchanged or decreases depending on the value of $A_0$.

The range of $K$ in linear scale is from $0.01$ to $100$.
When $K=0.01$, the Rician channel approaches the Rayleigh channel ($K=0$).
When $K=100$, the Rician channel approaches the deterministic channel ($K\to\infty$).
It can be seen that for fixed $\beta$, $\bar{p}_{up}$ is a constant for $K=0$ and is irrelevant to $A_0$ (nor $N$, $\theta_B$).
When $K>10$, $\bar{p}_{up}$ approaches to a certain value that depends on $A_0$ which in turn depends on $N$ and $\theta_B$.

\begin{figure}
\centering
\includegraphics[scale=0.9]{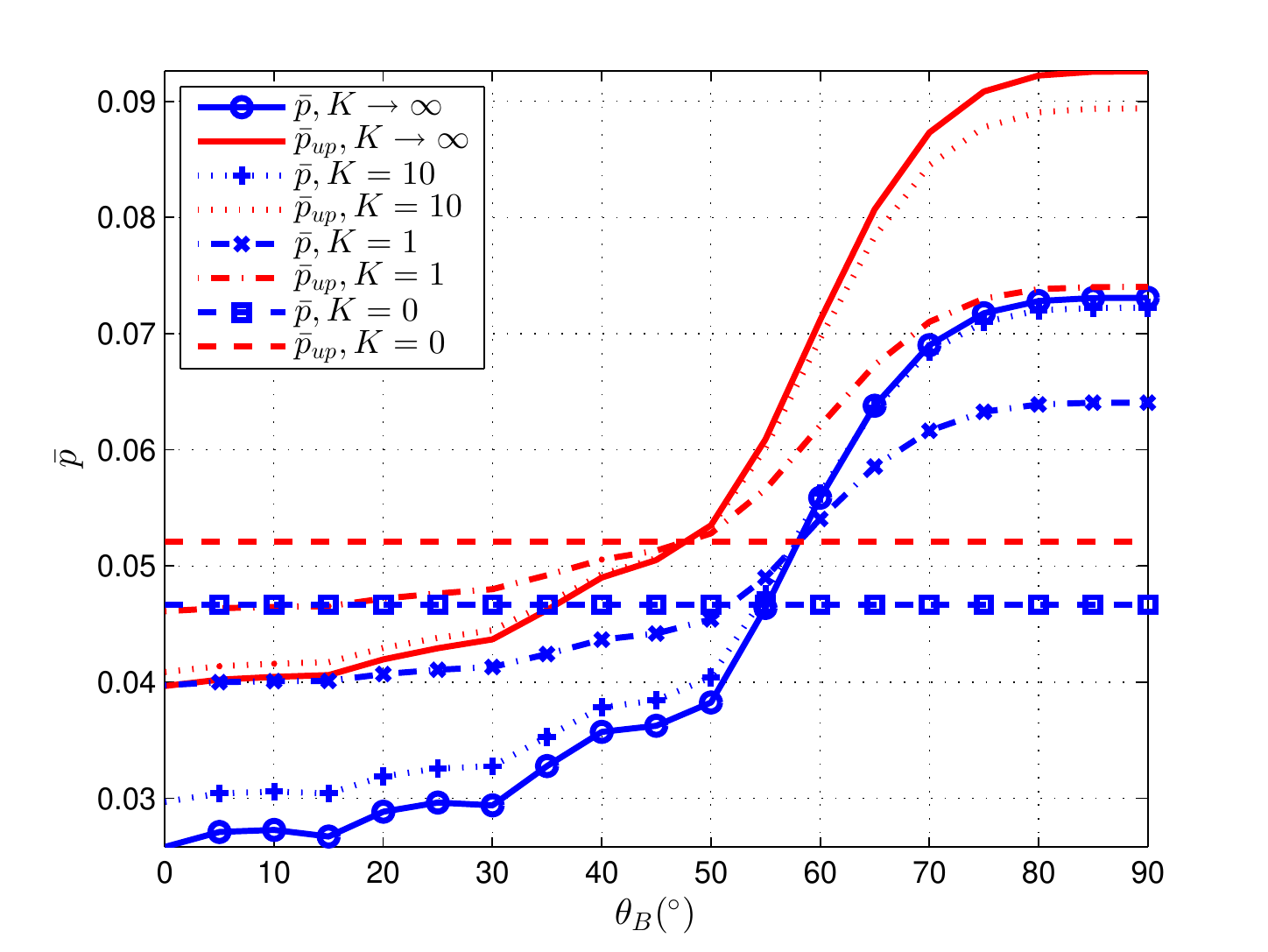}
\caption{$\bar{p}$ and $\bar{p}_{up}$ versus $\theta_B$ for different $K$. $\beta=3$, $N=8$, $\Delta d=0.5\lambda$. $P_t/\sigma_n^2=40$\,dB, $R_B=3.4594$\,bps/Hz, $R_s=1$\,bps/Hz, $\lambda_e=1\times10^{-4}$}
\label{fig:chp3_p_and_bounds_DoE_beta_3_L}
\end{figure}

While it is relatively straightforward to analyze the properties of $\bar{p}_{up}$ with respect to $(K,\beta, N,\theta_B)$, the properties of $\bar{p}$ cannot be easily analyzed according to (\ref{eq:chp3_meanSSOP_Ri_2}).
Thus, numerical results are used.
An example of $\bar{p}$ and $\bar{p}_{up}$ versus $\theta_B$ for $\beta=3$ and $N=8$ is given in Fig.\,\ref{fig:chp3_p_and_bounds_DoE_beta_3_L}.
$\beta=3$ is a typical value for some indoor scenarios such as home and factory\,\cite{goldsmith2005wireless}.
Typical values of $K$ are chosen as 0, 1, 10 and $\infty$.

It can be seen that $\bar{p}$ and $\bar{p}_{up}$ increase in the range $\theta_B\in[0,\frac{\pi}{2}]$, except for $K=0$.
When $K=0$, the curves are flat because $\bar{p}$ and $\bar{p}_{up}$ are irrelevant to $\theta_B$, according to (\ref{eq:chp3_meanSSOP_Ra}) and (\ref{eq:chp3_SSOP_Ra_up}).
By comparing $\bar{p}_{up}$ and $\bar{p}$, it can be observed that the upper bound reflects the trend very well.
It can also be seen that for both $\bar{p}$ and $\bar{p}_{up}$, the curve for $K=10$ is closer to that for $K\to\infty$, while the curve for $K=1$ is closer to that for $K=0$.

In Fig.\,\ref{fig:chp3_p_and_bounds_DoE_beta_3_L}, there is a pivot point at $\theta_B=48.35^{\circ}$ where $A_0=2\pi$, and all curves of $\bar{p}_{up}$ come across.
When $\theta_B<48.35^{\circ}$, $A_0<2\pi$; thus $\bar{p}_{up}$ decreases with $K$;
When $\theta_B>48.35^{\circ}$, $A_0>2\pi$; thus $\bar{p}_{up}$ increases with $K$;
As $K$ changes, the curve of $\bar{p}_{up}$ pivots around this point and approaches the curves for $K=0$ or $K\to\infty$.

For completeness, Fig.\,\ref{fig:chp3_p_and_bounds_N_beta_3_L} shows an example of $\bar{p}$ and $\bar{p}_{up}$ versus $N$ for $\beta=3$ and $\theta_B=0^{\circ}$.
It can be seen that $\bar{p}$ and $\bar{p}_{up}$ decrease to different floor levels depending on $K$.
The same behavior has been shown in Fig.\,\ref{fig:chp3_p_N_De_ULA} where $K=\infty$ and $\beta=2$.
In addition, for both $\bar{p}$ and $\bar{p}_{up}$, the curves for $K=10$ are closer to those for $K\to\infty$, while the curves for $K=1$ are closer to those for $K=0$.

\begin{figure}
\centering
\includegraphics[scale=0.9]{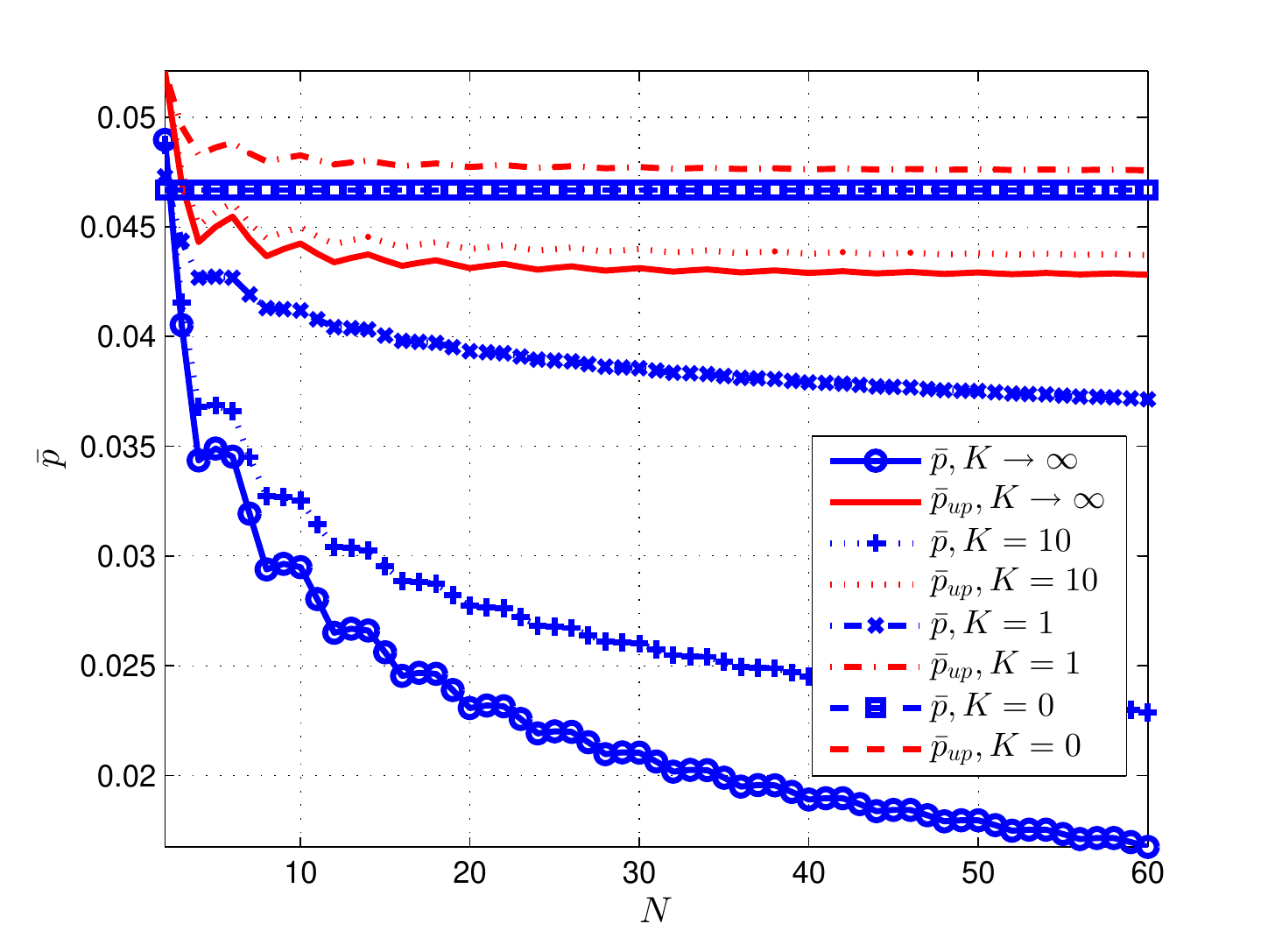}
\caption{$\bar{p}$ and $\bar{p}_{up}$ versus $N$ for different $K$. $\beta=3$, $\theta_B=0^{\circ}$, $\Delta d=0.5\lambda$. $P_t/\sigma_n^2=40$\,dB, $R_B=3.4594$\,bps/Hz, $R_s=1$\,bps/Hz, $\lambda_e=1\times10^{-4}$}
\label{fig:chp3_p_and_bounds_N_beta_3_L}
\end{figure}

Although both $\bar{p}$ and $\bar{p}_{up}$ decreases with $N$, $\bar{p}$ converges with a much slower speed, which causes a larger gap between $\bar{p}$ and $\bar{p}_{up}$ as $N$ increases.
The differences of the gaps between $\bar{p}$ and $\bar{p}_{up}$ for different $K$ is not very obvious in Fig.\,\ref{fig:chp3_p_and_bounds_DoE_beta_3_L}, because $N$ is small.

In summary, the properties of $A_0$ with respect to $N$ and $\theta_B$ can be extended to $\bar{p}_{up}$, because there is a straightforward relationship between $\bar{p}_{up}$ and $A_0$ for any $K$ and $\beta$.
On the other hand, the numerical results show that while $\bar{p}$ has similar properties to $A_0$ with respect to $N$ and $\theta_B$, the gaps between $\bar{p}$ and $\bar{p}_{up}$ increase as $N$.
Therefore, in the next section, the tightness of $\bar{p}_{up}$ will be examined.

\subsection{Tightness of Upper Bound}
\label{chp3:result:mnbv}

In this section, the tightness of the upper bound is examined via numerical results with respect to $(K,\beta, N,\theta_B)$.
An example of $\bar{p}$ and $\bar{p}_{up}$ for different $K$ and $\beta$ with $N=8$ and $\theta_B=0^{\circ}$ is shown in Fig.\,\ref{fig:chp3_p_and_p_up_K_beta}.
At lower region of $K$, the channel approaches the Rayleigh channel.
Thus, $\bar{p}$ and $\bar{p}_{up}$ converge to the certain values that only depend on $\beta$ according to (\ref{eq:chp3_meanSSOP_Ra}) and (\ref{eq:chp3_SSOP_Ra_up}).
At higher region of $K$, the channel approaches the deterministic channel.
$\bar{p}$ and $\bar{p}_{up}$ converge to the certain values that depend on $\beta$ and $G(\theta,\theta_B)$,  according to (\ref{eq:chp3_SSOP_De}) and (\ref{eq:chp3_SSOP_De_up}).

It can also be seen that when $\beta=2$, the curves for $\bar{p}$ and $\bar{p}_{up}$ emerge as $K$ increases, which corresponds to $\bar{p}=\bar{p}_{up}$ for the deterministic channel.
For other values of $\beta$, as $K$ increases, the gaps between $\bar{p}$ and $\bar{p}_{up}$ increases.

\begin{figure}
\centering
\includegraphics[scale=0.9]{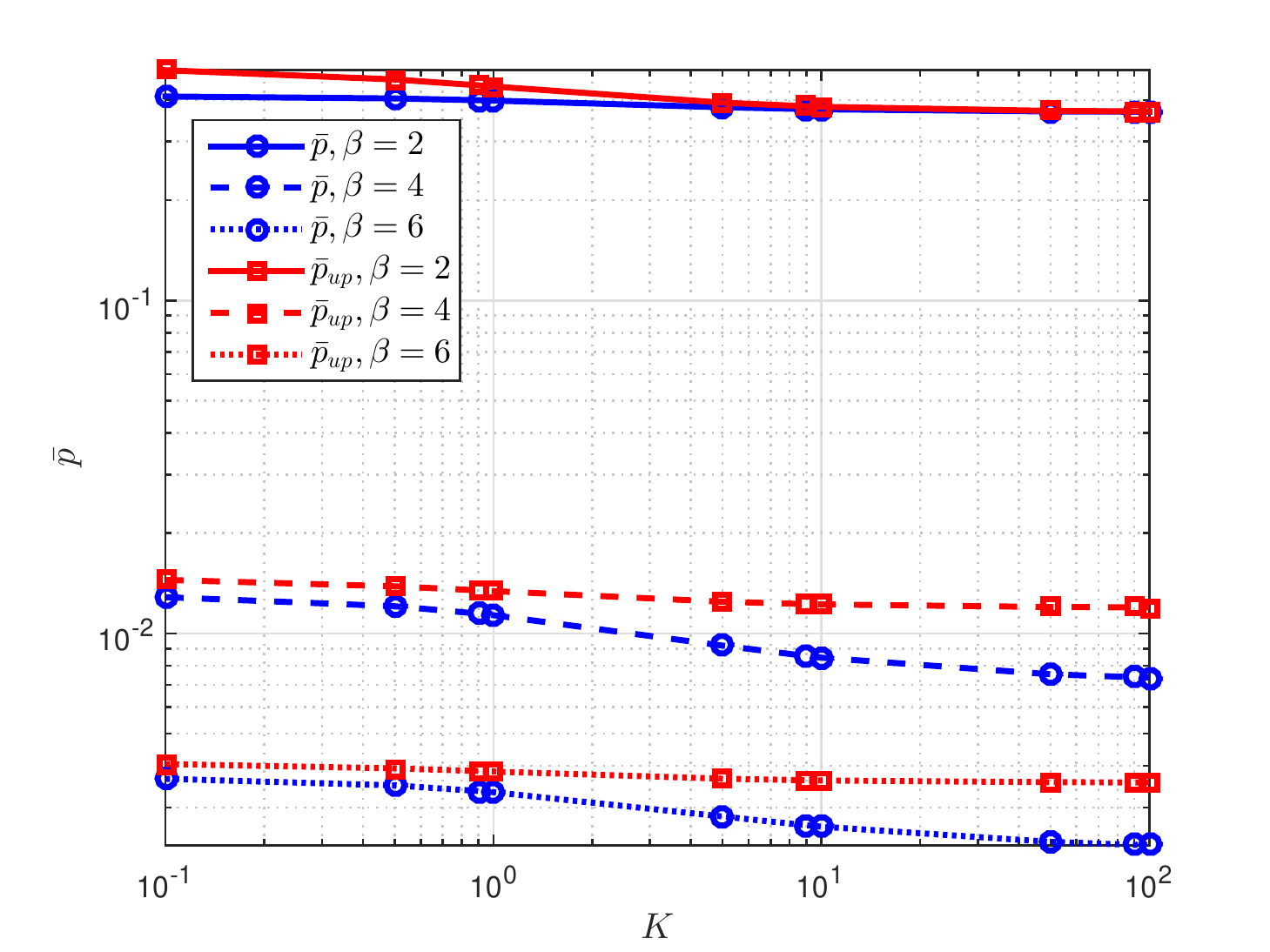}
\caption{$\bar{p}$ and $\bar{p}_{up}$ for different $K$ and $\beta$. $N=8$, $\theta_B=0^{\circ}$, $\Delta d=0.5\lambda$. $P_t/\sigma_n^2=40$\,dB, $R_B=3.4594$\,bps/Hz, $R_s=1$\,bps/Hz, $\lambda_e=1\times10^{-4}$}
\label{fig:chp3_p_and_p_up_K_beta}
\end{figure}

In this section, the ratio between $\bar{p}_{up}$ and $\bar{p}$ is used to measure the tightness of $\bar{p}_{up}$. 
Let $\eta$ denote the ratio,
\begin{align}\label{eq:chp3_eta}
	\eta=\frac{\bar{p}_{up}}{\bar{p}}.
\end{align}
$\eta\geq 1$.
The smaller value of $\eta$, the tighter $\bar{p}_{up}$ is.
In Fig.\,\ref{fig:chp3_p_and_p_up_K_beta}, it can be deduced that when $\beta=2$, $\eta$ will take the maximum value at $K=0$ and approach the minimum value $\eta=1$ at $K\to\infty$.
On the contrary, when $\beta>2$, $\eta$ will take the minimum value at $K=0$ and approach the maximum value at $K\to\infty$.
Thus, in the following, the extreme cases $K=0$ and $K\to\infty$ are used to study the range of $\eta$ for different $N$, $\theta_B$ and $\beta$.

\nomenclature{$\eta$}{ratio of $\bar{p}_{up}$ to $\bar{p}$}

In Fig.\,\ref{fig:chp3_eta_DoE_bounds_ULA}, $\eta$ is plotted against $\theta_B$ for $K=0$ and $K\to\infty$ for all $\beta$.
The ULA has $N=8$ elements and $\Delta d=0.5\lambda$.
For Rayleigh channel, both $\bar{p}$ and $\bar{p}_{up}$ are irrelevant to $\theta_B$, thus $\eta$ is flat across $\theta_B\in[0,90^{\circ}]$. 
For the deterministic channel, when $\beta=2$, $\eta=1$; 
when $\beta>2$, $\eta$ in general decrease with $\theta_B$.

\begin{figure}
\centering
\includegraphics[scale=0.9]{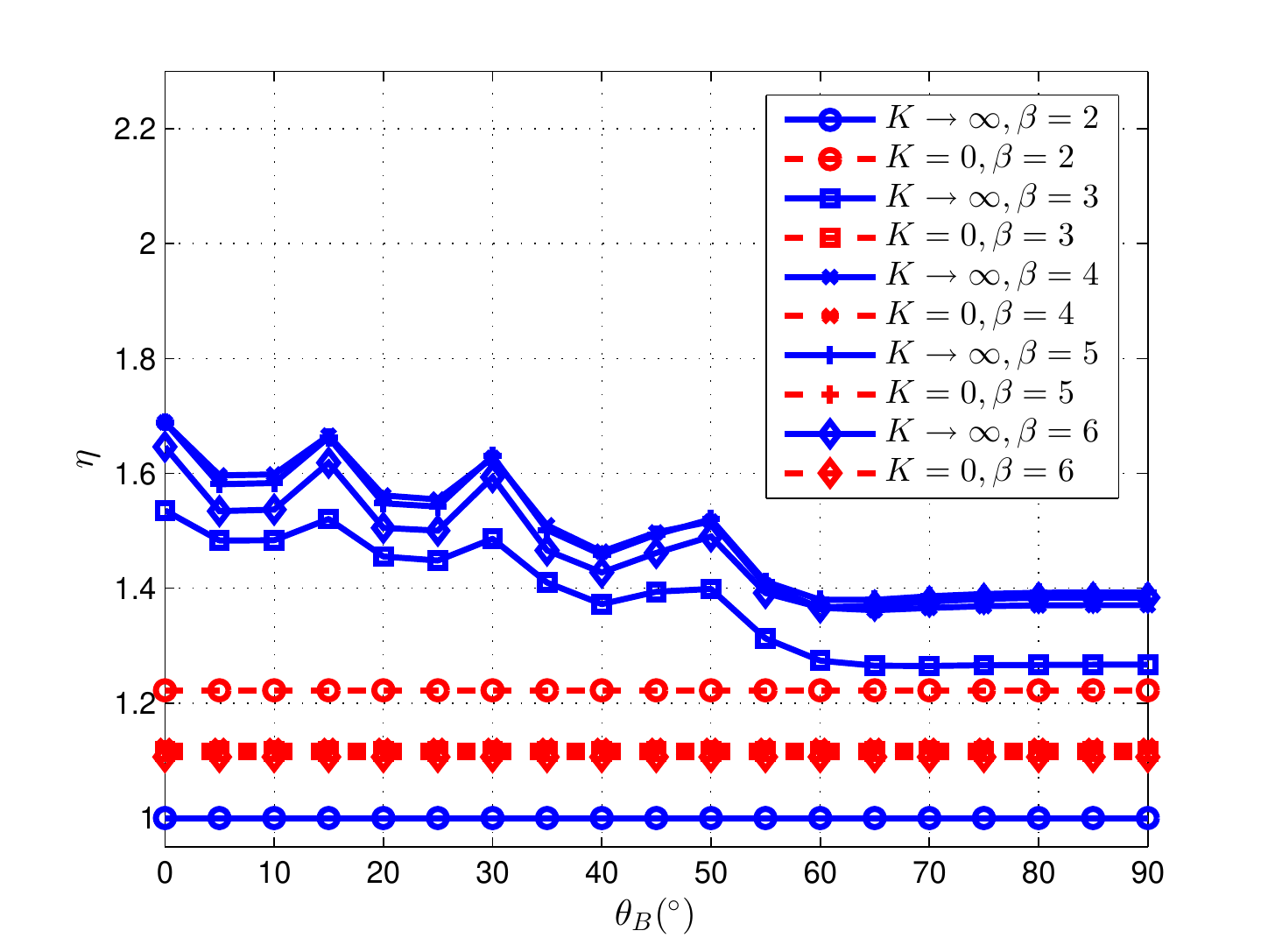}
\caption{$\eta$ versus $\theta_B$ for deterministic and Rayleigh channels for all $\beta$, $N=8$}
\label{fig:chp3_eta_DoE_bounds_ULA}
\end{figure}

Comparing the curves for both the deterministic and the Rayleigh channels, it is noticed that when $\beta>2$, the ratios are located closely in a cluster.
However, there does not exist monotonic relationship between $\eta$ and $\beta$.
For example, when $\beta=6$, $\eta$ for the deterministic channel is smaller than that when $\beta=4$.

In Fig.\,\ref{fig:chp3_eta_N_bounds_ULA}, $\eta$ is plotted against $N$ for $K=0$ and $K\to\infty$ for all $\beta$.
The ULA has $\Delta d=0.5\lambda$ and $\theta_B=0^{\circ}$.
For the Rayleigh channel, $\eta$ is flat across $N$ for all $\beta$.
For the deterministic channel, $\eta$ in general increases with $N$ when $\beta>2$, which verifies the observation from Fig.\,\ref{fig:chp3_p_and_bounds_N_beta_3_L}.
It can be seen that when as $N$ increases, $\eta$ does not converge to a certain value, but increases instead.
This means that in the larger region of $N$, $\bar{p}_{up}$ does not serve the purpose of predicting the behavior of $\bar{p}$ anymore. 

\begin{figure}
\centering
\includegraphics[scale=0.9]{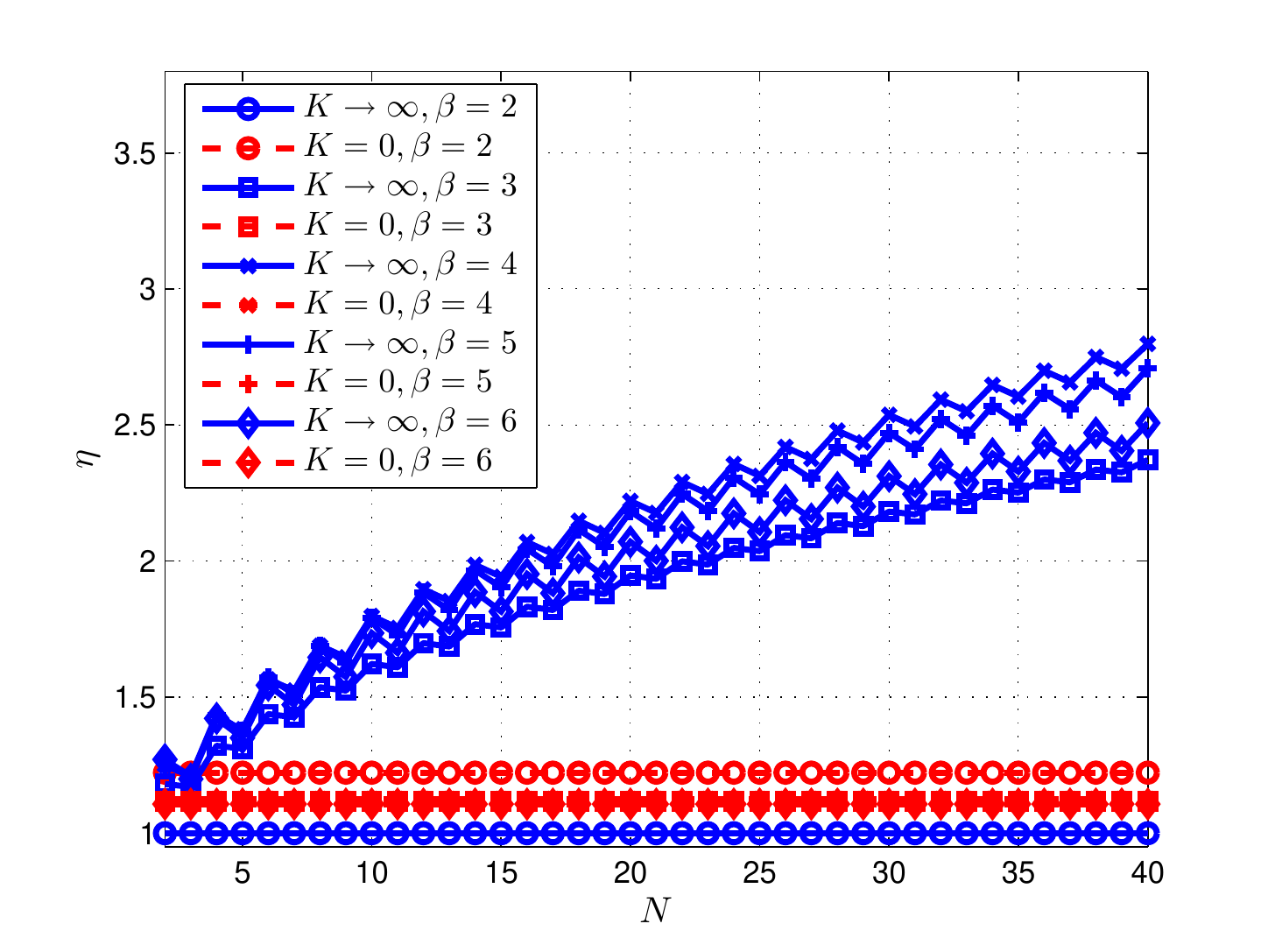}
\caption{$\eta$ versus $N$ for deterministic and Rayleigh channels for all $\beta$, $\theta_B=0^{\circ}$}
\label{fig:chp3_eta_N_bounds_ULA}
\end{figure}

Notice that in both Fig.\,\ref{fig:chp3_eta_DoE_bounds_ULA} and Fig.\,\ref{fig:chp3_eta_N_bounds_ULA}, when $\theta_B$ or $N$ change, $\eta$ fluctuates in a different way than how $\bar{p}_{up}$ changes with $\theta_B$ or $N$ in Fig.\,\ref{fig:chp3_p_DoE_De_ULA} and Fig.\,\ref{fig:chp3_p_N_De_ULA}.
The tightness for each inequality in (\ref{eq:chp3_JI_1}) and (\ref{eq:chp3_JI_2}) depends on the properties of $\bar{p}$.
The combination of two inequalities makes the tightness of upper bound hard to investigate.

In summary, when $\beta=2$, $\eta$ decreases with $K$ till the minimum value $\eta=1$;
when $\beta>2$, $\eta$ increases with $K$ till certain value that depends on $N$ and $\theta_B$, and the values of $\eta$ for different $\beta$ stay in a cluster.
For given $\beta$ and $K$, $\eta$ generally decreases with $\theta_B$ and increases with $N$.
In a lower region of $N$, e.g., $N<10$, the value of $\eta$ is smaller than 2.

\section{Conclusions}
\label{chp3:concl}

In this chapter, the secure transmission to Bob with ER based beamforming in presence of PPP distributed is investigated with a ULA.
The ER is created by beamforming based on the SSO and the physical layer security is quantitatively measured by the SSOP from the spatial aspect.
The concepts of ER and SSOP are applicable to a general array geometry and a general fading channel.

The exact expression of the SSOP is obtained, which can be used for numerical simulations; in the mean time, its analytic upper bound is obtained to facilitate analytical analysis.
Both analytical and numerical results show that the SSOP and its upper bound in general increase with the DoE angle (i.e., Bob's angle) in the range $[0,\frac{\pi}{2}]$ for a ULA with any number of elements and half-wavelength spacing; and they asymptotically approach certain values depending on the DoE angle when the number of elements increases. 
These properties can also be verified by observing array patterns.
The tightness of the upper bound (i.e., the ratio $\eta$) of the SSOP is also examined by numerical results, which shows that $\eta$ has a monotonic relationship with $K$, but a non-linear relationship with $\beta$.
It is worth noticing that $\eta$ increases with the number of elements, which makes it less useful in predicting the behavior of SSOP when the number of elements is very large.

\chapter{Comparison of Spatial Secrecy Outage Probability for Uniform Linear and Circular Arrays}
\label{chp4}

\section{Introduction}
\label{chp4:sec1}

In this chapter, the security performance of the ER-based beamforming with the UCA is studied and compared with the ULA.
Compared to the ULA, the UCA has a symmetric geometry around 360$^{\circ}$ and is more flexible on the choices of array configuration.
Thus, it is more appealing to applications that require a wider range of coverage.
In addition to the theoretical analysis, the practical issue, i.e., the mutual coupling, is examined towards the wireless security.

The expressions for the SSOP and its upper bound in Chapter\,3 are generally applicable to any array type.
Different array geometries should have different behaviors in terms of security.
By studying the secrecy performance of the UCA, a comparison can be made with the ULA, which will provide some insights on choosing an array geometry for different situations to achieve higher level of security.

To investigate the SSOP for the UCA, the system model and the methodology used in Chapter\,3 are reused except for the ULA geometry.
The analytic expressions of the pattern area for the UCA are derived, based on which the SSOP and the tightness of the upper bound for the UCA are evaluated, and the comparison with the ULA with respect to the array parameters is made.

The analysis about the SSOP and its upper bound in Chapter\,3 reveals the important role of the array factor; it determines the shape of the ER and affects the SSOP.
However, the mutual coupling distorts the array factor in practice and has different impact for different array geometries.
In this chapter, a practical beamformer is built on WARP and the  mutual coupling effect is numerically analyzed with WARP experiments and NEC simulation results towards the security performance.

This chapter is organized as follows. 
In Section\,\ref{chp4:sec2}, the system model is briefly introduced with the focus on the UCA and the expressions of the SSOP and its upper bound for the UCA are derived.
In Section\,\ref{chp4:sec3}, the SSOP and its upper bound are analyzed for the UCA with respect to the array parameters; a parallel comparison with the ULA is made.
In Section\,\ref{chp4:sec4}, the theoretical analysis and numerical results are given for the generalized Rician channel.
In Section\,\ref{chp4:sec5}, the mutual coupling is introduced and the conclusions are reached via studying the experimental and simulation results.
In Section\,\ref{chp4:sec6}, the conclusions of this chapter are given.

\section{System Model and SSOP for UCA}
\label{chp4:sec2}
\subsection{System Model with UCA}
\label{chp4:sec2:ownvw}

Consider a dense wireless communications system where the AP communicates to Bob in presence of a large number of Eves.
While the AP is equipped with an antenna array, Bob and Eves have a single antenna.
Users are distributed by a homogeneous PPP $\Phi_e$ with density $\lambda_e$.

The UCA has $N$ elements that are equally allocated on a circle with radius $R$ with spacing $\Delta d$.
An example of UCA is shown in Fig.\,\ref{fig:chp2_UCA}. 
To avoid ambiguity, the subscript `$_L$' and `$_C$' are used to distinguish between the ULA and the UCA hereinafter.
The DoE angle is set to Bob's angle, i.e., $\theta_{\text{doe}}=\theta_B$.
$G_C(\theta,\theta_B)$ can be obtained by substituting $\theta_{\text{doe}}=\theta_B$ into (\ref{eq:chp2_AF_UCA}),
\begin{align} \label{eq:chp4_AF_UCA}
G_C(\theta,\theta_B) =\frac{1}{\sqrt{N}}\sum_{i=1}^N e^{jkR[\cos(\theta_B-\psi_i)-\cos(\theta-\psi_i)]},
\end{align}
where $k=2\pi/\lambda$ and $\psi_i=2\pi(i-1)/N$.

\nomenclature{$_L$}{uniform linear array}
\nomenclature{$_C$}{uniform circular array}

The term `array dimension', denoted by $l_a$, is used to refer to the size of the array.
For the ULA, the array dimension is the array length, i.e., $l_{a,L}=(N-1)\Delta d$; for the UCA, the array dimension is the diameter, i.e., $l_{a,C}=2R$.
The relationship between $\Delta d$ and $l$ for the ULA and the UCA is then given by 
\begin{align}
	\Delta d_L &=\frac{l_{a,L}}{N-1}, \label{eq:chp4_spacing_and_dim_ULA} \\
	\Delta d_C &=l_{a,C}\sin(\frac{\pi}{N}). \label{eq:chp4_spacing_and_dim_UCA}
\end{align}

Unlike the ULA, the configuration of the UCA is more flexible, i.e., the spacing is not necessarily set to $0.5\lambda$.
To uniformly investigate and compare the ULA and the UCA, the term `array configuration' is used to refer to $(N,l_a)$.
In this chapter, the ULA and the UCA are set to either the same $l_a$ or the same $\Delta d$ for any given $N$.

\nomenclature{$l_a$}{array dimension}

While it is the reflection symmetry for the ULA, the UCA has both reflection symmetry and rotational symmetry.
Similar to Proposition\,\ref{prop:chp3_theta_B_range}, the following proposition can be deduced.
\begin{proposition}\label{prop:chp4_theta_B_range}
The array pattern for $G_C(\theta,\theta_B)$ repeats itself every $\frac{2\pi}{N}$ for $\theta_B$.
As the first element of the UCA lies on the positive x-axis, $G_C(\theta,\theta_B)$ is symmetric regarding to $\theta_B=\frac{i\pi}{N}$, $i\in\mathbb{Z}$.
Therefore, it suffices to study $G_C(\theta,\theta_B)$ in range of $\theta_B\in[0,\frac{\pi}{N}]$.
Nevertheless, as often in comparison with the ULA, the range of $\theta_B$ is set to $[0,\frac{\pi}{2}]$ for the UCA.
\end{proposition}

The large-scale path loss and the Rician small-scale fading are considered in addition to the AWGN.
Assume that the AP has the knowledge of Bob's CSI or coordinates, but has no knowledge of Eves' CSI except for their distribution, i.e., the PPP distribution.
Substituting $G_C(\theta,\theta_B)$ in (\ref{eq:chp4_AF_UCA}) for the general expression $G(\theta,\theta_B)$, $|\tilde{h}_C|^2$ can be obtained from the general expression $|\tilde{h}|^2$ in (\ref{eq:chp3_h_tilde_square}).
Thus, the received signal power $P_r(z)$ in (\ref{eq:chp3_receivedpower}) and the channel capacity $C(z)$ in (\ref{eq:chp3_channelcapacity}) for the UCA can be written based on $|\tilde{h}_C|^2$.

\subsection{SSOP and Its Upper Bound for UCA}
\label{chp4:sec2:ozix}

$\Theta$ is defined in the same way for the UCA as for the ULA, as shown in (\ref{eq:chp3_erdefinition2}).
For the UCA, $D(\theta)$ in (\ref{eq:chp3_ERboundary}) relies on $|\tilde{h}_C|^2$, which in turn depends on $G_C(\theta,\theta_B)$.
Similarly, $\bar{p}_C$ can be obtained by substituting $G_C(\theta,\theta_B)$ into the general expressions $\bar{p}$ in (\ref{eq:chp3_meanSSOP_Ri_2}), 
\begin{align}\label{eq:chp4_meanSSOP_Ri}
	\bar{p}_C&=1-\int_{-\infty}^{\infty}\int_{-\infty}^{\infty} \text{exp}\Big\{-\frac{\lambda_e}{2}c_0^{\frac{2}{\beta}}\int_0^{2\pi}\Big[\frac{KG_C^2(\theta,\theta_B)}{K+1} \nonumber \\
	& +\frac{x^2+y^2}{K+1}+\frac{2\sqrt{K}G_C(\theta,\theta_B)}{K+1}x\Big]^{\frac{2}{\beta}}\,\mathrm{d}\theta\Big\} \frac{e^{-(x^2+y^2)}}{\pi} \,\mathrm{d}x\,\mathrm{d}y.
\end{align}
It can be seen that it is not tractable to analytically analyze.

The upper bound $\bar{p}_{up,C}$ can be obtained by substituting $A_{0,C}$ into (\ref{eq:chp3_meanSSOP_up_3}),
\begin{align}\label{eq:chp4_meanSSOP_up}
	\bar{p}_{up,C}= 1-\text{exp}\Big\{-\lambda_e\pi\Big[\frac{c_0K}{2\pi(K+1)}A_{0,C}+\frac{c_0}{K+1}\Big]^{\frac{2}{\beta}}\Big\},
\end{align}
where $A_{0,C}$ is the pattern area for the UCA and is given by
\begin{align}\label{eq:chp4_A0C}
	A_{0,C}=\int_0^{2\pi}G_C^2(\theta,\theta_B)\,\mathrm{d}\theta.
\end{align}
Notice that (\ref{eq:chp4_meanSSOP_Ri}) and (\ref{eq:chp4_meanSSOP_up}) are for the generalized Rician channel.

For the special cases, i.e., the deterministic channel and Rayleigh fading channel, $\bar{p}_C$ and $\bar{p}_{up,C}$ can be obtained by substituting $K\to\infty$ and $K=0$ into (\ref{eq:chp4_meanSSOP_Ri}) and (\ref{eq:chp4_meanSSOP_up}), respectively, which is the same as Theorem\,\ref{th:chp3_p_De}, Theorem\,\ref{th:chp3_p_Ra}, Theorem\,\ref{th:chp3_p_up_De} and Theorem\,\ref{th:chp3_p_up_Ra}.
Notice that for the Rayleigh fading channel, $\bar{p}$ in (\ref{eq:chp3_meanSSOP_Ra}) and $\bar{p}_{up}$ in (\ref{eq:chp3_SSOP_Ra_up}) are regardless of the array geometry.

To obtain the analytic expression for $\bar{p}_{up}$ in (\ref{eq:chp4_meanSSOP_up}), the analytic expression of $A_{0,C}$ is required.
Here, $A_{0,C}$ is directly given.
\begin{theorem}\label{th:chp4_A_0C}
\begin{align}\label{eq:chp4_A_0C_2}
	A_{0,C}=2\pi+2\pi\sum_{n=1}^{N-1}J_0(kRW_n) \sum_{l=-\infty}^{\infty}(-1)^{ln+lN} J_{lN}(kRW_n) e^{jlN\theta_B},
\end{align}
where $W_n=2\sin(\frac{n}{N}\pi)$; $J_0(x)$ and $J_{lN}(x)$ are the Bessel function of the first kind with order zero and $lN$, respectively.
\end{theorem}
The proof of Theorem\,\ref{th:chp4_A_0C} is in Appendix\,\ref{appdx:bessel:ownbe}.
$A_{0,C}$ is obtained in the form of double summation of Bessel functions of the first kind.

\nomenclature{$W_n$}{intermediate variable in $A_{0,C}$}

the probabillity $\bar{p}_C$ in (\ref{eq:chp4_meanSSOP_Ri}) and $\bar{p}_{up,C}$ in  (\ref{eq:chp4_meanSSOP_up})  are determined by the channel parameters (i.e., $K$, $\beta$) as well as $G(\theta,\theta_B)$ in (\ref{eq:chp4_AF_UCA}) or $A_{0,C}$ in (\ref{eq:chp4_A_0C_2}), both of which depend on $N$, $R$ and $\theta_B$.
Thus, in this section, the impact of $N$, $R$ and $\theta_B$ is investigated for different channel parameters $K$ and $\beta$.
As mentioned in Proposition\,\ref{prop:chp3_SSOP_up_analysis}, for the deterministic channel when $\beta=2$, $\bar{p}_C=\bar{p}_{up,C}$, in which case $\bar{p}_{up}$  reduces to
\begin{align}\label{eq:chp4_p_De_beta_is_2}
	\bar{p}_{up,C}=1-\text{exp}\Big(-\frac{\lambda_ec_0}{2}A_{0,C}\Big).
\end{align}
First, the investigation of the properties of $\bar{p}_{up,C}$ in (\ref{eq:chp4_p_De_beta_is_2}) is carried out with respect to $N$, $R$ and $\theta_B$.
Then the properties of $\bar{p}_C$ and $\bar{p}_{up,C}$ for the generalized Rician channel will be examined.

Due to the positive correlation between $\bar{p}_{up,C}$ and $A_{0,C}$, the properties of $A_{0,C}$ are first examined.
Notice that $A_{0,C}$ in (\ref{eq:chp4_A_0C_2}) contains the complex component $e^{jlN\theta_B}$.
Since $A_{0,C}$ is the pattern area, it should be a real value.
In Appendix\,\ref{appdx:bessel:bxv}, $A_{0,C}$ is further derived to obtain the `real-value' form.
However, for even and odd $N$, the `real-value' forms of $A_{0,C}$ are slightly different.
Here, the `real-value' forms of $A_{0,C}$ are directly given,
\begin{align}
	A_{0,C,even}&=2\pi+2\pi \sum_{n=1}^{N-1} J_0^2(kRW_n)+4\pi\sum_{n=1}^{N-1} J_0(kRW_n)\sum_{l=1}^{\infty}(-1)^{ln} J_{lN}(kRW_n)\cos(lN\theta_B), \label{eq:chp4_A0C_even} \\
	A_{0,C,odd}&=2\pi+2\pi\sum_{n=1}^{N-1} J_0^2(kRW_n)+4\pi\sum_{n=1}^{N-1} J_0(kRW_n)\sum_{l=1}^{\infty} J_{2lN}(kRW_n) \cos(2lN\theta_B). \label{eq:chp4_A0C_odd}
\end{align}
$A_{0,C,even}$ in (\ref{eq:chp4_A0C_even}) and $A_{0,C,odd}$ in (\ref{eq:chp4_A0C_odd}) are only slight different, which allows us to use the same method to approximate and analyze them.
$A_{0,C,even}$ is taken as an example in this chapter; therefore, the subscript $_{even}$ is omitted in the expressions for convenience.

\nomenclature{$_{even}$}{even number}
\nomenclature{$_{odd}$}{odd number}

\section{Impact of Array Parameters on SSOP for UCA}
\label{chp4:sec3}
\subsection{Impact of DoE Angle}
\label{chp4:sec3:opbwec541}

To analytically analyze $A_{0,C}$ in (\ref{eq:chp4_A0C_even}), an appropriate approximation is needed.
To this end, $A_{0,C}$ is re-written by
\begin{align}\label{eq:chp4_A_0C_owjeonwp}
	 A_{0,C}=2\pi+\sum_{n=1}^{N-1}A_{0,C,n},
\end{align}
where $A_{0,n,C}$, $n=1,...,N-1$, is the summation term in (\ref{eq:chp4_A0C_even}),
\begin{align}\label{eq:chp4_A_0C_n}
	A_{0,C,n}=2\pi J_0^2(kRW_n)+4\pi J_0(kRW_n)\sum_{l=1}^{\infty}(-1)^{ln} J_{lN}(kRW_n)\cos(lN\theta_B).
\end{align}

\begin{proposition}\label{prop:chp4_W_n_range}
Because $\frac{n}{N}$ is not an integer for $n=1,...,N-1$, $W_n=2\sin(\frac{n}{N}\pi)\neq 0$.
When $N$ is even, $\sin(\frac{n}{N}\pi)=1$ at $n=\frac{N}{2}$.
Therefore, $W_n$ is in the range $(0,2]$ and $kRW_n$ is in the range $(0,2kR]$, for $n=1,...,N-1$.
\end{proposition}

Similar to the method used in Section\,\ref{chp3:analysis:sve}, the most dominant $A_{0,C,n}$ should be chosen to approximate $A_{0,C}$.
The value of $A_{0,C,n}$ is dominated by $J_0(kRW_n)$.
On one hand, $J_{lN}(x)$ in general decreases as $l$ increases.
On the other hand, the inner summation of $J_{lN}(kRW_n)$ in (\ref{eq:chp4_A_0C_n}) is weighted by $(-1)^{ln}$, which further cancels the impact of $J_{lN}(kRW_n)$ on the overall summation.

An example of $J_{lN}(x)$ is shown by the upper plot in Fig.\,\ref{fig:chp4_besselj_UCA}, where $N=8$ and $R=\frac{l_a}{2}=\frac{(N-1)\Delta d_L}{2}$.
The UCA has the same array dimension as the ULA with the same $N$.
For $\Delta d_L=0.5\lambda$, $2kR=21.99$.

Because the order of $J_{lN}(x)$ increases in the step of $N$ in (\ref{eq:chp4_A_0C_n}), $J_{lN}(x)$ with higher orders vanishes quickly in lower range of $x\in(0,2kR]$.
It can be seen that up till around $x=5$, $J_{lN}(x)$ is negligible for $l\geq 1$.
For the values of $n$ that satisfy $kRW_n<5$, $J_0(kRW_n)$ is absolutely dominant in $A_{0,C,n}$ in (\ref{eq:chp4_A_0C_n}).
In addition, for the whole range $x\in[0,2kR]$, only the first few $J_{lN}(x)$, i.e., $l=1,2$, are comparable to $J_0(x)$.

\begin{figure}
\centering
\includegraphics[scale=0.9]{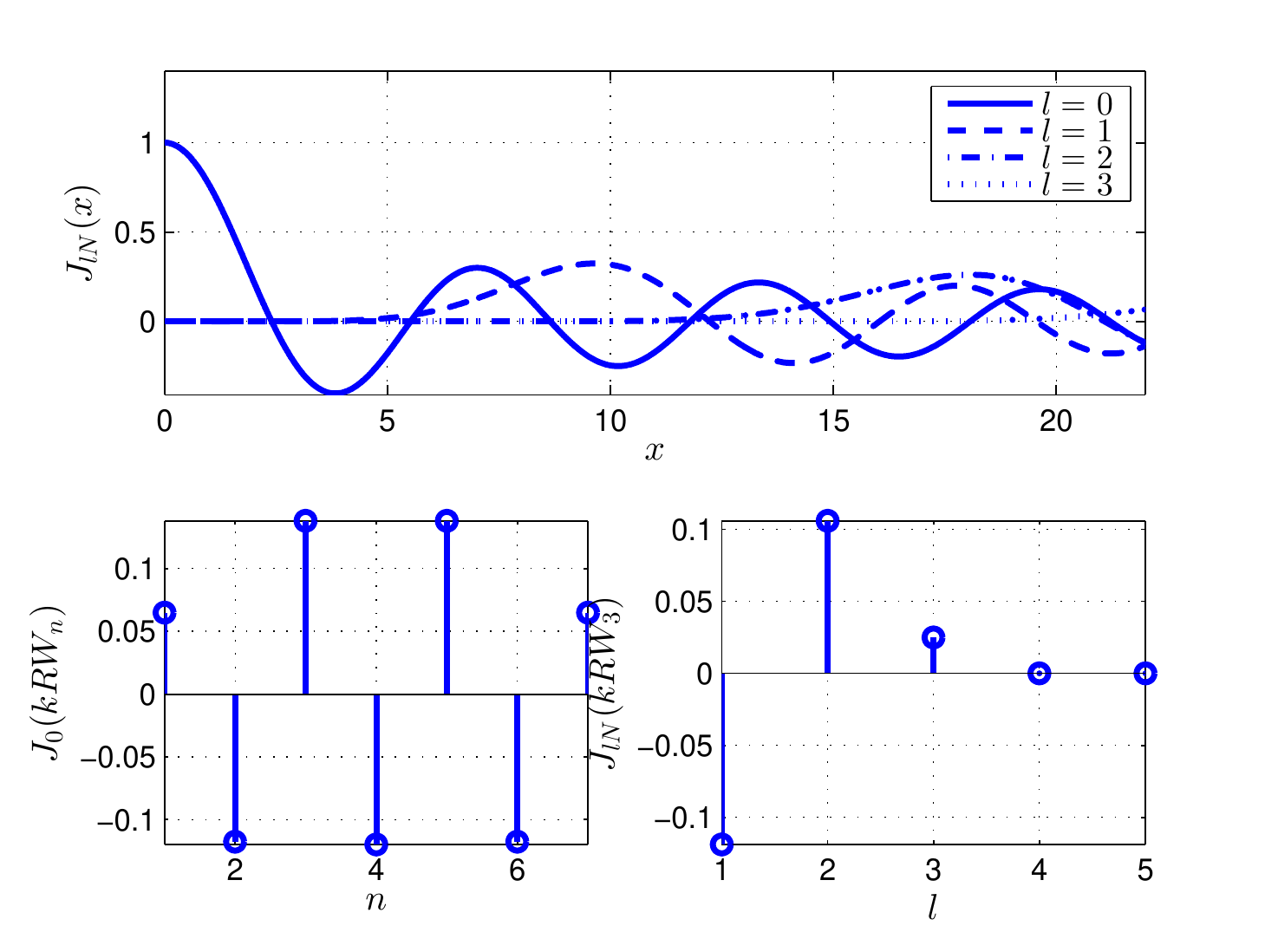}
\caption{Upper plot: $J_{lN}(x)$; lower left plot: $J_0(kRW_n)$ versus $n$; lower right plot: $J_{lN}(kRW_3)$ versus $l$. $N=8$, $R=1.75\lambda$.}
\label{fig:chp4_besselj_UCA}
\end{figure}

Take $N=8$ as an example in the following.
$kRW_n$ are written by
\begin{align}
	kRW_n=\frac{2\pi}{\lambda}\cdot\frac{(N-1)\Delta d}{2}\cdot2\sin(\frac{n}{N}\pi)=(N-1)\sin(\frac{p}{N}\pi)\pi=7\sin(\frac{p}{8}\pi)\pi.
\end{align}
The range of $kRW_n$ is $[7\pi\sin\frac{\pi}{8},7\pi]$.
The value of $J_0(kRW_n)$ in this range is shown in the lower left plot in Fig.\,\ref{fig:chp4_besselj_UCA}.
It can be seen that there is no absolutely dominant term of $J_0(kRW_n)$, because unlike $A_{0,L,n}$ that has the term $\frac{N-n}{N}$, there is no such term for $A_{0,C,n}$.

Based on the previous analysis of $A_{0,C,n}$, $A_{0,C}$ in (\ref{eq:chp4_A_0C_owjeonwp}) can be  approximated  by choosing the $A_{0,C,n}$ that has the most significant $J_0(kRW_n)$.
In this case, $J_0(kRW_3)$ has the largest absolute value, which makes the biggest impact in the summation of $A_{0,C}$.
Thus, $A_{0,C,3}$ is picked to approximate $A_{0,C}$,
\begin{align}
	A_{0,C}\approx 2\pi+2\pi J_0^2(kRW_3)+4\pi J_0(kRW_3)\sum_{l=1}^{\infty}(-1)^{3l} J_{8l}(kRW_3)\cos(8l\theta_B).
\end{align}
In the lower left plot in Fig.\,\ref{fig:chp4_besselj_UCA}, it shows that when $l\geq 3$, $J_{lN}(kRW_3)$ becomes negligible. 
Thus, only $l=1,2$ are taken into consideration,
\begin{align}\label{eq:chp4_A0C_approx}
	A_{0,C}\approx 2\pi+2\pi J_0^2(kRW_3)+4\pi J_0(kRW_3)[J_{16}(kRW_3)\cos(16\theta_B)-J_{8}(kRW_3)\cos(8\theta_B)].
\end{align}
Then, the approximation of $\bar{p}_{up,C}$ can be obtained by substituting (\ref{eq:chp4_A0C_approx}) into (\ref{eq:chp4_p_De_beta_is_2}).

For $A_{0,L,n}$ in (\ref{eq:chp3_A_0L_n_2}), $\theta_B$ exists for all $A_{0,L,n}$, $n=1,...,N-1$.
Compared to $A_{0,L,n}$, $\theta_B$ has less impact on $A_{0,C,n}$,
because for smaller $n$ where $kRW_n<5$, $J_{lN}(kRW_n)$ is negligible, which means $\theta_B$ does not impact these $A_{0,C,n}$; for the rest $A_{0,C,n}$, $\theta_B$ affects only the summation terms for $l>1$, which are not dominant. 
This means that $A_{0,C}$ does not change as much with $\theta_B$ as $A_{0,L}$.
Given the positive correlation between $A_0$ and $\bar{p}_{up}$, $\bar{p}_{up,C}$ does not change as much with $\theta_B$ as $\bar{p}_{up,L}$.

An example $\bar{p}_{up,C}$ and its approximation versus $\theta_B$ are shown in Fig.\,\ref{fig:chp4_p_DoE_De} together with $\bar{p}_{up,L}$ with the same $N=8$ and $l_a=3.5\lambda$.
It can be seen that the approximation is very close to $\bar{p}_{up,C}$.
By comparing $\bar{p}_{up,C}$ and $\bar{p}_{up,L}$, it can be seen that $\bar{p}_{up,C}$ varies much less in the range $\theta_B\in[0,\frac{\pi}{2}]$ than $\bar{p}_{up,L}$.
It can also be seen that in the lower range of $\theta_B$, e.g., $[0^{\circ},50^{\circ}]$, $\bar{p}_{up,L}$ is smaller than $\bar{p}_{up,C}$.

\begin{figure}
\centering
\includegraphics[scale=0.9]{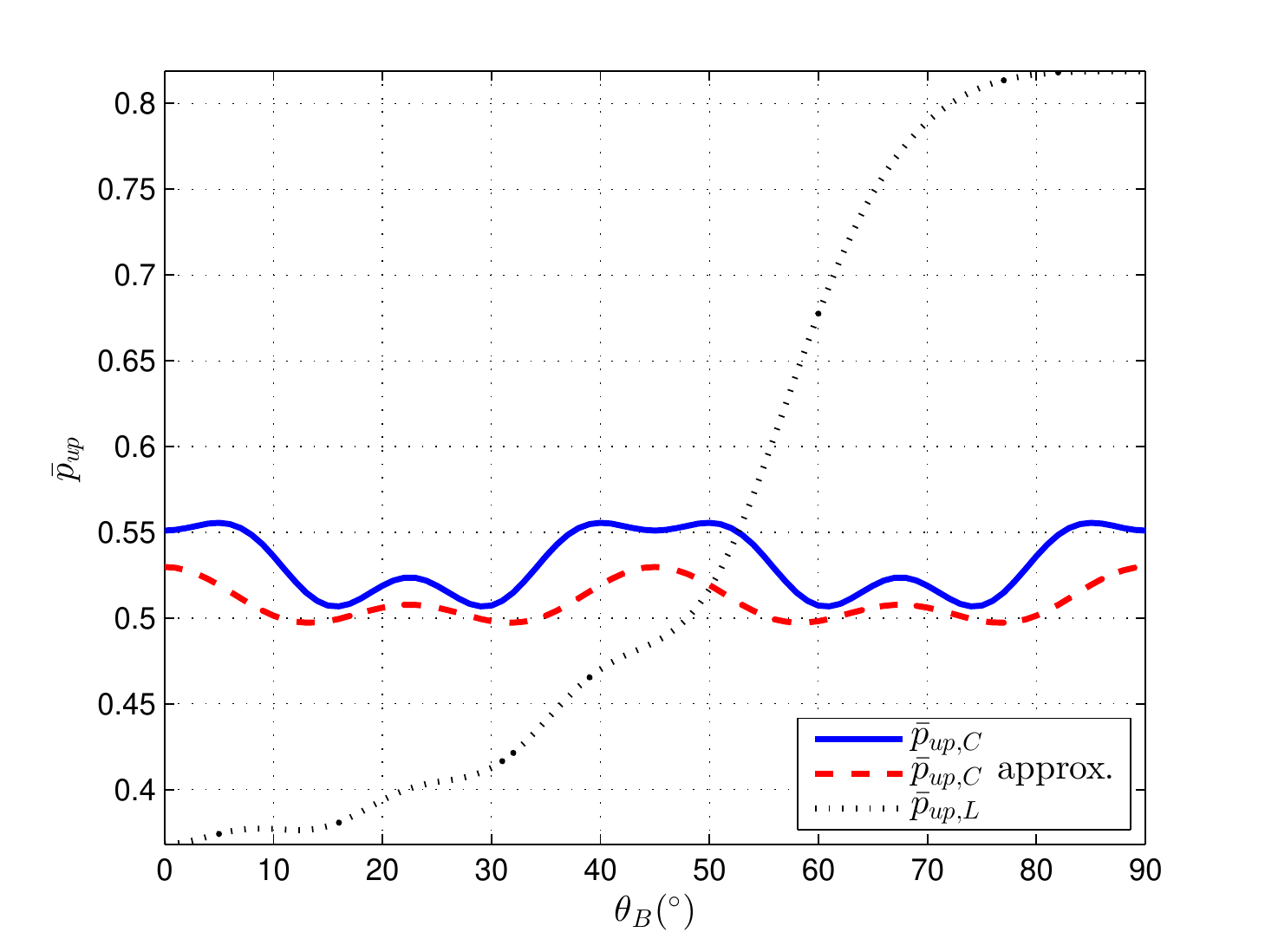}
\caption{$\bar{p}_{up,L}$, $\bar{p}_{up,C}$ and the approximation of $\bar{p}_{up,C}$ versus $\theta_B$. $N=8$, $R=1.75\lambda$. $P_t/\sigma_n^2=40$\,dB, $R_B=3.4594$\,bps/Hz, $R_s=1$\,bps/Hz, $\lambda_e=1\times10^{-4}$}
\label{fig:chp4_p_DoE_De}
\end{figure}

For the ULA, the difference between the maximum and minimum value of $\bar{p}_{up,L}$ is 0.3784.
For the UCA, the difference between the maximum and minimum value of $\bar{p}_{up,C}$ is 0.0636.
From both theoretical analysis and numerical results, it can be seen that in the range $\theta_B\in[0,\frac{\pi}{2}]$, $\bar{p}_{up,C}$ is more constant than $\bar{p}_{up,L}$.

\subsection{Impact of Array Configuration}
\label{chp4:sec3:njowq}

\subsubsection{Impact of Number of Elements}

As discussed in Section\,\ref{chp4:sec3:opbwec541}, $J_{lN}(x)$ is negligible for high order $lN$ in the low region of $x$.
Let $x_0$ denote the upper limit where $J_{lN}(x)$ is negligible in the range $x\in[0,x_0]$ for certain $lN$.
$x_0$ depends on the order ${lN}$.
For example, in the upper plot in Fig.\,\ref{fig:chp4_besselj_UCA}, $J_8(x)$ is negligible in the range $x\in[0,5]$ and $J_{16}(x)$ is negligible in the range $x\in[0,12]$.
As the order $lN$ increases, $x_0$ increases.

As shown in Proposition\,\ref{prop:chp4_W_n_range}, for fixed $R$, the range of $x=KRW_n$ is fixed, i.e., $(0,2kR]$.
As $N$ increases, $x_0$ also increases.
Once $x_0$ becomes larger than $2kR$, all $J_{lN}(x)$ for $l\geq 1$ are negligible in the range $(0,2kR]$.
Thus, for sufficiently large $N$, $A_{0,C}$ in (\ref{eq:chp4_A0C_even}) can be approximated by
\begin{align}\label{eq:appdx_bessel_A0C_asymptotic}
	A_{0,C}\approx 2\pi+2\pi \sum_{n=1}^{N-1} J_0^2(kRW_n).
\end{align}

For fixed $R$, the asymptotic behavior of $A_{0,C}$ with $N$ can be analyzed through (\ref{eq:appdx_bessel_A0C_asymptotic}).
As $N$ increases, $W_n=2\sin(\frac{n}{N}\pi)$ takes more samples of $\sin x$ in the range of $x\in(0,\pi]$, thus $J_0^2(kRW_n)$ takes more samples of $J_0^2(x)$ in the range $x\in(0,2kR]$.
Because $J_0^2(x)$ is non-negative, the more samples are taken, the larger the summation of $A_{0,C}$ is.
However, when $N$ is not very large, (\ref{eq:appdx_bessel_A0C_asymptotic}) is not valid and there does not exist a simple monotonic relationship between $A_{0,C}$ and $N$.

Due to the positive correlation between $A_{0,C}$ and $\bar{p}_{up,C}$, $\bar{p}_{up,C}$ has the same behavior with repect to $N$.
The upper plot in Fig.\,\ref{fig:chp4_p_N_and_R_De} shows the examples of $\bar{p}_{up,C}$ versus $N$, where $R$ is fixed to $1.75\lambda$.

\begin{figure}
\centering
\includegraphics[scale=0.9]{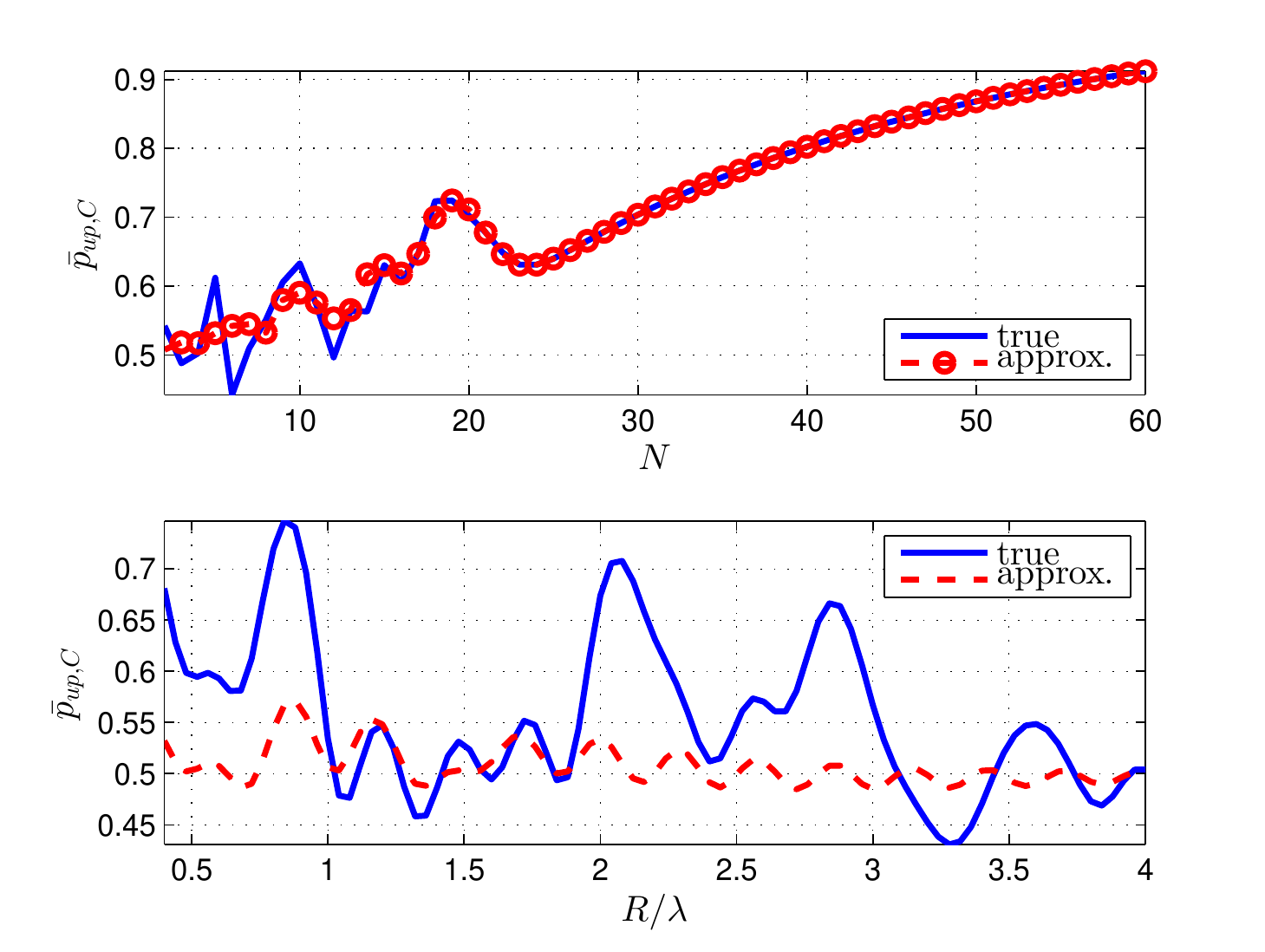}
\caption{Upper plot: $\bar{p}_{up,C}$ and the approximation versus $N$. $R=1.75\lambda$, $\theta_B=0^{\circ}$. Lower plot: $\bar{p}_{up,C}$ and the approximation versus $R$. $N=8$, $\theta_B=0^{\circ}$. $P_t/\sigma_n^2=40$\,dB, $R_B=3.4594$\,bps/Hz, $R_s=1$\,bps/Hz, $\lambda_e=1\times10^{-4}$}
\label{fig:chp4_p_N_and_R_De}
\end{figure}

In the lower region of $N$, besides $J_0(kRW_n)$, other orders of $J_{lN}(kRW_n)$ still contribute to the summation of $A_{0,C}$ in (\ref{eq:chp4_A0C_even}), which leading to the fluctuating behavior.
After $N\geq 19$, in the range of $(0,2kR]$, the summation of other orders of $J_{lN}(kRW_n)$ becomes less significant and the approximation in (\ref{eq:appdx_bessel_A0C_asymptotic}) is very close to the true value.
After $N>25$, the asymptotic behavior of $\bar{p}_{up,C}$ is almost linearly increasing with $N$.

Compared with $\bar{p}_{up,L}$, $\bar{p}_{up,C}$ in general increases with $N$ and there is no upper limit in theory, whereas
$\bar{p}_{up,L}$ in general decreases with $N$ and approaches to certain values depending on $\theta_B$.

\subsubsection{Impact of Array Dimension}

The impact of $R$ can be analyzed from (\ref{eq:chp4_A_0C_owjeonwp}) and (\ref{eq:chp4_A_0C_n}) without any approximation.
For $n=1,...,N-1$ and $l\geq 0$, the envelopes of $J_{lN}(kRW_n)$ decreases and approaches zero with different speed as $R$ increases.
Thus, the summation of $A_{0,C}$ also in general decreases and approaches certain value as $R$ increases. 
But due to the difference in the converging speed of $J_{lN}(kRW_n)$, there will be some fluctuations.

From (\ref{eq:chp4_A_0C_n}), it can be seen that, if $R$ is sufficiently large, all $J_{lN}(kRW_n)$ approach zero.
As a result, the value of $A_{0,C}$ approaches the value $2\pi$, which gives the limit of $\bar{p}_{up,C}$ by $1-\text{exp}(-\lambda_ec_0\pi)$.
However, this is only a theoretical limit. 
Because, in fact, when $R$ approaches infinity, the expression of $G_C(\theta,\theta_B)$ in (\ref{eq:chp4_AF_UCA}) no longer holds true for the far-field condition in Section\,\ref{chp2:antennas:UCA}.

In the lower plot in Fig.\,\ref{fig:chp4_p_N_and_R_De}, $\bar{p}_{up,C}$ versus $R$ is shown for fixed $N=8$.
It can be seen that the true value fluctuates as $R$ increases, because the curve is a superposition of $J_{lN}(kRW_n)$ with different orders $lN$.
In general, $\bar{p}_{up,C}$ decreases and approaches the value of $1-\text{exp}(-\lambda_ec_0\pi )$, which is 0.5025 in this case.
However, in the low region of $R$, e.g., $R<2\lambda$, the decreasing behavior is not very obvious.

The approximation in (\ref{eq:chp4_A0C_approx}) is also plotted as a comparison in the lower plot in Fig.\,\ref{fig:chp4_p_N_and_R_De}.
It can be seen that in general, the approximation is inaccurate.
This is because the approximation in (\ref{eq:chp4_A0C_approx}) is obtained for $N=8$ and $R=1.75\lambda$, which works the best only for that particular array configuration.
This can be verified by the proximity of the true value and the approximation at $R=1.75\lambda$ in the lower plot in Fig.\,\ref{fig:chp4_p_N_and_R_De}.

\subsection{Impact of Array Parameters on Array Pattern}
\label{chp4:sec3:vnweiwo}

The same as the ULA, the behavior of $A_{0,C}$ with respect to $N$, $R$ and $\theta_B$ can be explained from the spatial aspect by looking into the array pattern of $G_C(\theta,\theta_B)$, which determines the shape of $\Theta$ in (\ref{eq:chp3_ERboundary}).
Similar to the ULA, the mainbeam of the UCA is investigated first.

The mainbeam for the UCA can also be characterized by $G_{\text{max}}$ and $\Delta\theta_{HP,C}$.
$G_{\text{max}}$ in (\ref{eq:chp2_max_gain}) is regardless of the array geometry.
For the UCA, $\Delta\theta_{HP,C}$ is directly given here.
\begin{proposition}\label{th:chp4_HPBW_C}
\begin{align}\label{eq:chp4_HPBW_C}
	\Delta\theta_{HP,C}=4\arcsin\frac{1.1264}{2kR}=4\arcsin\frac{1.1264}{kl_a}.
\end{align}
\end{proposition}
The proof of Proposition\,\ref{th:chp4_HPBW_C} is in Appendix\,\ref{appdx:bessel:yter}.
$\Delta\theta_{HP,C}$ is only determined by $R$ (i.e., $l_a$).
Since $\arcsin(\cdot)$ is monotonically increasing, $\Delta\theta_{HP,C}$ is reversely proportional to $R$.

Both $\Delta\theta_{HP,L}$ in (\ref{eq:chp3_HPBW_L}) and $\Delta\theta_{HP,C}$ in (\ref{eq:chp4_HPBW_C}) are reversely proportional to the array dimension $l_a$, since (\ref{eq:chp3_HPBW_L}) can be rewritten with respect to $l_a$ using (\ref{eq:chp4_spacing_and_dim_ULA}).
\begin{align}\label{eq:chp4_HPBW_L}
	\Delta\theta_{HP,L}= 2\Big[\theta_B-\arcsin\Big(\sin\theta_B-\frac{N-1}{N}\frac{2.782}{kl_a}\Big)\Big].
\end{align}

\begin{proposition}\label{th:chp4_HPBW_L_C}
When $\theta_B=0$, $\Delta\theta_{HP,L}$ and $\Delta\theta_{HP,C}$ is approximately the same given the same $l_a$. 
\end{proposition}
\begin{lemma}\label{le:chp4_arcsin}
For $0<x\ll 1$, $2\arcsin(\frac{x}{2})\approx\arcsin(x)$
\end{lemma}
The following proof of Proposition\,\ref{th:chp4_HPBW_L_C} requires Lemma\,\ref{le:chp4_arcsin}, the proof of which is given in Appendix\,\ref{appdx:bessel:nvoeor}.
\begin{proof}
When $\theta_B=0$, $\Delta\theta_{HP,L}$ in (\ref{eq:chp4_HPBW_L}) is given by
\begin{align}\label{eq:chp4_HPBW_L_boresight}
	\Delta\theta_{HP,L}= 2\arcsin\Big(\frac{N-1}{N}\frac{2.782}{kl_a}\Big).
\end{align}
According to (\ref{eq:chp4_HPBW_L_boresight}), the following approximation holds true,
\begin{align}\label{eq:chp4_owenvpwoaga}
	\frac{N-1}{N}\frac{2.782}{kl_a}\approx\frac{2.782}{kl_a}=2\frac{ 1.391}{kl_a}\approx 2\frac{1.1264}{kl_a}.
\end{align}

Notice that the right side in (\ref{eq:chp4_owenvpwoaga}) is twice the input of $\arcsin(\cdot)$ in (\ref{eq:chp4_HPBW_C}).
Usually, both $\frac{N-1}{N}\frac{2.782}{kl_a}$ and $\frac{1.1264}{kl_a}$ are far less than 1.
For example, when $N=8$, $\Delta d=0.5\lambda$ and $R=\frac{(N-1)\Delta d}{2}$, $\frac{N-1}{N}\frac{2.782}{kl_a}=0.1107$ and $\frac{1.1264}{kl_a}=0.0512$.
According to (\ref{eq:chp4_HPBW_C}), (\ref{eq:chp4_HPBW_L_boresight}), (\ref{eq:chp4_owenvpwoaga}) and Lemma\,\ref{le:chp4_arcsin}, $\Delta\theta_{HP,L}\approx\Delta\theta_{HP,C}$ given the same $l_a$ for the ULA and the UCA.
Thus, the proof is completed.
\end{proof}

Two groups of array patterns for the UCA with different $N$ and $\theta_B$ are shown in Fig.\,\ref{fig:chp4_patterns_UCA}. 
The radius $R$ is fixed to $0.6533\lambda$, which gives the same $\Delta\theta_{HP,C}$ for all patterns. 
The first group is for fixed $N=8$ and different $\theta_B$.
Discrete values of $\theta_B$ are chosen in the range $(0,\frac{\pi}{N})$, as mentioned in Proposition\,\ref{prop:chp4_theta_B_range}.
The second group is for fixed $\theta_B$ and different $N$.

\begin{figure}
\centering
\includegraphics[scale=0.9]{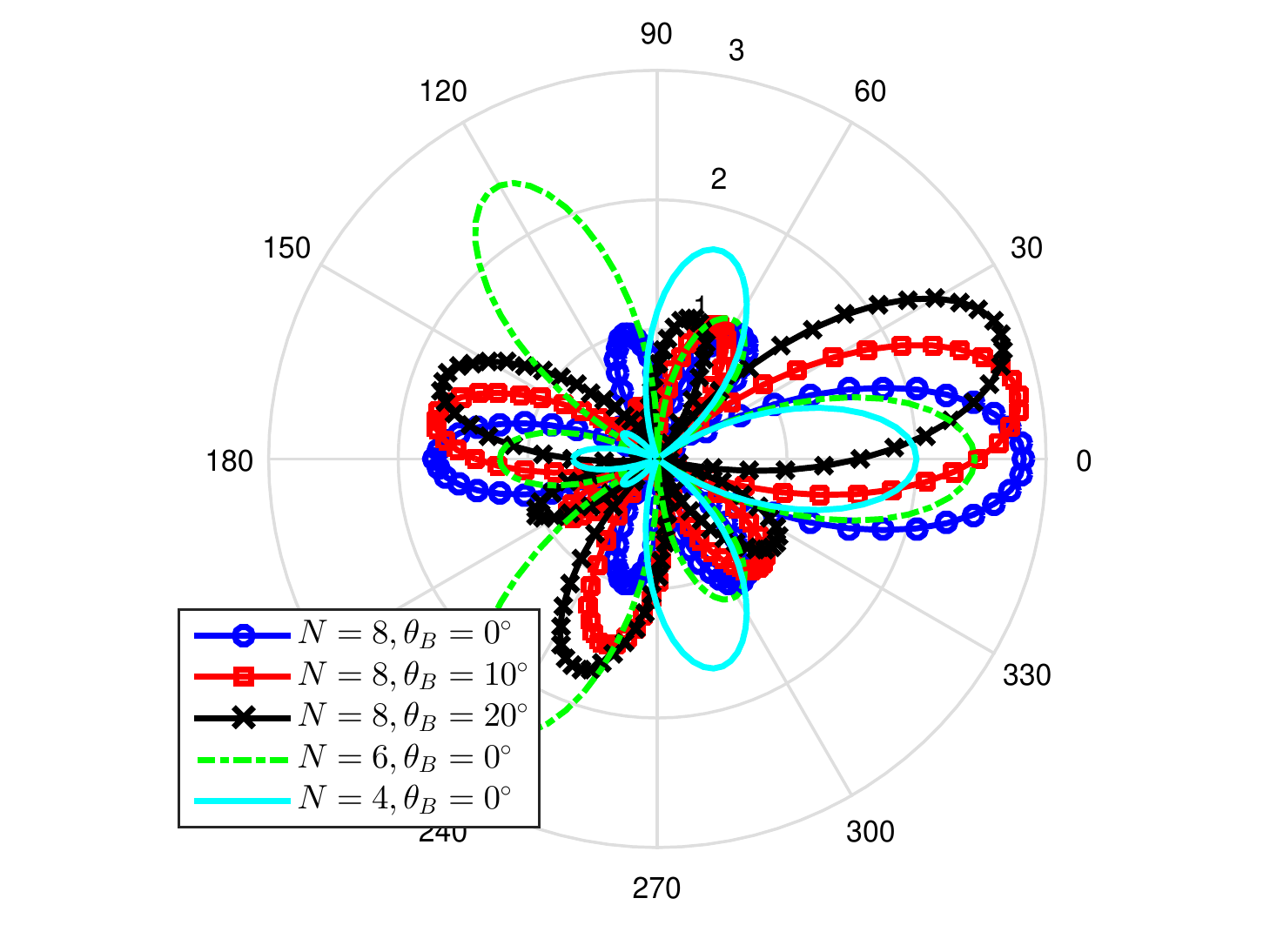}
\caption{Array patterns of UCA for different $N$ and $\theta_B$, $R=0.6533\lambda$}
\label{fig:chp4_patterns_UCA}
\end{figure}

By observing the first group of patterns, it can be seen that the mainbeam for different patterns is of the same length due to the same $N$, and the mainbeam widths are of the same due to fixed $R$.
Thus, it can be deduced that the area of the mainbeam stays more or less the same.
It can also be observed that the sidelobe level (SSL) does not change much.
Thus, the total area $A_{0,C}$ is rather constant over the range of $\theta_B$.

\nomenclature{SSL}{sidelobe level}

By observing the second group of the array patterns, it can be seen that while the mainbeam width stays the same due to fixed $R$, the mainbeam length increases with $N$. 
Thus, the area of the mainbeam increases along with $N$.
In the same time, the SSL changes dramatically as $N$ changes.
For $N=8$, the sidelobes are comparable to the mainbeam, which leads to a complex relationship between the total area $A_{0,C}$ and $N$.

In Fig.\,\ref{fig:chp4_patterns_UCA_2}, the patterns of the ULA and the UCA with the same $(N,l_a)$ are shown. 
It can be seen that when $\theta_B=0^{\circ}$, the mainbeams of the ULA and the UCA are approximately the same.
The SSL for the ULA is very low compared to the mainbeam, leading to the fact that the pattern area of the ULA is mainly contributed by the mainbeam;
for the UCA, there are larger sidelobes, resulting in a bigger pattern area than the ULA.

\begin{figure}
\centering
\includegraphics[scale=0.9]{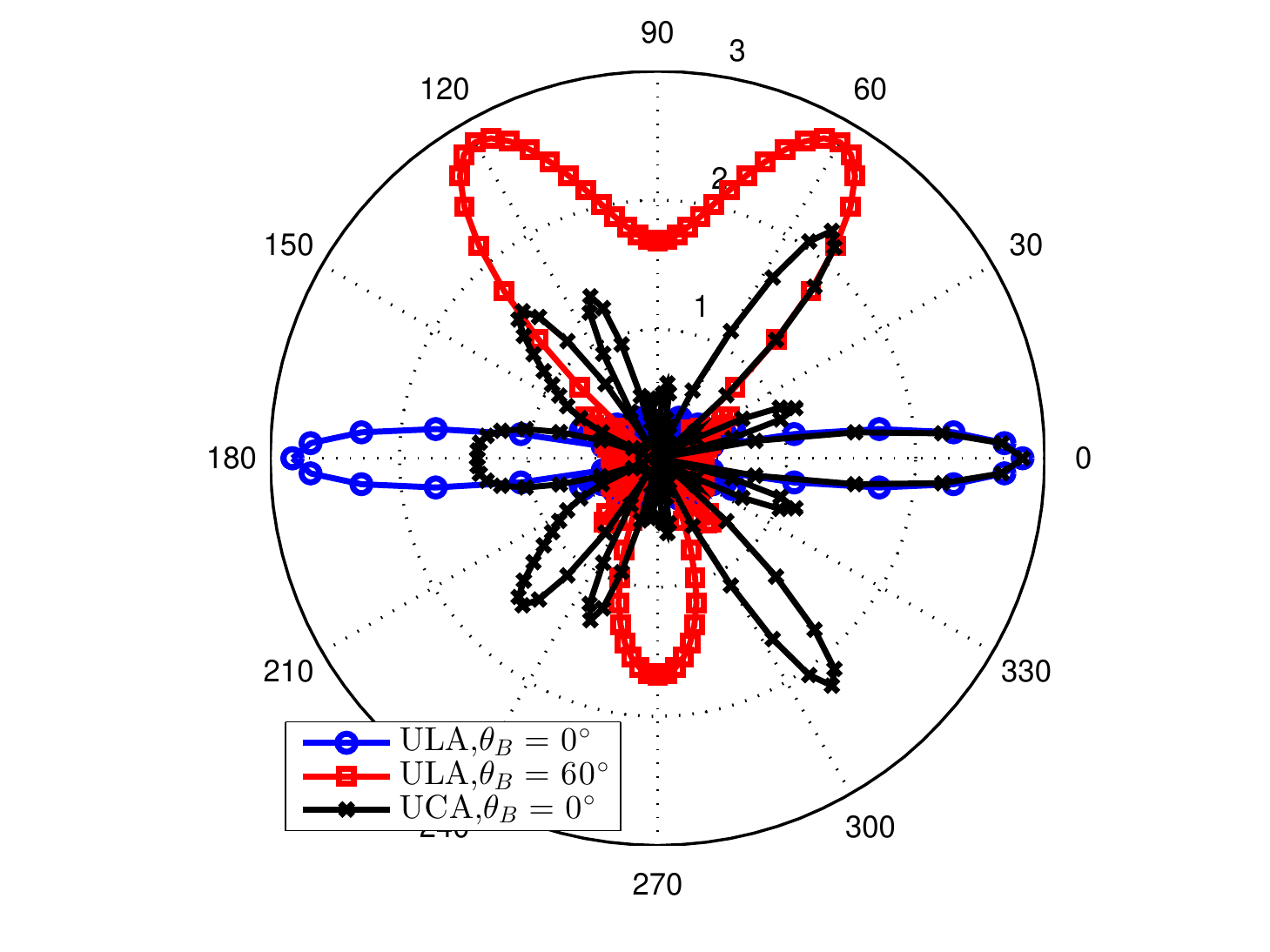}
\caption{Array patterns of ULA and UCA for $N=8$ and $l_a=\frac{N-1}{2}\Delta d$, $\Delta d=0.5\lambda$}
\label{fig:chp4_patterns_UCA_2}
\end{figure}

When $\theta_B=60^{\circ}$, $\Delta\theta_{HP,L}$ is larger than when $\theta_B=0^{\circ}$, according to Fig.\,\ref{fig:chp3_HPBW_DoE_ULA} and Fig.\,\ref{fig:chp4_patterns_UCA_2}, which results in a larger $A_{0,L}$.
On the contrary, $A_{0,C}$ stays more or less constant in $\theta_B\in[0,90^{\circ}]$.
Thus, as $\theta_B$ increases, $A_{0,L}$ grows bigger than $A_{0,C}$.
The above analysis explains the comparison of $\bar{p}_{up,L}$ and $\bar{p}_{up,C}$ in Fig.\,\ref{fig:chp4_p_DoE_De}.

For the UCA, when $R$ increases and $(N,\theta_B)$ are fixed, the mainbeam becomes narrower and smaller.
However, the SSL changes in a complex way, which makes it hard to get useful conclusions from the patterns. 
Thus, the patterns are not shown here.

\section{Numerical Results for Generalized Rician Channel Model}
\label{chp4:sec4}
\subsection{SSOP and Its Upper Bound for UCA}
\label{chp4:sec4:pwoepw}

In Section\,\ref{chp4:sec2:ozix}, $\bar{p}_C$ and $\bar{p}_{up,C}$ have been derived for the generalized Rician channel.
In Section\,\ref{chp4:sec3}, the special case when $K\to\infty$ and $\beta=2$ is analyzed, where $\bar{p}_C=\bar{p}_{up,C}$.
With the aid of the same method used for the ULA, in this section, the behaviors of $\bar{p}_C$ and $\bar{p}_{up,C}$  with respect to $N$, $R$ and $\theta_B$ are analyzed for the generalized Rician channel.
In addition, the tightness of $\bar{p}_{up,C}$ is studied via $\eta_C$ according to the general definition in (\ref{eq:chp3_eta}).

In fact, the relationship between $\bar{p}_{up,C}$ and $A_{0,C}$ is exactly the same as that for the ULA, because the general expression of $\bar{p}_{up}$ in (\ref{eq:chp3_meanSSOP_up_2}) applies to any array type.
Thus, the properties of $\bar{p}_{up,C}$ with respect to $N$, $R$ and $\theta_B$ is similar to those of $A_{0,C}$.
In the same way, it is natural to conjecture that the properties of $\bar{p}_C$ with respect to $N$, $R$ and $\theta_B$ are also similar to $A_{0,C}$, but with some deviation, depending on the particular channel parameter and array parameter.

To avoid repetition, the detailed analysis is referred to in Section\,\ref{chp3:result:wier}.
Here, some examples of $\bar{p}_C$ and $\bar{p}_{up,C}$ are used to verify the previous conclusions.
In Fig.\,\ref{fig:chp4_p_and_bounds_R_beta_3_C}, $\bar{p}_C$ and $\bar{p}_{up,C}$ versus $R$ for a typical value $\beta=3$ are shown, where $N=8$ and $\theta_B=0^{\circ}$.

\begin{figure}
\centering
\includegraphics[scale=0.9]{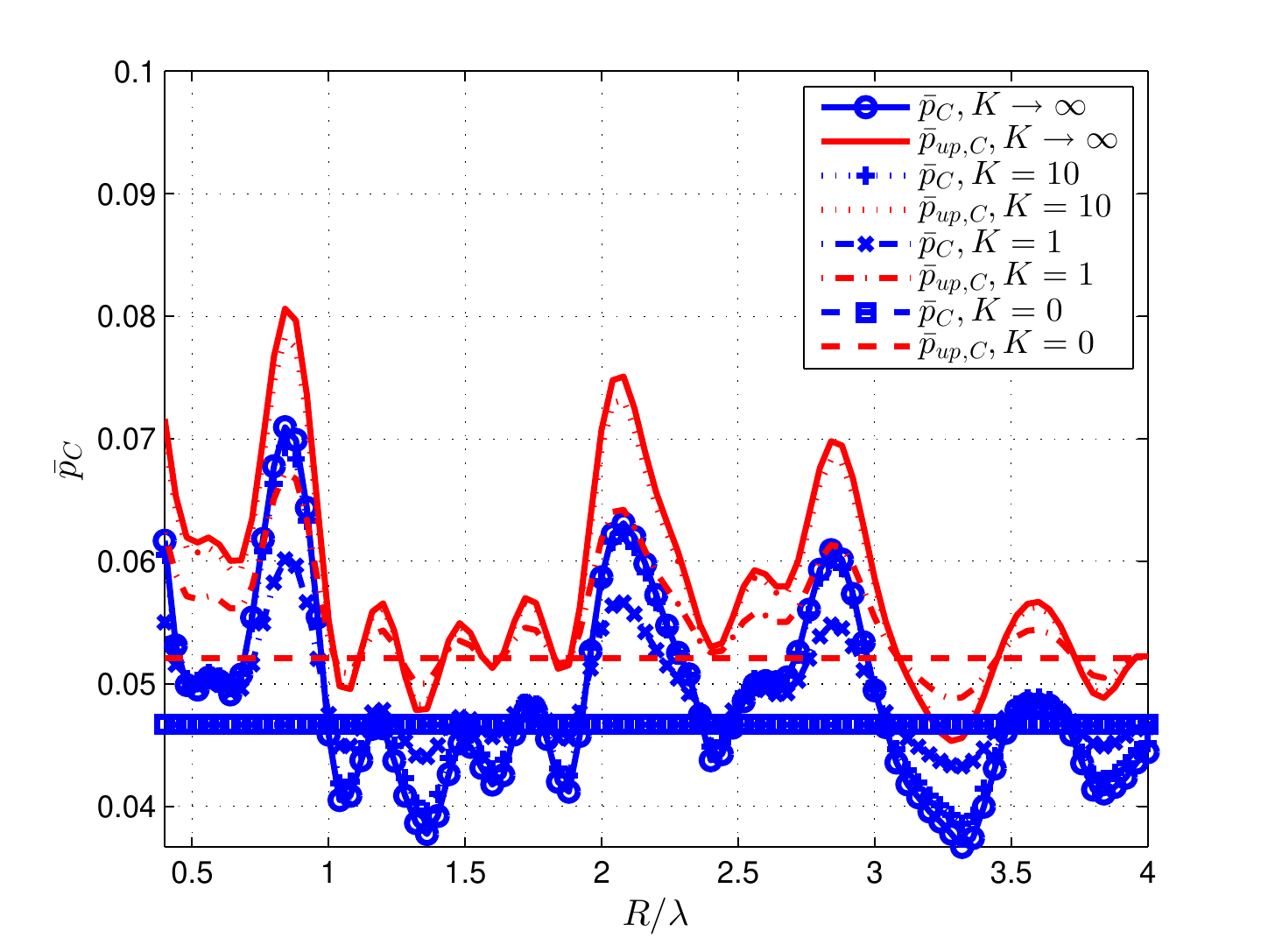}
\caption{$\bar{p}_C$ and $\bar{p}_{up,C}$ versus $R$ for different $K$. $\beta=3$, $N=8$, $\theta_B=0^{\circ}$. $P_t/\sigma_n^2=40$\,dB, $R_B=3.4594$\,bps/Hz, $R_s=1$\,bps/Hz, $\lambda_e=1\times10^{-4}$}
\label{fig:chp4_p_and_bounds_R_beta_3_C}
\end{figure}

It can be seen that for the Rayleigh channel (i.e., $K=0$), both curves for $\bar{p}_C$ and $\bar{p}_{up,C}$ are flat and are identical to those of the ULA in Fig.\,\ref{fig:chp3_p_and_bounds_DoE_beta_3_L}, because they do not rely on $G(\theta,\theta_B)$, according to (\ref{eq:chp3_meanSSOP_Ra}) and (\ref{eq:chp3_SSOP_Ra_up}).

Comparing the curves for $K\to\infty$ and $\beta=3$ in Fig.\,\ref{fig:chp4_p_and_bounds_R_beta_3_C} with the lower plot in Fig.\,\ref{fig:chp4_p_N_and_R_De} when $\beta=2$, it can be seen that the two curves have very similar fluctuating behavior with respect to $R$.
Furthermore, comparing the curves with different $K$ (except for $K=0$) in Fig.\,\ref{fig:chp4_p_and_bounds_R_beta_3_C}, it can be seen that all curves have very similar behavior with respect to $R$, which verifies that the properties of $\bar{p}_C$ and $\bar{p}_{up,C}$ with respect to $N$, $R$ and $\theta_B$ is similar to those of $A_{0,C}$.
The examples of $\bar{p}_C$ and $\bar{p}_{up,C}$ versus $\theta_B$ and $N$ are in Fig.\,\ref{fig:appdx_fig_p_and_bounds_DoE_beta_3_C} and\,\ref{fig:appdx_fig_p_and_bounds_N_beta_3_C} in Appendix\,\ref{appdx:fig:cmvs}. 
Similar conclusions can be concluded from those two figures.

The same as the ULA, when $K=10$, both curves of $\bar{p}_C$ and $\bar{p}_{up,C}$ are close to those when $K\to\infty$; 
when $K=1$, both curves of $\bar{p}_C$ and $\bar{p}_{up,C}$ are close to those when $K=0$, which can be observed in Fig.\,\ref{fig:chp4_p_and_bounds_R_beta_3_C}.
However, for different $(K,\beta)$ and $(N,R,\theta_B)$, the tightness of upper bound $\eta_C$ is different.

In summary, it can be seen from the numerical results that the properties of $\bar{p}_C$ and $\bar{p}_{up,c}$ with respect to $N$, $R$ and $\theta_B$ are in general consistent with those of $A_{0,C}$.

\subsection{Tightness of Upper Bound for UCA}
\label{chp4:sec4:psodpv}

In Section\,\ref{chp3:result:mnbv}, it has been concluded that $\eta_C$ decreases with $K$ when $\beta=2$; and $\eta_C$ increases with $K$ when $\beta>2$.
For fixed $\beta$ and $(N,R,\theta_B)$, $\eta_C$ is bounded by two extreme cases, i.e, $K=0$ and $K\to\infty$.
Thus, in the following, the examples of $\eta_C$ for $K=0$ and $K\to\infty$ are given.

In Fig.\,\ref{fig:chp4_eta_R_bounds_UCA}, $\eta_C$ versus $R$ is shown for all $\beta$.
It can be seen that when $\beta=2$, $\eta_C=1$, because this is the simple case, in which $\bar{p}_C=\bar{p}_{up,C}$.
When $\beta>2$, $\eta_C$ for different $\beta$ is located in a cluster for given $R$ and has no monotonic relationship with $\beta$, which is the same as $\eta_L$.

\begin{figure}
\centering
\includegraphics[scale=0.9]{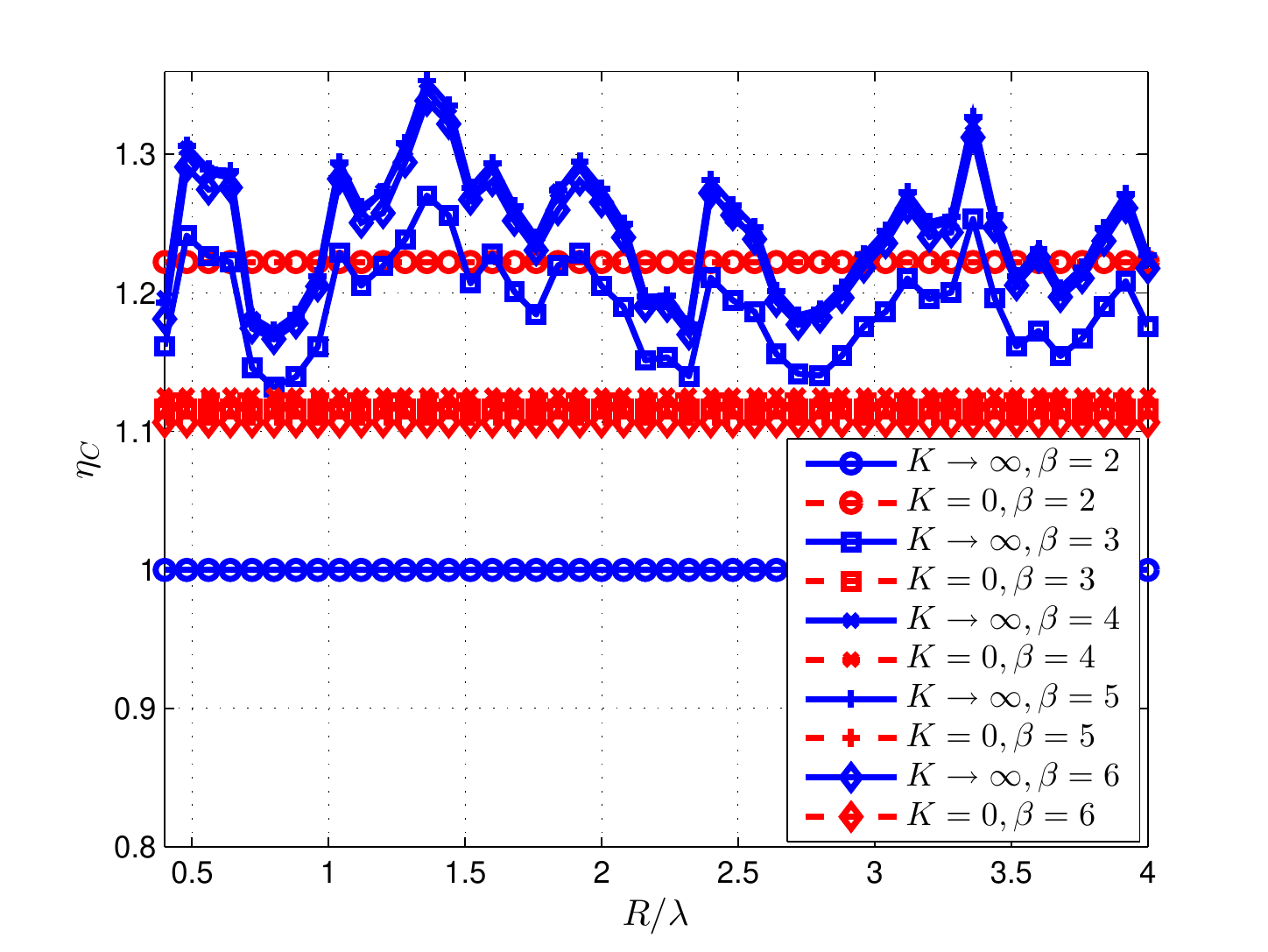}
\caption{$\eta_C$ versus $R$ for the deterministic for all $\beta$, $N=8$, $\theta_B=0^{\circ}$}
\label{fig:chp4_eta_R_bounds_UCA}
\end{figure}

In Fig.\,\ref{fig:chp4_eta_R_bounds_UCA}, as $R$ changes, there are some fluctuations for $\eta_C$. 
However, there is no obvious increasing or decreasing in the whole range of $R$.
More results of $\eta_C$ for different $\theta_B$ and $N$ are in Fig.\,\ref{fig:appdx_fig_eta_DoE_bounds_UCA} and\,\ref{fig:appdx_fig_eta_N_bounds_UCA} in Appendix\,\ref{appdx:fig:snoa}.

Take $\beta=3$ as an example to compare $\eta_C$ with $\eta_L$, and the results are shown in Fig.\,\ref{fig:chp4_eta_N_and_DoE_bounds}.
Since $\eta$ is always a constant when $K=0$, only the results $\eta$ for $K\to\infty$ are plotted.
In the upper plot in Fig.\,\ref{fig:chp4_eta_N_and_DoE_bounds}, $\eta$ versus $\theta_B$ is shown.
While $\eta_L$ in general decreases with $\theta_B\in[0,\frac{\pi}{2}]$, $\eta_C$ is more constant and smaller than $\eta_L$.
In the lower plot in Fig.\,\ref{fig:chp4_eta_N_and_DoE_bounds}, $\eta$ versus $N$ is shown.
Both $\eta_L$ and $\eta_C$ increases with $N$.
However, $\eta_C$ is smaller than $\eta_L$ for any $N$, except for $N=2$. 
Furthermore, $\eta_C$ converges to certain value, while $\eta_L$ keeps increasing.

\begin{figure}
\centering
\includegraphics[scale=0.9]{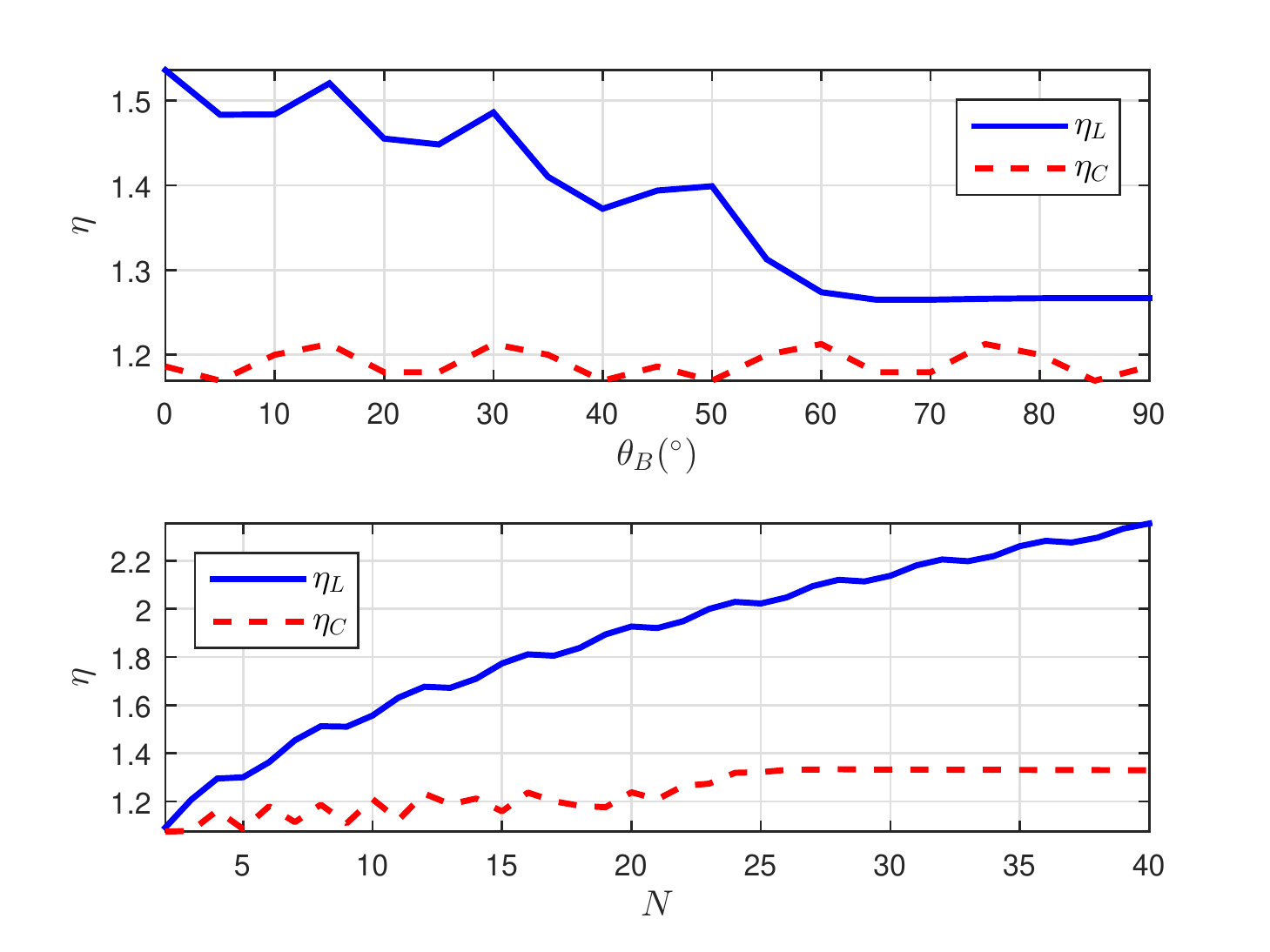}
\caption{Upper plot: $\eta_C$ versus $\theta_B$ for $N=8$ and $R=1.75\lambda$; lower plot: $\eta_C$ versus $N$ for $R=1.75\lambda$ and $\theta_B=0^{\circ}$. $K\to\infty$ and $\beta=3$}
\label{fig:chp4_eta_N_and_DoE_bounds}
\end{figure}

Based on the previous analysis, it can be concluded that not only is $\eta_C$ smaller than $\eta_L$ in general, but it is also more constant for changing $N$, $R$ and $\theta_B$, which indicates that the upper bound is tighter for the UCA.
It is worth noticing that the range of $\eta_C$ is mostly within $[1,1.4]$, which is very tight.
Thus, $\bar{p}_{up,C}$ can provide a good approximation for $\bar{p}_C$.

\section{Impact of Mutual Coupling}
\label{chp4:sec5}
\subsection{Array Factor and Mutual Coupling}
\label{chp4:sec5:nviewoe}

This section primarily focuses on the impact of the mutual coupling on the maximum gain $G_{\text{max}}$ in (\ref{eq:chp2_max_gain}), because both $P_{rB}$ and $C_B$ depend on $\tilde{h}_B$ in (\ref{eq:chp3_h_tilde_Bob}), which is calculated based on $G_{\text{max}}$.
In addition, the distortion to the array pattern is numerically measured.

As mentioned in Section\,\ref{chp2:antennas}, $G(\theta,\theta_B)$ is calculated based on the assumption that each element in the array is omni-directional, and $G_{\text{max}}=\sqrt{N}$, which is independent from $\theta_B$.
However, due to the mutual coupling, the pattern of each element in the array is not omni-directional.
Thus, the maximum gain in practice becomes angle dependent.
Let $G_{\text{max,mc}}$ denote the distorted maximum gain.
\begin{align}
	G_{\text{max,mc}}=G_{\text{max}}f(\theta_B)=\sqrt{N}f(\theta_B),
\end{align}
where $f(\theta_B)$ is the attenuation function for $G_{\text{max}}$, $0<f(\theta_B)\leq 1$.
The subscript $_{\text{mc}}$ is short for `mutual coupling'.

\nomenclature{$_{\text{mc}}$}{mutual coupling}
\nomenclature{$f(\theta_B)$}{maximum gain attenuation depending on $\theta_B$}

According to (\ref{eq:chp3_channelcapacity}), the message transmission condition (i.e., $C_B\geq R_B$) used for the secrecy outage formulation in Section\,\ref{chp3:metric:ER} can be converted into
\begin{align}\label{eq:chp4_C_B_R_B}
	P_{rB}=\frac{P_t}{d_B^{\beta}}|\tilde{h_B}|^2\geq\sigma_n^2(2^{R_B}-1).
\end{align}
Using $G_{\text{max,mc}}$ in (\ref{eq:chp3_h_tilde_square}), $|\tilde{h_B}|^2$ that is subjected to the mutual coupling can be given by
\begin{align}\label{eq:chp4_h_tilde_B_square}
	|\tilde{h_B}|^2 =\frac{KG_{\text{max,mc}}^2}{K+1}+\frac{1}{K+1}g_{Re}^2+\frac{1}{K+1}g_{Im}^2+\frac{2\sqrt{K}G_{\text{max,mc}}}{K+1}g_{Re}.
\end{align}
Because $G_{\text{max,mc}}$ changes with $\theta_B$, $|\tilde{h_B}|^2$ also changes with $\theta_B$.
To guarantee (\ref{eq:chp4_C_B_R_B}), $P_t$ need to be adjusted according to $d_B$ for certain channel parameters and $R_B$, which causes complexity in adjusting $P_t$ based on $G_{\text{max,mc}}$, i.e., $f(\theta_B)$. 
If $P_t$ is not adjustable, e.g., the current AP in Wi-Fi networks, (\ref{eq:chp4_C_B_R_B}) may not be guaranteed.

Pearson's correlation coefficient, denoted by $\rho$, is used to measure the correlation between two random variables $X$ and $Y$.
\begin{align}\label{chp3_eq:corrcoef}
  \rho=\frac{cov(X,Y)}{std(X)\cdot std(Y)},
\end{align}
where $cov(\cdot,\cdot)$ stands for the covariance and $std(\cdot)$ the standard deviation. 
$\rho$ takes value from -1 to 1, where 1 means total positive correlation, 0 means no dependence at all, and -1 means total negative correlation. 
The array patterns obtained from experiments and simulations can be regarded as samples of variables. 
Thus, the correlation coefficient $\rho$ between different patterns can be calculated to indicate how close they are.
The larger $\rho$, the more alike two patterns are.
Notice that in the expression of $\rho$ in (\ref{chp3_eq:corrcoef}), $cov(X,Y)$ is normalized against their the average values of $X$ and $Y$.
Thus, the value of $\rho$ between two patterns only indicates the closeness of shapes, or the likeness, between two patterns, but does not reflect the maximum gain attenuation.

\nomenclature{$\rho$}{Pearson's correlation coefficient}
\nomenclature{$P_rB$}{Bob's received signal power}

The mutual coupling effect was observed when building a practical transmit beamformer on WARP\,\cite{warpProject}.
First, the WARP experiment set-up is introduced before demonstrating the results.
Then the experiment results are shown and compared with NEC results.
The mutual coupling is difficult to analytically calculate because of the complex electromagnetic boundaries in the near field.
Therefore, both WARP experiments and NEC simulations are used to study the impact of the mutual coupling.

\subsection{WARP Experiments}
\label{chp4:sec5:pzoxsp}

In this section, WARPLab is used to build the transmit beamformer with antenna arrays.
An introduction to the WARP hardware and how to build a MISO   communications system on WARPLab is in Section\,\ref{chp2:mc}.
This section explains how to build a practical beamformer with calibratied phase and how to measure the beam pattern in an anechoic chamber, in order to observe the mutual coupling effect.

\subsubsection{Beamformer Set-Up}

WARP node can hold up to four RF interfaces, each of which can be regarded as a transceiver in the WARPLab design.
In this section, one/two nodes with multiple antennas, i.e., RF interfaces, are used to act as the AP, and another node with a single RF interface as a general user, in order to measure the pattern.
To simulate transmissions in Wi-Fi network, such as 802.11n, the carrier frequency is set to $f_0=2.484$\,GHz, which is the center frequency of Wi-Fi channel 14.
This is to avoid co-channel interferences from other wireless devices.

In WARPLab, data packets are generated and processed in MATLAB. 
So the data source and sink are both in MATLAB.
The transmit and receive RF interfaces are also controlled by MATLAB. 
A system diagram that describes the transmission from the AP to the user is shown in Fig.\,\ref{fig:chp4_WARP_commsyst}.
This diagram is based on the one in Fig.\,\ref{fig:appdx_WARP_commsyst2}, but with the focus on the RF end, in order to introduce the phase calibration later.
The packet that contains the preamble and the payload is first generated and configured in MATLAB; 
then it is passed into the buffer on WARP board and later sent over the air via the transmit RF interfaces.
The receive RF interface captures and stores the packet in the buffer again before sending it back to MATLAB for post-processing.
More details of the system diagram and the structure of the RF interface are shown in Fig.\,\ref{fig:appdx_WARP_transceiver}, \ref{fig:appdx_WARP_commsyst} and \ref{fig:appdx_WARP_commsyst2}.

\begin{figure}
\centering
\includegraphics[scale=1]{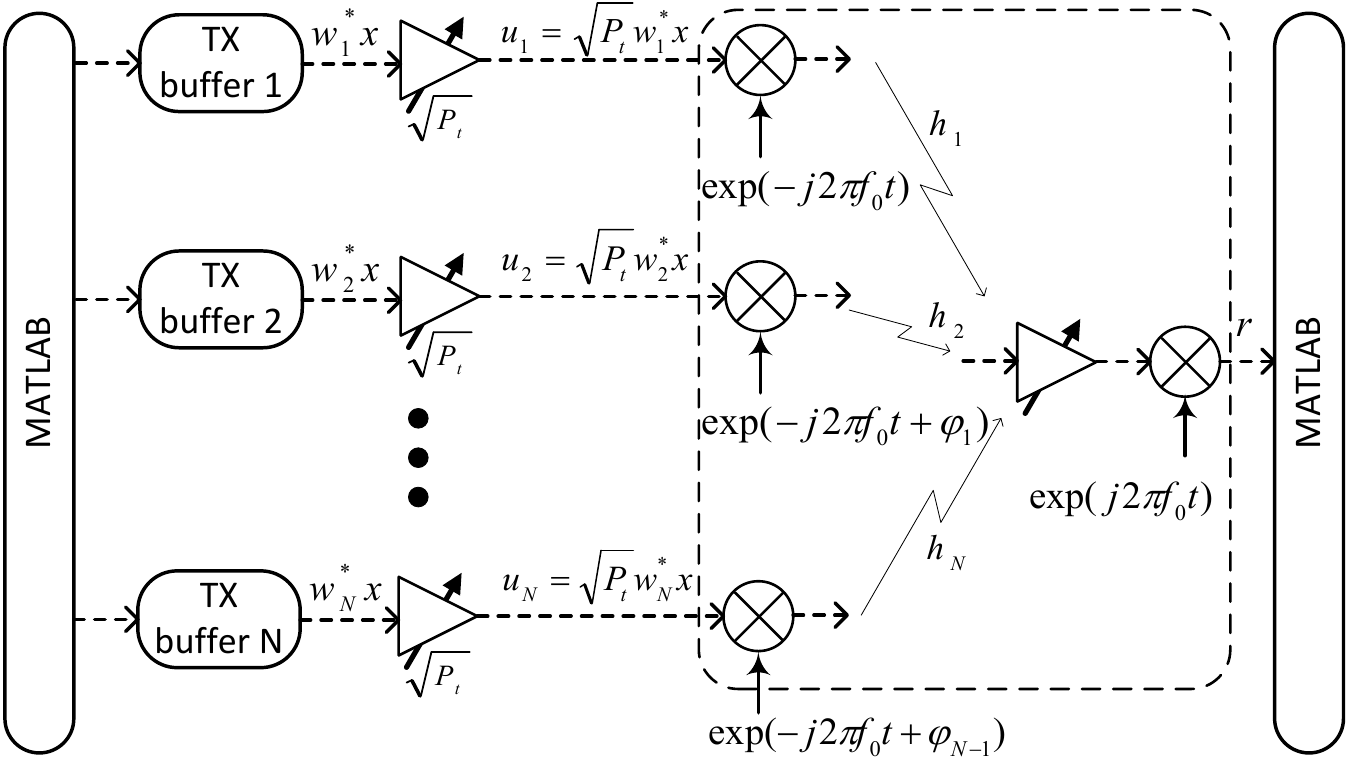}
\caption{MISO system on WARPLab}
\label{fig:chp4_WARP_commsyst}
\end{figure}

\nomenclature{$f_0$}{carrier frequency}

The transmit beamforming is explained combined with the system model in Section\,\ref{chp3:syst:model}.
As shown in Fig.\,\ref{fig:chp4_WARP_commsyst}, 
$x$ is an arbitrary symbol in the packet and the same packet is transmitted over all channels with different weights.
In MATLAB, the packet is pre-coded by $\mathbf{w}$.
Each weighted symbol, i.e., $w_i^*x$, $i=1,...,N$, is delivered to the corresponding buffer and waits to be transmitted. 
In the RF interface, an appropriate level of transmit power $P_t$ is chosen.
The discrete symbols are converted into analog signals, up-converted and transmitted via the antenna.
On the receiver side, the signal is captured, down-converted and converted into discrete signal.
Note that all arrows are in dashed lines, because it only represents the data flow and some not-so-relevant blocks/processing are omitted.
The block in the dashed diagram includes the analog processing and is considered as a discrete channel to the transmitter and the receiver.
The channel input is $u_i=\sqrt{P_t}w_i^*x$, $i=1,...,N$ and the channel output is $r$.

In this thesis, the beamformer is built in the baseband. 
Thus, pure sinusoidal signals of an intermediate frequency 5\,MHz are used as payload.
Each packet has a preamble that consists of a long-training symbol (LTS), a guard interval and a pilot.
The LTS is known by both transmitter and receiver and is used for sample-level synchronization by correlation at the receiver.
The payloads are pre-coded by $\mathbf{w}$.
The transmitted signals are superimposed in the air and form beam patterns.
The figure of the data structure is shown in Fig.\,\ref{fig:chp4_WARP_datapacket}.

\nomenclature{LTS}{long-training symbol}

\begin{figure}
\centering
\includegraphics[scale=1]{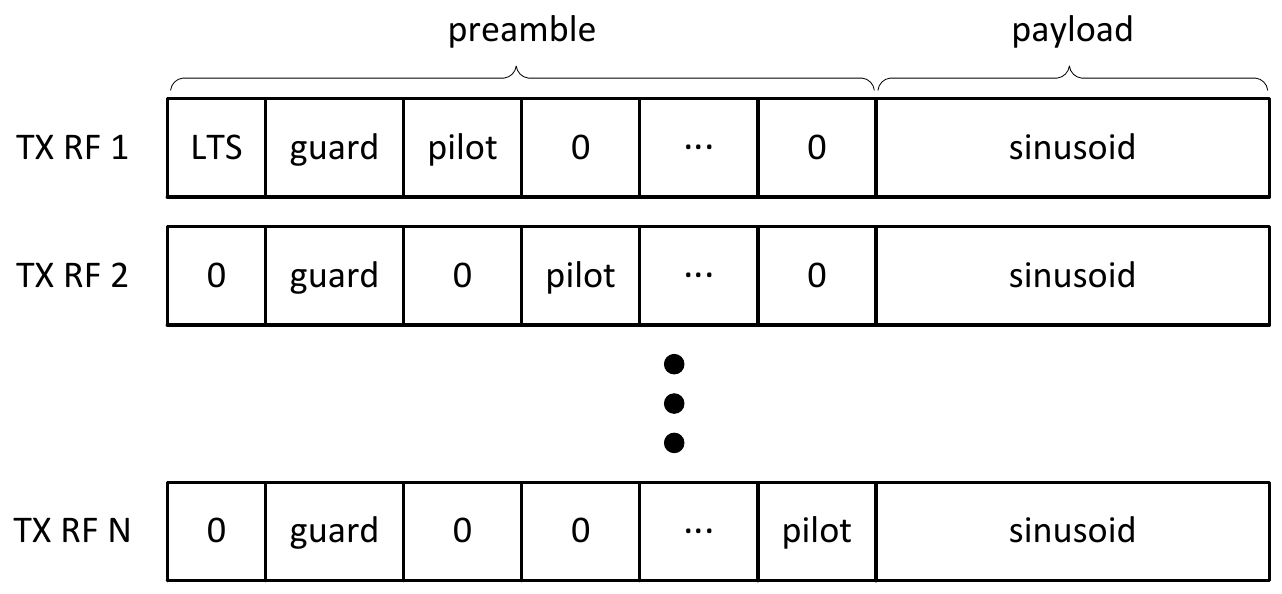}
\caption{Data packet structure}
\label{fig:chp4_WARP_datapacket}
\end{figure}

At the receiver side, the packets via different channels are superimposed.
The received signal is down-converted and sampled.
The received packet that contains received symbol $r$ is stored in the buffer.
The RSSI is also recorded.
RSSI can be converted into dBm values.
However, in our case, the absolute power level is not of interest.
The transmit power is fixed at an appropriate level, so that the patterns can be measured by appropriately normalized values at different $\theta$, i.e., $r(\theta)$.
In addition, the amplifiers at the receiver side is set to minimum levels, to avoid the thermal noise being amplified.

There are two main challenges in building the transmit beamformer that forms actual patterns over the air, i.e., carrier frequency offset (CFO) and random initial phase in each radio interface.
The details of the solution to the CFO problem and the calibration process for the random initial phase are in Appendix\,\ref{appdx:warp}.
.

\subsubsection{Beamformer Measurements}

The following explains how to measure the pattern $G(\theta,\theta_{\text{doe}})$ in an anechoic chamber, where the channel between the transmit and receive antennas is similar to a free-space path loss channel.
As shown in Fig.\,\ref{fig:chp4_WARP_chamber}, the transmit antenna array is composed of commercial 2.4\,GHz dipole antennas and is put on a rotating platform; the receive antenna is put on a fixed platform at the other end of the chamber.
The two WARP nodes are set up as described in Section\,\ref{appdx:warpnec:wioo} and \ref{appdx:warpnec:rqe} (notice that the receive WARP node is hidden from this view).

\begin{figure}
\centering
\includegraphics[width=0.8\textwidth]{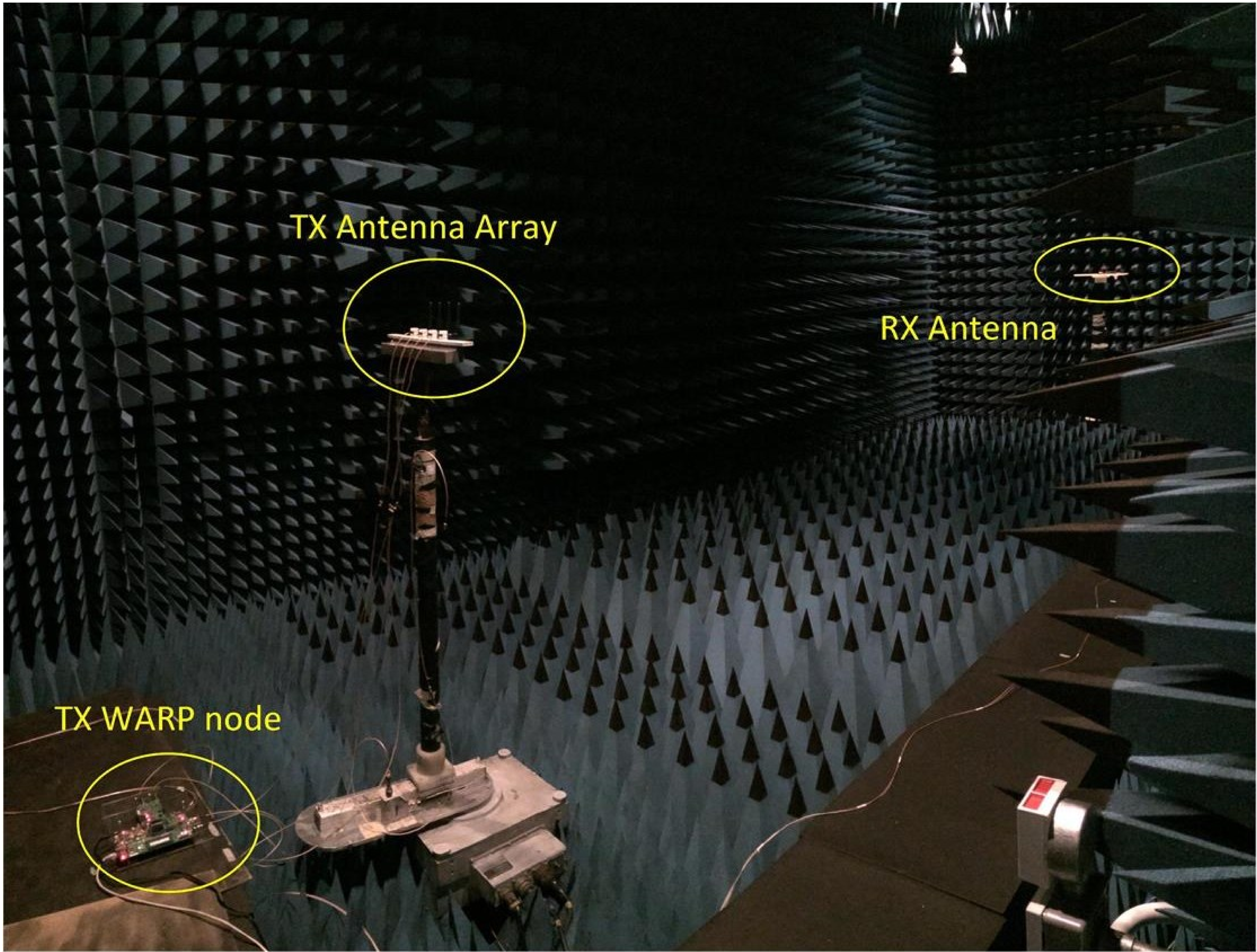}
\caption{Set-up in the anechoic chamber}
\label{fig:chp4_WARP_chamber}
\end{figure}

At the transmitter side, the sinusoid signal that is pre-coded by $\mathbf{w}$ in (\ref{eq:ch3_BF_Weights}) is sent.
For both the ULA and the UCA, it is sufficient to measure the DoE angle $\theta_{\text{doe}}\in[0,\frac{\pi}{2}]$.
For all measurements, the total transmit power is set to a fixed value that is large enough so that the background noise can be ignored.
For example, $P_t=0$\,dBm and a free-space path loss is 54.07\,dB for 5\,m, which gives a receive SNR roughly at 45.93\,dB if the background noise power is assumed to be $-100$\,dBm.

Since the distance between the array and the receive antenna is fixed, the array factor $G(\theta,\theta_{\text{doe}})$ can be directly measured by $P_r=r^2$, according to (\ref{eq:chp3_receivedpower}).
For each $\theta_{\text{doe}}$, the received signal $P_r$ is measured at a discrete step in the range of $\theta\in[-\frac{\pi}{2},\frac{\pi}{2}]$.
Then $G(\theta,\theta_{\text{doe}})$ can be calculated as the average value of all data samples that are received at angle $\theta$.

\subsection{Experiment Results}
\label{chp4:sec5:nvweiwewvfe}

As stated in Section\,\ref{chp2:antennas:UCA}, the mutual coupling is the coupling effect between two neighbor antenna elements.
The first experiment measures how the pattern of dipole (i.e., omnidirectional pattern) is changed when it is in a 4-element ULA with $\Delta d=0.5\lambda$.
This ULA is placed along y-axis as shown in Fig.\,\ref{fig:chp2_ULA} and the pattern of the first element is measured. 
For this purpose, only the first element is activated, while keeping the other elements inactive.
The measurement is taken every $5^{\circ}$ in the range of $\theta\in[-90^{\circ},90^{\circ}]$.
The pattern of the first element is plotted in Fig.\,\ref{fig:chp4_WARP_pattern_DoE_ULA_1stelement}.
For ease of comparison, the pattern is normalized with regard to its own maximum gain, i.e., 1.

\begin{figure}
\centering
\includegraphics[scale=0.9]{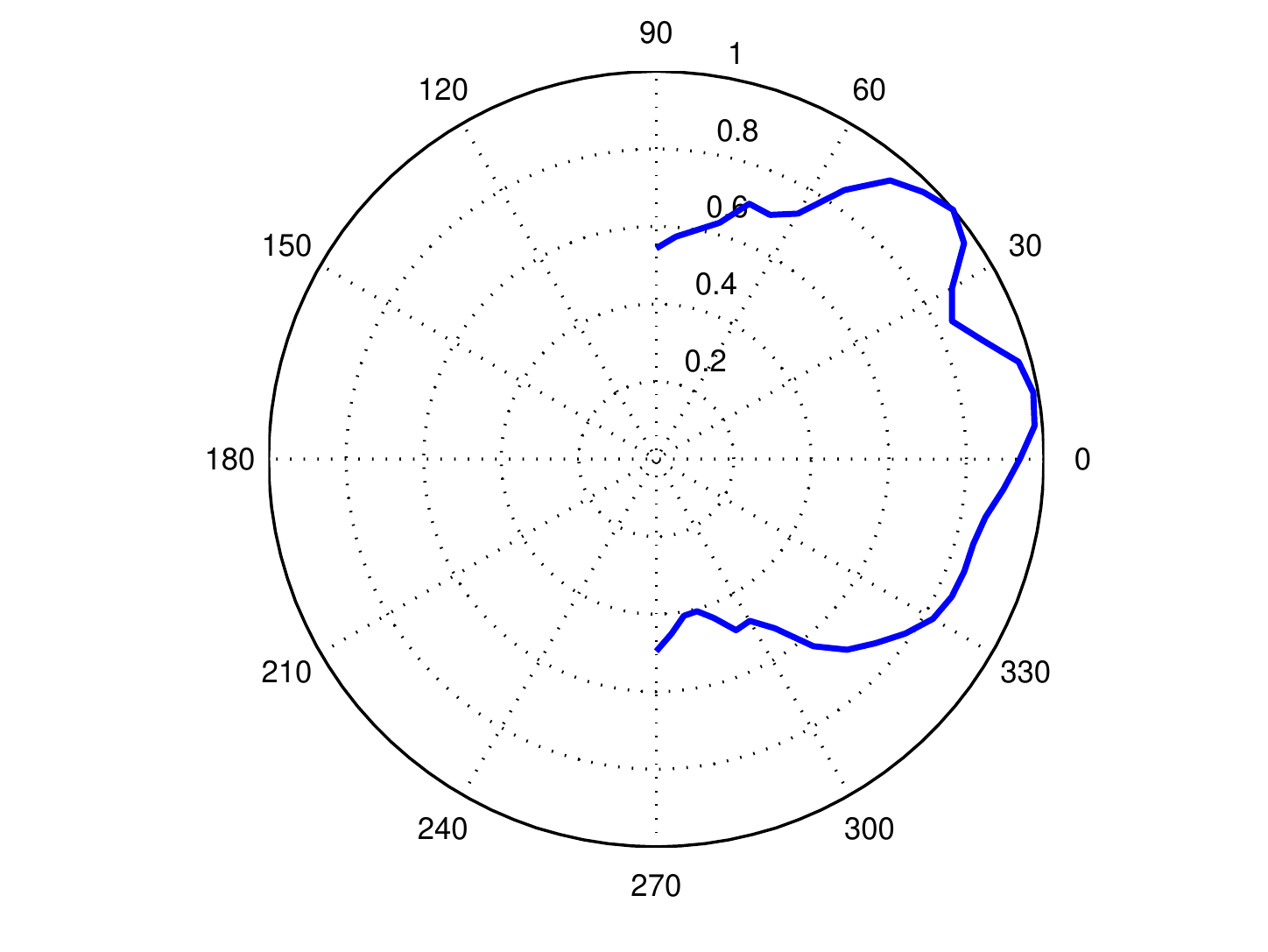}
\caption{1st WARP experiment: Pattern of the 1st element in ULA}
\label{fig:chp4_WARP_pattern_DoE_ULA_1stelement}
\end{figure}

It can be seen that the patterns shown in Fig.\,\ref{fig:chp4_WARP_pattern_DoE_ULA_1stelement} is not circular, because the radiation power from the first element induces electromagnetic fields in the nearby antennas, which interfere with its own radiation pattern.
To be more specific, the maximum value is a bit off $\theta=0^{\circ}$; at the angle $\theta=\pm90^{\circ}$, the array gain is approximately half of the maximum gain. 
In addition, the imperfection in WARP experiments causes some fluctuations to the pattern, which makes it less smooth.

The second experiment measures the pattern of the whole ULA. 
For comparison, 7 independent measurements are carried out for $\theta_B\in[0^{\circ},90^{\circ}]$ in steps of $15^{\circ}$. 
For each measurement, the same method is used to measure $G(\theta,\theta_B)$ as the first experiment.
All 7 patterns are normalized with respect to the maximum value of all patterns and are plotted in Fig.\,\ref{fig:chp4_WARP_pattern_DoE_ULA}.

\begin{figure}
\centering
\includegraphics[scale=0.9]{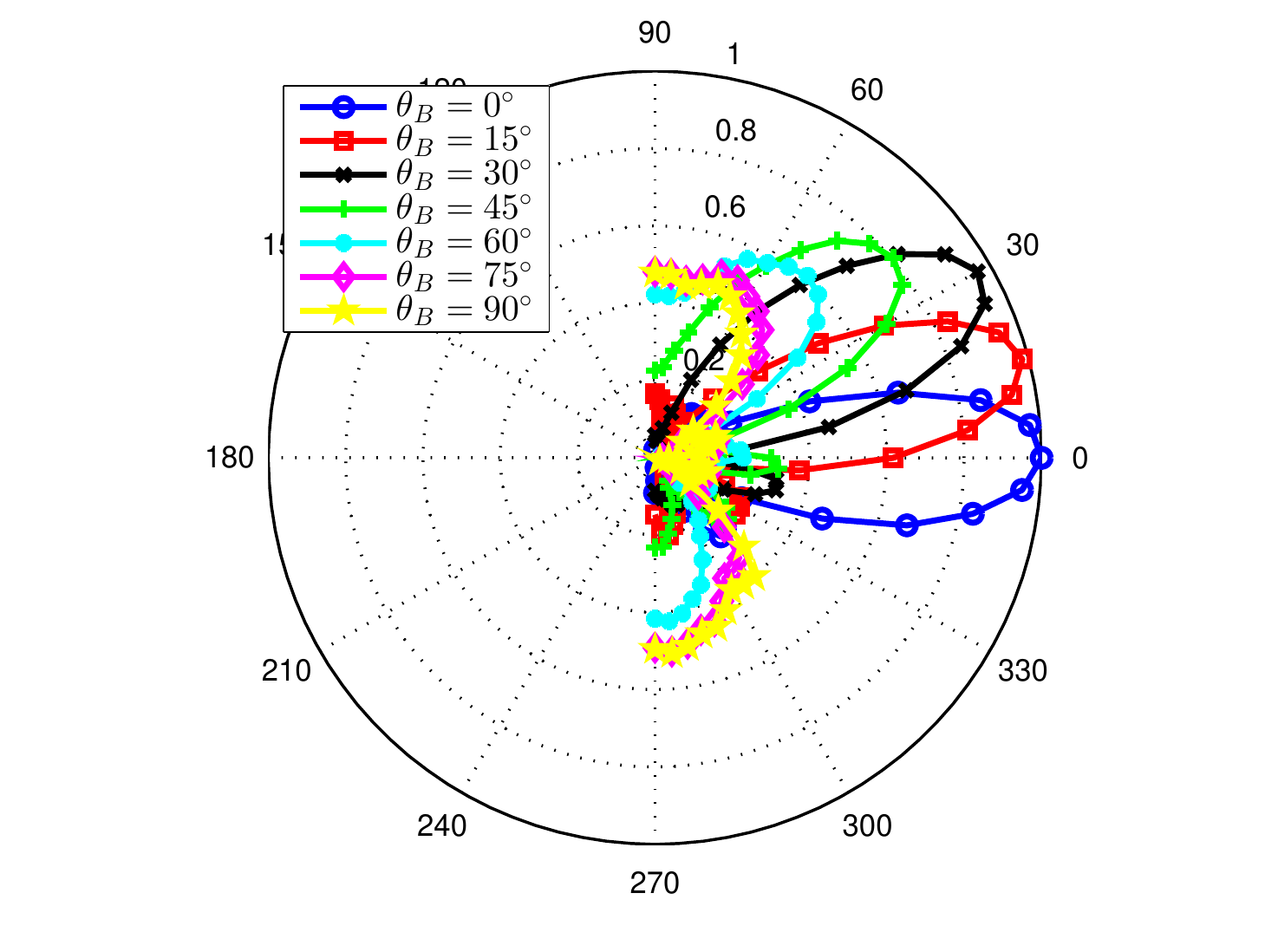}
\caption{2nd WARP experiment: $G_L(\theta,\theta_B)$}
\label{fig:chp4_WARP_pattern_DoE_ULA}
\end{figure}

As can be seen in Fig.\,\ref{fig:chp4_WARP_pattern_DoE_ULA}, the maximum gain stays more or less the same when $\theta_B<30^{\circ}$, however it drops as $\theta_B$ approaches $90^{\circ}$, which is not observed from the theoretical patterns in Fig.\,\ref{fig:chp3_patterns_ULA} when $N$ is fixed.
This is because the overall pattern is the superposition of individual patterns of every active element, which are distorted as shown in Fig.\,\ref{fig:chp4_WARP_pattern_DoE_ULA_1stelement}, where the pattern has larger gain near $\theta_B=0^{\circ}$ and smaller gain near $\theta_B=90^{\circ}$. 
Furthermore, the pointing becomes worse as $\theta_B$ increases, because when $\theta_B>30^{\circ}$, $G(\theta_B,\theta_B)$ is no longer the maximum value.

To compare with the ULA, the pattern of the UCA with $N=8$ and $\Delta d=0.5\lambda$ is measured.
The antenna elements are placed as shown in Fig.\,\ref{fig:chp2_UCA}.
Similar to the second experiment, several values $\theta_B\in\{0^{\circ},5^{\circ},10^{\circ},15^{\circ},20^{\circ}\}$ are chosen.
All patterns are normalized and plotted in Fig.\,\ref{fig:chp4_WARP_pattern_DoE_UCA}.

\begin{figure}
\centering
\includegraphics[scale=0.9]{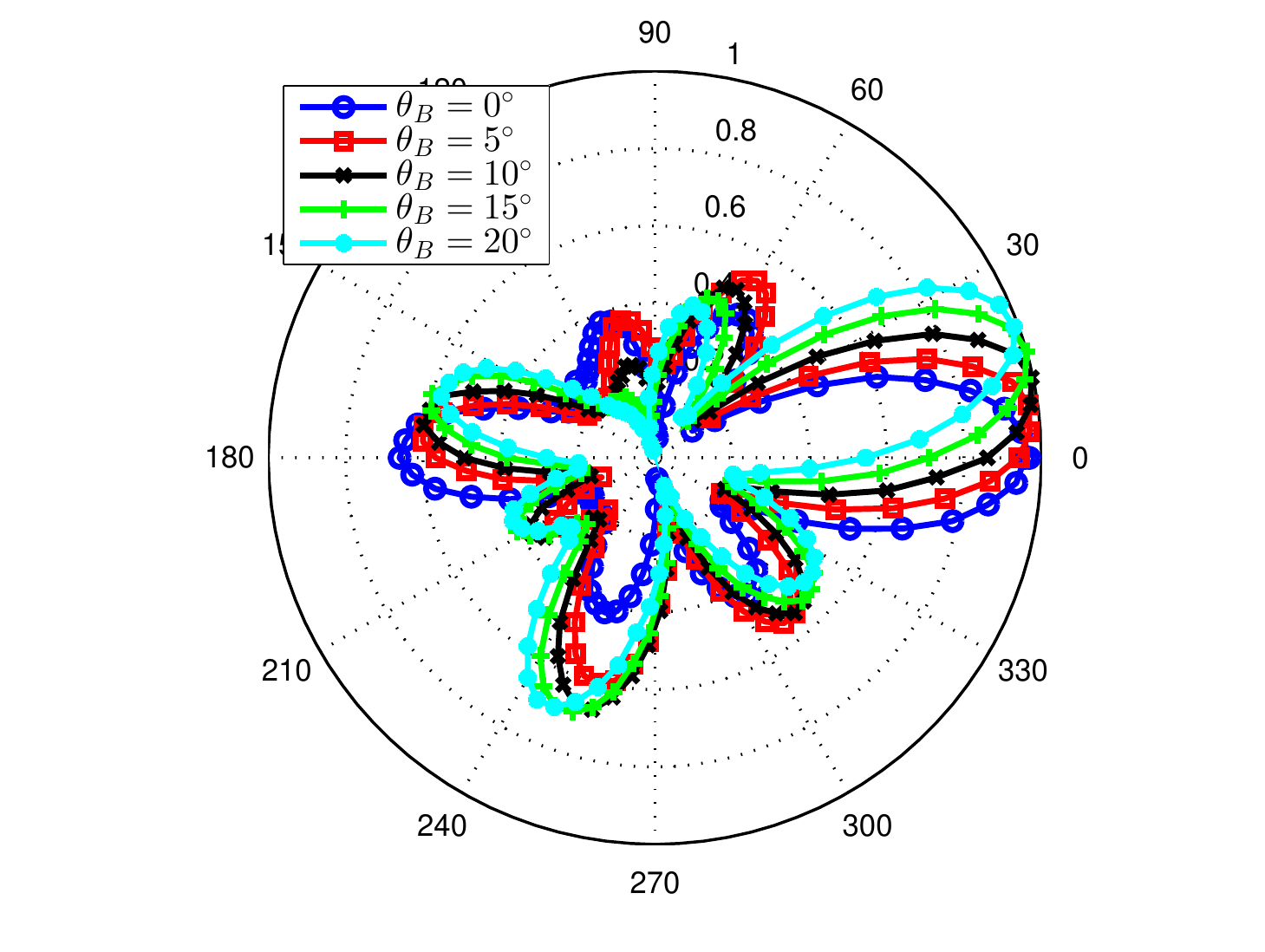}
\caption{3rd WARP experiment: $G_C(\theta,\theta_B)$}
\label{fig:chp4_WARP_pattern_DoE_UCA}
\end{figure}

As can be seen in Fig.\,\ref{fig:chp4_WARP_pattern_DoE_UCA}, the maximum gain stays nearly the same for all $\theta_B$, which is similar to the patterns shown in Fig.\,\ref{fig:chp4_patterns_UCA} for fixed $N$. 
The reason is that although each element's pattern is distorted, the distortion is symmetric.
Thus, the overall array pattern is less affected, which makes the UCA less sensitive to the mutual coupling.

\subsection{NEC Simulations}
\label{chp4:sec5:ieurhf}

It suffices to observe the mutual coupling effect from the measurements in WARP experiment.
However, it is not suitable for numerical analysis of the mutual coupling, because the observed phenomenon in WARP experiment is a result of many combined factors, such as the mutual coupling, noise and other imperfections in the real experiments.
In addition, it takes considerable time to complete the measurements for different array configuration $(N,l)$.

Unlike WARP experiments, $G(\theta,\theta_B)$ is directly calculated based on the input model instead of measuring the received signal $r$. 
To verify the observation of the mutual coupling effect on the ULA and the UCA in the WARP experiments, the NEC simulations are run corresponding to the three WARP experiments in Section\,\ref{chp4:sec5:nvweiwewvfe}. 

The input file of the specifications in the NEC simulation has been explained in Section\,\ref{appdx:warpnec:vbjhbie}.
For a fair comparison to the WARP experiments in Anechoic chamber, the free-space environment without a ground is chosen. 
For easy comparison, the patterns are uniformly normalized with regard to their maximum value.
The simulation results are shown in Fig.\,\ref{fig:chp4_NEC_pattern_DoE_ULA_1stelement}-\ref{fig:chp4_NEC_pattern_DoE_UCA}.

\begin{figure}
\centering
\includegraphics[scale=0.9]{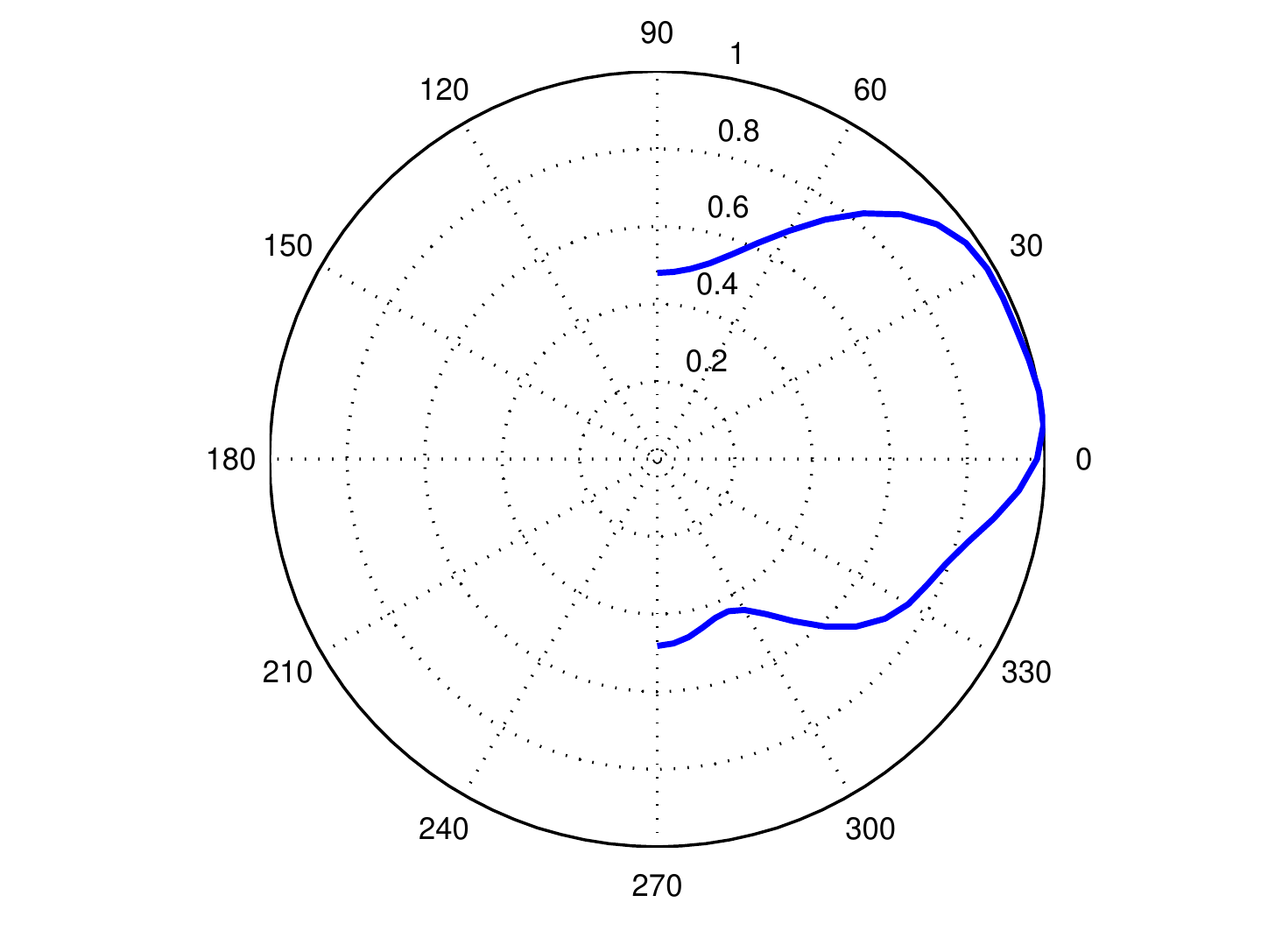}
\caption{1st NEC simulation: Pattern of the 1st element in ULA}
\label{fig:chp4_NEC_pattern_DoE_ULA_1stelement}
\end{figure}

\begin{figure}
\centering
\includegraphics[scale=0.9]{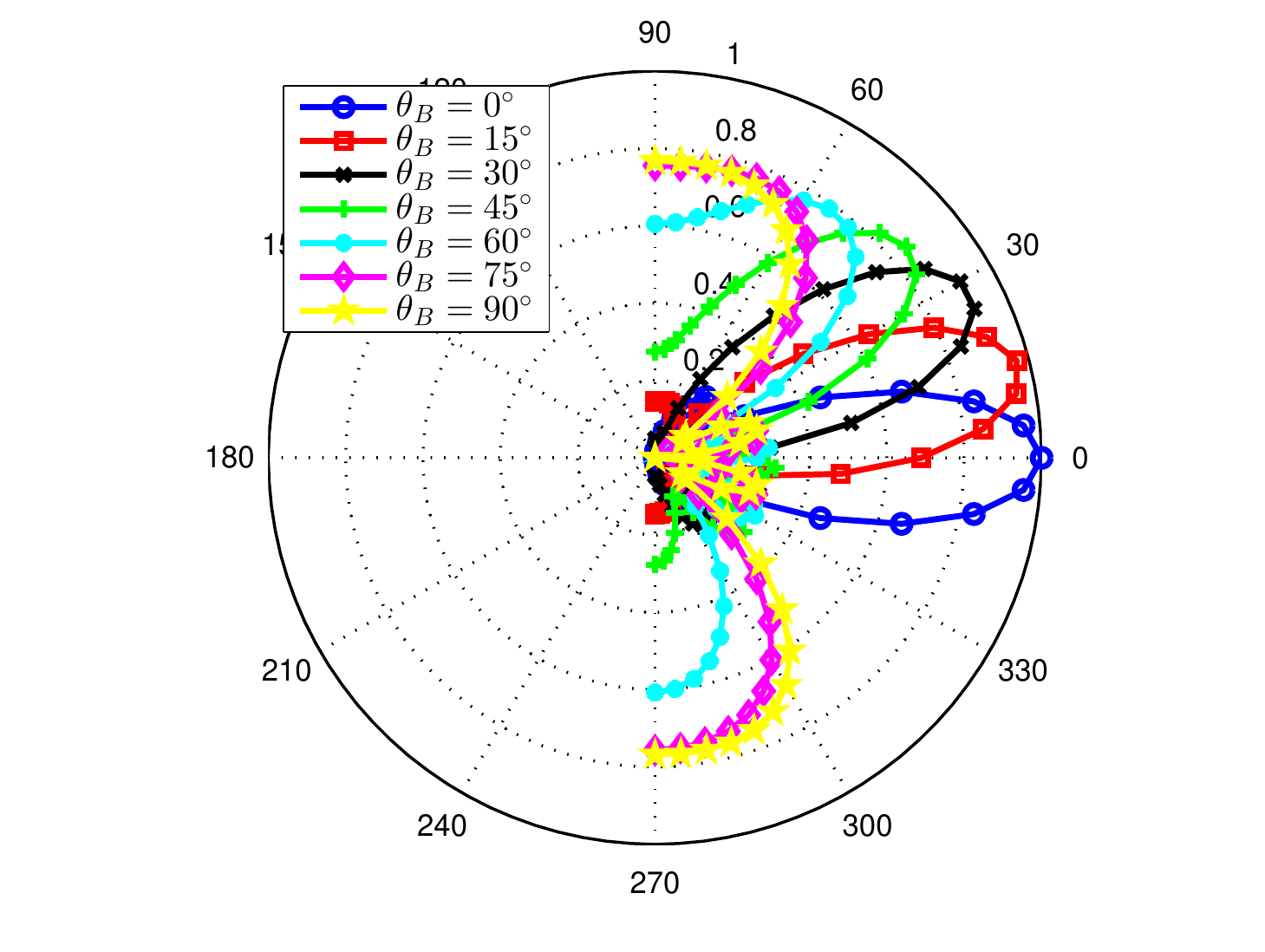}
\caption{2nd NEC simulation: $G_L(\theta,\theta_B)$}
\label{fig:chp4_NEC_pattern_DoE_ULA}
\end{figure}

\begin{figure}
\centering
\includegraphics[scale=0.9]{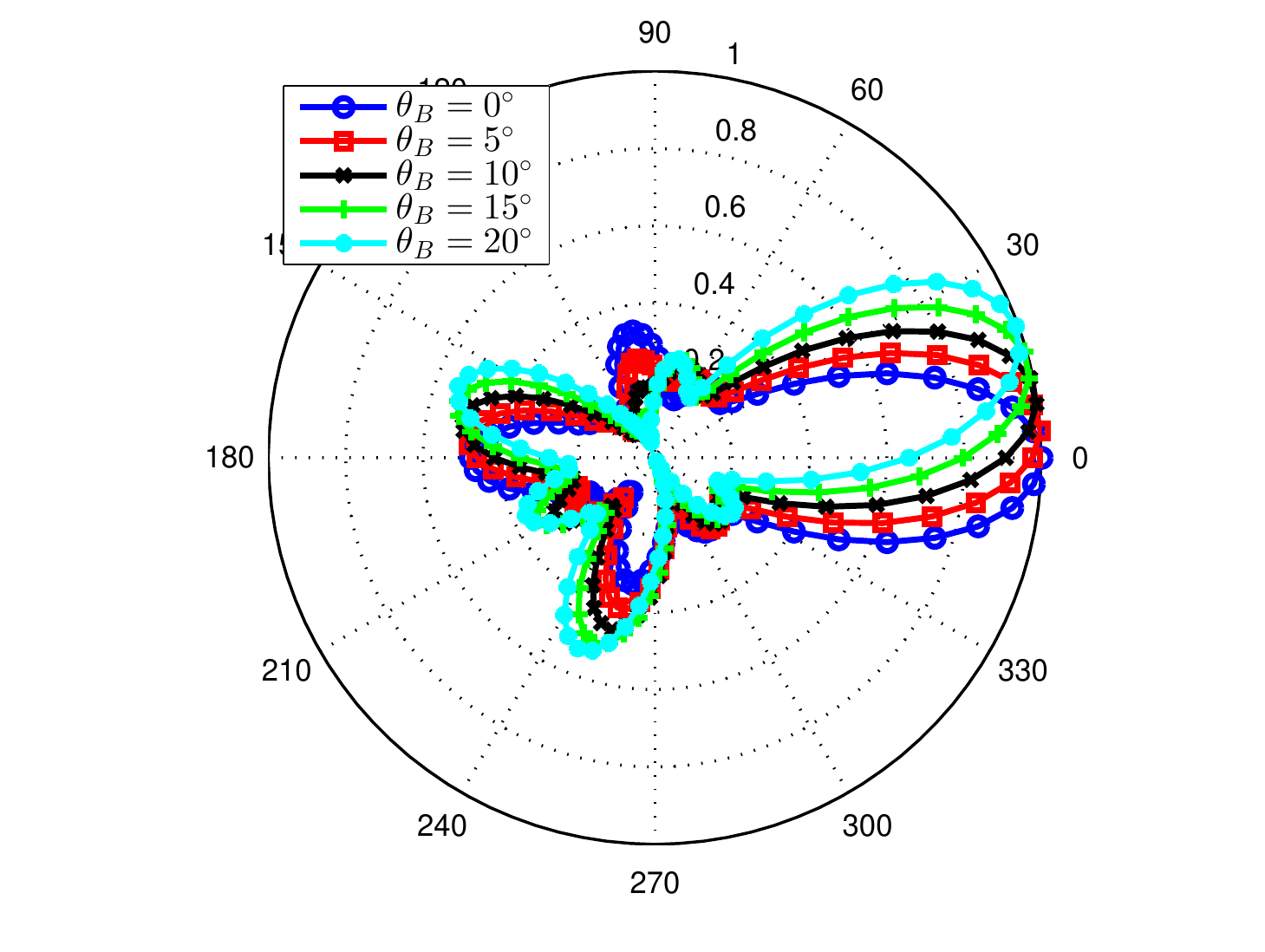}
\caption{3rd NEC simulation: $G_C(\theta,\theta_B)$}
\label{fig:chp4_NEC_pattern_DoE_UCA}
\end{figure}

Comparing Fig.\,\ref{fig:chp4_NEC_pattern_DoE_ULA_1stelement} with Fig.\,\ref{fig:chp4_WARP_pattern_DoE_ULA_1stelement}, it can be seen that both patterns are distorted due to mutual coupling effect. 
More importantly, the patterns resemble each other.
However, it is also obvious that there is more variation in the WAPR pattern.
This is because NEC simulation only includes the mutual coupling and excludes other realistic factors, such as technical malfunctions on fabrication, inaccurate alignment of array geometry and so on.

Comparing Fig.\,\ref{fig:chp4_NEC_pattern_DoE_ULA} with Fig.\,\ref{fig:chp4_WARP_pattern_DoE_ULA}, the same phenomenon of the maximum gain attenuation for the ULA can be found.
It is worth noticing that for the NEC and WARP patterns, the amount of maximum gain attenuation is different. 
This is mainly caused by the power calibration for each RF in the WARP experiments.
Comparing Fig.\,\ref{fig:chp4_NEC_pattern_DoE_UCA} with Fig.\,\ref{fig:chp4_WARP_pattern_DoE_UCA}, it can be seen that for the UCA, there is almost no maximum gain attenuation.

\subsection{Result Analysis for SSOP}
\label{chp4:sec5:vnekweiaaa}

The array patterns measured from the WARP experiments and NEC simulations for both ULA and UCA are shown in Section\,\ref{chp4:sec5:nvweiwewvfe} and\,\ref{chp4:sec5:ieurhf}.
This section examines these results more closely and analyzes them in terms of the security performance.

Through the comparison of the NEC and WAPR results, it can be seen they are very similar.
In Fig.\,\ref{fig:chp4_pattern_DoE_WARP_NEC_correlation}, the correlation coefficients between the WARP and NEC results for the ULA and the UCA are shown.
The angle range for the ULA is $\theta_B\in[0^{\circ},90^{\circ}]$, and for the UCA, it is $\theta_B\in[0^{\circ},20^{\circ}]$.
It can be seen that for both ULA and UCA, the NEC and WAPR patterns are highly correlated with $\rho>0.85$.
The value of $\rho$ ranges from -1 to 1.
Unfortunately there is no rigorous threshold value for $\rho$ that yields high correlation.
Fig.\,\ref{fig:chp4_WARP_NEC_pattern_UCA_example} shows the WARP and NEC patterns when $\rho=0.8567$ as marked in Fig.\,\ref{fig:chp4_pattern_DoE_WARP_NEC_correlation}.
It can be seen that the main beams of the two patterns are almost identical.
The shapes of sidelobes of the two patterns are very alike.
In this case, $\rho=0.85$ is considered as a high correlation.

\begin{figure}
\centering
\includegraphics[scale=0.9]{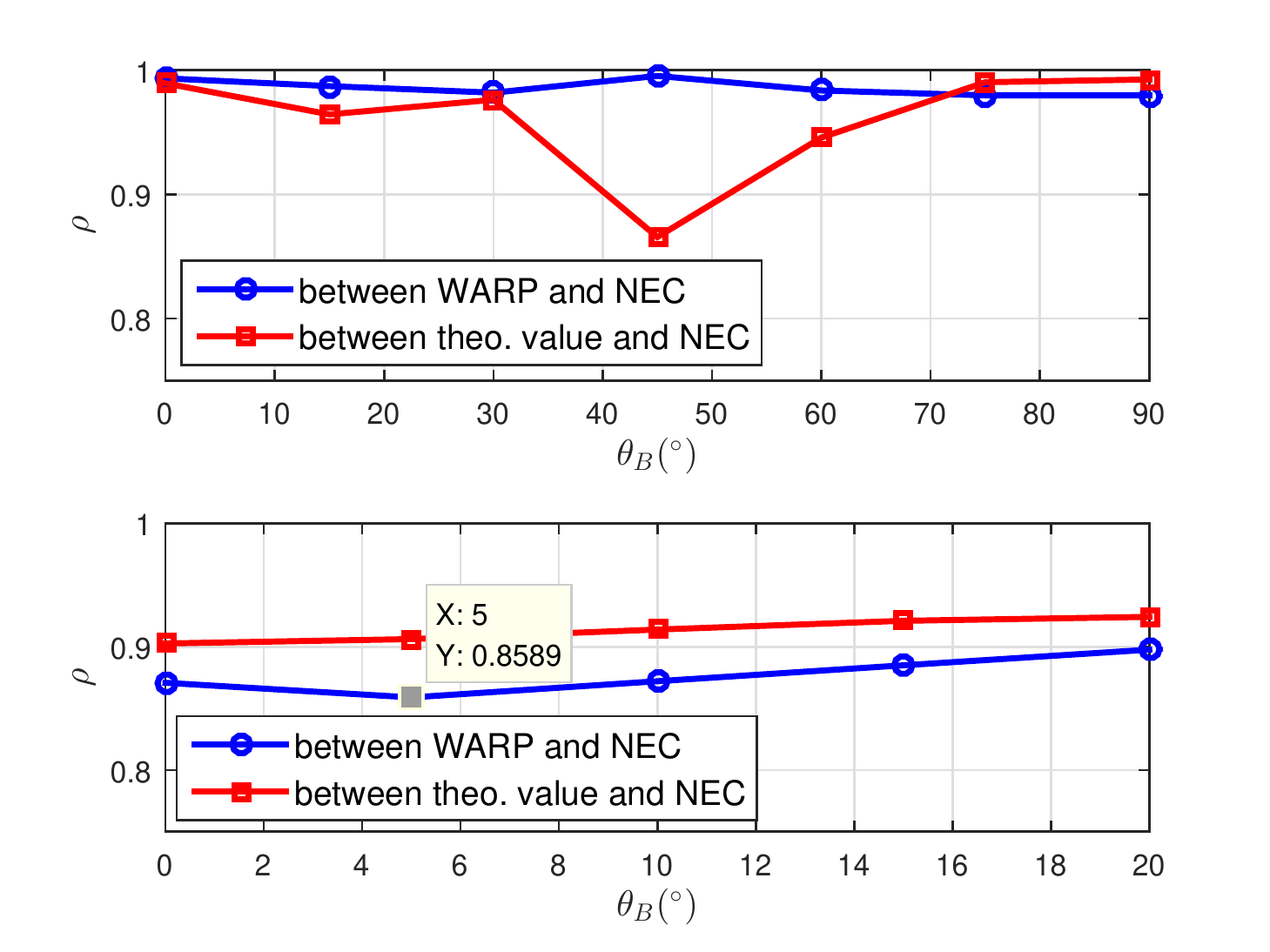}
\caption{$\rho$ between different patterns for ULA (upper plot) and UCA (lower plot)}
\label{fig:chp4_pattern_DoE_WARP_NEC_correlation}
\end{figure}

\begin{figure}
\centering
\includegraphics[scale=0.9]{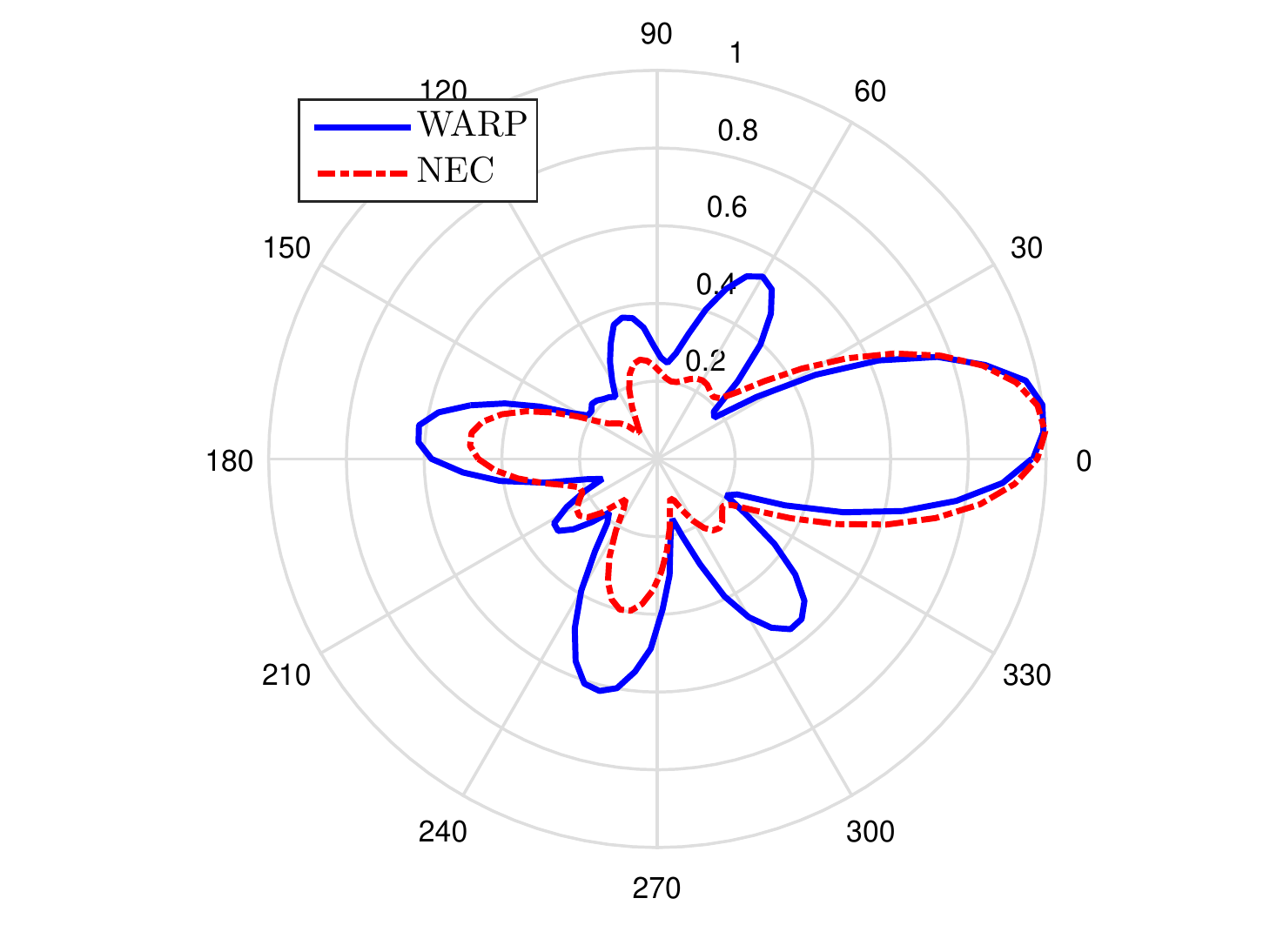}
\caption{WARP and NEC patterns for the same UCA, $\theta_B=5^{\circ}$}
\label{fig:chp4_WARP_NEC_pattern_UCA_example}
\end{figure}

As stated in Section\,\ref{chp4:sec5:ieurhf}, due to the absence of many practical factors, NEC simulations reflects purely the mutual coupling effect.
Given the high correlation between the WARP and NEC results, the NEC results can be used to study the mutual coupling effect.
Therefore, from now on, the mutual coupling effect is studied via NEC simulations.

In Fig.\,\ref{fig:chp4_pattern_DoE_WARP_NEC_correlation}, the correlation coefficients between the NEC results and the theoretical patterns for both ULA and UCA are shown as well.
For the ULA, most values of $\rho$ are larger than $0.94$, except for $\rho\approx0.85$ at $\theta_B=45^{\circ}$.
For the UCA, all values of $\rho$ are larger than $0.9$.

The high correlation between the NEC and the theoretical patterns suggests that the mutual coupling does not alter the shape of the pattern very much, although $G_{\text{max},L}$ is affected.
This means that for the UCA, the conclusions reached about the impact of $(N,R,\theta_B)$ on $A_0$ are still valid.
For the ULA, the power loss of $P_r(z_B)$ in (\ref{eq:chp4_C_B_R_B}) can be compensated by increasing $P_t$, which, however, changes $\bar{p}$.

In Section\,\ref{chp4:sec5:nvweiwewvfe} and\,\ref{chp4:sec5:ieurhf}, the patterns of ULA with $N=4$ and UCA with $N=8$ have been shown.
To fairly compare the ULA and the UCA, the ULA and the UCA are set with the same $N=8$ and $\Delta d=0.5\lambda$.
The NEC results are generated in the same way as in Section\,\ref{chp4:sec5:ieurhf}.
For the ULA, choose $\theta_B\in[0^{\circ},90^{\circ}]$ in a step of $15^{\circ}$; and for UCA, choose $\theta_B\in[0^{\circ},20^{\circ}]$ in a step of $5^{\circ}$.
Then, the results of the UCA are expanded into the range $\theta_B\in[0^{\circ},90^{\circ}]$ according to Proposition\,\ref{prop:chp4_theta_B_range}.

From the simulated patterns, the maximum gain can be recorded for each pattern, from which $f(\theta_B)$ can be calculated.
The results of $f(\theta_B)$ are shown in the upper plot in Fig.\,\ref{fig:chp4_p_out_DoE_MC}.
It can be seen that $f_L(\theta_B)$ decreases with $\theta_B$ and the minimum value of $f_L(\theta_B)$ is about 0.7; while $f_C(\theta_B)$ is more or less flat.
It is worth noticing that for the ULA, the attenuation is not very big when $\theta_B<30^{\circ}$.

\begin{figure}
\centering
\includegraphics[scale=0.9]{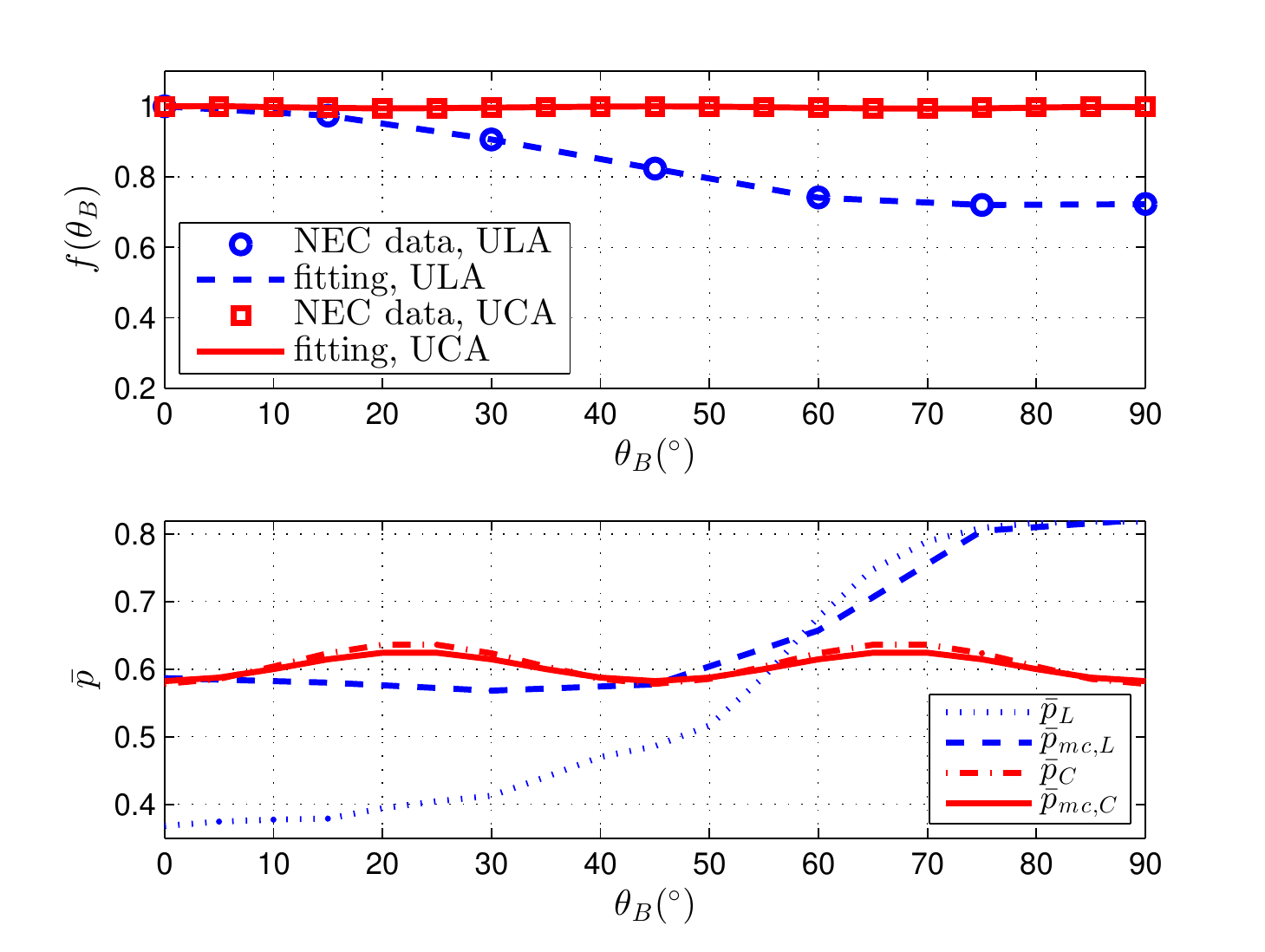}
\caption{Upper figure: mathematical fitting for $f(\theta_B)$; Lower figure: $\bar{p}$ versus $\theta_B$ with or without mutual coupling effect}
\label{fig:chp4_p_out_DoE_MC}
\end{figure}

Based on the results, $f(\theta_B)$ can be numerically fitted via MATLAB toolbox.
For example, two fitting functions are given below and the results are shown in the upper figure in Fig.~\ref{fig:chp4_p_out_DoE_MC}.
\begin{align}
	f_L(x)=&-0.8639x^6+3.783x^5-6.166x^4+4.855x^3-2.054x^2+0.199x+1, \\
	f_C(x)=&-0.2517x^6+1.186x^5-2.058x^4+1.587x^3-0.5048x^2+0.0386x+0.9999.
\end{align}

To illustrate the impact of the mutual coupling on $\bar{p}$, a special case when $K=\infty$ and $\beta=2$ is used.
Then (\ref{eq:chp4_C_B_R_B}) is reduced to
\begin{align}
	P_{rB}=\frac{P_t}{d_B^2}Nf^2(\theta_B)\geq\sigma_n^2(2^{R_B}-1).
\end{align}
When $P_t$ is adjustable, in order to compensate the attenuation $f(\theta_B)$, $P_t$ is increased by $\frac{1}{f^2(\theta_B)}$.
$\bar{p}_{mc}$ is used to denote the averaged SSOP subject to the mutual coupling.

In the lower plot in Fig.\,\ref{fig:chp4_p_out_DoE_MC}, $\bar{p}$ and $\bar{p}_{mc}$ versus $\theta_B$ are plotted for both ULA and UCA.
The curves for $\bar{p}$ are calculated based on (\ref{eq:chp3_SSOP_De}); while the curves for $\bar{p}_{mc}$ are based on the NEC simulation results.
The results show that for the UCA, $\bar{p}_{mc,C}$ is approximately equal to $\bar{p}_C$, which indicates that the UCA is less sensitive to the mutual coupling.
For the ULA, because the attenuation at $\theta_B=0^{\circ}$ is the least, the power increase affects the most.
Thus, there is bigger increase at lower region of $\theta_B$.

In summary, the mutual coupling affects the patterns of both ULA and UCA.
The pattern shape of both array are little affected.
The maximum gain of the ULA attenuates more as $\theta_B$ increases;  in the low region of $\theta_B$, the ULA suffers less from the mutual coupling.
Compared to the ULA, the UCA is less insensitive to the mutual coupling for all $\theta_B$.

\section{Conclusions}
\label{chp4:sec6}

In this chapter, the secure transmission to Bob with ER based beamforming in presence of PPP distributed is investigated with a UCA, which is compared in parallel with a ULA in the previous chapter via both analytic expressions and numerical results.
The analysis shows that for a UCA, the variation of the SSOP with the DoE angle (i.e., Bob's angle) is much smaller than that for a ULA; but the SSOP asymptotically increases with the number of elements for a UCA rather than converging to certain values;
as the radius increases, the SSOP gradually decreases with some fluctuations and approaches a fixed value.
The behaviors of the SSOP with respect to the array parameters can also be verified by the array pattern.

In complement to the theoretical analysis, the mutual coupling is investigated with the focus on the impact on the array pattern.
The experiments on WARP and the NEC simulations show that while the shape of the patterns is not severely affected, the maximum gain attenuation for the ULA is severe.
This means that the properties of the SSOP are still valid to a large extent;
however, the mutual coupling degrades the security performance for the ULA, especially in the large region of Bob's angle.

Compared to the ULA, the SSOP of a UCA is more constant in the whole range of Bob's angle, although the SSOP of a ULA is smaller at the bore-sight direction.
For the UCA, the tightness of the upper bound does not change much even for a large number of elements.
From the practical point of view, the UCA is less sensitive to the mutual coupling over Bob's angle range, and is more flexible on the choices of array configurations.
Thus, the UCA is a better choice in creating and optimizing the the SSOP.
In the following chapter, the UCA is chosen as an example to develop optimization algorithms.

\chapter{Array Configuration Optimization of Uniform Circular Arrays}
\label{chp5}

\section{Introduction}
\label{chp5:intro}

In this chapter, the security performance of the ER-based beamforming with the adjustable UCA is enhanced.
To this end, the system performance metric, i.e., the SSOP, is to be minimized.
Two numerical optimization algorithms are developed for different transmit power constraints, and are examined against the mutual coupling effect in practice.

The conclusions from Chapter\,3 and\,4, where the security performance of the ER-based beamforming with the ULA and the UCA are studied and compared with respect to the channel parameters and the array parameters, provide some insights on the possibility to optimize the array parameters to achieve higher level of security.
The goal of this chapter is to enhance the security by designing the optimization algorithms based on the previous observations and conclusions.

To achieve this goal, the array parameters need to be jointly analyzed via numerical methods, which provides more accurate results than the pure theoretical analysis.
First, the key parameters that affect the security performance are identified, and are used to formulate the optimization problem.
Then, based on the analysis on the key parameters, two numerical algorithms are developed as solutions to the optimization problem for both adjustable transmit power and fixed transmit power.
The algorithms can be generalized for any channel parameter.
In addition, the mutual coupling is investigated for a wider range of parameters via NEC simulations than in Chapter\,4, and its impact on the optimization algorithms is revealed.

This chapter is organized as follows.
In Section\,\ref{chp5:syst}, the system model for the UCA as well as the basic concepts that are used in this chapter are introduced.
In Section\,\ref{chp5:prob_form}, the optimization problem is formulated.
In Section\,\ref{chp5:analysis}, the key parameters in the optimization problem are jointly analyzed.
In Section\,\ref{chp5:opt_alg}, two numerical optimization algorithms are developed and the error analysis for the configurable beamforming technique is given.
In Section\,\ref{chp5:mc_err}, more detailed analysis of the mutual coupling on the UCA is given and its impact on the optimization algorithms is investigated.
In Section\,\ref{chp5:concl}, the conclusions of this chapter are given.

\section{System Model}
\label{chp5:syst}
\subsection{System Model with Adjustable UCA}
\label{chp5:syst:rqewvz}

Consider a dense wireless communications system with a large number of users that are distributed by a homogeneous PPP with density $\lambda_e$.
The AP wishes to transmit to Bob in presence of Eves.
The system model is similar to Section\,\ref{chp3:syst:model} and\,\ref{chp4:sec2:ownvw}, except that the AP is equipped with an adjustable UCA.
Assume that Bob's CSI, or coordinates, is available at the AP, while the knowledge of Eves' CSI or coordinates are not known.
In addition, assume that the channel does not vary between the current transmission and the next transmission.
Thus, the array configuration can be adjusted according to Bob's CSI or coordinates for the next transmission.

An example of the adjustable UCA is shown in Fig.\,\ref{fig:chp5_UCA_systemmodel}.
Unlike in Chapter\,3 and\,4, the total number of elements and the number of active elements are distinguished in this chapter.
The active elements are the elements that are used during transmission, while others remain silent or unused.
To avoid ambiguity, $N_{\text{max}}$ is used for the total number of elements, while $N$ denotes the number of active elements.
In addition, $P_{t,\text{max}}$ is used to denote the maximum available transmit power, while $P_t$ refers to the actual transmit power, which is also adjustable.
Therefore, for certain UCA, $N_{\text{max}}$ and $P_{t,\text{max}}$ are fixed values, while $N$ and $P_t$ are adjustable and do not exceed $N_{\text{max}}$ and $P_{t,\text{max}}$, respectively.

\nomenclature{$N_{\text{max}}$}{total number of elements in the array}
\nomenclature{$P_{t,\text{max}}$}{maximum available transmit power}

\begin{figure}
\centering
\includegraphics[scale=1]{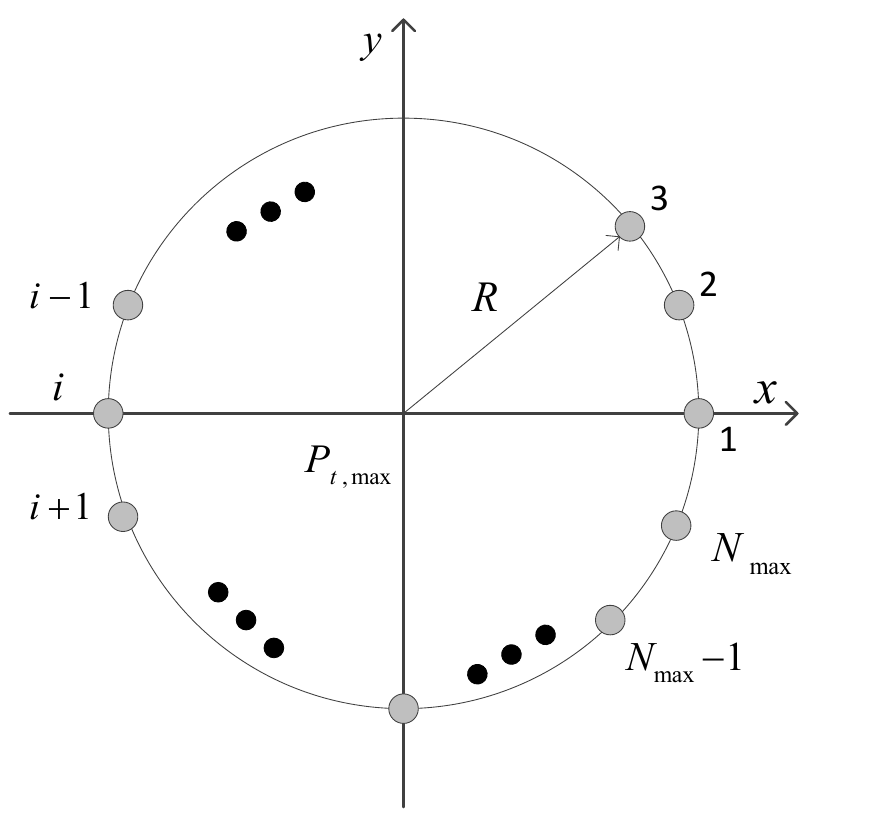}
\caption{UCA with $N_{\text{max}}$, $R$ and $P_{t,\text{max}}$}
\label{fig:chp5_UCA_systemmodel}
\end{figure}

Although there is no limit for $R$ in theory, it is usually not very large in practice.
For example, a 1\,m-diameter array is considered very large for an indoor AP.
The commercial circular WLAN phased array terminal FCI-3710 developed by Fidelity Comtech has 15.24\,cm radius.
In this chapter, the range is chosen as $R\in[5\text{\,cm},25\text{\,cm}]$, i.e., $[0.4\lambda,2\lambda]$ for the $2.4$\,GHz carrier frequency.
In addition, without the need to compare with the ULA, $(N,R)$ is used as the array configuration instead of $(N,l_a)$.

In this chapter, it is assumed that $N_{\text{max}}$ is fixed for certain value, e.g., $8$ elements, because it is usually fixed for certain device.
For example, the aforementioned FCI-3710 has 8 elements.
Another example is that a single WARP node can hold 4 RF interfaces.
On the other hand, although the radius $R$ is difficult to change during transmission, it can be chosen as certain optimum value against the security performance.
Thus, the array configuration $(N,R)$ as well as the transmit power $P_t$ can be adapted or optimized according to Bob's dynamic location. 

The channel gain vector $\mathbf{h}(z)$ for the generalized Rician channel in (\ref{eq:chp3_CH_Gain}) is used.
Thus, the received signal, channel capacity and etc. can be calculated according to the general expressions in Section\,\ref{chp3:syst:model}.

\subsection{Array Mode}
\label{chp5:syst:bnmo}

$N_{\text{max}}$ is usually unadjustable once the array is installed.
On the contrary, it is relatively simple to control $N$.
For example, the antenna elements in the array could be electronically switched on and off, or it can be done from the baseband by not generating data packets or being weighted with zero.
Thus, the AP can decide $N$ for the next transmission.

For the UCA with $N_{\text{max}}$ elements, the range of $N$ is from $1$ to $N_{\text{max}}$.
However, only the values of $N$ that give a UCA are considered.
In other words, a $N$-element sub-array is picked.
The discussion of the non-uniform circular array is beyond the scope of this thesis.

For convenience, the term `array mode' is used to refer to the sub-array. 
Let $M_{ij}$ denote a particular array mode.
The set $\{M_{ij}\}$ refers to all possible array configurations for certain UCA with $N_{\text{max}}$.
The first index $_i$ is associated with the number of active elements $N$.
$\{M_{i}\}$ is a subset of $\{M_{ij}\}$, which contains all $M_{ij}$ with the same $N$, but different angles.
The second index $_j$ is associated with the angle of $M_{ij}$ in $\{M_{i}\}$.
Take the 8-element UCA in Fig.\,\ref{fig:chp2_UCA} as an example, there are 2, 4, 8 elements that can form the sub-arrays. 
When $N=4$, there are two array modes with $45^{\circ}$ angle difference; one has elements (1,3,5,7) and the other has elements (2,4,6,8).

\nomenclature{$M_{ij}$}{array mode}

In this chapter, the index $_i$ is in descending order of $N$.
$i=1$ is assigned to the $N_{\text{max}}$-element UCA, and there is only one array mode in $\{M_1\}$, i.e., $M_1$. 
For the same example in Fig.\,\ref{fig:chp2_UCA}, $\{M_2\}$ has 4 elements and $\{M_3\}$ has 2 elements.
The index $_j$ is in ascending order of first element in the sub-arrays.
For the same example in Fig.\,\ref{fig:chp2_UCA}, the first element in the sub-array for $M_{21}$ is element 1, and the first element in the sub-array for $M_{22}$ is element 2.
In Fig.\,\ref{fig:chp5_arraymode}, $\{M_{ij}\}$ for the UCA with $N_{\text{max}}=8$ is shown.
There are 7 array modes in total.
In particular, $\{M_3\}$ has $2$ elements, which is equivalent to a 2-element ULA.

\begin{figure}
\centering
\includegraphics[scale=1]{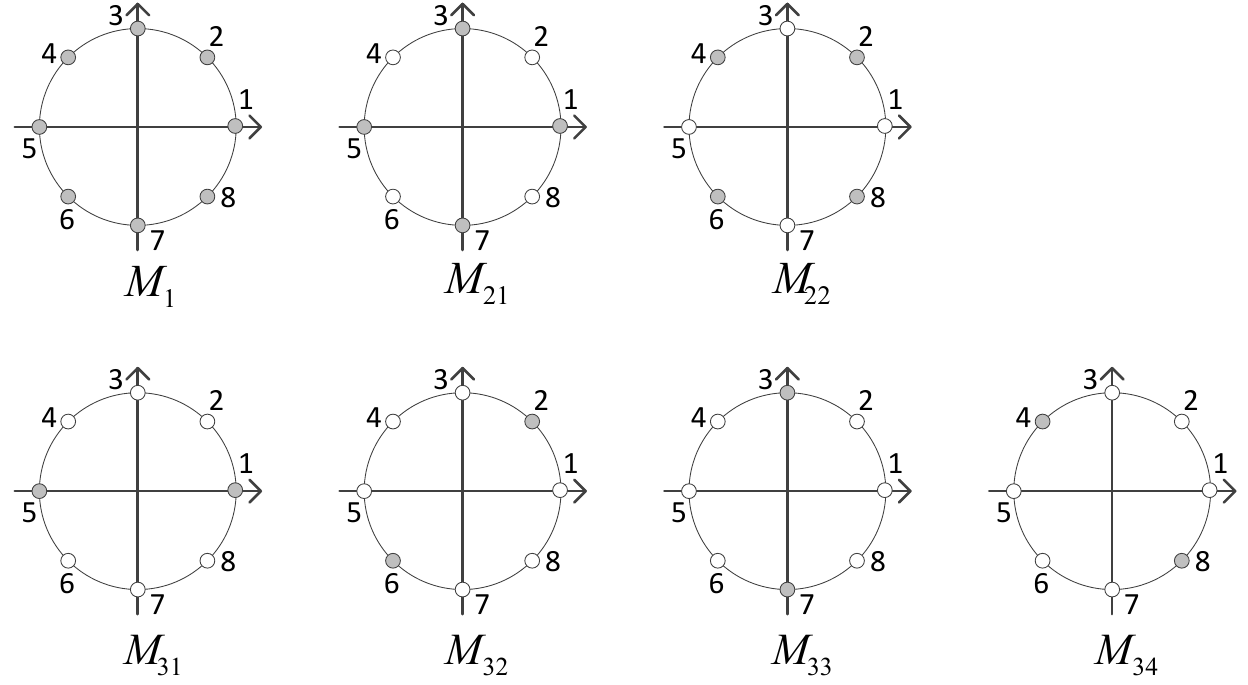}
\caption{$\{M_{ij}\}$ for an UCA with $N_{\text{max}}=8$; gray elements are active.}
\label{fig:chp5_arraymode}
\end{figure}

The purpose of distinguishing $M_{ij}\in\{M_i\}$ is that they have different $\theta_{\text{doe}}$ given the same $\theta_B$ and different $G(\theta,\theta_{\text{doe}})$ generates different SSOP.
For example, when $\theta_B=5^{\circ}$, the DoE angle for $M_{21}$ is $\theta_{\text{doe}}=\theta_B=5^{\circ}$; 
however, for $M_{22}$, it is $\theta_{\text{doe}}=\theta_B-45^{\circ}=-40^{\circ}$ due to the $45^{\circ}$ rotation between $M_{21}$ and $M_{22}$.

\subsection{Coverage Zone for Bob}
\label{chp5:syst:vndkawia}

The intuition brought by the adjustable UCA is that $N$ can be adjusted to achieve a lower SSOP; in the mean time, $C_B\geq R_B$ needs to be guaranteed.
In (\ref{eq:chp4_C_B_R_B}), $C_B\geq R_B$ is converted into
\begin{align}\label{eq:chp5_constraint_ownevow}
	P_{rB}\geq \sigma_n^2(2^{R_B}-1).
\end{align}
According to (\ref{eq:chp3_receivedpower}) and (\ref{eq:chp3_h_tilde_square}), $P_{rB}$ is given by,
\begin{align}\label{eq:chp5_P_rB}
	P_{rB} = \frac{P_t}{d_B^{\beta}}\Big( \frac{KN}{K+1}+\frac{1}{K+1}g_{Re}^2+\frac{1}{K+1}g_{Im}^2+\frac{2\sqrt{KN}}{K+1}g_{Re} \Big).
\end{align}
$P_{rB}$ depends on $N$.
In other words, $N$ be adjusted based on the current Bob's CSI (e.g., $P_{rB}$) for the next transmission.

For certain $P_t$ and $N$, $P_{rB}$ is a random variable due to random $d_B$ and the Rician fading (i.e., $g_{Re}$ and $g_{Im}$).
If a certain channel realization (i.e., $g_{Re}$ and $g_{Im}$) is known, $d_B$ can be estimated, which can be used to guide the adjustment of $P_t$ or $N$ for the next transmission.
However, due to the unknown fading realization, it is impossible to estimate $d_B$ by $P_{rB}$ alone.

In order to make progress, the mean value $\mathbb{E}[P_r(z_B)]$ is used to provide guidance on the adjustment of $P_t$ or $N$.
Using (\ref{eq:chp5_P_rB}), $\mathbb{E}[P_r(z_B)]$ is obtained by
\begin{align}\label{eq:chp5_meanP_rB}
	\mathbb{E}[P_r(z_B)]=\frac{P_t}{d_B^{\beta}}\frac{KN+1}{K+1}.
\end{align}
In this way, the small-scale fading is averaged out.
Replacing $P_{rB}$ with $\mathbb{E}[P_r(z_B)]$ in (\ref{eq:chp5_constraint_ownevow}), it can be derived that
\begin{align}\label{eq:chp5_constraint}
	\frac{P_t}{d_B^{\beta}}\frac{KN+1}{K+1}\geq \sigma_n^2(2^{R_B}-1).
\end{align}
For fixed $P_t$, $N$ and ($K,\beta$), $d_B$ should not exceed certain threshold to satisfy (\ref{eq:chp5_constraint}), which can be expressed by
\begin{align}
	d_B\leq\Big[ \frac{P_t}{\sigma_n^2(2^{R_B}-1)}\frac{KN+1}{K+1} \Big]^{\frac{1}{\beta}}.
\end{align}
Let $d_{th,N}$ denote the threshold for $d_B$ for certain channel condition (i.e., $K$ and $\beta$) and certain $P_t$,
\begin{align}\label{eq:chp5_d_th}
	d_{th,N}=\Big[\frac{P_t}{\sigma_n^2(2^{R_B}-1)}\frac{KN+1}{K+1}.\Big]^{\frac{1}{\beta}}.
\end{align}

\nomenclature{$d_{th,N}$}{coverage distance for $N$ active antennas}

For certain $P_t$, $K$ and $\beta$, upon acquisition of $P_{rB}$, the AP will calculate $d_B$ using $P_{rB}$ to replace $\mathbb{E}[P_r(z_B)]$ in (\ref{eq:chp5_meanP_rB}) and assumes that the calculated $d_B$ is the true distance.
Then, this value is compared with $d_{th,N}$ for all available $N$.
All the values of $N$ that satisfy $d_B\leq d_{th,N}$ are eligible for the next transmission, from which the optimum value of $N$ that gives the minimum SSOP will be chosen.
Since the channel realization for the next transmission is assumed to be unchanged, the estimated $d_B$, even though not the true value, suffices to provide guidance for the adjustment of $N$.
For convenience, in this chapter, Bob's distance is still used to refer to the assumed Bob's distance by the AP, unless otherwise stated.

The coverage zone is defined based on $d_{th,N}$ for different $N$.
Fig.\,\ref{fig:chp5_N_zone} shows an example of $d_{th,N}$ for the UCA with $N_{\text{max}}=8$. 
$d_{th,N}$ in (\ref{eq:chp5_d_th}) is proportional to $N$.
The concentric circles shows the coverage distance for different $N$.
For the maximum available transmit power $P_{t,\text{max}}$, $d_{th,N_{\text{max}}}$ is the maximum coverage distance, which is denoted by $d_{\text{max}}$,
\begin{align}\label{eq:chp5_d_max}
	d_{\text{max}}=\Big[\frac{P_{t,\text{max}}}{\sigma_n^2(2^{R_B}-1)}\frac{KN_{\text{max}}+1}{K+1}\Big]^{\frac{1}{\beta}}.
\end{align}

\nomenclature{$d_{\text{max}}$}{maximum coverage distance}

\begin{figure}
\centering
\includegraphics[scale=1]{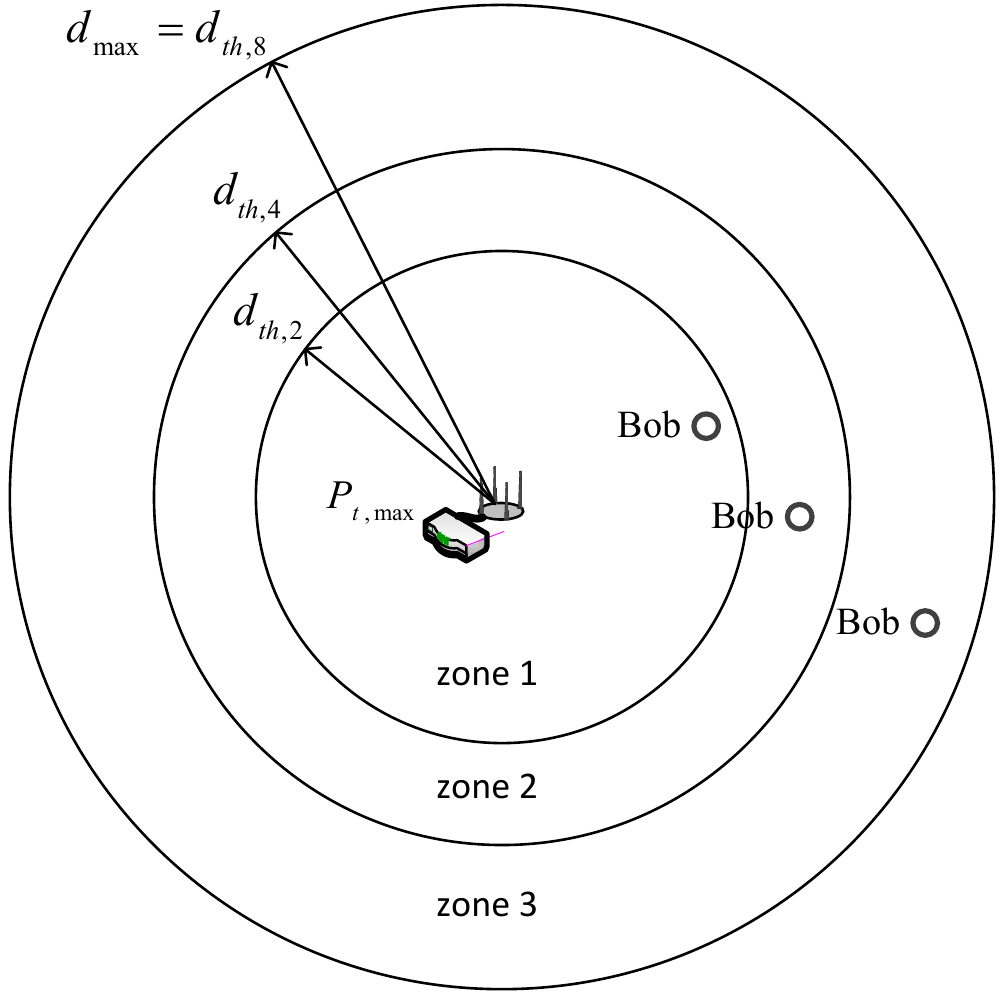}
\caption{$d_{th,N}$ for the UCA with $N_{\text{max}}=8$ and $P_{t,\text{max}}$}
\label{fig:chp5_N_zone}
\end{figure}

The term 'zone' is used to describe the annular/circular area divided by $d_{th,N}$ with different $N$.
In Fig.\,\ref{fig:chp5_N_zone}, there are 3 zones for the UCA with 8 elements and $P_{t,\text{max}}$.
Zone 1 is the area within the $d_{th,2}$ curve;
zone 2 is the annular area between the $d_{th,2}$ and $d_{th,4}$ curves;
zone 3 is the annular area between the $d_{th,4}$ and $d_{th,8}$ curves.

For Bob being in different zone, an appropriate number of active elements $N$ should be used.
For example, when Bob is in zone 2, $N$ can be either $4$ or $8$, but not $2$.
For the UCA with $N_{\text{max}}$ elements and $P_{t,\text{max}}$, Bob's reliable transmission can be guaranteed only when $d_B\leq d_{\text{max}}$.

Since Bob is uniformly distributed in the coverage zone, it is interesting to know the area of these zones. 
Let $S$ be the total coverage area and $S^{(k)}$ be the area of zone $k$, $k=1,2,3$.
For example, the zone area in Fig.\,\ref{fig:chp5_N_zone} can be calculated by
\begin{align}
	&S^{(1)}=\pi d_{th,2}^2, \label{eq:chp5_zone1_area} \\
	&S^{(2)}=\pi (d_{th,4}^2-d_{th,2}^2), \label{eq:chp5_zone2_area} \\
	&S^{(3)}=\pi (d_{\text{max}}^2-d_{th,4}^2), \label{eq:chp5_zone3_area} \\
	&S=\sum_k S^{(k)}= \pi d_{\text{max}}^2 \label{eq:chp5_zone_area}.
\end{align}
Thus, the probability that Bob is in zone $k$, denoted by $q^{(k)}$, is
\begin{align}\label{eq:chp5_zone_prob}
	q^{(k)}=\frac{S^{(k)}}{S}.
\end{align}


\nomenclature{$^{(k)}$}{variable related to zone $k$}
\nomenclature{$S$}{coverage area}
\nomenclature{$q^{(k)}$}{probability that Bob is in zone $k$}

\section{Problem Formulation}
\label{chp5:prob_form}

The aim of this chapter is to enhance the security level of the transmission from the AP to Bob with the ER-based beamforming on the adjustable UCA.
As the system performance metric, the SSOP is to be minimized.
Given the difference from the system models in Section\,\ref{chp5:syst:rqewvz} and in Section\,\ref{chp4:sec2:ownvw}, the expressions of the SSOP and its upper bound for the adjustable UCA  are slightly different.

In Section\,\ref{chp4:sec2:ownvw}, $c_0=\frac{P_t}{\sigma_n^2(2^{R_B-R_S}-1)}$ is regarded as a constant.
However, the impact of $P_t$ is considered in this chapter.
For convenience, a new constant that excludes $P_t$ is defined.
Denoted by $c_1$, it is given by
\begin{align}
	c_1=\frac{1}{\sigma_n^2(2^{R_B-R_S}-1)}.
\end{align}
Therefore, $\bar{p}_C$ in (\ref{eq:chp4_meanSSOP_Ri}) and $\bar{p}_{up,C}$ in (\ref{eq:chp4_meanSSOP_up}) can be re-written by
\begin{align}
	&\bar{p}_C=1-\int_{-\infty}^{\infty}\int_{-\infty}^{\infty} \text{exp}\Big\{-\frac{\lambda_e}{2}(c_1P_t)^{\frac{2}{\beta}}\int_0^{2\pi}\Big[\frac{KG_C^2(\theta,\theta_B)}{K+1} \nonumber \\
	& \qquad +\frac{x^2+y^2}{K+1}+\frac{2\sqrt{K}G_C(\theta,\theta_B)}{K+1}x\Big]^{\frac{2}{\beta}}\,\mathrm{d}\theta\Big\} \frac{e^{-(x^2+y^2)}}{\pi} \,\mathrm{d}x\,\mathrm{d}y, \label{eq:chp5_meanSSOP_Ri} \\
	&\bar{p}_{up,C}= 1-\text{exp}\Big\{-\lambda_e\pi(c_1P_t)^{\frac{2}{\beta}}\Big[\frac{KA_{0,C}}{2\pi(K+1)}+\frac{1}{K+1}\Big]^{\frac{2}{\beta}}\Big\}, \label{eq:chp5_meanSSOP_up}
\end{align}
where different forms of $A_{0,C}$ are given in (\ref{eq:chp4_A0C}), (\ref{eq:chp4_A_0C_2}), (\ref{eq:chp4_A0C_even}) and (\ref{eq:chp4_A0C_odd}).

There are four groups of parameters that affect $\bar{p}_C$ in (\ref{eq:chp5_meanSSOP_Ri}) and $\bar{p}_{up,C}$ in (\ref{eq:chp5_meanSSOP_up}).
The first group includes some constants, i.e., the system requirements parameters $(R_B,R_s)$, the noise variance $\sigma_n^2$ and the density $\lambda_e$.
The second group includes the channel parameters, i.e.,  $K$ and $\beta$.
The third group includes $P_t$ and $(N,R)$, which can be controlled by the AP and forms the basis of the optimization problem.
The fourth group includes Bob's location, $(d_B,\theta_B)$.

As concluded in Section\,\ref{chp4:sec4:psodpv}, $\bar{p}_{up,C}$ is tight to $\bar{p}_C$, and can be used to predict the behaviors of $\bar{p}_C$ with respect to the $(N,R,\theta_B)$.
Furthermore, according to Proposition\,\ref{prop:chp3_SSOP_up_analysis}, for the deterministic channel when $\beta=2$, $\bar{p}_{up,C}=\bar{p}_C$.
$\bar{p}_{up,C}$ is also positively correlated to $A_{0,C}$ which is tractable to analytically analyze.
Therefore, in this chapter, the optimization of $\bar{p}_C$ starts from the deterministic channel when $\beta=2$.
Then, the developed optimization algorithms are extended to the generalized Rician channel.

In Section\,\ref{chp4:sec3}, $(N,R)$ in the third group of parameters together with $\theta_B$ have been separately studied.
For the purpose of designing optimization algorithms, they will be jointly investigated for the adjustable UCA with $P_t$ and $M_{ij}$.

In this chapter, two scenarios are considered together.
For Bob being at a particular location, the most appropriate $M_{ij}$ and $P_t$ can be chosen according to $(d_B,\theta_B)$.
For Bob being randomly located in the coverage area, while $N_{\text{max}}$ is assumed to be fixed as mentioned in Section\,\ref{chp5:syst:rqewvz}, $R$ can be designed for all possible Bob's locations to achiever higher security level.
In other words, the average security performance over all possible Bob's locations should to be evaluated.

Based on the previous analysis, the optimization of $\bar{p}_C$ can be stated as follows.
In order to enhance the security level of the system, $\bar{p}_C$ is to be minimized by adjusting $P_t$ and $M_{ij}$ (i.e., $N$)  and designing $R$ according to Bob's dynamic location $(d_B,\theta_B)$ for given $(K,\beta)$ and constants $\lambda_e$ and $c_1$.
In addition, to guarantee a reliable transmission to Bob, as mentioned in Section\,\ref{chp5:syst:vndkawia}, the constraint $C_B\geq R_B$ needs to satisfied by adjusting $P_t$ and $N$ according to $d_B$.
The optimization problem can be formulated by
\begin{align}
	&\min \bar{p}_C(P_t,N,R,d_B,\theta_B)\;\forall (d_B,\theta_B), \label{eq:chp5_problem} \\
	&\;\text{s.t.}\; C_B\geq R_B. \label{eq:chp5_problem_constraint}
\end{align}

As previously discussed, first, the optimization problem will  be analyzed for $K\to\infty$ and $\beta=2$, in which case $\bar{p}_C=\bar{p}_{up,C}$, and is given by
\begin{align}\label{eq:chp5_p_De_beta_is_2}
	\bar{p}_C=1-\text{exp}\Big(-\frac{\lambda_ec_1P_tA_{0,C}}{2} \Big).
\end{align}
According to (\ref{eq:chp5_constraint_ownevow}) and (\ref{eq:chp5_P_rB}), the constraint in (\ref{eq:chp5_problem_constraint}) can be written as
\begin{align}\label{eq:chp5_problem_constraint2}
	P_tN\geq \sigma_n^2(2^{R_B}-1)d_B^{\beta}.
\end{align}

In practice, it is not always applicable or desirable to implement the transmit power control (TPC) on the downlink transmission.
For example, typical Wi-Fi systems do not implement TPC.
Most WLAN APs, such as Cisco WLAN controller, only provide a few power levels, and the dynamic TPC is not supported during transmission.
Thus, all transmissions are sent at the same power, e.g., 35\,dBm\,\cite{1400008}.
This can be partly attributed to the fact that energy consumption is not a critical issue because the AP is normally connected to the power line and the implementation of TPC is not at no cost.
Furthermore, there is also doubt about the effectiveness of a fine-grained TPC in the indoor environment\,\cite{shrivastava2007understanding}. 
On the other hand, there is no doubt about the usefulness of TPC in wireless communications.
Therefore, in this chapter, the two transmit power constraints, i.e., with/without adjusting $P_t$ are studied separately.

\nomenclature{TPC}{transmit power control}

In Chapter\,4, approximations of $\bar{p}_{up,C}$ have been used to analyze the behavior of $A_{0,C}$ and $\bar{p}_C$.
However, for the purpose of optimizing $\bar{p}_C$, the numerical results, which are more accurate, will be mainly relied on in this chapter.
In the next section, the third and fourth groups of parameters, i.e., $(P_t,N,R)$ and $(d_B,\theta_B)$, will be jointly studied.

\section{Problem Analysis}
\label{chp5:analysis}
\subsection{Array Dimension and Averaged SSOP over Bob's Locations}
\label{chp5:analysis:vneiwenve}

In Section\,\ref{chp4:sec3:njowq}, the impact of $R$ has been analyzed.
In the low region of $R$, e.g., $[0.4\lambda,2\lambda]$, $\bar{p}_C$ fluctuates without obvious decreasing as $R$ changes.
Thus, it is likely that the minimum value of $\bar{p}_C$ is not given by the largest or smallest value of $R$.

Examples of $\bar{p}_C$ versus $R$ for different $\theta_B$ and $N$ are shown in Fig.\,\ref{fig:chp5_p_R_All}.
For the purpose of MATLAB simulation, $R$ takes value every 1\,cm in $[0.4\lambda,2\lambda]$.
In the upper plot, some typical values of $\theta_B$, i.e., $\theta_B=0^{\circ}$, $10^{\circ}$, $20^{\circ}$, are taken for the UCA with $N_{\text{max}}=8$. 
In the lower plot, all possible $N$ are taken for UCA with $N_{\text{max}}=8$, i.e., $N=2,4,8$.

\begin{figure}
\centering
\includegraphics[scale=0.9]{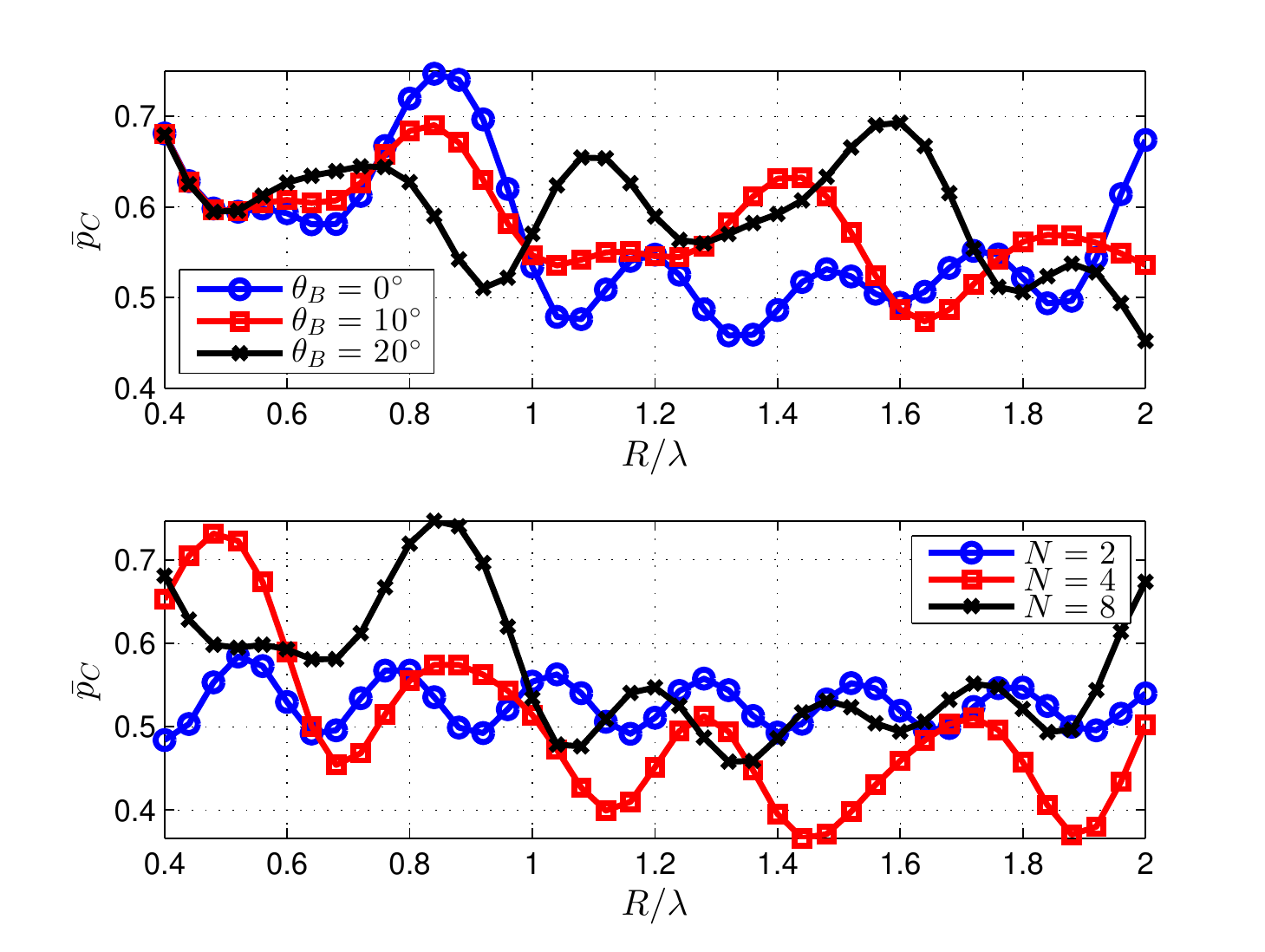}
\caption{$\bar{p}_C$ versus $R$. Upper plot: $N=8$; Lower plot $\theta_B=0^{\circ}$}
\label{fig:chp5_p_R_All}
\end{figure}

The fluctuating behavior of $\bar{p}_C$ with respect to $R$ can be observed for different $\theta_B$ and $N$.
For different $\theta_B$ and $N$, the local minimum of $\bar{p}_C$ is given by different value of $R$.
It suggests that different $R$ is required to minimize $\bar{p}_C$ for different $\theta_B$ and different $d_B$ (i.e., $N$).
However, $R$ can only be a particular value.
Therefore, $R_{opt}$ for Bob's locations needs to be found.

MMSE method is used to find $R_{opt}$ in a certain range of $R$ that produces the minimum $\bar{p}_C$ for all $(d_B,\theta_B)$.
First, $R_{opt}$ for $\theta_B\sim \mathcal{U}(0,2\pi)$ is to be found, based on which $R_{opt}$ for all $(d_B,\theta_B)$ is found.

To establish the cost function, imagine that $R$ is adjustable, which provides the hypothetical function of $\bar{p}_{C,\text{min}}$ with respect to $\theta_B$.
Notice that the value of $\bar{p}_{C,\text{min}}$ for each $\theta_B$ is in fact given by a different value of $R$, which is not practical.
To find $R_{opt}$, let the mean-square error, denoted by $\text{err}(R)$, be the mean square of the difference between $\bar{p}_C$ and $\bar{p}_{C,\text{min}}$ over the range $\theta_B\in[0,2\pi]$,
\begin{align}
	\text{err}(R)=\mathbb{E}_{\theta_B}[(\bar{p}_C-\bar{p}_{C,\text{min}})^2].
\end{align}
Thus, $R_{opt}$ can be found by 
\begin{align}\label{eq:chp5_mmse1}
	R_{opt}=\arg\min_R \text{err}(R).
\end{align}
(\ref{eq:chp5_mmse1}) can be converted into the following expression, the derivation of which is in Appendix\,\ref{appdx:bessel:mnkeorie}.
\begin{align}\label{eq:chp5_mmse2}
	R_{opt}=\arg\min_R \bar{\bar{p}}_C,
\end{align}
where $\bar{\bar{p}}_C$ is the averaged SSOP over Bob's angles and is defined by
\begin{align}\label{eq:chp5_p2bar_1}
	\bar{\bar{p}}_C=\frac{1}{2\pi}\int_{0}^{2\pi}\bar{p}_C\,\mathrm{d}\theta_B.
\end{align}

\nomenclature{$\text{err}(\cdot)$}{one-dimension mean square error}
\nomenclature{$_{opt}$}{optimum value}
\nomenclature{$\bar{\bar{p}}_C$}{averaged spatial secrecy outage probability over all possible Bob's angles or locations}

Next, $R_{opt}$ for all $(d_B,\theta_B)$ will be found.
In polar coordinates, the mean square error over two-dimension is given by 
\begin{align}\label{eq:chp5_eiruel}
	\text{err}_2(R)=\frac{1}{S}\int_{0}^{2\pi}\int_{0}^{d_{\text{max}}}d_B(\bar{p}_C-\bar{p}_{C,\text{min}})^2\,\mathrm{d}d_B\,\mathrm{d}\theta_B.
\end{align}
where $\frac{1}{S}$ is the probability that Bob is at a certain location $(d_B,\theta_B)$.
Thus, $R_{opt}$ can be found by 
\begin{align}\label{eq:chp5_mmse3}
	R_{opt}=\arg\min_R \text{err}_2(R).
\end{align}
(\ref{eq:chp5_mmse3}) can be converted into the following expression. The derivation is referred to in Appendix\,\ref{appdx:bessel:owgvow}.
\begin{align}\label{eq:chp5_mmse4}
	R_{opt}=\arg\min_R \Big(q^{(1)}\bar{\bar{p}}_C^{(1)}+q^{(2)}\bar{\bar{p}}_C^{(2)}+q^{(3)}\bar{\bar{p}}_C^{(3)}\Big),
\end{align}
where $\bar{\bar{p}}_C^{(k)}$ is the averaged SSOP over Bob's angles in zone $k$ and is defined by
\begin{align}\label{eq:chp5_p2bar_2}
	\bar{\bar{p}}_C^{(k)}=\frac{1}{2\pi}\int_{0}^{2\pi}\bar{p}_C^{(k)}\,\mathrm{d}\theta_B.
\end{align}
It can be seen that (\ref{eq:chp5_p2bar_2}) is based on (\ref{eq:chp5_p2bar_1}).
In the following, the property of the averaged SSOP over Bob's angles $\bar{\bar{p}}_C$ will be studied with respect to $R$ for fixed $N$, in order to find $R_{opt}$.

\nomenclature{$\text{err}_2(\cdot)$}{two-dimension mean square error}

Substituting the expression of $\bar{p}_C$ in (\ref{eq:chp5_meanSSOP_Ri}) into (\ref{eq:chp5_p2bar_1}), the expression of $\bar{\bar{p}}_C$ can be obtained,
\begin{align}\label{eq:chp5_p2bar_3}
	\bar{\bar{p}}_C&=1-\frac{1}{2\pi}\int_{-\infty}^{\infty}\int_{-\infty}^{\infty} \int_0^{2\pi}\text{exp}\Big\{-\frac{\lambda_e}{2}(c_1P_t)^{\frac{2}{\beta}}\int_0^{2\pi}\Big[\frac{KG_C^2(\theta,\theta_B)}{K+1} \nonumber \\
	& +\frac{x^2+y^2}{K+1}+\frac{2\sqrt{K}G_C(\theta,\theta_B)}{K+1}x\Big]^{\frac{2}{\beta}}\,\mathrm{d}\theta\Big\} \frac{e^{-(x^2+y^2)}}{\pi} \,\mathrm{d}\theta_B\,\mathrm{d}x\,\mathrm{d}y.
\end{align}
Although (\ref{eq:chp5_p2bar_3}) can be numerically calculated, it is untraceable to analytically analyze.
Thus, the upper bound, denoted by $\bar{\bar{p}}_{up,C}$, is required for theoretical analysis.

\begin{theorem}\label{th:chp5_p2bar_up}
\begin{align}\label{eq:chp5_p2bar_up}
	\bar{\bar{p}}_{up,C} = 1-{\exp}\Big\{-\lambda_e \pi \Big[ \frac{c_0K\bar{A}_{0,C}}{2\pi(K+1)}+\frac{c_0}{K+1}  \Big]^{\frac{2}{\beta}}  \Big\},
\end{align}
where $\bar{A}_{0,C}$ is the expectation of $A_{0,C}$ over $\theta_B$ and is given by
\begin{align}\label{eq:chp5_meanA_0}
	\bar{A}_{0,C}=\mathbb{E}_{\theta_B}[A_{0,C}]=\frac{1}{2\pi}\int_0^{2\pi} A_{0,C} \,\mathrm{d}\theta_B.
\end{align}
\end{theorem}
The proof of Theorem\,\ref{th:chp5_p2bar_up} is in Appendix\,\ref{appdx:bessel:oieor}.
The tightness of $\bar{\bar{p}}_{up,C}$ can be analyzed with the help of the following proposition.
\begin{proposition}\label{prop:chp5_tightness}
For random variable $X\sim\mathcal{U}(a,b)$, the smaller $(b-a)$ is, the tighter the two inequalities in Lemma\,\ref{le:chp3_jensensinequality} are.
\end{proposition}
The proof of Proposition\,\ref{prop:chp5_tightness} is in Appendix\,\ref{appdx:bessel:uirtnvb}. 
In the derivation to obtain $\bar{\bar{p}}_{up,C}$, $A_{0,C}$ is the random variable with respect to $\theta_B$.
According to Proposition\,\ref{prop:chp5_tightness}, the less variation of $A_{0,C}$, the tighter $\bar{\bar{p}}_{up,C}$ is.
As discussed in Section\,\ref{chp4:sec3:opbwec541}, the variation of $A_{0,C}$ in the range $\theta_B\in[0,\frac{\pi}{2}]$ is not big compared to the ULA.
Thus, $\bar{\bar{p}}_{up,C}$ can be used to analytically analyze $\bar{\bar{p}}_C$.

\nomenclature{$\bar{\bar{p}}_{up,C}$}{upper bound of $\bar{\bar{p}}_C$}

To obtain the analytic expression of $\bar{\bar{p}}_{up,C}$, $\bar{A}_{0,C}$ is directly given here.
\begin{theorem}\label{th:chp5_meanA0}
\begin{align}\label{eq:chp5_A0C_average}
	\bar{A}_{0,C}=2\pi+2\pi\sum_{n=1}^{N-1} J_0^2(2kR\sin\frac{n\pi}{N}).
\end{align}
\end{theorem}
The proof of Theorem\,\ref{th:chp5_meanA0} is in Appendix\,\ref{appdx:bessel:twqwrq}.
It can be seen that (\ref{eq:chp5_A0C_average}) is the same as the approximation for $A_{0,C}$ in (\ref{eq:appdx_bessel_A0C_asymptotic}).
Therefore, $\bar{A}_{0,C}$ in general decreases with some fluctuations as $R$ increases.

\nomenclature{$\bar{A}_{0,C}$}{averaged $A_{0,C}$}

According to (\ref{eq:chp5_p2bar_up}), $\bar{\bar{p}}_{up,C}$ is positively correlated with $\bar{A}_{0,C}$.
Thus, the behavior of $\bar{\bar{p}}_{up,C}$ and $A_0$ is consistent with respect to $R$.
As previously stated, $\bar{\bar{p}}_{up,C}$ is tight to $\bar{\bar{p}}_C$. 
Thus, it can be conjectured that $\bar{\bar{p}}_C$ decreases in general with some fluctuations as $R$ increases.
This can be verified by the numerical results shown in Appendix\,\ref{appdx:fig:ldfsjo}.

Because $\bar{\bar{p}}_C$ fluctuates in a certain range of $R$, there must exist at least one local minimum.
Numerical results are used to find $R_{opt}$ in (\ref{eq:chp5_mmse2}) and (\ref{eq:chp5_mmse4}).
For example, chose $N_{\text{max}}=8$ and $R\in[0.4\lambda,2\lambda]$.
The results of $\bar{\bar{p}}_C$ for all possible $\theta_B$ and $\big(q^{(1)}\bar{\bar{p}}_C^{(1)}+q^{(2)}\bar{\bar{p}}_C^{(2)}+q^{(3)}\bar{\bar{p}}_C^{(3)}\big)$ for all possible $(d_B,\theta_B)$ are shown in Fig.\,\ref{fig:chp5_p2bar_R}. 
It can be seen that for both curves, there are more than one local minima.
In the range $R\in[0.4\lambda,2\lambda]$, $R_{opt}$ can be chosen by the smallest local minimum.
Coincidentally, for both curves $R_{opt}$ is $1.76\lambda$ in the range $R\in[0.4\lambda,2\lambda]$.

\begin{figure}
\centering
\includegraphics[scale=0.9]{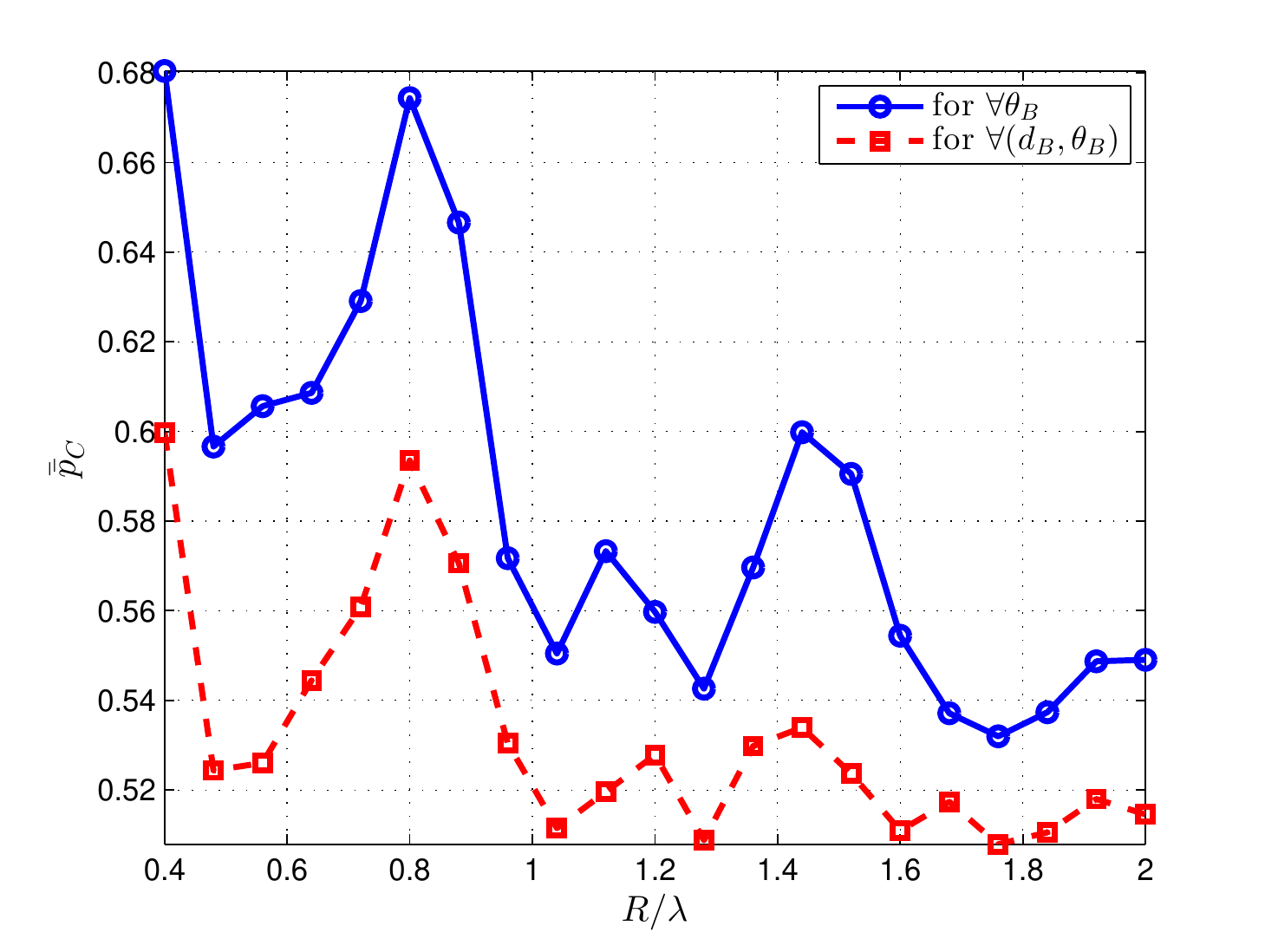}
\caption{$\bar{\bar{p}}_C$ versus $R$ for all $\theta_B$ and $(d_B,\theta_B)$, $N=8$}
\label{fig:chp5_p2bar_R}
\end{figure}

\subsection{Transmit Power and Number of Elements}
\label{chp5:analysis:vneiwnv}

In this section, the relationship between $P_t$ and $N$ is studied.
Both $P_t$ and $N$ are included by $\bar{p}_C$ in (\ref{eq:chp5_p_De_beta_is_2}) and the constraint in (\ref{eq:chp5_problem_constraint2}).
According to (\ref{eq:chp5_p_De_beta_is_2}), there is a monotonically increasing relationship between $P_t$ and $\bar{p}_C$; according to (\ref{eq:chp5_problem_constraint2}), the larger $P_t$, the easier to satisfy the constraint for fixed $d_B$.
However, $\bar{p}_C$ is determined by the product of $P_tA_{0,C}$, while the constraint has the product of $P_tN$.
Therefore, the trade-off between $P_t$ and $N$ is studied.

The key to the problem lies in the relationship between $A_{0,C}$ and $N$.
In Section\,\ref{chp4:sec3:njowq}, the asymptotic behavior of $\bar{p}_C$ is almost linear when $N$ is very large.
However, when $N$ is not very large, there does not exist a simple monotonic relationship between $A_{0,C}$ and $N$; instead, the change is non-linear, e.g. $N\leq 8$, which is a common setting for indoor devices.
This means that when $N$ increases, $P_tN$ increases, but $P_tA_{0,C}$ (i.e., $\bar{p}_C$) could increase or decrease.

The product of $P_tN$ should be adapted to Bob's distance $d_B$, according to (\ref{eq:chp5_problem_constraint2}).
When Bob moves farther away form the AP, i.e., $d_B$ increases, the AP should respond by increasing $P_t$ or $N$, in order to satisfy $C_B\geq R_B$.
Alternatively, if $d_B$ does not change, the AP can double $P_t$ while halving $N$, to seek a lower $\bar{p}_C$.
Next, the trade-off between adjusting $P_t$ or $N$ is studied for the previous two cases.

For the first case, imagine a situation where Bob's distance $d_B$ changes. Assume that the benchmark of Bob's distance is $d_{B0}$ and the AP uses $N_0$ active elements with transmit power $P_{t0}$. As Bob moves closer to or farther away from the AP, the AP needs to adjust either $P_t$ or $N$ to keep $P_{t}N/d_{B}^\beta$ unchanged,  in order to satisfy $C_B\geq R_B$.

\nomenclature{$d_{B0}$}{benchmark value for $d_B$}
\nomenclature{$P_{t0}$}{benchmark value for $P_t$}
\nomenclature{$N_0$}{benchmark value for $N$}

The results are plotted in Fig.\,\ref{fig:chp5_P_t_and_N_movingd_B}, where $N_0=8$, $R=1.6\lambda$ and $\theta_B=0^{\circ}$.
Bob's distance $d_B$ is normalized against the benchmark $d_{B0}$.
In the upper plot, it shows how $P_t$ is changed with $d_B$ when $N$ is fixed, i.e., $N=N_0$
In the middle plot, it shows how $N$ is changed with $d_B$ when $P_t$ is fixed, i.e., $P_t=P_{t0}$.
In order to satisfy $C_B\geq R_B$, when fixing either $P_t$ or $N$, the other parameter increases along with $d_B^{\beta}$.

\begin{figure}
\centering
\includegraphics[scale=0.9]{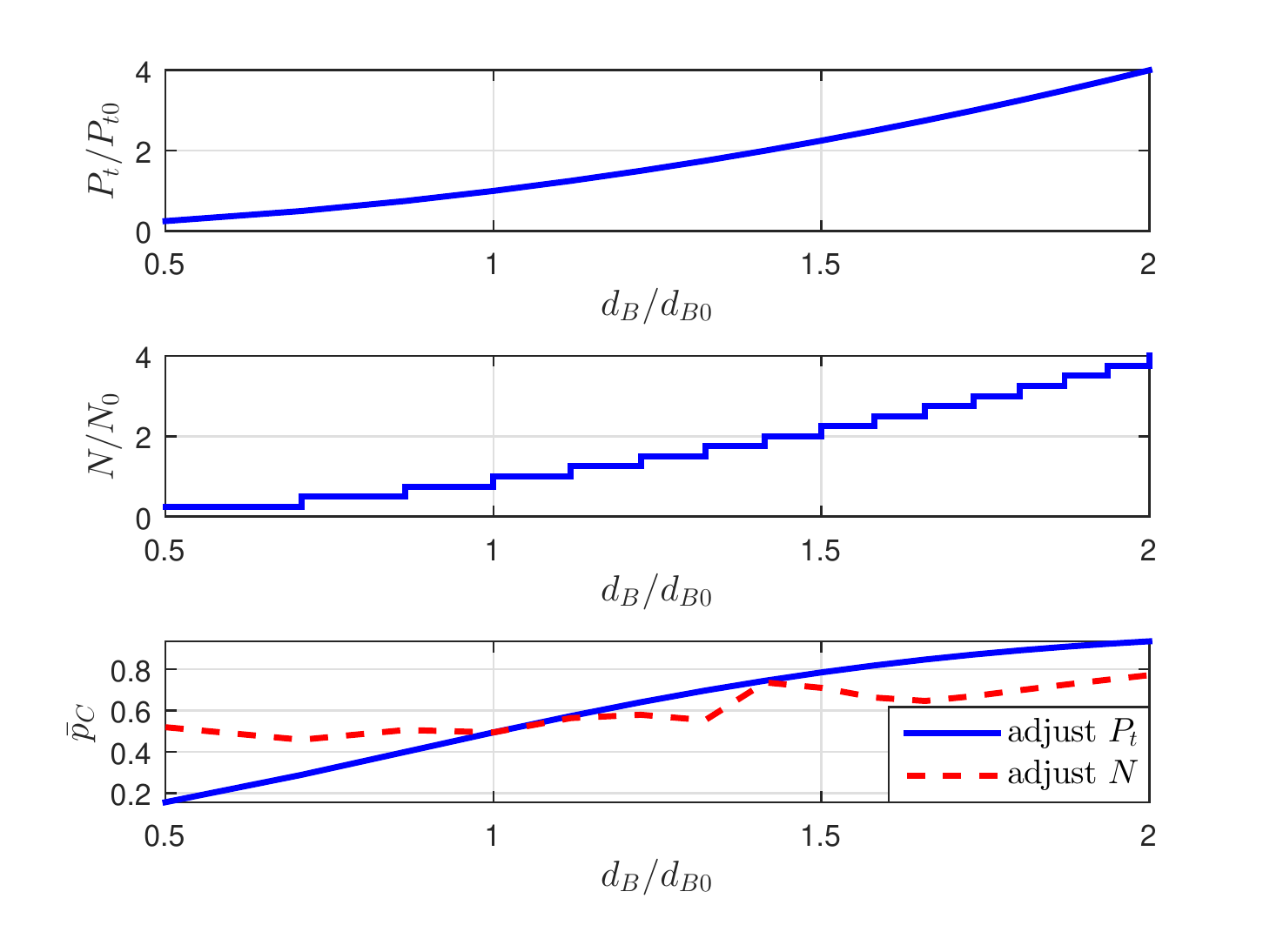}
\caption{Plots for moving $d_B$. Upper plot: $P_t/P_{t0}$ versus $d_B/d_{B0}$ for fixed $N$; middle plot: $N/N_{0}$ versus $d_B/d_{B0}$ for fixed $P_t$; lower plot: $\bar{p}_C$ versus $d_B/d_{B0}$ for adjustable $P_t$ with fixed $N$, and adjustable $N$ with fixed $P_t$. $\beta=2$}
\label{fig:chp5_P_t_and_N_movingd_B}
\end{figure}

The lower plot in Fig.\,\ref{fig:chp5_P_t_and_N_movingd_B} shows the change of $\bar{p}_C$ when changing $P_t$ or $N$. 
As discussed previously, $\bar{p}_C$ increases along with $P_t$, whereas $\bar{p}_C$ increases in general with $N$ with some fluctuations. 
Compared the two curves, it can be seen that the variation of $\bar{p}_C$ when changing $P_t$ is bigger than that of $N$, which suggests that $P_t$ has a more dominant role in $\bar{p}_C$.

For the second case, the trade-off between changing $P_t$ and changing $N$ is studied when $d_B$ is fixed.
For fixed $d_B$, the product of $P_t$ and $N$ should be kept constant, i.e., $P_tN=P_{t0}N_0$.
The AP starts with $N=N_0$ elements and $P_t=P_{t0}$, then $P_t$ and $N$ are adjusted.
The numerical results are shown in Fig.\,\ref{fig:chp5_P_t_and_N_fixd_B}, where $N_0=8$, $R=1.6\lambda$ and $\theta_B=0^{\circ}$.
The upper plot shows the change of $P_t$ versus $N$.
The lower plot shows how $\bar{p}_C$ changes with $N$.

\begin{figure}
\centering
\includegraphics[scale=0.9]{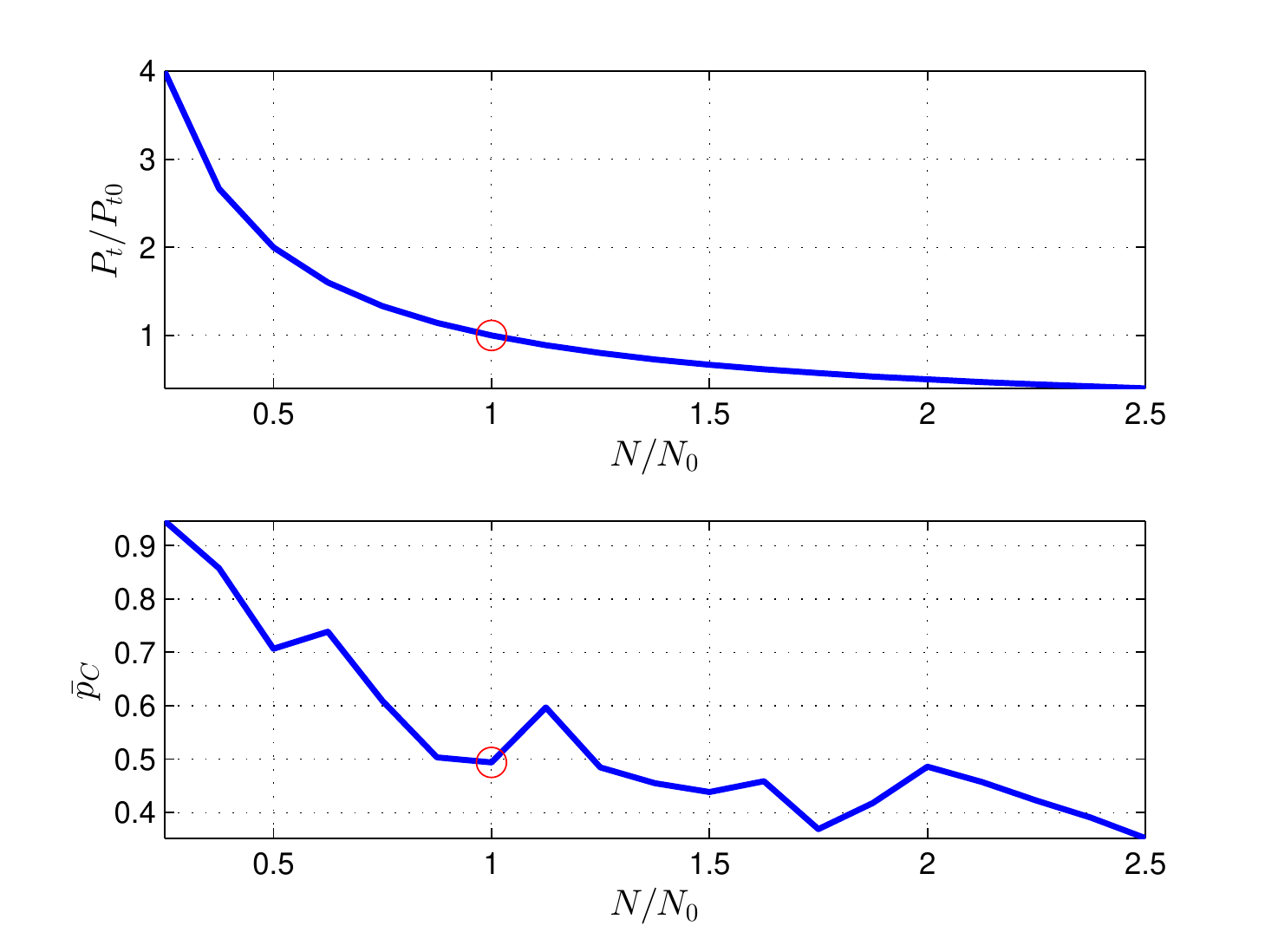}
\caption{Plots for fixed $d_B$. Upper plot: $P_t/P_{t0}$ versus $N/N_{0}$; lower plot: $\bar{p}_C$ versus $N/N_{0}$. $\beta=2$}
\label{fig:chp5_P_t_and_N_fixd_B}
\end{figure}

It can be seen from the upper plot that for fixed $d_B$, when $N$ increases, $P_t$ decreases to keep the product of $P_tN$ unchanged.
From the lower plot, it can be seen that $\bar{p}_C$ generally decreases along with $P_t$, although there is some fluctuations due to the increase of $N$.
Thus, the same conclusion can be drawn as in the first case, which is that $P_t$ plays the dominant role in the pair $(P_t,N)$ in terms of $\bar{p}_C$.

In the previous cases, $\theta_B$ is chosen as a particular angle, i.e., $\theta_B=0^{\circ}$.
To eliminate the impact of $\theta_B$, the averaged value $\bar{A}_{0,C}$ is studied versus $N$ for a fixed radius $R=1.6\lambda$.
The results are shown in Fig.\,\ref{fig:chp5_p_out_P_t_N_A0}.
An auxiliary line of $y=x$ is drawn.
It can be seen that in general, the change of $\bar{A}_{0,C}$ versus $N$ is smaller than the gradient of $y=x$.
Thus, $P_t$ is the more dominant factor in $P_tA_{0,C}$.

\begin{figure}
\centering
\includegraphics[scale=0.9]{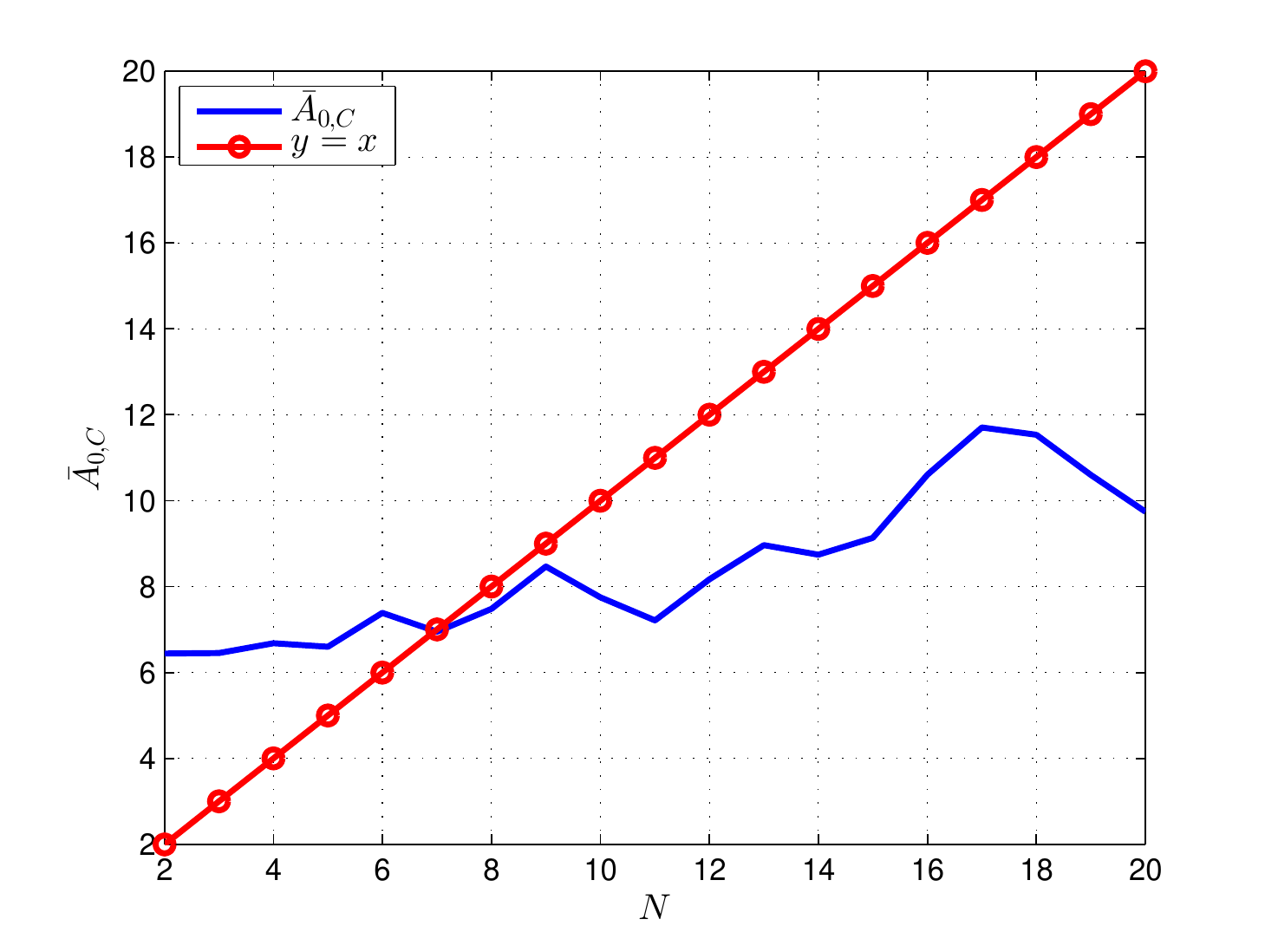}
\caption{$\bar{A}_{0,C}$ versus $N$. $R=1.5\lambda$}
\label{fig:chp5_p_out_P_t_N_A0}
\end{figure}

When Bob moves closer to the AP, i.e., $d_B/d_{B0}<1$, reducing $P_t$ leads to a smaller $\bar{p}_C$ than reducing $N$, as shown by the lower plot in Fig.\,\ref{fig:chp5_P_t_and_N_movingd_B}.
When Bob stays at the same distance, reducing $N$ leads to the increase of $P_t$ and the increase of $\bar{p}_C$, as shown by the lower plot in Fig.\,\ref{fig:chp5_P_t_and_N_fixd_B}.
When Bob moves further beyond the the maximum coverage distance is $d_{\text{max}}$, although it is preferable to increase $N$ rather than $P_t$, it is unlikely to do so because it is hard to change $N_{\text{max}}$ once the array is installed.

Based on the above analysis, it can be concluded that for given UCA with $N_{\text{max}}$ elements, when $(P_t,N)$ are adjustable, it is preferable to adjust $P_t$ than $N$. 
In other words, it is preferable not to change $N$ when $P_t$ is adjustable.

\subsection{Array Mode and Bob's Location }
\label{chp5:analysis:nbvpi}

For adjustable $P_t$, there is no need to change $N$, i.e., $M_{ij}$.
However, when $P_t$ is fixed, $N$ can be adjusted according to $d_B$. 
In addition, different $M_{ij}$ in $\{M_i\}$ gives different $\bar{p}_C$ for the same $\theta_B$.
Thus, in this section, the task is to find the optimum $M_{ij}$ that gives the minimum $\bar{p}_C$ for certain $(d_B,\theta_B)$ and satisfies $C_B\geq R_B$, when $P_t$ is fixed.

$\bar{p}_C$ for different $M_{ij}$ of UCA with 8-elements and $R=1.6\lambda$ is shown in Fig.\,\ref{fig:chp5_p_DoE_arraymode}.
The angle range is chosen to be $\theta_B\in[0^{\circ},90^{\circ}]$, because it is the smallest range where the patterns of $\{M_3\}$ do not repeat themselves.
$\bar{p}_C$ of $\{M_2\}$ and $\{M_3\}$ is plotted for $\theta_{B}\in[0^{\circ},90^{\circ}]$ in the upper and lower plots, respectively.
It can be seen that, different $M_{ij}$ in $\{M_i\}$ has different $\bar{p}_C$.
The minimum curve $\bar{p}_{C,\text{min}}$ is shown by the dotted curve.
For each $\theta_B$, one array mode in $\{M_i\}$ gives $\bar{p}_{C,\text{min}}$.

\begin{figure}
\centering
\includegraphics[scale=0.9]{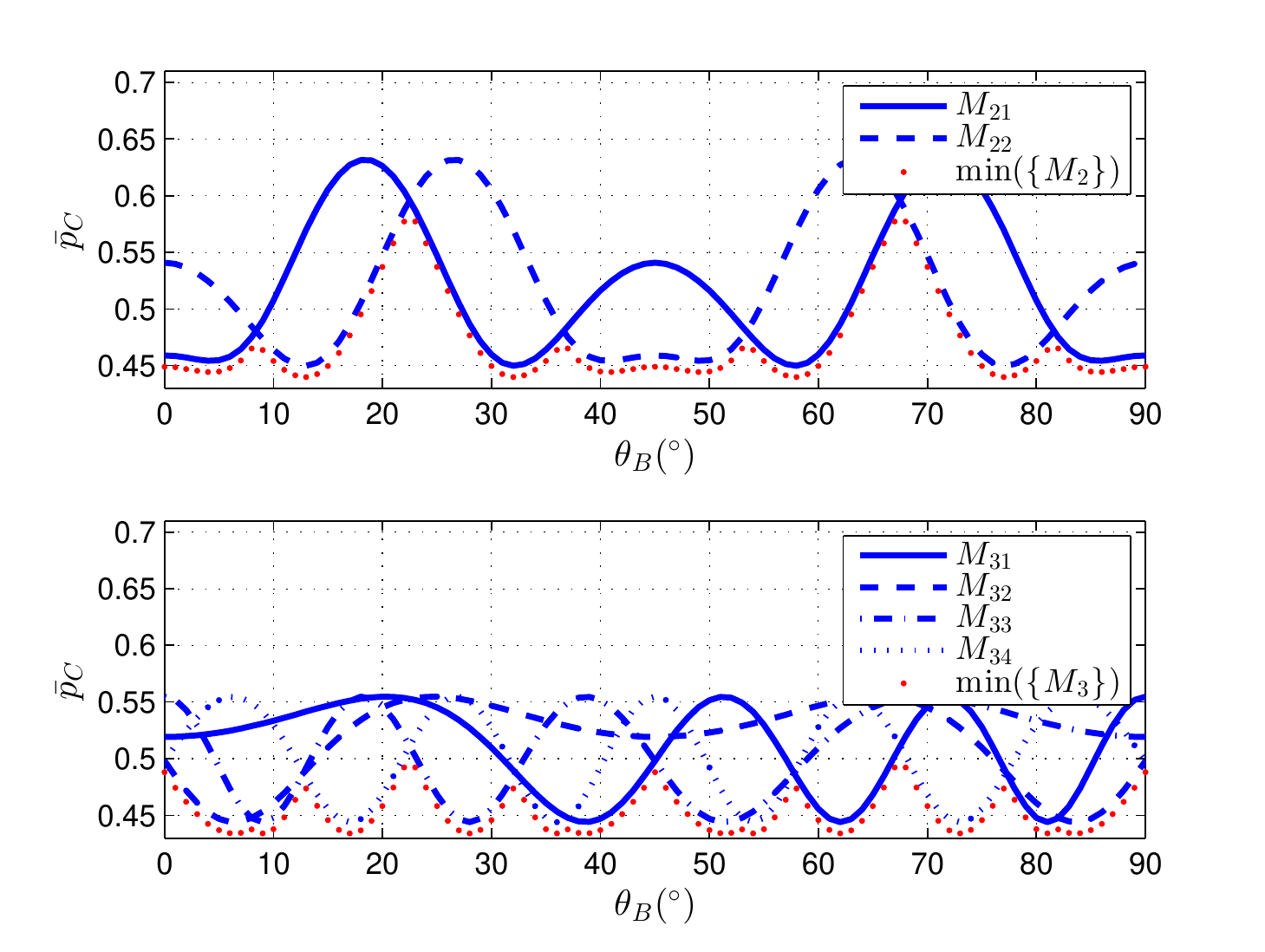}
\caption{$\bar{p}_C$ for different $M_{ij}$ over $\theta_B$; upper plot: $\{M_{2}\}$; lower plot: $\{M_{3}\}$. For better vision, the minimum value curve is drawn with an offset by -0.01.}
\label{fig:chp5_p_DoE_arraymode}
\end{figure}

Take the example in Section\,\ref{chp5:syst:bnmo} for instance.
When $\theta_B=5^{\circ}$, the AP can choose $\bar{p}_{C,\text{min}}$ from $M_{21}$ and $M_{22}$.
This is especially true for $\{M_3\}$, where $N=2$.
Because $\{M_3\}$ can also be regarded as linear array, which has larger SSOP near $\theta_{\text{doe}}=90^{\circ}$.
However, this can be avoided by choosing an appropriate array mode that turns $\theta_{\text{doe}}=90^{\circ}$ for one array mode to $\theta_{\text{doe}}=0^{\circ}$ for another array mode.

Because there is no monotonic relationship between $A_{0,C}$ and $N$, as shown in Section\,\ref{chp4:sec3:njowq}, larger $N$ does not necessarily mean smaller or lager $\bar{p}_C$.
In Fig.\,\ref{fig:chp5_p_DoE_differentN}, $\bar{p}_C$ for different $N$ is shown for the UCA with 8 elements and $R=1.6\lambda$.
For example, when $\theta_B=0^{\circ}$, $N=4$ produces smaller $\bar{p}_C$ than that of $N=8$.

\begin{figure}
\centering
\includegraphics[scale=0.9]{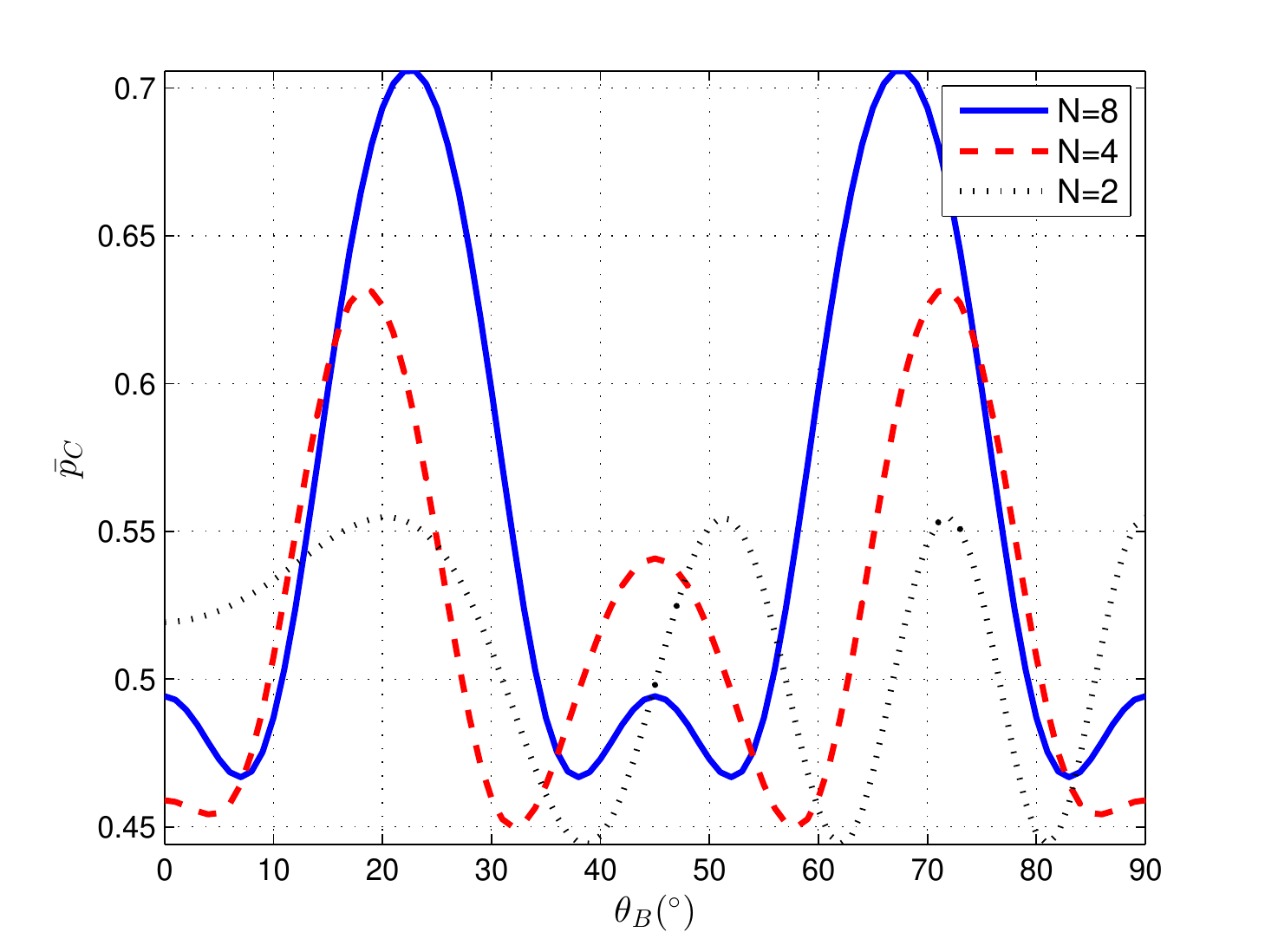}
\caption{$\bar{p}_C$ versus $\theta_B$ for different $N$}
\label{fig:chp5_p_DoE_differentN}
\end{figure}

An important issue when choosing $M_{ij}$ is that the constraint $C_B\geq R_B$ must be satisfied.
As introduced in Section\,\ref{chp5:syst:vndkawia}, the choice of $M_{ij}$ is limited by Bob being in different zone.
For example, when Bob is zone 3 and $\theta_B=0^{\circ}$, although $M_{21}$ gives $\bar{p}_{C,\text{min}}$, only $M_1$ can be used, so that $C_B\geq R_B$ can be satisfied.
Combining the observations from Fig.\,\ref{fig:chp5_p_DoE_arraymode} and\,\ref{fig:chp5_p_DoE_differentN}, the optimum $M_{ij}$ can be chosen according to $(d_B,\theta_B)$ to achieve $\bar{p}_{C,\text{min}}$.

For the generalized Rician channel with any $\beta$ and $K$, the conclusions regarding to $P_t$, $N$, $R$ and $M_{ij}$ are still valid, because the behavior of $\bar{p}_C$ and $\bar{\bar{p}}_C$ with respect to $N$, $R$ (and $\theta_B$ for $\bar{p}_C$) is consistent with that of $A_{0,C}$ and $\bar{A}_{0,C}$.
In summary, $R_{opt}$ can be found by numerically searching the minimum value of $\bar{\bar{p}}_C$ for either all possible Bob's angles or all possible Bob's locations in a local range of $R$; when $P_t$ is adjustable, there is no need to adjust $M_{ij}$; when $P_t$ is fixed, the optimum $M_{ij}$ can be numerically found according to Bob's location.

\section{Optimization Algorithms}
\label{chp5:opt_alg}


Based on the empirical results in Section\,\ref{chp5:analysis}, two numerical algorithms for two different constraints, i.e., adjustable and fixed transmit power are developed.
For each constraint, the optimum $(N,R)$ (i.e., $M_{ij}$) is found according to Bob's dynamic location $(d_B,\theta_B)$.

When $P_t$ is adjustable, it is preferable to adapt $P_t$ rather than $N$ according to $d_B$.
Thus, the optimization problem degrades to finding $R_{opt}$ that minimizes $\bar{\bar{p}}_C$ for all $\theta_B\in[0,2\pi]$, which is shown in (\ref{eq:chp5_mmse2}).
When $P_t$ is fixed, adjusting $M_{ij}$ according to $(d_B,\theta_B)$, which is called configurable beamforming technique, can produce $\bar{\bar{p}}_{C,min}$.
In the meantime, $R_{opt}$ that gives the minimum $\bar{\bar{p}}_C$ for all $(d_B,\theta_B)$ should be found in (\ref{eq:chp5_mmse4}).
For both constraints, assume that the total number of elements $N_{\text{max}}$ is fixed.

The numerical optimization algorithms are provided because there exist the analytical expressions only for the upper bounds $\bar{p}_{up,C}$ and $\bar{\bar{p}}_{up,C}$, which cannot provide accurate solutions for $R_{opt}$ and the optimum $M_{ij}$.
In addition, the developed numerical methods can be applied to arbitrary value of any parameters, such as $\beta$ and $K$.

\subsection{Numerical Optimization for Radius}
\label{chp5:opt_alg:nviewowe}

When $P_t$ is adjustable, the optimization of $R$ is for fixed $N_{\text{max}}$, which is straightforward according to (\ref{eq:chp5_mmse2}).
When $P_t$ is fixed, according to (\ref{eq:chp5_mmse4}), the optimization of radius $R$ is for a combination of different $N$, which is based on (\ref{eq:chp5_mmse2}). 
Thus, only the numerical optimization algorithm for adjustable $P_t$ (i.e., fixed $N$) is introduced.

The numerical implementation of the algorithm, referred to as Algorithm\,1, is illustrated by flowcharts in Fig.\,\ref{fig:chp5_algorithm1} to Fig.\,\ref{fig:chp5_algorithm1b}.
The detailed algorithm is available in Appendix\,\ref{appdx:opt:one}.
As shown in Fig.\,\ref{fig:chp5_algorithm1}, the algorithm first takes in some parameters and computes the iteration numbers.
Then, $\bar{\bar{p}}_C$ is calculated, which is used to find $R_{opt}$.

\begin{figure}
\centering
\includegraphics[scale=1]{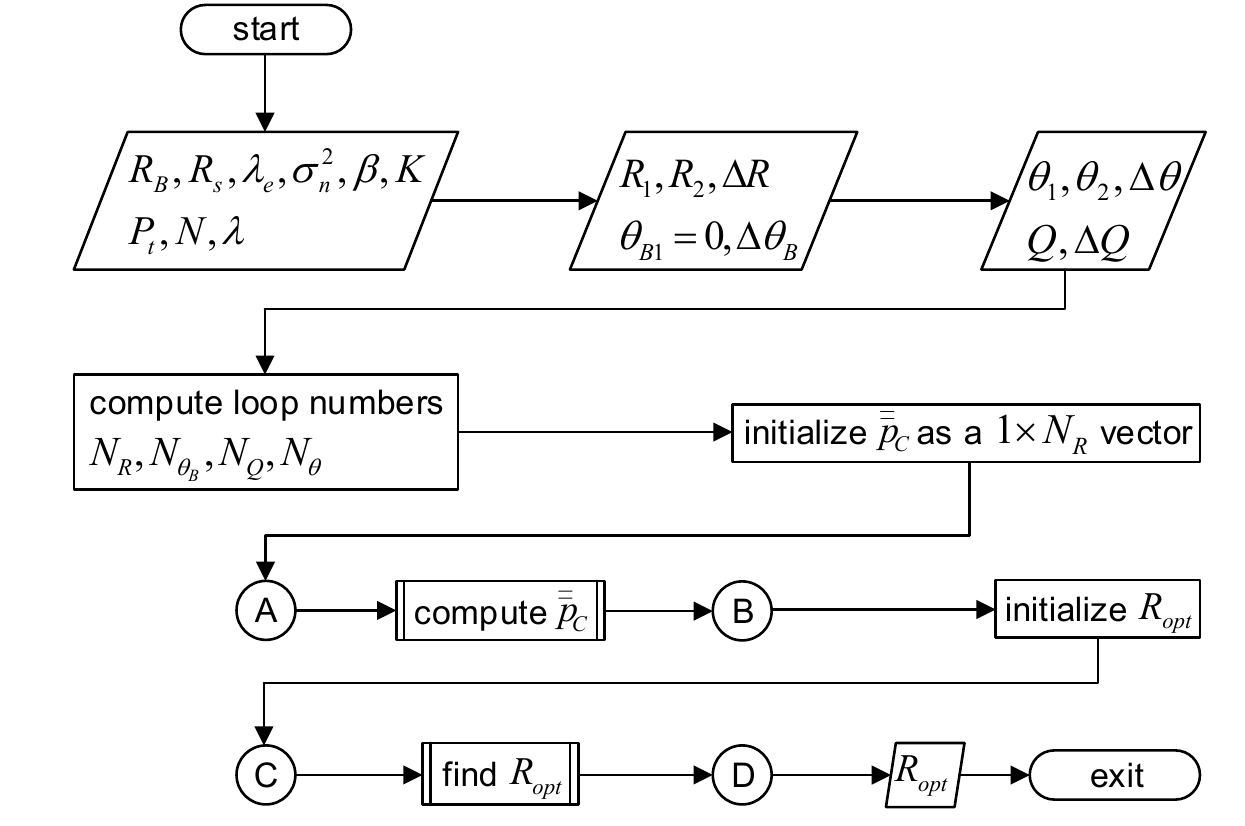}
\caption{Flowchart for Algorithm\,1}
\label{fig:chp5_algorithm1}
\end{figure}

\begin{figure}[t]
\centering
\includegraphics[scale=1]{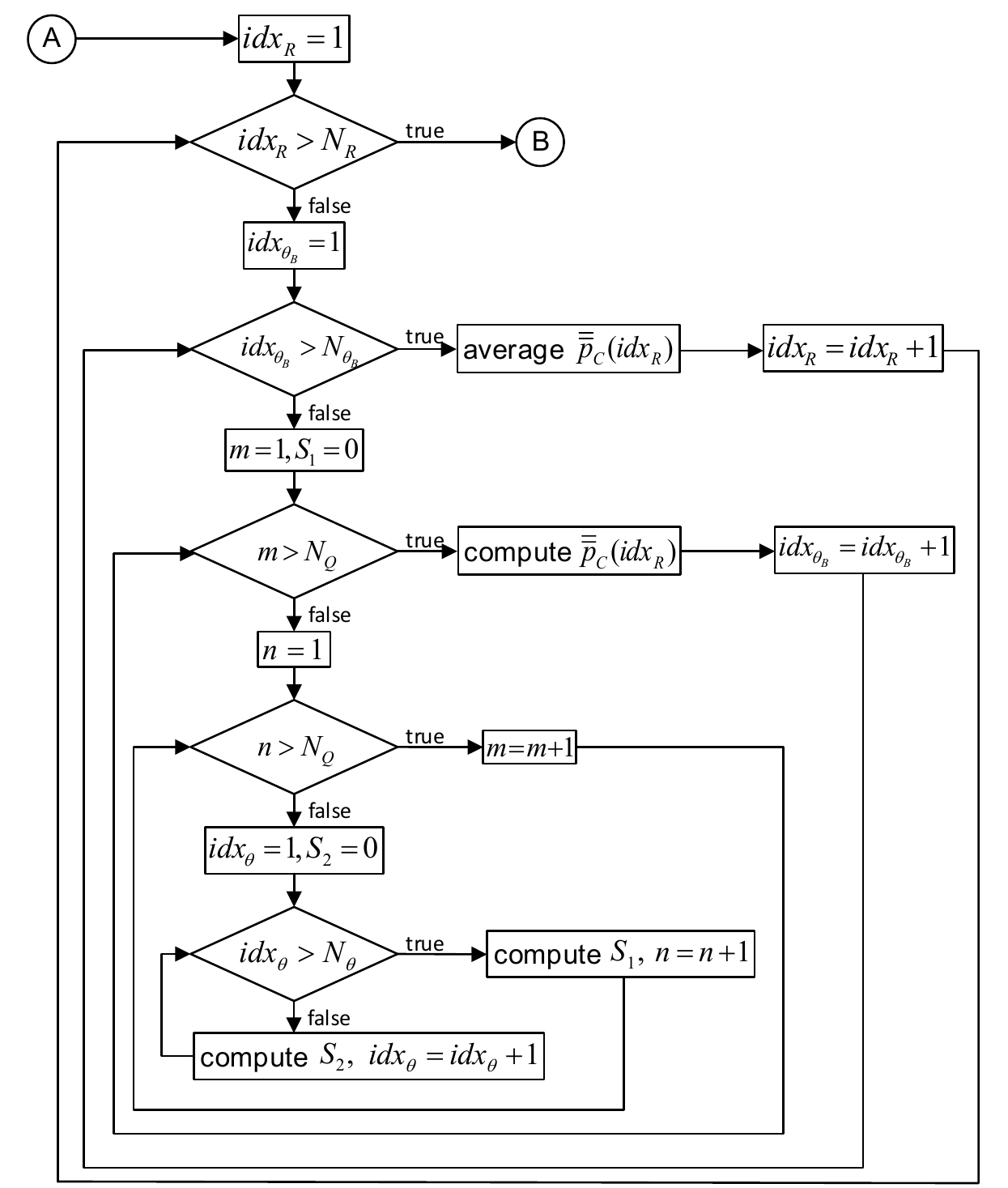}
\caption{Flowchart for `compute $\bar{\bar{p}}_C$' in Algorithm\,1}
\label{fig:chp5_algorithm1a}
\end{figure}

\begin{figure}[t]
\centering
\includegraphics[scale=1]{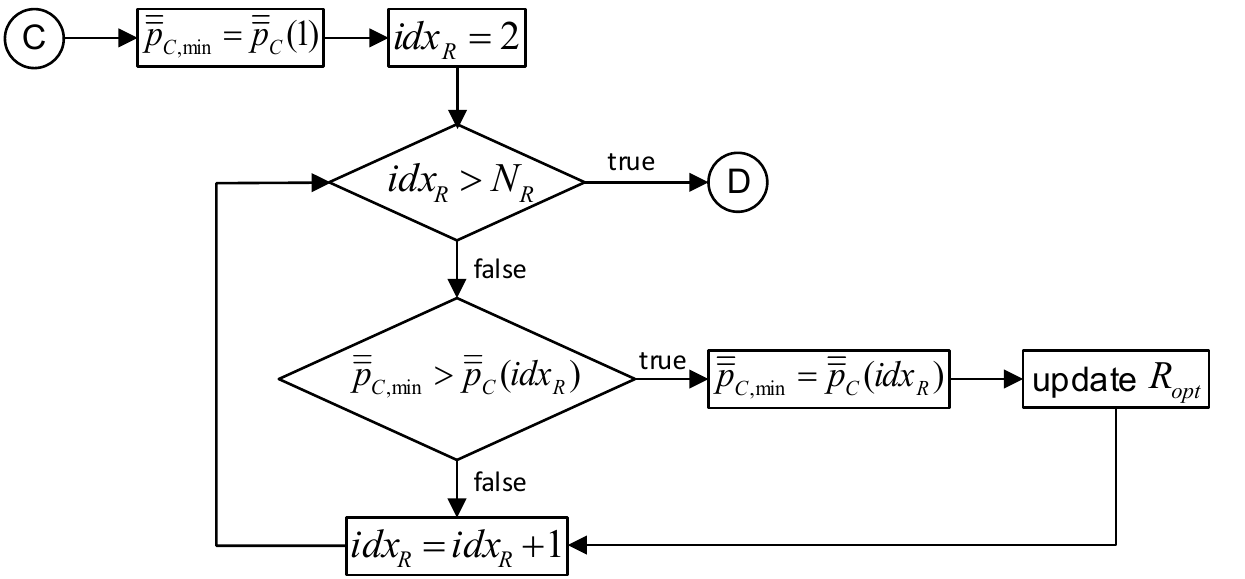}
\caption{Flowchart for `find $R_{opt}$' in Algorithm\,1}
\label{fig:chp5_algorithm1b}
\end{figure}

The continuous ranges of $R$, $\theta_B$, $\theta$ are sampled and become discrete.
The more samples are taken, the more accurate the result is; however, the computing complexity also increases.
It can be seen from Fig.\,\ref{fig:chp5_algorithm1} to Fig.\,\ref{fig:chp5_algorithm1b}  that the running time of Algorithm\,1 main depends on the `compute $\bar{\bar{p}}_C$' function shown in Fig.\,\ref{fig:chp5_algorithm1a}.
For convenience, assume that the basic computing unit, i.e., `compute $S_2$' takes 1 unit time length.
The asymptotic running time of Algorithm\,1 is $\mathcal{O}(N_R N_{\theta_B} N_Q^2 N_{\theta})$.
There is no specific restriction on the sampling interval as long as the chosen resolution generates a reasonable value.
Table\,\ref{tab:alg1} shows an example of the runninig time of Algorithm\,1 in MATLAB for different sampling interval for $R\in[0.4\lambda,2\lambda]$.

\begin{table}
\renewcommand{\arraystretch}{1.3}
\caption{Running time of Algorithm\,1}
\label{tab:alg1}
\centering
\begin{tabular}{|l||c|c|c|}
\hline
 $R$(cm)    & 1 & 0.5 & 0.1  \\ \hline
 $N_R$      & 21     & 41      & 201      \\ \hline
 time(sec)  & 2.6913 & 5.0563  & 24.5662  \\ \hline
\end{tabular}
\end{table}

\nomenclature{$\mathcal{O}(\cdot)$}{Big O notation}

One of the practical issues mentioned in Section\,\ref{chp5:syst:rqewvz} is the range of $R$.
As can be seen in Algorithm\,1, the range of $R$ is an input of the algorithm.
Thus, $R_{opt}$ is only the optimum value in this range.
For a better understanding, $\bar{\bar{p}}_C$ is shown in a larger range $R\in[0.4\lambda,4\lambda]$ in Fig.~\ref{fig:chp5_p2bar_R_longrange}.

Previously in Fig.\,\ref{fig:chp5_p2bar_R}, $R_{opt}$ is $1.76\lambda$ in the range of $[0.4\lambda,2\lambda]$, which gives $\bar{\bar{p}}_{C,min}=0.532$.
For the increased radius range, there are three more local minimums at $2.4\lambda$, $3.12\lambda$ and $3.76\lambda$, which give $\bar{\bar{p}}_{C,min}$ as $0.5086$, $0.5107$ and $0.5086$, respectively.

Take $R=2.4\lambda$ for example, there is a big increase in the radius (i.e., by $0.64\lambda$), compared to $R=1.76\lambda$, which increases the difficulty in the deployment of the UCA.
However, the improvement of $\bar{\bar{p}}_C$ (i.e., 0.0234) is not significant. Thus, the trade-off can be decided according to the specific applications.

\begin{figure}
\centering
\includegraphics[scale=0.9]{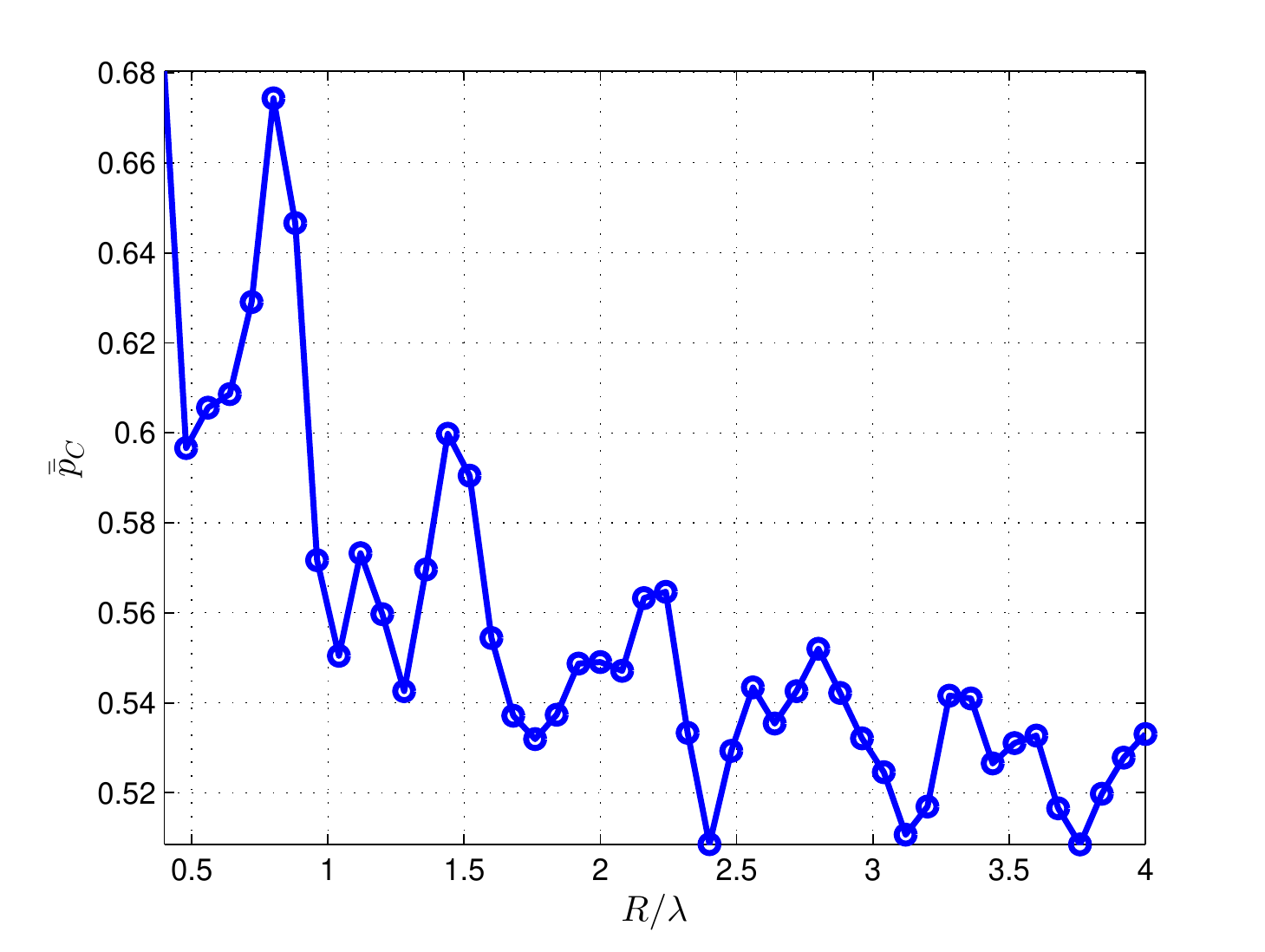}
\caption{$\bar{\bar{p}}_C$ versus $R$, $N=8$}
\label{fig:chp5_p2bar_R_longrange}
\end{figure}

The process to find $R_{opt}$ for all $(d_B,\theta_B)$ is similar to Algorithm\,1.
The difference is that instead of $\bar{\bar{p}}_C$, $\sum_k q^{(k)}\bar{\bar{p}}_C^{(k)}$ is calculated.
Thus, Algorithm\,1 is repeated for all $k$, and in each iteration, $q^{(k)}$ needs to be calculated.

To illustrate the security enhancement of Algorithm\,1, the same example as in Fig.\,\ref{fig:chp5_p2bar_R} is used. 
Let $\rho_1$ define the ratio of the difference between the value of $\bar{\bar{p}}_C$ at certain $R$ and the value of $\bar{\bar{p}}_C$ at $R_{opt}$ to the value of $\bar{\bar{p}}_C$ at that $R$.
The results of $\rho_1$ versus $R$ are shown in Fig.\,\ref{fig:chp5_performance_measure1}.
It can be seen that up to more than $20\%$ improvement can be achieved by choosing $R_{opt}$.

\begin{figure}
\centering
\includegraphics[scale=0.9]{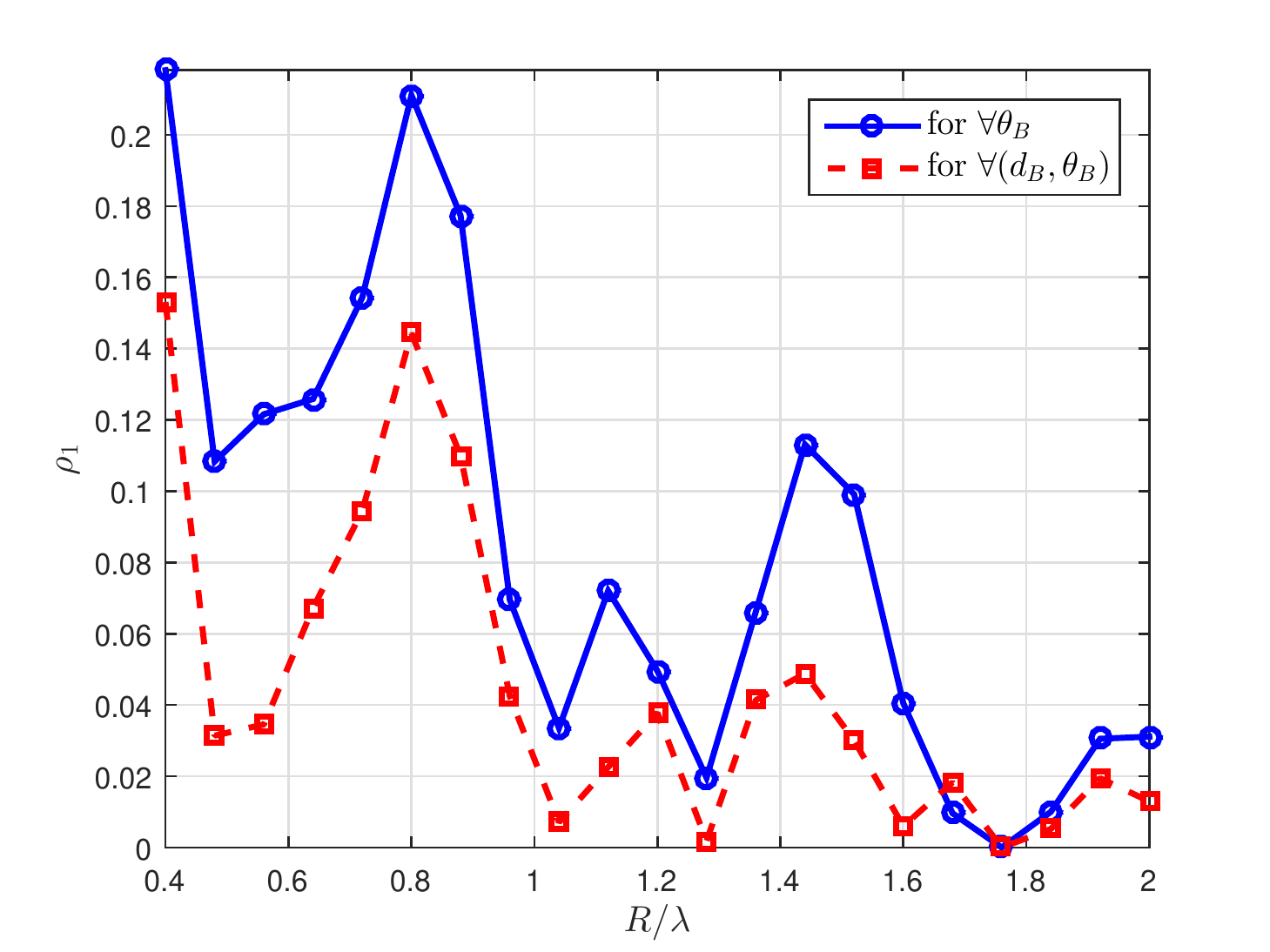}
\caption{$\rho_1$ versus $R$, $N=8$}
\label{fig:chp5_performance_measure1}
\end{figure}


\subsection{Configurable Beamforming Technique}
\label{chp5:opt_alg:nvnnnnc}

When $P_t$ is fixed, $M_{ij}$ can be adjusted to minimize the SSOP according to $(d_B,\theta_B)$.
The first step is to determine which zone Bob is in, according to $d_B$.
This is to determine the available array modes for Bob.
The second step is then to choose the optimum $M_{ij}$ from the available array modes according to $\theta_B$.
The previous process can be transformed into searching the optimum $M_{ij}$ according to $(d_B,\theta_B)$ in look-up tables.
This section shows how to create the look-up tables that store the optimum $M_{ij}$ for $(d_B,\theta_B)$.

For this purpose, the same example shown in Fig.\,\ref{fig:chp5_p_DoE_arraymode} is used.
$\bar{p}_{C,\text{min}}$ for $\{M_2\}$ and $\{M_3\}$ are plotted together with the curve for $M_1$ in Fig.\,\ref{fig:chp5_p_DoE_arraymode_2}.
It can be seen that in general, less number of active elements $N$ gives smaller $\bar{p}_{C,\text{min}}$.
But it is also noticed that for certain value of $\theta_B$, larger $N$ generates smaller $\bar{p}_{C,\text{min}}$, e.g., $\theta_B=0^{\circ}$.

\begin{figure}
\centering
\includegraphics[scale=0.9]{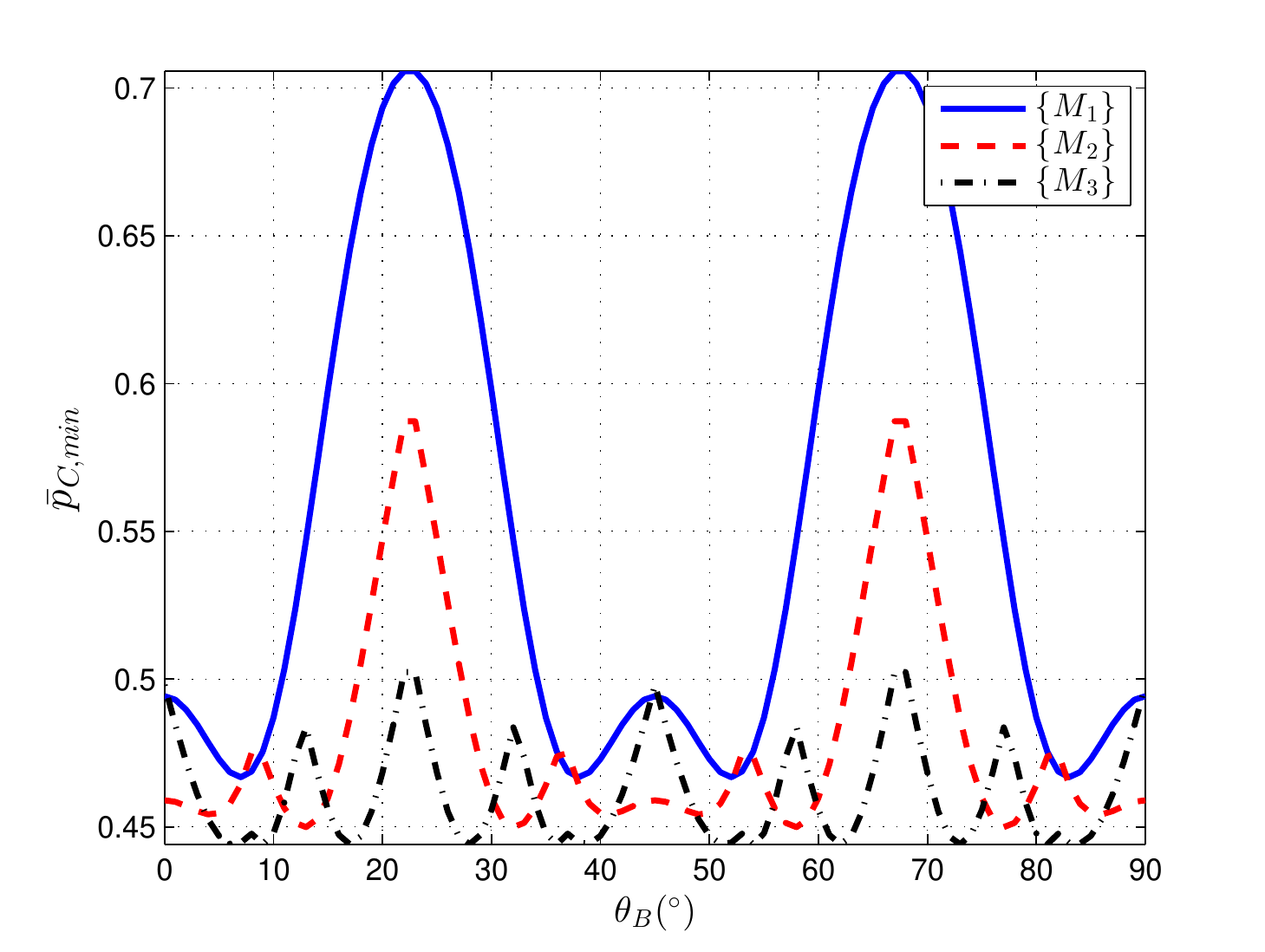}
\caption{$\bar{p}_{C,\text{min}}$ for $\{M_1\}$, $\{M_2\}$ and $\{M_3\}$}
\label{fig:chp5_p_DoE_arraymode_2}
\end{figure}

For certain $\theta_B$, once the minimum curve is picked, the optimum $M_{ij}$ can be subsequently decided by using Fig.\,\ref{fig:chp5_p_DoE_arraymode}.
For example, when $\theta_B=45^{\circ}$, it can be seen from Fig.\,\ref{fig:chp5_p_DoE_arraymode_2} that the dashed curve, i.e., $\{M_2\}$, gives $\bar{p}_{C,\text{min}}$.
Then, in Fig.\,\ref{fig:chp5_p_DoE_arraymode}, it can be seen in the upper plot that the dashed curve, i.e., $M_{22}$, gives the smaller value of $\bar{p}_C$ at $\theta_B=45^{\circ}$.
Therefore, $M_{22}$ is chosen for $\theta_B=45^{\circ}$.

So far, only $\theta_B$ is taken into consideration. 
According to the analysis in Section\,\ref{chp5:syst:vndkawia}, when Bob is in different zone, there is a lower bound of $N$ that should be used to guarantee $C_B\geq R_B$.
For the previous example when $\theta_B=45^{\circ}$, if Bob is in zone 3 in Fig.\,\ref{fig:chp5_N_zone}, only $M_1$, i.e., $N=8$, can be used, because using $M_{22}$ leads to $C_B<R_B$.

For Bob being in different zones, there are limited number of array modes that can be chosen to guarantee $C_B\geq R_B$. 
When Bob is in zone 1, all array modes are available. 
When Bob is in zone 2, $\{M_1\}$ and $\{M_2\}$ are available.
When Bob is in zone 3, only $M_1$ is available.
In Fig.\,\ref{fig:chp5_arraymode_DoE}, the optimum $M_{ij}$ for Bob being in different zone is plotted for $\theta_B\in[0^{\circ},90^{\circ}]$.
The y-axis shows the index of $M_{ij}$.
For convenience, $M_1$, $M_{21}$, $M_{22}$, $M_{31}$, $M_{32}$, $M_{33}$, $M_{34}$ are indexed from 1 to 7.

\begin{figure}
\centering
\includegraphics[scale=0.9]{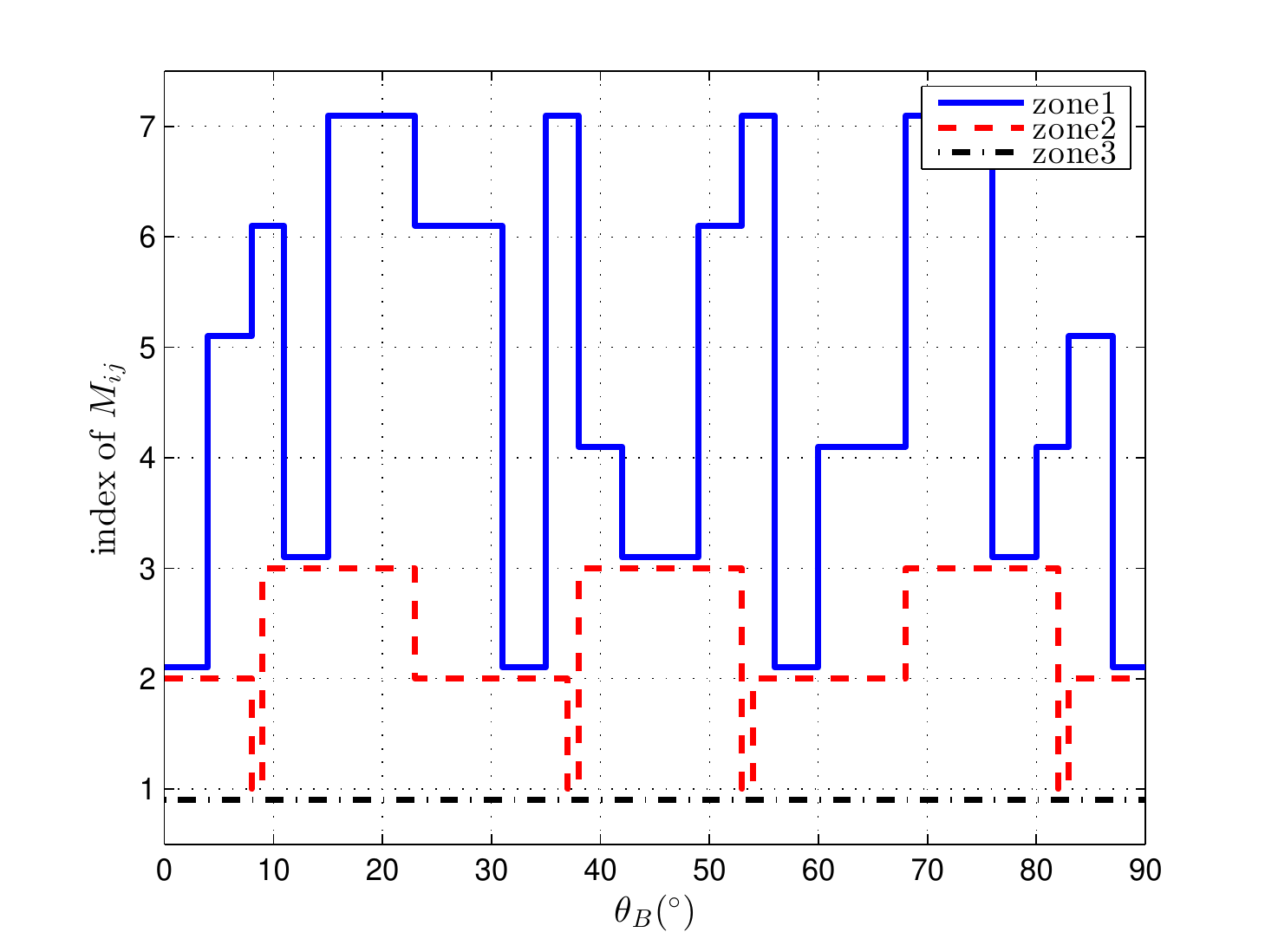}
\caption{Index of optimum $M_{ij}$ versus $\theta_B$ for Bob in zone 1 to zone 3. For better vision, the curves for zone 1 and zone 3 are drawn with offsets of 0.1 and -0.1, respectively.}
\label{fig:chp5_arraymode_DoE}
\end{figure}

The plots in Fig.\,\ref{fig:chp5_arraymode_DoE} can be converted into look-up tables.
Denoted by $T^{(k)}$, it stores the $M_(ij)$ (or its index) that generates $\bar{p}_{C,\text{min}}$.
In Fig.\,\ref{fig:chp5_arraymode_DoE}, the plot for Bob in zone $k$ is converted into table $T^{(k)}$.
$T^{(1)}$ and $T^{(2)}$ are shown in Table\,\ref{tab:one} and\,\ref{tab:two}, respectively.
Here, the angle resolution is taken as $0.5^{\circ}$.
Nevertheless, the resolution can be any practical value.
For the outer zone 3, $T^{(3)}$ has only one entry, i.e., $M_1$.
Therefore $T^{(3)}$ is not explicitly shown here.

\nomenclature{$T^{(k)}$}{look-up table for Bob in zone $k$}

\begin{table}
\renewcommand{\arraystretch}{1.3}
\caption{Look-up table $T^{(1)}$}
\label{tab:one}
\centering
\begin{tabular}{|l|c||l|c||l|c|}
\hline
 $\theta_B(^{\circ})$ & $M_{ij}$ & $\theta_B(^{\circ})$ & $M_{ij}$ & $\theta_B(^{\circ})$ & $M_{ij}$ \\ \hline
 $[0,3.5]$            & $M_{21}$ & $[34.5,37.5]$        & $M_{34}$ & [67.5,75.5]          & $M_{34}$ \\ \hline
 $[3.5,7.5]$          & $M_{32}$ & $[37.5,41.5]$        & $M_{31}$ & [75.5,79.5]          & $M_{22}$ \\ \hline
 $[7.5,10.5]$         & $M_{33}$ & $[41.5,48.5]$        & $M_{22}$ & [79.5,82.5]          & $M_{31}$ \\ \hline
 $[10.5,14.5]$        & $M_{22}$ & $[48.5,52.5]$        & $M_{33}$ & [82.5,86.5]          & $M_{32}$ \\ \hline
 $[14.5,22.5]$        & $M_{34}$ & $[52.5,55.5]$        & $M_{34}$ & [86.5,90]            & $M_{21}$ \\ \hline 
 $[22.5,30.5]$        & $M_{33}$ & $[55.5,59.5]$        & $M_{21}$ &                      &          \\ \hline
 $[30.5,34.5]$        & $M_{21}$ & $[59.5,67.5]$        & $M_{31}$ &                      &          \\ \hline
\end{tabular}
\end{table}

\begin{table}
\renewcommand{\arraystretch}{1.3}
\caption{Look-up table $T^{(2)}$}
\label{tab:two}
\centering
\begin{tabular}{|l|c||l|c||l|c|}
\hline
 $\theta_B(^{\circ})$ & $M_{ij}$ & $\theta_B(^{\circ})$ & $M_{ij}$ & $\theta_B(^{\circ})$ & $M_{ij}$ \\ \hline
 $[0,7.5]$            & $M_{21}$ & $[36.5,37.5]$        & $M_{1}$  & [67.5,81.5]          & $M_{22}$ \\ \hline
 $[7.5,8.5]$          & $M_{1}$ & $[37.5,52.5]$         & $M_{22}$ & [81.5,82.5]          & $M_{1}$ \\ \hline
 $[8.5,22.5]$         & $M_{22}$ & $[52.5,53.5]$        & $M_{1}$  & [82.5,90]            & $M_{21}$ \\ \hline
 $[22.5,36.5]$        & $M_{21}$ & $[53.5,67.5]$        & $M_{21}$ &                      &         \\ \hline
\end{tabular}
\end{table}

The look-up table $T^{(k)}$ can be generated and stored ready in the AP. 
After the AP acquired Bob's location $(d_B,\theta_B)$, the corresponding $T^{(k)}$ is chosen according to $d_B$; then the optimum $M_{ij}$ is decided according to $\theta_B$ in $T^{(k)}$.
The procedure of creating the look-up table $T^{(k)}$ is illustrated by flowcharts in Fig.\,\ref{fig:chp5_algorithm2} to Fig.\,\ref{fig:chp5_algorithm2b}, which is referred to as Algorithm\,2.
The details of Algorithm\,2 are available in Appendix\,\ref{appdx:opt:two}.

Similar to Algorithm\,1, the continuous ranges of $\theta_B$ and $\theta$ are sampled.
The running time of the algorithm depends on the resolution of the sampling.
It can be seen from Fig.\,\ref{fig:chp5_algorithm2} to Fig.\,\ref{fig:chp5_algorithm2b} that the running time mainly depends on the `compute $\bar{p}_C$' function shown in Fig.\,\ref{fig:chp5_algorithm2a}.
For convenience, assume that the basic computing unit, i.e., `compute $S_2$' takes 1 unit time length.
The asymptotic running time of Algorithm\,2 is $\mathcal{O}(N_{M_{ij}} N_{\theta_B} N_Q^2 N_{\theta})$.
Table\,\ref{tab:alg2} shows an example of the runninig time of Algorithm\,2 in MATLAB for different sampling interval for $\theta_B\in[0,90^{\circ}]$.

\begin{table}
\renewcommand{\arraystretch}{1.3}
\caption{Running time of Algorithm\,2}
\label{tab:alg2}
\centering
\begin{tabular}{|l||c|c|c|}
\hline
 $\theta_B$($^{\circ}$)    & 1 & 0.5 & 0.1  \\ \hline
 $N_{\theta_B}$      & 91     & 181      & 901      \\ \hline
 time(sec)  & 0.4105 & 0.7900  & 3.7975  \\ \hline
\end{tabular}
\end{table}

\begin{figure}
\centering
\includegraphics[scale=1]{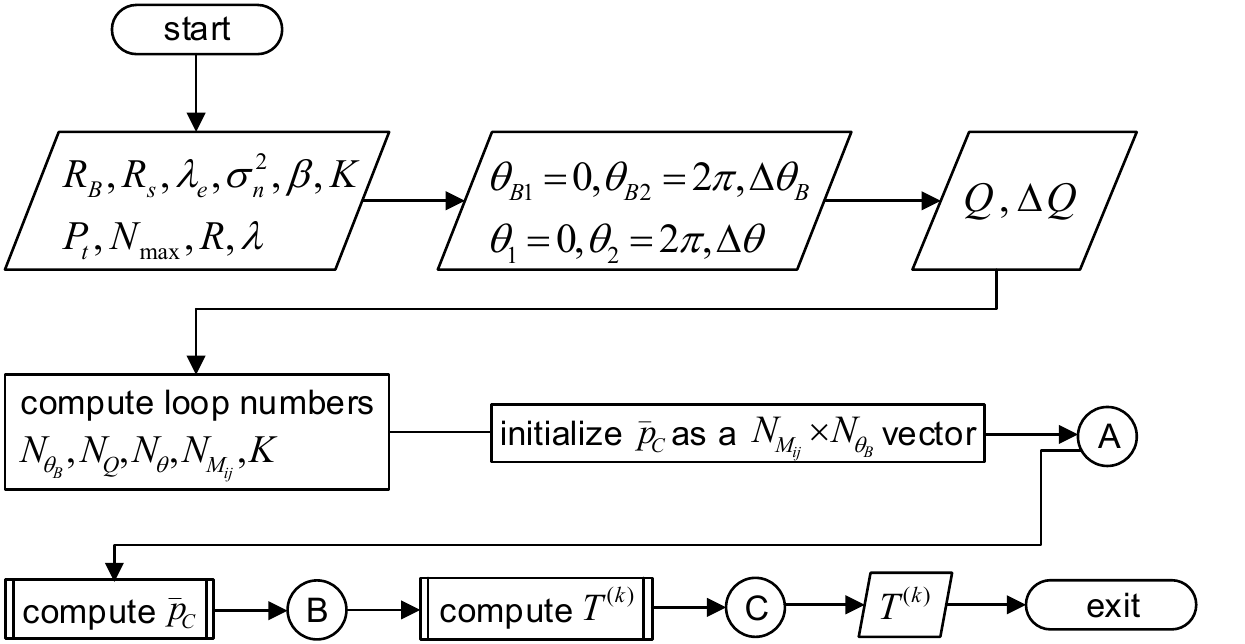}
\caption{Flowchart for Algorithm\,2}
\label{fig:chp5_algorithm2}
\end{figure}

\begin{figure}[t]
\centering
\includegraphics[scale=1]{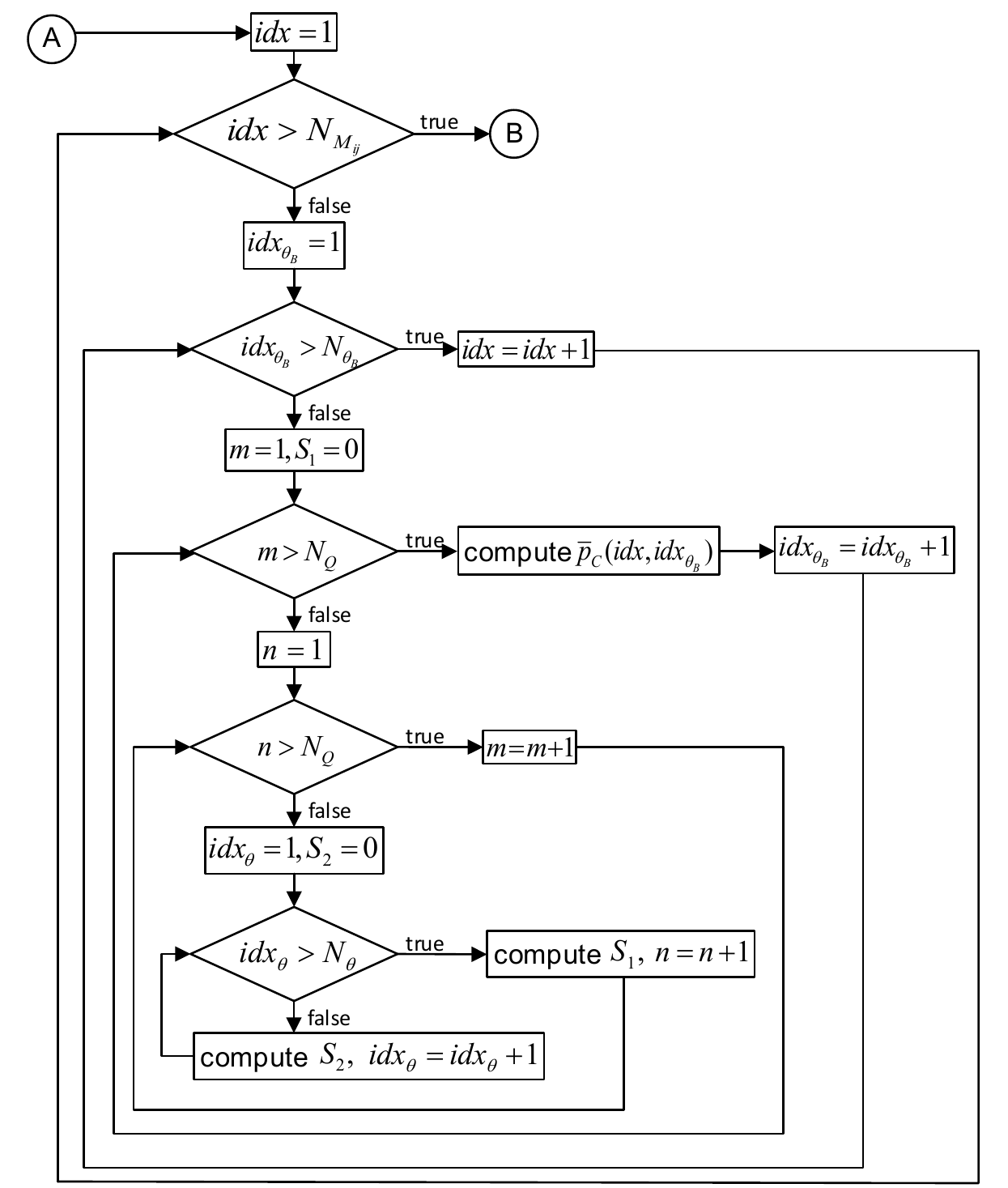}
\caption{Flowchart for `compute $\bar{p}_C$' in Algorithm\,2}
\label{fig:chp5_algorithm2a}
\end{figure}

\begin{figure}[t]
\centering
\includegraphics[scale=1]{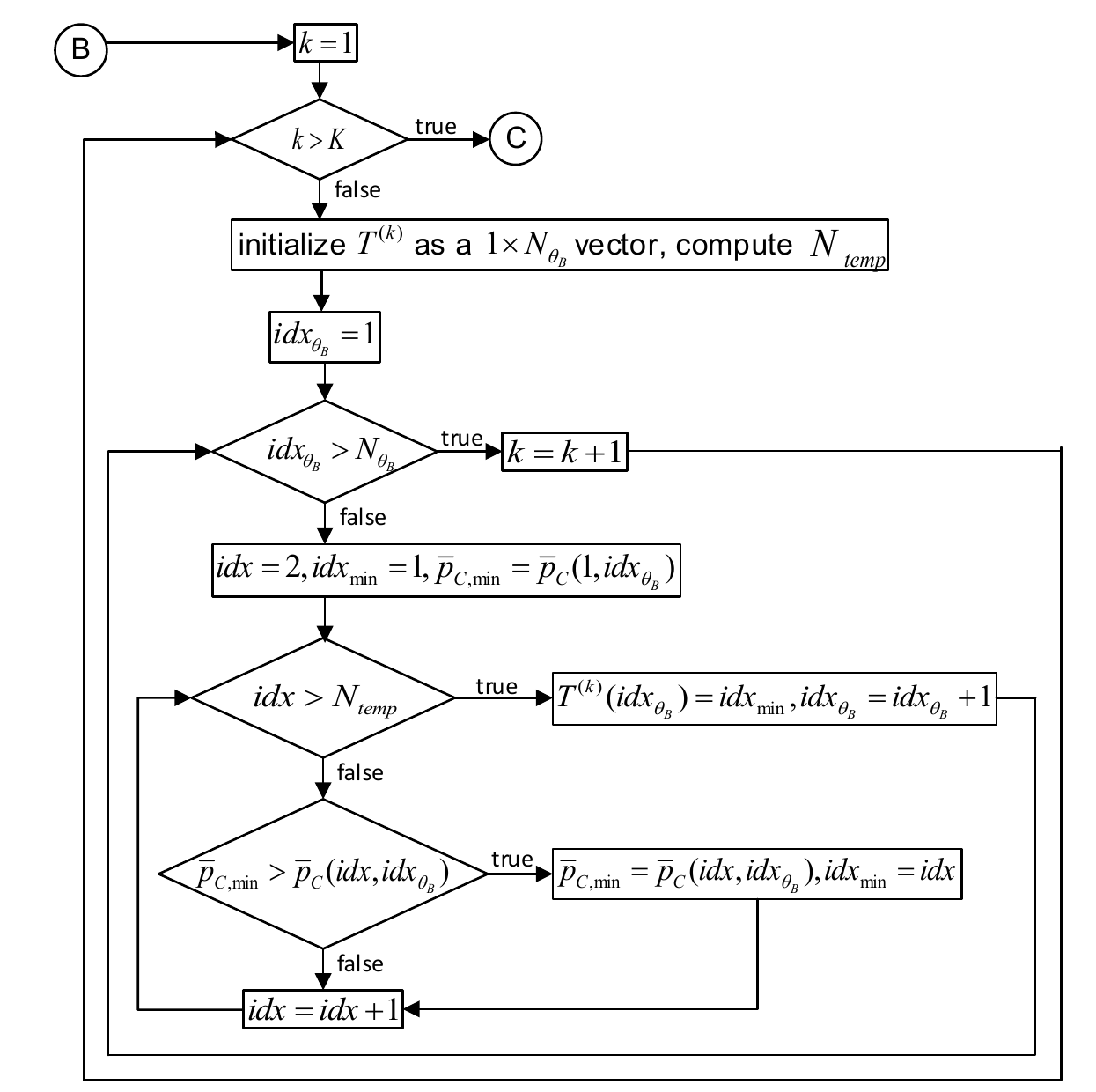}
\caption{Flowchart for `compute $T^(k)$' in Algorithm\,2}
\label{fig:chp5_algorithm2b}
\end{figure}

To illustrate the security enhancement of Algorithm\,2, the same example as in Fig.\,\ref{fig:chp5_p_DoE_arraymode_2} is used. 
Let $\rho_2$ define the ratio of the difference between $\bar{p}_{C,\text{min}}$ for $\{M_1\}$ and the minimum value $\bar{p}_{C,\text{min}}$ for $\{M_1,M_2,M_3\}$ to $\bar{p}_{C,\text{min}}$ for $\{M_1\}$.
The results of $\rho_2$ versus $\theta_B$ are shown in Fig.\,\ref{fig:chp5_performance_measure2}.
It can be seen that an improvement ranging from about $5\%$ to $47\%$ can be achieved by Algorithm\,2.

\begin{figure}
\centering
\includegraphics[scale=0.9]{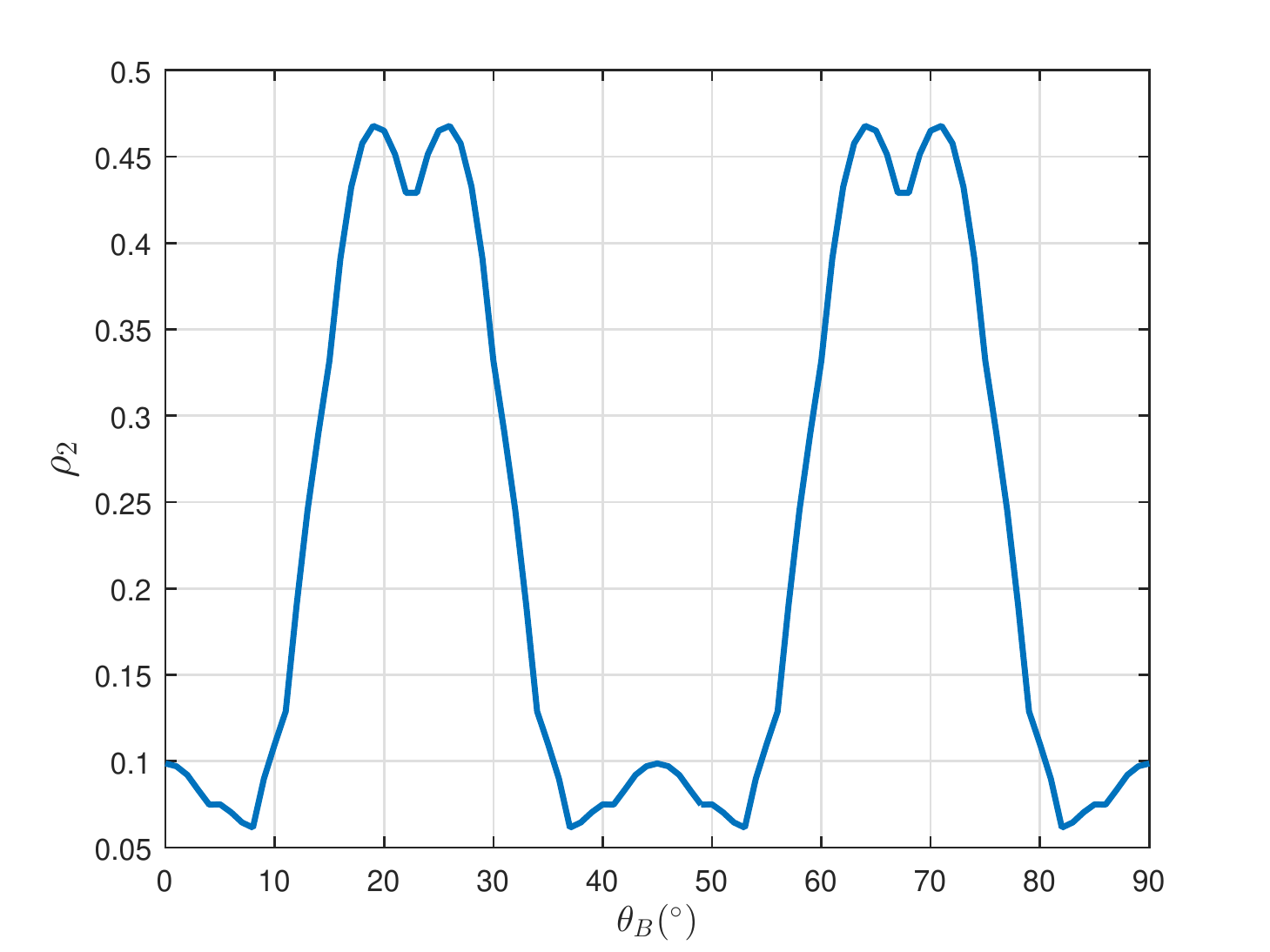}
\caption{$\rho_2$ versus $\theta_B$, $N=8$}
\label{fig:chp5_performance_measure2}
\end{figure}

\subsection{Error Analysis on Configurable Beamforming Technique}
\label{chp5:opt_alg:oeioqwjf}

In Section\,\ref{chp5:opt_alg:nvnnnnc}, the numerical optimization process to find the optimum $M_{ij}$ according to $(d_B,\theta_B)$ is shown.
It is assumed in Section\,\ref{chp5:syst:rqewvz} that Bob's CSI or $(d_B,\theta_B)$ is known by the AP, e.g., via channel estimation and feedback channel.
However, the estimated CSI could be erroneous, which means that AP's knowledge of $(d_B,\theta_B)$ could be erroneous.
In section\,\ref{chp5:syst:vndkawia}, it has been stated that the true value of $d_B$ is not vital; thus the error on $d_B$ is not considered here.
In this section, the impact of the angle error on the  configurable beamforming technique is studied.
Furthermore, the impact of $R$ on the error performance is evaluated.

Let $\hat{\theta}_B$ denote the erroneous estimation of $\theta_B$.
There are different types of errors that can lead to an erroneous estimation of $\theta_B$, e.g., imperfect feedback from Bob to the AP.
In this section, without the discussion of the detailed types of errors, a generalized error on $\theta_B$, i.e., the uniform distribution $\hat{\theta}_B\sim \mathcal{U}(0,2\pi)$, is assumed.

\nomenclature{$\hat{\theta}_B$}{erroneous estimation of $\theta_B$}

For the configurable beamforming technique, the look-up tables $T^{(k)}$ are created and stored in the AP.
The erroneous $\hat{\theta}_B$ could lead to a wrong decision of the optimum $M_{ij}$, thus leading to an increase in $\bar{p}_C$.
Take $T^{(2)}$ in Table\,\ref{tab:two} as an example to illustrate the impact of $\hat{\theta}_B$. 
For example, assume that Bob is in zone 2 and the true angle is $\theta_B=30^{\circ}$.
According to $T^{(2)}$, the optimum array mode is $M_{21}$, which corresponds to the SSOP $\bar{p}_C=0.4597$, as it can be found in the upper plot in Fig.\,\ref{fig:chp5_p_DoE_arraymode}.
If $\hat{\theta}_B=50^{\circ}$, the array mode would be $M_{22}$, accoding to $T^{(2)}$.
The SSOP given by $M_{22}$ at $\theta_B=30^{\circ}$ is $0.6055$, which is larger than that of $M_{21}$.

While $\hat{\theta}_B$ could lead to a wrong array mode, thus an increased $\bar{p}_C$, it is also possible that $\hat{\theta}_B$ does not lead to a wrong array mode.
For the same example that $\theta_B=30^{\circ}$, if $\hat{\theta}_B=55^{\circ}$, the chosen array mode would still be $M_{21}$ according to $T^{(2)}$, which happens to be the optimum array mode. 
Therefore, how $\bar{p}_C$ is affected depends on both $\hat{\theta}_B$ and the particular look-up table $T^{(k)}$.

Let $M'_{ij}$ be the chosen array mode based on $\hat{\theta}_B$ at $(d_B,\theta_B)$.
Let $\Delta \bar{p}_C$ denote the difference between $\bar{p}_C$ using $M'_{ij}$ and $\bar{p}_C$ using $M_{ij}$.
Therefore, $\Delta \bar{p}_C$ depends on both $(d_B,\theta_B)$ and $\hat{\theta}_B\sim\mathcal{U}(0,2\pi)$.
The mean value of $\Delta \bar{p}_C$ at $(d_B,\theta_B)$ is then $\mathbb{E}_{\hat{\theta}_B}[\Delta \bar{p}_C]$.
Since Bob is randomly distributed in the coverage zone, the mean value of the increased SSOP over all possible $(d_B,\theta_B)$, denoted by $\text{err}_{\hat{\theta}_B}$, can be calculated by
\begin{align}\label{eq:chp5_erroranalysis}
	\text{err}_{\hat{\theta}_B}=\mathbb{E}_{\hat{\theta}_B,z_B}[\Delta \bar{p}_C]=\frac{1}{S}\int_{0}^{2\pi}\int_{0}^{d_{\text{max}}}d_B\mathbb{E}_{\hat{\theta}_B}[\Delta \bar{p}_C]\,\mathrm{d}d_B\,\mathrm{d}\theta_B.
\end{align}
$\text{err}_{\hat{\theta}_B}$ in (\ref{eq:chp5_erroranalysis}) can be numerically calculated in the similar way to (\ref{eq:chp5_eiruel}).

\nomenclature{$\Delta \bar{p}_C$}{SSOP increase caused by $\hat{\theta}_B$}

In Fig.\,\ref{fig:chp5_ErrAnalysis}, the result of $\text{err}_{\hat{\theta}_B}$ versus $R$ is shown, where $N_{\text{max}}=8$.
It can be seen that the error performance varies with $R$.
At $R=1.68\lambda$, $\text{err}_{\hat{\theta}_B}$ is the smallest, which means that in terms of error performance, $R=1.68\lambda$ is the optimum value in the range of $R\in[0.4\lambda,2\lambda]$.
Notice that in Fig.\,\ref{fig:chp5_p2bar_R}, $R=1.76\lambda$ gives the smallest value of $\bar{\bar{p}}_C$ for all $(d_B,\theta_B)$.
Thus, the minimum averaged SSOP and the minimum error are not given by the same radius.

\begin{figure}
\centering
\includegraphics[scale=0.9]{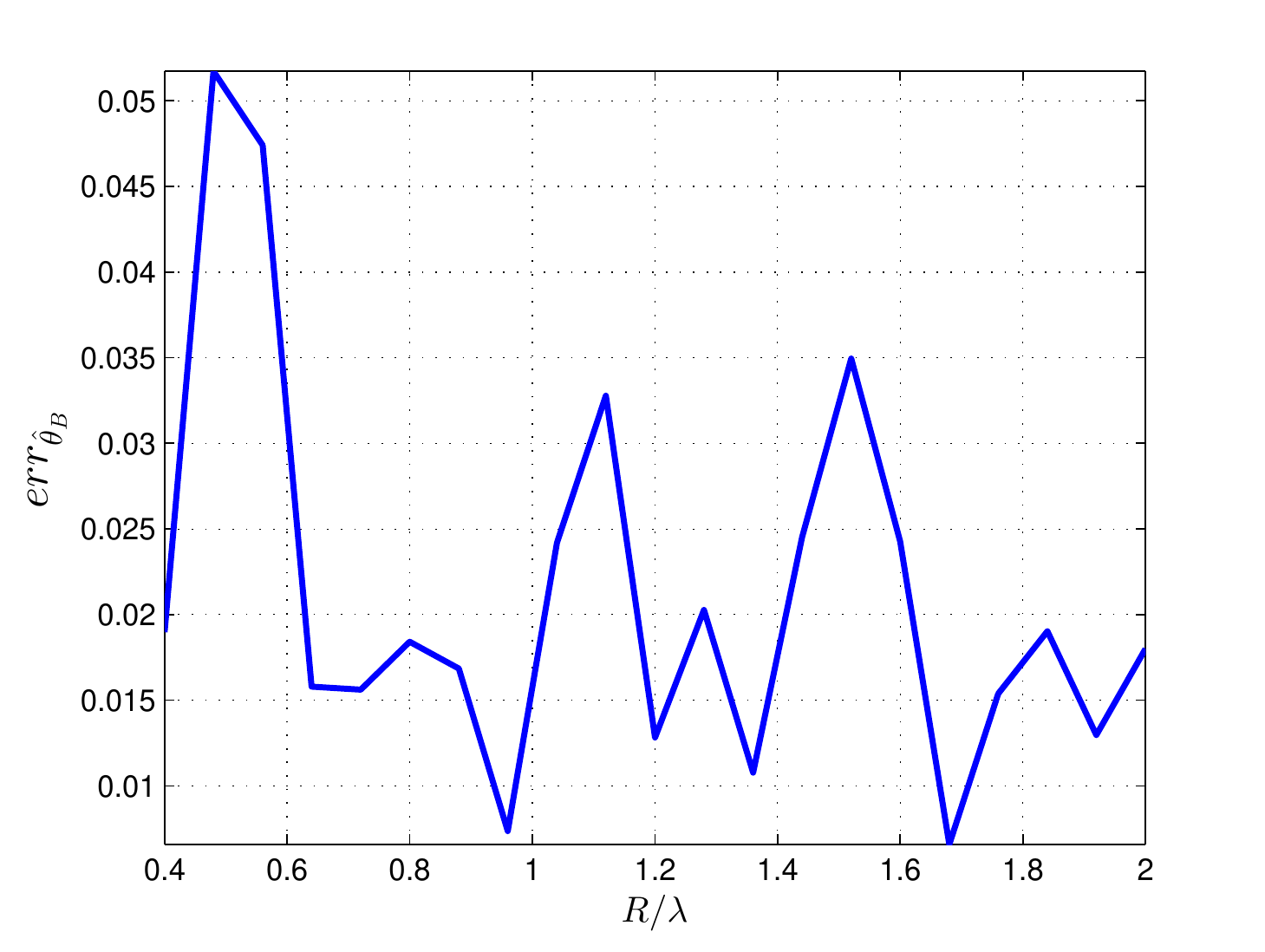}
\caption{$\text{err}_{\hat{\theta}_B}$ versus $R$}
\label{fig:chp5_ErrAnalysis}
\end{figure}

Comparing Fig.\,\ref{fig:chp5_ErrAnalysis} and Fig.\,\ref{fig:chp5_p2bar_R} jointly, it is easily noticed that there exists no one-to-one mapping between them.
While $\bar{\bar{p}}_C$ is the security metric, $\text{err}_{\hat{\theta}_B}$ refers to the reliability or the resistance to $\hat{\theta}_B$.
Thus, there is a trade-off between 'security' and 'reliability' when choosing the value of $R$.
It is also worth noticing that the method used in this section applies to a generalized error distribution, since $\text{err}_{\hat{\theta}_B}$ in (\ref{eq:chp5_erroranalysis}) does not limit for the uniform distribution.

\section{Impact of Mutual Coupling}
\label{chp5:mc_err}
\subsection{NEC Results for a Range of Radius}
\label{chp5:mc_err:euyr}

The analysis of $\bar{p}_C$ in Section\,\ref{chp5:analysis} and the optimization algorithms in Section\,\ref{chp5:opt_alg} are based on the numerical results that do not include the mutual coupling.
Section\,\ref{chp4:sec5} mainly investigated the mutual coupling for the UCA  in terms of $\rho$ between the true and NEC patterns as well as $G_{\text{max}}$.
In this section, further results of the mutual coupling is given for a range of radius, in order to investigate its impact on the optimization algorithms in Section\,\ref{chp5:opt_alg}.

To have a comprehensive understanding of the maximum gain attenuation, the patterns of the UCA in the range $R\in[0.4\lambda,2\lambda]$ and several typical values of $\theta_B$ are simulated in the NEC tool.
Then, $G_{\text{max}}$ of each pattern is recorded and plotted in Fig.\,\ref{fig:chp5_Gmax_R_MC}.
It can be seen that for different value of $R$, there exists an different extent of attenuation of $G_{\text{max}}$.
For example, the UCA shown in Fig.\,\ref{fig:chp4_NEC_pattern_DoE_UCA} has a radius of $R=0.6533\lambda$, which does has a very small variation of $G_{\text{max}}$.
For other values of $R$ in this range, there exists a relatively larger variation.
The minimum value of $G_{\text{max}}$ for UCA in the range $R\in[0.4\lambda,2\lambda]$ is higher than 0.8, which is larger than that of the ULA with the same number of elements as shown in Fig.\,\ref{fig:chp4_p_out_DoE_MC}.

\begin{figure}
\centering
\includegraphics[scale=0.9]{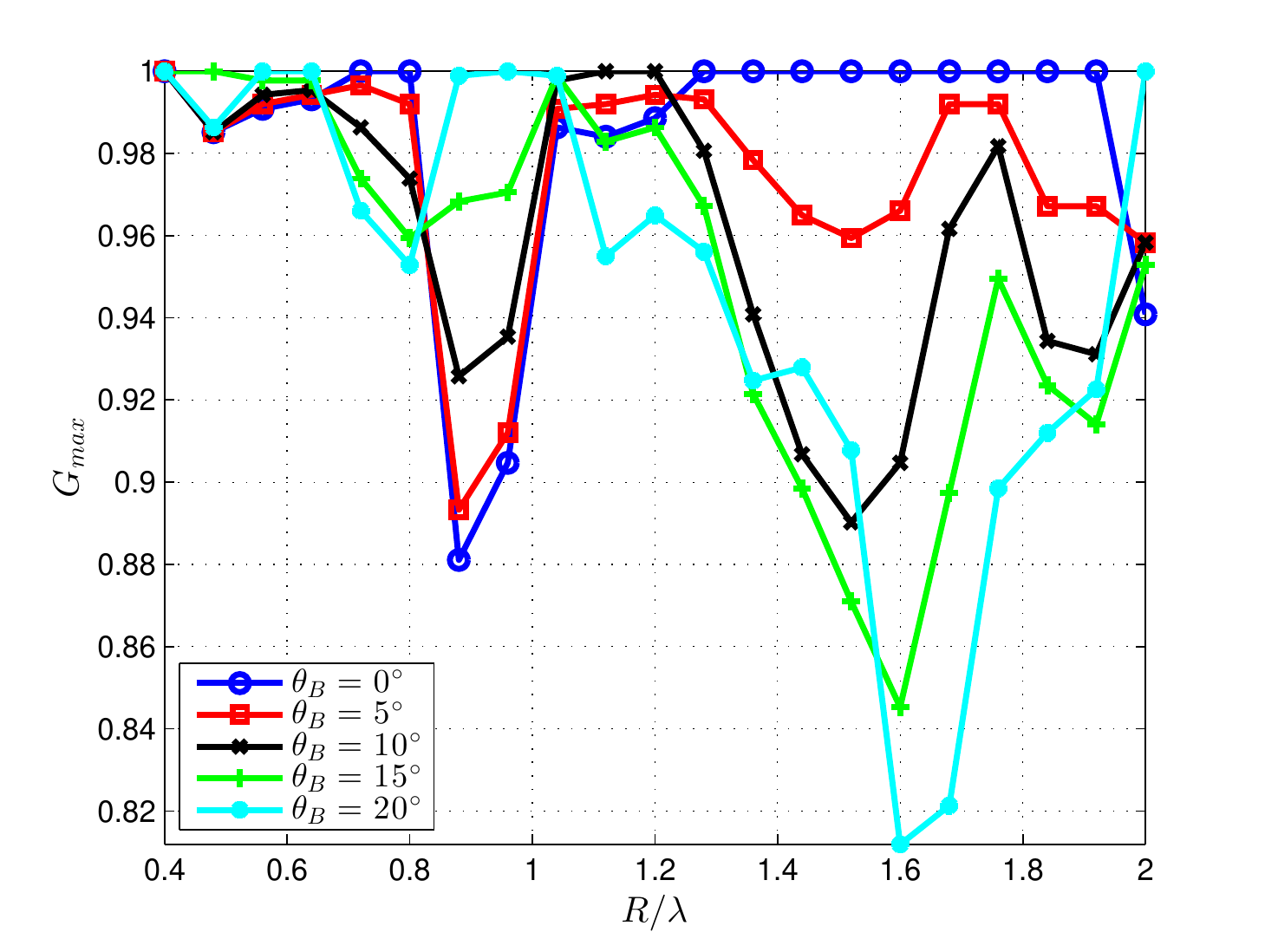}
\caption{$G_{\text{max}}$ versus $R$ for different $\theta_B$, $N_{max}=8$}
\label{fig:chp5_Gmax_R_MC}
\end{figure}

Beside $G_{\text{max}}$, the patterns of $G_C(\theta,\theta_B)$ are altered to a different extent for different $R$.
First, Fig.\,\ref{fig:chp5_MC_pattern_example} shows an example to illustrate the difference caused by the mutual coupling.
$N_{max}=8$, $R=0.8\lambda$, $\theta_B=0^{\circ}$.
The pattern without the mutual coupling is numerically calculated according to (\ref{eq:chp4_AF_UCA}).
The pattern with the mutual coupling is calculated by the NEC simulation.
Since $G_{\text{max}}=\sqrt{8}$, the NEC pattern is normalized to $\sqrt{8}$ before compared with the theoretical pattern. 
It can be seen that there is not much difference in the mainbeam. 
However, there is a bigger difference in the sidelobes due to the mutual coupling. 

\begin{figure}
\centering
\includegraphics[scale=0.9]{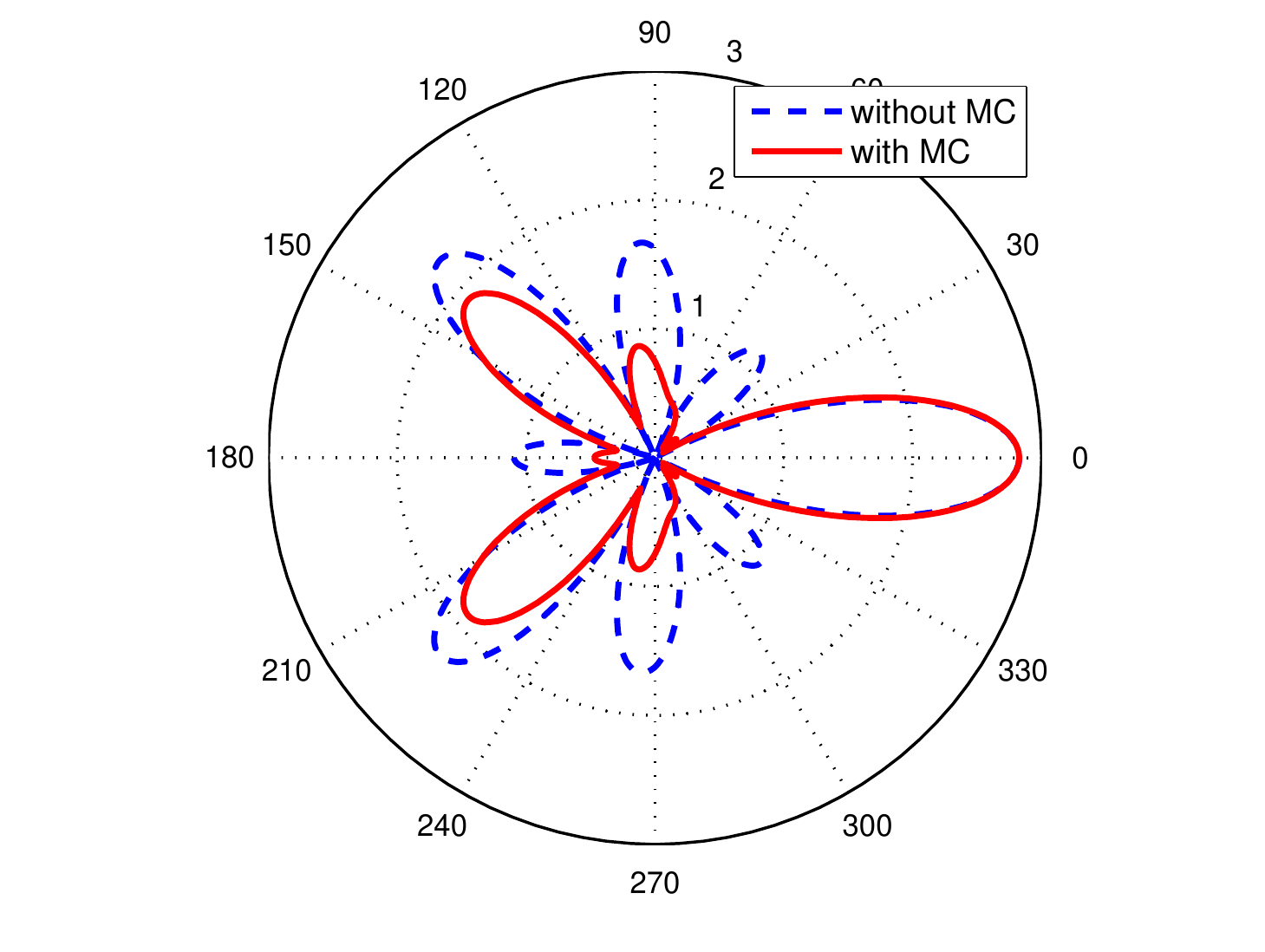}
\caption{Example of theoretical and NEC simulated patterns, $N_{max}=8$, $R=0.8\lambda$, $\theta_B=0^{\circ}$}
\label{fig:chp5_MC_pattern_example}
\end{figure}

Next, the patterns of the UCA with a wide range of configurations are simulated in NEC.
For $N_{max}=8$, typical values are chosen, i.e., $\theta_B=0^{\circ},5^{\circ},10^{\circ},15^{\circ},20^{\circ}$ in the range $R=[0.4\lambda,2\lambda]$.
The correlation coefficient $\rho$ between the theoretical and NEC patterns is calculated.
The results are shown in Fig.\,\ref{fig:chp5_MC_pattern_correlations}.

\begin{figure}
\centering
\includegraphics[scale=0.9]{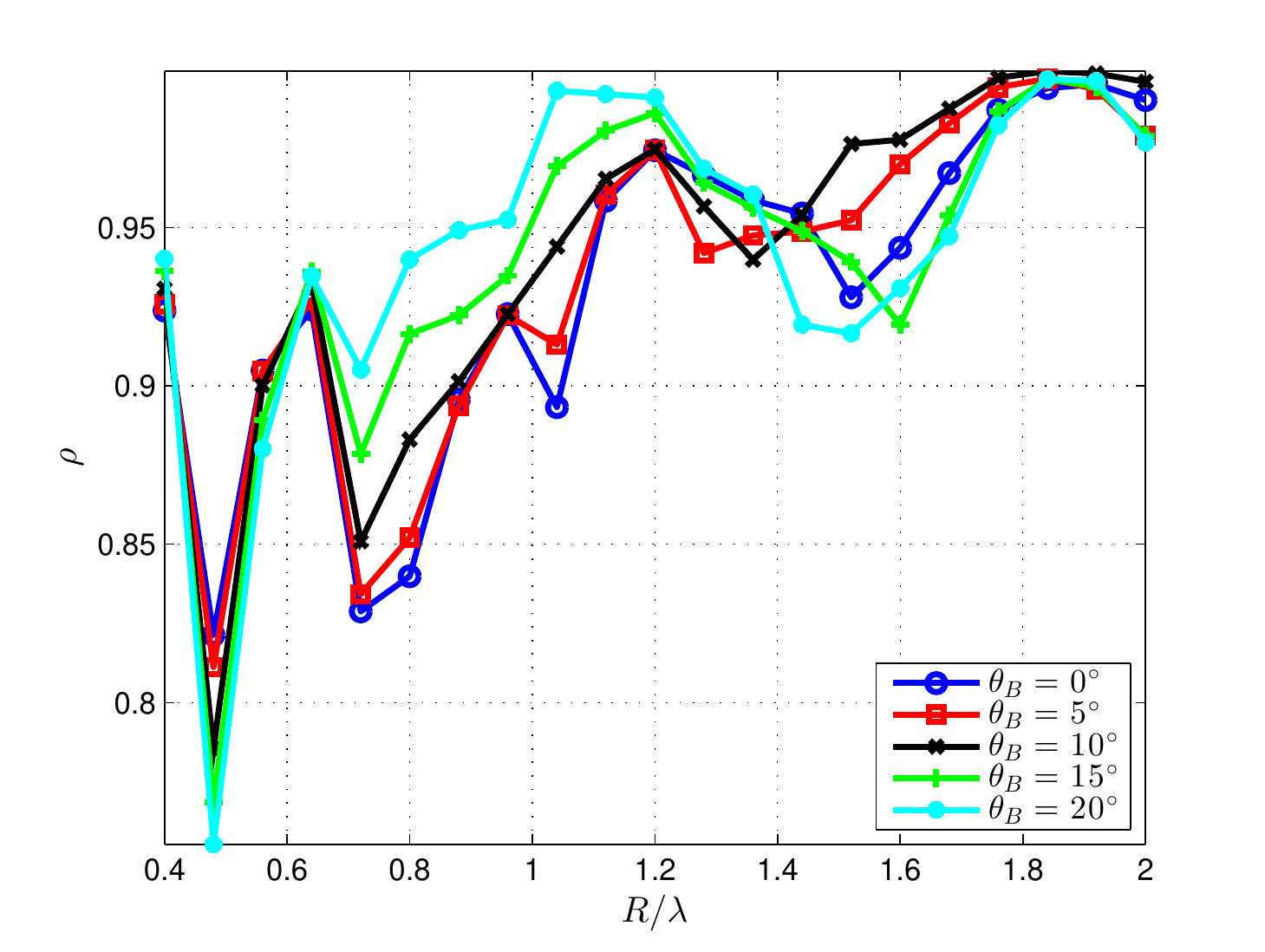}
\caption{Correlation coefficient of theoretical and NEC simulated patterns, $N_{max}=8$}
\label{fig:chp5_MC_pattern_correlations}
\end{figure}

It can be seen that $\rho$ is generally above $0.8$ in the range of $R=[0.4\lambda,2\lambda]$, except for $R=0.48\lambda$.
This shows that the mutual coupling does not cause a significant distortion to the pattern of UCA.
The high correlation between the theoretical and NEC patterns indicates that the optimization algorithm, which is based on empirical results on the theoretical patterns, can still work when considering the mutual coupling.
On the other hand, there do exist some differences between the theoretical and NEC patterns, which means that when calculating $R_{opt}$ or $T^{(k)}$ in line\,\ref{alg1:line19} in Algorithm\,1 and line\,\ref{alg2:line24} in Algorithm\,2, the NEC data instead of the theoretical data should be used. 
In addition, it is interesting to see that as $R$ increases, $\rho$ also increases in general.
Because the mutual coupling originates from the proximity of neighbor antennas, the closer the neighbor antennas are, the stronger the mutual coupling effect is.

\subsection{Mutual Coupling Effect for Adjustable Transmit Power}
\label{chp5:mc_err:iweownl}

To find $R_{opt}$ in (\ref{eq:chp5_mmse2}) when $P_t$ is adjustable, the optimization algorithm should be based on the NEC simulation data instead of the theoretical calculation.

$\bar{p}_C$ with the mutual coupling is calculated according to in Algorithm\,1, and the results are shown in the upper plot in Fig.\,\ref{fig:chp5_MC_p_p2bar_R}.
To compare with the upper plot in Fig.\,\ref{fig:chp5_p_R_All}, the same array configuration is adopted, i.e., $N_{max}=8$ and $R\in[0.4\lambda,2\lambda]$. 
In addition, the same values of $\theta_B$ are chosen.

\begin{figure}
\centering
\includegraphics[scale=0.9]{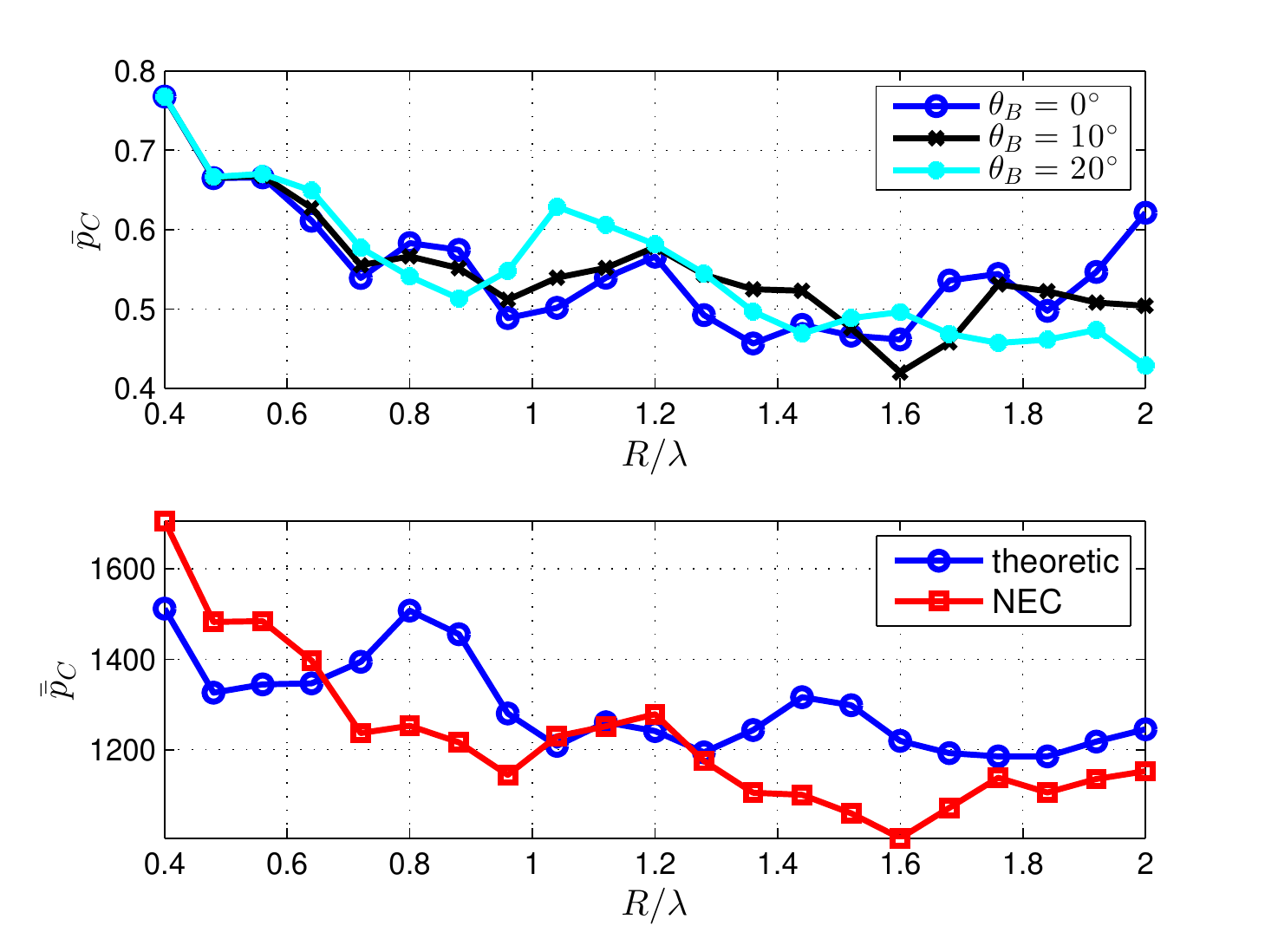}
\caption{Upper plot: $\bar{p}_C$ versus $R$. Lower plot: $\bar{\bar{p}}_C$ versus $N_{max}=8$}
\label{fig:chp5_MC_p_p2bar_R}
\end{figure}

Compared to Fig.\,\ref{fig:chp5_p_R_All}, it is not hard to notice the similarity between the theoretical and NEC simulated curves for the same $\theta_B$, which can be explained by the high correlation between them, as shown in Fig.\,\ref{fig:chp5_MC_pattern_correlations}.
However, some differences can be observed.

Because of the differences between the theoretical and NEC simulated results, $\bar{\bar{p}}_C$ in Fig.\,\ref{fig:chp5_p2bar_R} needs to be re-calculated based on the NEC results, in order to find $R_{opt}$.
The lower plot in Fig.\,\ref{fig:chp5_MC_p_p2bar_R} shows $\bar{\bar{p}}_C$ based on the NEC results in comparison with the theoretical curve.
It can be seen that the optimum value for the NEC result is $R_{opt}=1.6\lambda$ compared to $R_{opt}=1.76\lambda$ for the theoretical result.

\subsection{Mutual Coupling Effect for Fixed Transmit Power}
\label{chp5:mc_err:ioewovn}

In this section, the impact of the mutual coupling on the configurable beamforming technique is studied.
The same as in Section\,\ref{chp5:mc_err:iweownl}, the NEC data that includes the mutual coupling should also be used in the calculation of the optimum $M_{ij}$ as well as $R_{opt}$.

As introduced in Section\,\ref{chp5:opt_alg:nvnnnnc}, the first step is to find out which zone Bob is in according to $d_B$.
Due to the maximum gain attenuation, the coverage zones, e.g. in Fig.\,\ref{fig:chp5_N_zone}, are distorted, because the zone boundary $d_{th,N}$ in (\ref{eq:chp5_d_th}) is proportional to $G_{\text{max}}$.

In Fig.\,\ref{fig:chp5_Gmax_DoE_MC}, $G_{\text{max}}$ of three different array modes for an 8-element UCA with $R=1.6\lambda$ is shown.
Recall that in Fig.\,\ref{fig:chp5_Gmax_R_MC}, it has already been shown that there exists attenuation at $R=1.6\lambda$ for $N=8$.
It can be seen that for sub-arrays with $N=4$ and $N=2$, there is also attenuation over $\theta_B$.

\begin{figure}
\centering
\includegraphics[scale=0.9]{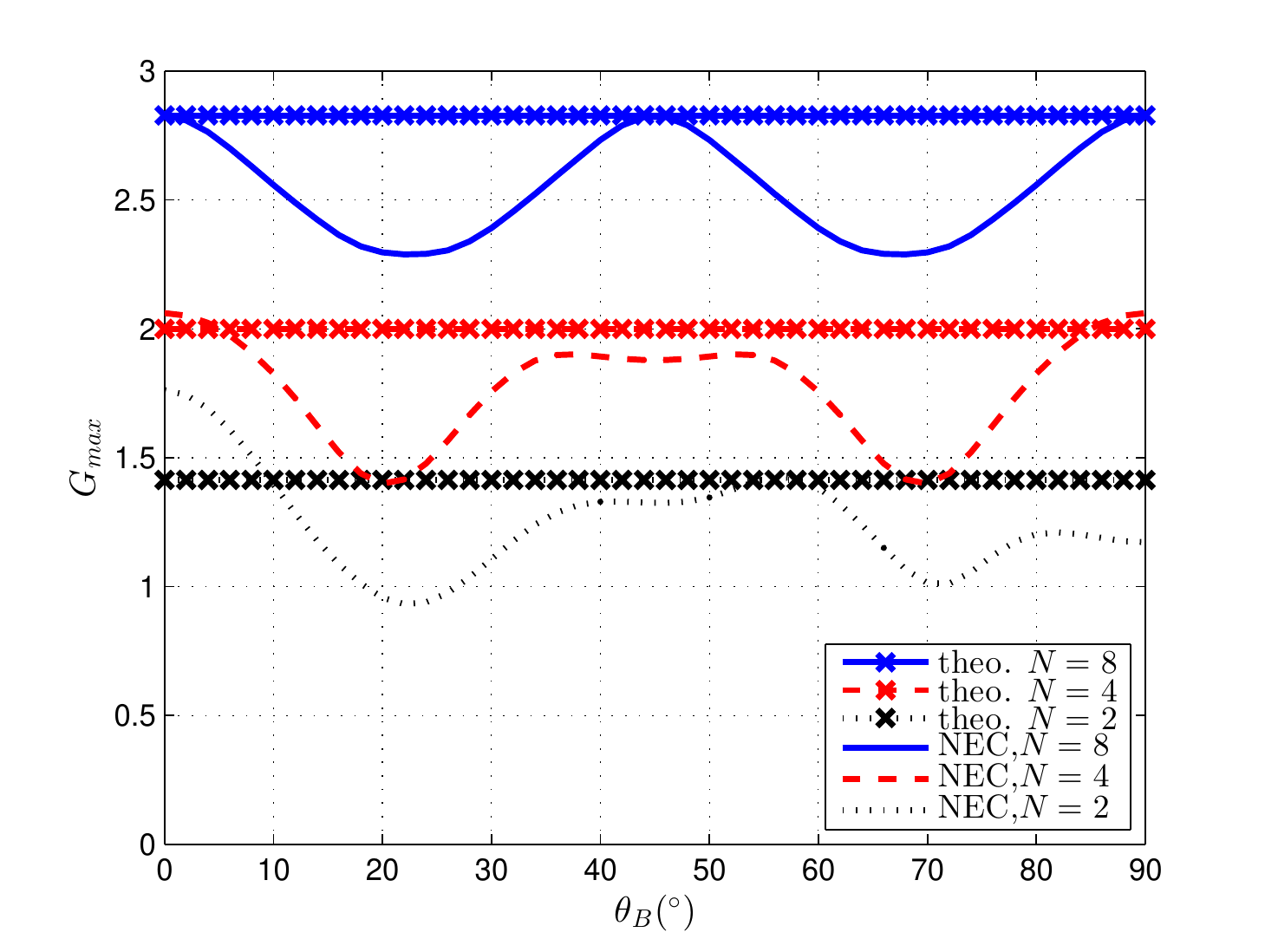}
\caption{Comparison of $G_{\text{max}}$ between theoretical and NEC results for $\theta_B\in[0^{\circ},90^{\circ}]$}
\label{fig:chp5_Gmax_DoE_MC}
\end{figure}

The attenuation of $G_{\text{max}}$ leads to the changes of the zone boundary, i.e., $d_{th,N}$, which in turn affects the choice of available array modes to certain Bob's location. 
For example, for Bob at $\theta_B=20^{\circ}$ and $d_B=d_{th,2}$, because the actual transmit gain at this angle is smaller than the theoretical value, $C_B\geq R_B$ cannot be guaranteed if the AP chooses $N=2$.
Thus, in practice, the AP need to keep track of the actual zone boundary, which can be done by converting Fig.\,\ref{fig:chp5_Gmax_DoE_MC} into a look-up table.

Besides the impact on the zone boundary, the mutual coupling also affects the configurable beamforming technique in the sense of creating the look-up tables shown in Algorithm\,2.
In practice, the computation in line\,\ref{alg2:line24} should be based on NEC simulation data instead of theoretical calculation.
To illustrate the difference, Fig.\,\ref{fig:chp5_p_DoE_differentN} is re-generated based on the NEC simulation results.
The results are shown in Fig.\,\ref{fig:chp5_MC_p_DoE_differentN}.

\begin{figure}
\centering
\includegraphics[scale=0.9]{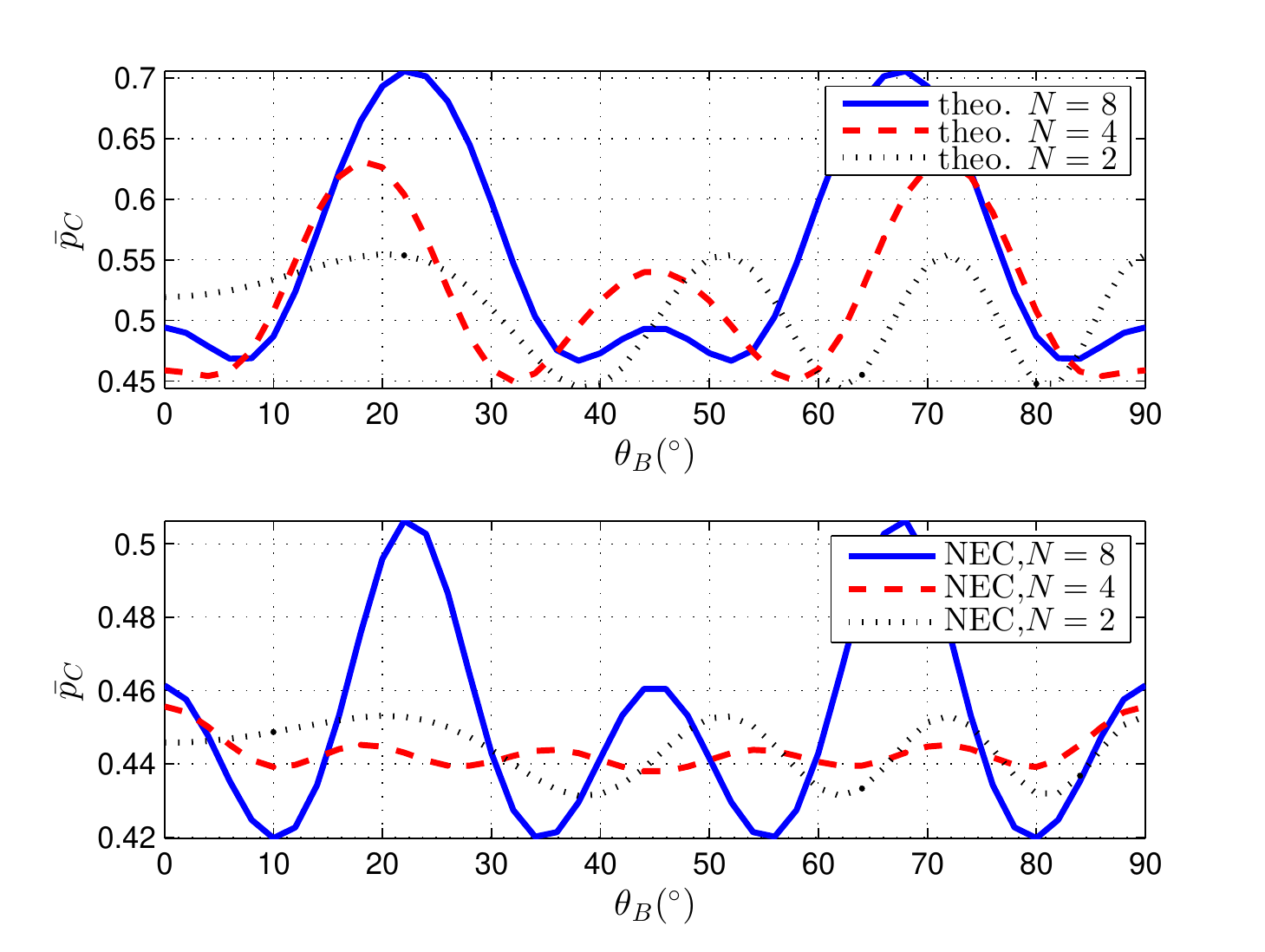}
\caption{$\bar{p}_C$ versus $\theta_B$ for different $N$. Upper plot: theoretical data without the mutual coupling; lower plot: NEC simulated data with the mutual coupling}
\label{fig:chp5_MC_p_DoE_differentN}
\end{figure}

For comparison, the theoretical results and the NEC simulation results are given in the upper and lower plots, respectively.
It can be seen that for the same $N$, $\bar{p}_C$ changes differently with $\theta_B$.
In addition, the relative position of the three curves, i.e., $N=8,4,2$, is also changed. 
Thus, the generation of look-up tables should be based on the practical data with the mutual coupling.

Since the mutual coupling affects $d_{th,N}$ and the generation of look-up tables, it is natural to conjecture that $R_{opt}$ in (\ref{eq:chp5_mmse4}) is also affected by the mutual coupling.
To solve (\ref{eq:chp5_mmse4}), $\bar{\bar{p}}_C^{(k)}$ and $q^{(k)}$ must be re-calculated based on the NEC data.
According to (\ref{eq:chp5_zone_prob}), $q^{(k)}$ is calculated from the zone area $S^{(k)}$.
Because $d_{th,N}$ is now angle dependent, $S^{(k)}$ in (\ref{eq:chp5_zone1_area})-(\ref{eq:chp5_zone_area}) can now be calculated by
\begin{align}
	&S^{(1)}=\frac{1}{2}\int_0^{2\pi}d_{th,2}^2(\theta_B)\,\mathrm{d}\theta_B, \label{eq:chp5_zone1_area_MC} \\
	&S^{(2)}=\frac{1}{2}\int_0^{2\pi}d_{th,4}^2(\theta_B)-d_{th,2}^2(\theta_B)\,\mathrm{d}\theta_B,  \label{eq:chp5_zone2_area_MC} \\
	&S^{(3)}=\frac{1}{2}\int_0^{2\pi}d_{\text{max}}^2(\theta_B)-d_{th,4}^2(\theta_B)\,\mathrm{d}\theta_B,  \label{eq:chp5_zone3_area_MC} \\
	&S=\sum_k S^{(k)}= \frac{1}{2}\int_0^{2\pi}d_{\text{max}}^2(\theta_B)\,\mathrm{d}\theta_B.  \label{eq:chp5_zone_area_MC}
\end{align}
The numerical result of $\bar{\bar{p}}_C$ versus $R$ is shown in Fig.\,\ref{fig:chp5_MC_p2bar_R}.
For comparison, $\bar{\bar{p}}_C$ that is based on the theoretical data is also shown.

\begin{figure}
\centering
\includegraphics[scale=0.9]{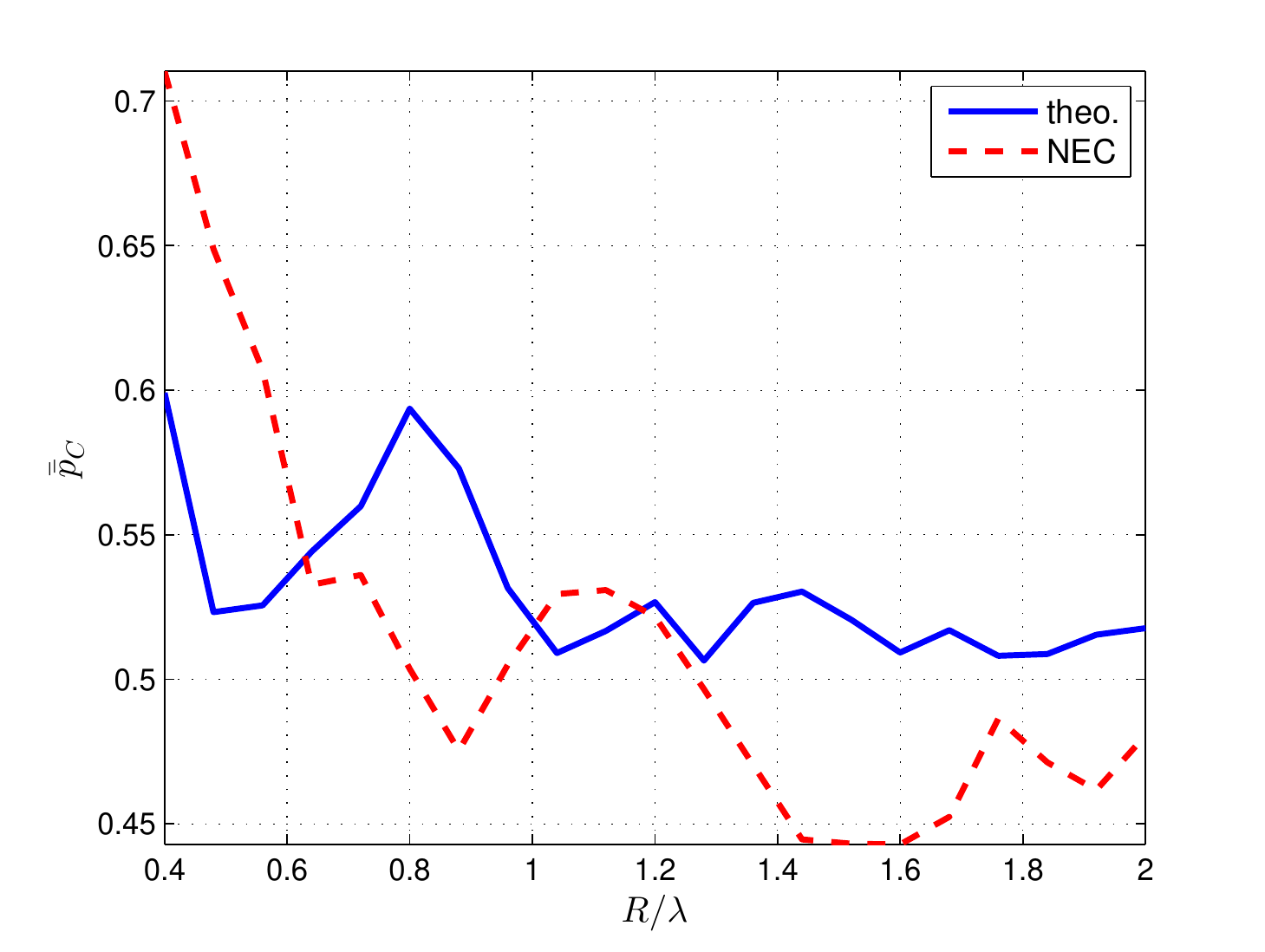}
\caption{$\bar{\bar{p}}_C$ versus $R$ with and without the mutual coupling, $N_{\text{max}}=8$}
\label{fig:chp5_MC_p2bar_R}
\end{figure}

It can be seen that the mutual coupling dramatically changes $\bar{\bar{p}}_C$ over the range of $R$. 
The minimum value of $\bar{\bar{p}}_C$ is 0.4429 at $R=1.6\lambda$ for the NEC result, compared to 0.5082 at $R=1.76\lambda$ for the theoretical result.
From the above analysis, it can be concluded that even though the configurable beamforming still works with the mutual coupling effect, its security performance and optimization are affected by the mutual coupling.
The generation of $T^{(k)}$ and the searching of $R_{opt}$ need to be based on the NEC data.

\section{Conclusions}
\label{chp5:concl}
In this chapter, the secure transmission to Bob with ER based beamforming in presence of PPP distributed is enhanced by optimizing the array configuration of the adjustable UCA according to Bob's location.
Based on the empirical results, two different optimization algorithms have been developed for the two scenarios, the fixed and adjustable transmit power.
For the adjustable transmit power scenario, the transmit power is adjusted according to Bob's distance and the radius is optimized for all Bob's angles to minimize the averaged SSOP.
For the fixed transmit power scenario, the optimum array mode is chosen according to Bob's location, which is called the configurable beamforming technique that generates look-up tables that stores the optimum array mode for Bob's specific location.
The two optimization algorithms have the potential to be generalized for any channel type and any array type.
It is worth noticing that the former scenario is more practical of relevance, e.g., base stations in cellular network, while the latter scenario is only applicable in applications with fixed transmit power.
In addition, the impact of the mutual coupling is evaluated for the optimization algorithms.
The high correlation coefficient between the theoretical and NEC data suggests that the mutual coupling in general does not harm the algorithms; 
however, the optimum values are changed and should be calculated based on the NEC data in the two numerical algorithms.

\chapter{Conclusions}
\label{chp6}

\section{Concluding Remarks}
\label{chp6:part1}

The goal of this thesis is to enhance the wireless security in the physical layer from the spatial aspect. 
For this purpose, the ER-based beamforming is investigated in the context of indoor wireless networks, e.g., the 802.11 WLAN, where the AP is equipped with antenna array.
The ER is created to protect the transmission to Bob in the presence of the PPP located Eves.
Various parameters, especially the array configuration, are examined towards security and numerical optimization algorithms are developed to enhance the security.

Although there is abundant research work that aims to enhance the wireless security from the physical layer, not many work considers the problem from the spatial aspect, i.e., creating physical region to enhance the security level.
The existing approaches in the field of the physical space security lack of the support from the information-theoretic secrecy and there is rarely work that considers the optimization of the array itself for further enhancement of security.
With the aid of the stochastic geometry tool (i.e., the PPP distribution) and the mathematical derivation of the analytic expression of the upper bound for SSOP, this work solves the former problem by characterizing the ER based on the information-theoretic parameter, which brings many practical work surveyed in the field of the physical space security close to the information-theoretic security, and the latter problem is investigated by studying the impact of the array parameters on the SSOP and optimizing the array configuration according to Bob's dynamic location(s).

The research work mainly focuses on two commonly used array geometries, the ULA and the UCA. 
However, the SSOP related expressions are generally applicable to any array geometry that has either analytic expression or numerical representations. 
Furthermore, although the Rician channel is used in the system model, the method also applies to other type of fading channel, e.g., Nakagami fading channel.
In addition to the theoretical analysis, a practical beamformer is built on WAPR hardware.
Observations from the experimental results in the anechoic chamber reveal the mutual coupling effect on the ER, which is also intensively simulated in NEC simulation tool.
The conclusions of each chapter are summarized in the following.

In Chapter\,3, the concept of the ER, the SSOP and the upper bound of the SSOP are established, which is useful to a general array geometry and a general fading channel.
Then, the commonly used ULA is examined in terms of the SSOP.
The analysis shows that in general the SSOP and its upper bound increase with the DoE angle and asymptotically approach certain values depending on the DoE angle when the number of elements increases. 
The ratio of the SSOP upper bound to the SSOP has a relatively small value for a small value of number of elements;
however, for a large value of number of elements, the ratio grows very large, thus cannot be used to predict the behavior of SSOP.

In Chapter\,4, the SSOP analysis for a UCA is conducted and compared with a ULA in addition with the mutual coupling study with WARP and NEC results.
It shows that the SSOP of a UCA is more constant in the whole range of Bob's angle, although the SSOP of a ULA is smaller at the bore-sight direction.
For the UCA, the tightness of the upper bound does not change much even for a large number of elements.
Furthermore,  the UCA is less sensitive to the mutual coupling over Bob's angle range, and is more flexible on the choices of array configurations.
Thus, the UCA is a better choice in creating and optimizing the the SSOP.

In Chapter\,5, based on the SSOP analysis on the UCA in Chapter\,4,  the array configuration of the adjustable UCA is optimized  according to Bob's location.
Two optimization algorithms are developed for the fixed and adjustable transmit power scenarios.
For the adjustable transmit power scenario, the transmit power is adjusted according to Bob's distance and the radius is optimized for all Bob's angles to minimize the averaged SSOP.
For the fixed transmit power scenario, the configurable beamforming technique generates look-up tables that stores the optimum array mode for Bob's specific location.
While the former scenario makes piratical sense, the latter scenario is only applicable in certain applications with fixed transmit power.
The mutual coupling in general does not harm the algorithms; 
however, the optimum values are changed and should be calculated based on the NEC data in the two numerical algorithms.

Overall, it is clear that the ER formulation and the performance metric SSOP are crucial to the study of the wireless security in the physical layer.
By studying the ULA and the UCA in the Rician fading channel, some valuable insights are obtained, which provides the basis of the optimization for the array configuration for user with dynamic location in a dynamic environment.
The mathematical derivation of analytic expressions, theoretical and numerical results as well as the experimental and simulation results show that this work has meet the objective mentioned in the introduction of the thesis.

\section{Future Work}
\label{chp6:part2}

The above conclusions demonstrated the contributions obtained in this research work and indicates that wireless security can be further enhanced by the ER-based beamforming technique in the physical layer.
In this section, the improvements on the current work is discussed first. 
Then, the possibilities to extend the ER-based beamforming and combine with other techniques are addressed.

\subsection{Improvements to Current Work}

In this research, the spatial distribution of Bob and Eves are modeled as the homogeneous PPP.
In practice, the system varies from a small local network with a dozen of users to an integrated network with hundreds of users.
The locations of the users may not fit in the defined distribution.
For example, the potential Eves may be largely located in the public area, such as the hall or lobby in the office building other than uniformly scattered in the whole floor.
In this case, the optimization of the radius of the UCA based on the defined distribution may be not accurate.
To solve this problem, a more realistic spatial distribution should be studied or measured.

Another assumption on the system is that the accurate knowledge of Bob's instantaneous CSI at the AP, which could be very difficult, because the acquisition of the CSI requires feedback mechanism, which is vulnerable to the impairments of the wireless channel as well as active attacks.
In fact, there is a lot of ongoing research about the imperfect CSI and different approaches are developed for different system requirements.
One direction of the improvement is to study the impact of imperfect CSI on the characterization of the SSOP given different level knowledge of Bob's CSI.
The other direction could be the development of ER-based optimizations to reduce such impact.
In addition, while in theory Eves stay silent and are assumed to be known only by the distribution, in practice, this constraint could be loosened if some sort of knowledge of Eves can be exploited, which will give more advantage to optimize the array configuration.

Besides the above aspects from the system model, some theoretical aspect can be further improved.
The upper bound of the SSOP is vital for the theoretical analysis.
It is derived based on Jensen's inequality.
However, in this thesis, the tightness of the upper bound is analytically examined for limited cases.
Accurate analysis relies on numerical results.
On one hand, the upper bound should be analyzed quantitatively for different parameters to gain better understanding;
on the other hand, a tighter bound should be found if possible.

Finally, the configuration beamforming technique can be further improved.
In this thesis, the UCA is considered as an example to demonstrate the potential of adjusting array configuration to achieve higher level of security.
This technique can also be applied to the ULA and even some irregular form of arrays, e.g., non-uniform circular array, which has been studied intensively from the pattern synthesis perspective.
Another improvement can be made on error analysis for the configuration beamforming technique.
Currently, the radius is optimized with the uniformly distributed angle error.
In practice, there could be a different kind of distribution, e.g., if more advanced DoA estimation technique is used, the error may be located in a small region with high probability.

\subsection{Extensions to Other Techniques}

In this research, the ER-based beamforming is developed based on MISO system in the context of indoor WLAN where the AP is equipped with multiple antennas.
The link between the physical region and the array factor can be carried over to other systems either for beamforming or for transmission of the AN.
Thus, this research can serve as the foundation to further improve some existing work in the area of physical space security.
In this section, some initial ideas combining the ER-based beamforming with other existing work are presented.

In\,\cite{li2012secure} the antenna array is installed on Bob that generates the AN.
However, no optimization is provided based on the array configuration.
The ER-based beamforming can be applied to formalize the OSR and guides the design for further enhancement of the system security level.
Similar gap exists for many practical work that exploits the advantage of beamforming, such as\,\cite{4595864,5357443,sheth2009geo}.
The ER-based beamforming is applicable no matter it is single-AP or multiple-AP systems.
For example, the UCA with optimized radius can replace the UCA used in\,\cite{4595864} to get a better control on the jointly created region.
In\,\cite{6502515}, the UCA alternatively transmit with two pre-defined patterns that has an overlapped region which is defined as the ER.
Through the ER-based beamforming, the two pre-defined patterns can be further optimized.

Inspired by the work in\,\cite{6502515} where the fast fading effect is created by the alternating two overlapped patterns in the time domain, similar idea could be achieved in the frequency domain, for example, in the OFDM-based system.
Instead of transmitting the packet in a time-division manner, different sub-carriers bear different segment of the packet and transmit with two (or even more) pre-defined patterns.
These patterns can be optimized by the ER-based beamforming in this thesis.
An even wilder idea could be combining the spatial and code division methods and apply the idea with the ER-based beamforming in CDMA systems.

\begin{spacing}{0.9}
\bibliographystyle{IEEEtran}
\cleardoublepage
\bibliography{IEEEabrv,bibliography}
\end{spacing}

\begin{appendices} 

\chapter{Mathematical Derivations and Proofs}
\label{appdx:bessel}

\section{Proof of Lemma\,\ref{le:chp3_h_tilde_squre}}
\label{appdx:bessel:nmve}

According to (\ref{eq:chp3_h_tilde_Ri}),
\begin{align}
	\tilde{h}&=\sqrt{\frac{K}{K+1}}G(\theta,\theta_B)+\sqrt{\frac{1}{K+1}}\frac{\mathbf{s}^H(\theta_B)\mathbf{g}}{\sqrt{N}} \nonumber \\
	&=\sqrt{\frac{K}{K+1}}G(\theta,\theta_B)+\sqrt{\frac{1}{K+1}}g,
\end{align}
where 
\begin{align}
	g=\frac{\mathbf{s}^H(\theta_B)\mathbf{g}}{\sqrt{N}}=\frac{\sum_{i=1}^N e^{j\phi_i(\theta_B)}g_i}{\sqrt{N}}.
\end{align}
Because $\phi_i(\theta_B)$ is deterministic and $g_i\sim{CN}(0,1)$, so $\phi_i(\theta_B)g_i\sim{CN}(0,1)$. 
Thus, the sum of $N$ i.i.d. complex Gaussian random variables $g_i$ is also a  complex Gaussian random variable with zero mean and variance $N$.
Therefore, $g$ is a complex Gaussian variable, $g\sim{CN}(0,1)$.

Let $g_{Re}$ and $g_{Im}$ denote the real and imaginary part of $g$, where $g_{Re}$ and $g_{Im}$ are joint normal variables, i.e., $g_{Re}, g_{Im}\sim N(0,\frac{1}{2})$.
Thus, 
\begin{align}
	\tilde{h}=\sqrt{\frac{K}{K+1}}G(\theta,\theta_B)+\sqrt{\frac{1}{K+1}}g_{Re}+j\sqrt{\frac{1}{K+1}}g_{Im}.
\end{align}
Then the square of the amplitude of $\tilde{h}$ is obtained by
\begin{align}
	|\tilde{h}|^2 &= \Big[\sqrt{\frac{K}{K+1}}G(\theta,\theta_B)+\sqrt{\frac{1}{K+1}}g_{Re}\Big]^2+ \frac{1}{K+1}g_{Im}^2 \nonumber \\
	&=\frac{K}{K+1}G^2(\theta,\theta_B)+\frac{1}{K+1}g_{Re}^2+\frac{1}{K+1}g_{Im}^2+\frac{2\sqrt{K}}{K+1}G(\theta,\theta_B)g_{Re}.
\end{align}

\section{Proof of Theorem\,\ref{th:chp3_A_0L}}
\label{appdx:bessel:veaien}

(\ref{eq:chp3_A_0}) can be written as
\begin{align}\label{eq:appdx_bessel_vniewoa}
	A_0=\int_0^{2\pi} G^2(\theta,\theta_B)\,\mathrm{d}\theta.
\end{align}
According to (\ref{eq:chp2_AF_ULA}), it can be derived that
\begin{align}
	G(\theta,\theta_B)=\frac{1}{\sqrt{N}}\sum_{i=1}^{N} e^{jk\Delta d(\sin\theta_B-\sin\theta)(i-1)},
\end{align}
where $k=\frac{2\pi}{\lambda}$.
Thus, (\ref{eq:appdx_bessel_vniewoa}) can be further derived,
\begin{align}
	A_0&=\int_0^{2\pi} G(\theta,\theta_B)\cdot G^*(\theta,\theta_B) \,\mathrm{d}\theta \nonumber \\
	   &=\int_0^{2\pi} \frac{1}{N}\sum_{i=1}^{N} e^{jk\Delta d(\sin\theta_B-\sin\theta)(i-1)}
	  \sum_{j=1}^{N} e^{-jk\Delta d(\sin\theta_B-\sin\theta)(j-1)} \,\mathrm{d}\theta \nonumber \\
	   &=\int_0^{2\pi} \frac{1}{N}\sum_{i,j} e^{jk\Delta d(\sin\theta_B-\sin\theta)(i-j)} \,\mathrm{d}\theta \nonumber \\
	   &=\int_0^{2\pi} \frac{1}{N}\sum_{i,j} e^{jk\Delta d\sin\theta_B(i-j)}
	             e^{-jk\Delta d\sin\theta(i-j)} \,\mathrm{d}\theta \nonumber \\
		 &=\frac{1}{N}\sum_{i,j} e^{jk\Delta d\sin\theta_B(i-j)}
	  \int_0^{2\pi} e^{-jk\Delta d\sin\theta(i-j)} \mathrm{d}\theta. \label{eq:appdx_bessel_aieow}
\end{align}
According to the integral representation of the Bessel function of the first kind, $J_n(x)=\frac{1}{2\pi}\int_{-\pi}^{\pi}e^{j(n\tau-x\sin\tau)}\mathrm{d}\tau$\,\cite{zwillinger2007table}, (\ref{eq:appdx_bessel_aieow}) can be further derived,
\begin{align}\label{eq:appdx_bessel_euirue}
  A_0=\frac{2\pi}{N}\sum_{i,j} J_0(k\Delta d(i-j))e^{jk\Delta d(i-j)\sin\theta_B},
\end{align}
where $A_0$ is the finite double summation. Next $A_0$ will be simplified from a double summation to a single summation.

$A_0$ is the summation of $N\times N$ terms, each of which is denoted by $A_{0,i,j}$,
\begin{align}
	A_{0,i,j}=\frac{2\pi}{N} J_0(k\Delta d(i-j))e^{jk\Delta d(i-j)\sin\theta_B}.
\end{align}
Notice that the only variable across all $A_{0,i,j}$ is the difference $i-j$.
So let $n=i-j$ and it can be derived that
\begin{align}
	A_{0,n}=\frac{2\pi}{N} J_0(k\Delta dn)e^{jk\Delta dn\sin\theta_B}.
\end{align}
Then, all the values of $n$ that are associated with $A_{0,n}$ are mapped into a table shown in Fig.\,\ref{fig:appdx_bessel_table_ULA}.

\begin{figure}
\centering
\includegraphics[scale=1]{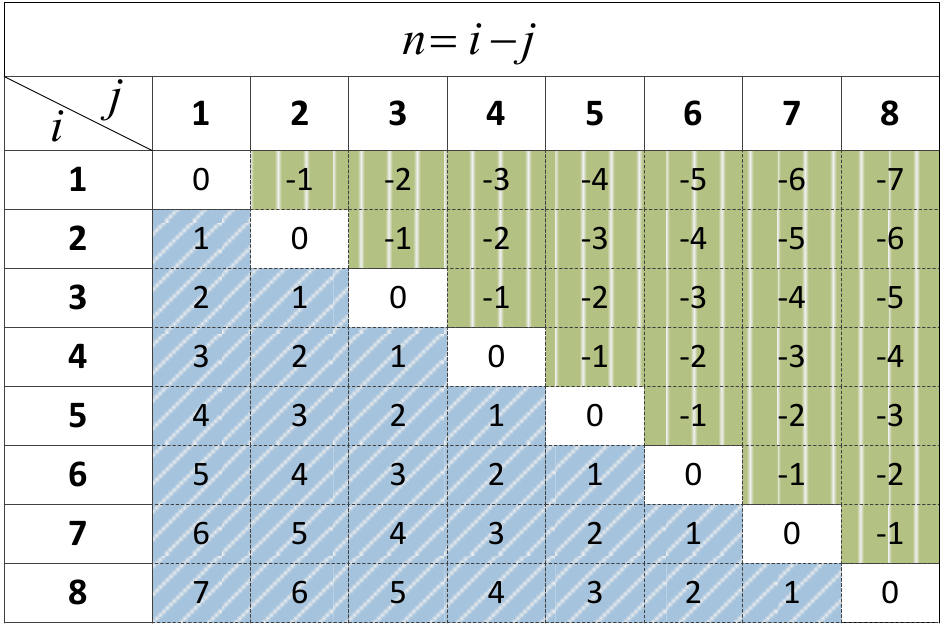}
\caption{Table for $A_{0,i,j}$}
\label{fig:appdx_bessel_table_ULA}
\end{figure}

Observing the table in Fig.\,\ref{fig:appdx_bessel_table_ULA}, it is noticed that i) the terms of $A_{0,n}$ on the diagonal lines can be combined, because they are the same; 
ii) becuase $J_m(-x)=(-1)^mJ_m(x)$, the terms of $A_{0,n}$ that have the same absolute value of $n$ can be added
\begin{align}
	A_{0,n}+A_{0,-n}&=\frac{2\pi}{N} [J_0(k\Delta dn)e^{jk\Delta dn\sin\theta_B}+J_0(-k\Delta dn)e^{-jk\Delta dn\sin\theta_B}] \nonumber \\
	&=\frac{4\pi}{N}J_0(k\Delta dn)\cos(k\Delta dn\sin\theta_B).
\end{align}
In addition, when $n=0$, $J_0(0)=1$ and $e^{j0}=1$.
Thus, $A_{0,0}=\frac{2\pi}{N}$.
Now sum up the terms of $A_{0,n}$ on each diagonal lines from $n=0$ to $p=N-1$ and obtain
\begin{align}\label{eq:appdx_bessel_ieur}
	A_0=2\pi+4\pi\sum_{n=1}^{N-1} \frac{N-n}{N}J_0(k\Delta dn)\cos(k\Delta dn\sin\theta_B).
\end{align}

\section{Proof of Theorem\,\ref{th:chp3_HPBW_L}}
\label{appdx:bessel:oitor}

In the range $\theta_B\in[0,\frac{\pi}{2}]$, both $\sin\theta_B$ and $\cos\theta_B$ are in the range $[0,1]$.
First, the zeros of $\frac{\partial}{\partial\theta_B}\Delta\theta_{HP}$ are to be found.
\begin{align} 
	&\frac{\partial}{\partial\theta_B}\Delta\theta_{HP}=0 \nonumber \\
	\Leftrightarrow & \sqrt{1-(\sin\theta_B-\frac{2.782}{Nk\Delta d})^2}=\cos\theta_B.
\end{align}
Then, take the power of two on both sides,
\begin{align}
	& 1-(\sin\theta_B-\frac{2.782}{Nk\Delta d})^2=\cos^2\theta_B  \nonumber \\
	\Leftrightarrow & 1-\cos^2\theta_B = (\sin\theta_B-\frac{2.782}{Nk\Delta d})^2 \nonumber \\
	\Leftrightarrow & \sin^2\theta_B = (\sin\theta_B-\frac{2.782}{Nk\Delta d})^2. \nonumber \\
\end{align}
Because $\sin\theta_B$ and $\frac{2.782}{Nk\Delta d}$ are both non-negative values, the solution of the following equation is
\begin{align}
	& \sin\theta_B = -(\sin\theta_B-\frac{2.782}{Nk\Delta d}) \nonumber \\
	\Leftrightarrow & \sin\theta_B = \frac{1.391}{Nk\Delta d}.
\end{align}

For $\Delta d=0.5\lambda$, it can be derived that
\begin{align}
	\sin\theta_B = \frac{1.391}{N\pi}.
\end{align}
Since $\frac{1.391}{N\pi}<1$, in the range $\theta_B\in[0,\frac{\pi}{2}]$, there is one zero of $\frac{\partial}{\partial\theta_B}\Delta\theta_{HP}$, which is at $\theta_B=\arcsin\frac{1.391}{N\pi}$

When $\theta_B\in[0,\arcsin\frac{1.391}{N\pi})$, $\frac{\partial}{\partial\theta_B}\Delta\theta_{HP}<0$;
when $\theta_B\in(\arcsin\frac{1.391}{N\pi},\frac{\pi}{2}]$, $\frac{\partial}{\partial\theta_B}\Delta\theta_{HP}>0$.
The proof is in the following.

Proof:
Because the square-root expression in the denominator is positive, it can be derived that
\begin{align} 
	&\frac{\partial}{\partial\theta_B}\Delta\theta_{HP}\lessgtr 0 \nonumber \\
	\Leftrightarrow & \sqrt{1-(\sin\theta_B-\frac{2.782}{N\pi})^2}\lessgtr \cos\theta_B \nonumber \\
	\Leftrightarrow & 1-(\sin\theta_B-\frac{2.782}{N\pi})^2 \lessgtr \cos^2\theta_B  \nonumber \\
	\Leftrightarrow & 1-\cos^2\theta_B -(\sin\theta_B-\frac{2.782}{N\pi})^2\lessgtr 0 \nonumber \\
	\Leftrightarrow & \sin^2\theta_B -(\sin\theta_B-\frac{2.782}{N\pi})^2\lessgtr 0 \nonumber \\
	\Leftrightarrow & (\sin\theta_B+\sin\theta_B-\frac{2.782}{N\pi})(\sin\theta_B-\sin\theta_B+\frac{2.782}{N\pi}) \lessgtr 0 \nonumber \\
	\Leftrightarrow & (2\sin\theta_B-\frac{2.782}{N\pi})(\frac{2.782}{N\pi}) \lessgtr 0 \nonumber \\
	\Leftrightarrow & \sin\theta_B \lessgtr \frac{1.391}{N\pi}.
\end{align}
The proof is completed.

\section{Proof of Theorem\,\ref{th:chp4_A_0C}}
\label{appdx:bessel:ownbe}

A similar method to the ULA is adopted to solve $A_{0,C}$ for UCA.
For convenience, the subscript $_C$ is omitted.

First, the array factor $G(\theta,\theta_B)$ for UCA is given by,
\begin{align}
	G(\theta,\theta_B)=\frac{1}{\sqrt{N}}\sum_{i=1}^N e^{jkR[\cos(\theta_B-\psi_i)-\cos(\theta-\psi_i)]},
\end{align}
where $k=\frac{2\pi}{\lambda}$ and $\psi_i=2\pi(i-1)/N$.
Then, the following derives $G(\theta,\theta_B)\cdot G^*(\theta,\theta_B)$.
\begin{align}\label{eq:appdx_bessel_mmew}
	G(\theta,\theta_B)\cdot G^*(\theta,\theta_B)&=\frac{1}{N}\sum_{i=1}^N\sum_{j=1}^N e^{jkR[\cos(\theta_B-\psi_i)-\cos(\theta-\psi_i)-\cos(\theta_B-\psi_j)+\cos(\theta-\psi_j)]} \nonumber \\
	&= \frac{1}{N} \sum_{i,j} e^{jkR[\cos(\theta_B-\psi_i)-\cos(\theta_B-\psi_j)]}\cdot e^{-jkR[\cos(\theta-\psi_i)-\cos(\theta-\psi_j)]}.
\end{align}
So far, the exponential terms that contain $\theta$ are separated to solve the integral.
In order to use $J_n(x)=\frac{1}{2\pi}\int_{-\pi}^{\pi}e^{j(n\tau-x\sin\tau)}\mathrm{d}\tau$, the triangle identity $\cos a-\cos b=-2\sin(\frac{a+b}{2})\sin(\frac{a-b}{2})$ is used to simply $\cos(\theta-\psi_i)-\cos(\theta-\psi_j)$,
\begin{align}\label{eq:appdx_bessel_wies}
	&\cos(\theta-\psi_i)-\cos(\theta-\psi_j) \nonumber \\
	=&-2\sin(\frac{\theta-\psi_i+\theta-\psi_j}{2})\sin(\frac{\theta-\psi_i-\theta+\psi_j}{2}) \nonumber \\
	=&2\sin(\theta-\frac{\psi_i+\psi_j}{2})\sin(\frac{\psi_i-\psi_j}{2}) \nonumber \\
	=&2\sin(\theta-\frac{i+j-2}{N}\pi)\sin(\frac{i-j}{N}\pi).
\end{align}
Let $W_{i,j}=2\sin(\frac{i-j}{N}\pi)$ and $Z_{i,j}=\frac{i+j-2}{N}\pi$.
Substituting (\ref{eq:appdx_bessel_wies}) into (\ref{eq:appdx_bessel_mmew}), it can be derived that
\begin{align}
	G(\theta,\theta_B)\cdot G^*(\theta,\theta_B)=\frac{1}{N}\sum_{i,j} e^{jkRW_{i,j}\sin(\theta_B-Z_{i,j})}\cdot e^{-jkRW_{i,j}\sin(\theta-Z_{i,j})}.
\end{align}
Now, (\ref{eq:chp4_A0C}) can be written as
\begin{align}\label{eq:appdx_bessel_vbnewuioa}
	A_0&=\int_0^{2\pi} G^2(\theta,\theta_B)\,\mathrm{d}\theta \nonumber \\
	   &=\int_0^{2\pi} G(\theta,\theta_B)\cdot G^*(\theta,\theta_B) \,\mathrm{d}\theta \nonumber \\  
		 &=\frac{1}{N}\sum_{i,j} e^{jkRW_{i,j}\sin(\theta_B-Z_{i,j})}
	          \int_0^{2\pi} e^{-jkRW_{i,j}\sin(\theta-Z_{i,j})}\,\mathrm{d}\theta \nonumber \\
	&=\frac{2\pi}{N}\sum_{i,j} J_0(kRW_{i,j}) e^{jkRW_{i,j}\sin(\theta_B-Z_{i,j})}.
\end{align}
Next $A_0$ will be simplified. 
Let $A_{0,i,j}$ denote each summation term in (\ref{eq:appdx_bessel_vbnewuioa}),
\begin{align}
	A_{0,i,j}=\frac{2\pi}{N} J_0(kRW_{i,j}) e^{jkRW_{i,j}\sin(\theta_B-Z_{i,j})}.
\end{align}
It is obvious that $W_{i,j}=-W_{j,i}$ and $Z_{i,j}=Z_{j,i}$.
Consider that $J_n(-x)=(-1)^nJ_n(x)$ and $J_0(x)$ is a real number, it can be deduced that $A_{0,j,i}=A^*_{0,i,j}$.
The proof is in the following.
\begin{align}
	A_{0,j,i}&=\frac{2\pi}{N} J_0(kRW_{j,i}) e^{jkRW_{j,i}\sin(\theta_B-Z_{j,i})} \nonumber \\
	&=\frac{2\pi}{N} J_0(-kRW_{i,j}) e^{-jkRW_{i,j}\sin(\theta_B-Z_{i,j})} \nonumber \\
	&=\frac{2\pi}{N} J_0(kRW_{i,j}) [e^{jkRW_{i,j}\sin(\theta_B-Z_{i,j})}]^* \nonumber \\
	&=A^*_{0,i,j}.
\end{align}
It is not hard to see that combining $A_{0,i,j}$ and $A_{0,j,i}$ as in Appendix\,\ref{appdx:bessel:veaien} will not give a simpler solution. 
Thus, other method should be used to simply (\ref{eq:appdx_bessel_vbnewuioa}).

From the expression of $W_{i,j}$ and $Z_{i,j}$, it is noticed that $W_{i,j+N}=-W_{i,j}$.
The proof is in the following.
\begin{align}
	W_{i,j+N}&=2\sin(\frac{i-j-N}{N}\pi) \nonumber \\
	&=2\sin(\frac{i-j}{N}\pi-\pi) \nonumber \\
	&=-2\sin(\frac{i-j}{N}\pi)  \nonumber \\
	&=-W_{i,j}.
\end{align}
Similarly, $\sin(\theta_B-Z_{i,j+N})=-\sin(\theta_B-Z_{i,j})$.
Thus, $A_{0,i,j+N}=A_{0,i,j}$.
The proof is in the following.
\begin{align}\label{eq:appdx_bessel_ompuk}
	A_{0,i,j+N}&= \frac{2\pi}{N} J_0(kRW_{i,j+N}) e^{jkRW_{i,j+N}\sin(\theta_B-Z_{i,j+N})}\nonumber \\
	&=\frac{2\pi}{N} J_0(-kRW_{i,j}) e^{jkR\cdot -1\cdot W_{i,j}\cdot -1 \cdot \sin(\theta_B-Z_{i,j})}\nonumber \\
	&=\frac{2\pi}{N} J_0(kRW_{i,j}) e^{jkRW_{i,j}\sin(\theta_B-Z_{i,j})}\nonumber \\
	&=A_{0,i,j}.
\end{align}
According to (\ref{eq:appdx_bessel_ompuk}), the summation of $A_{0,i,j}$ can be formulated in another way. 
To better illustrate the new summation, an extended table is created, as shown in Fig.\,\ref{fig:appdx_bessel_table_UCA}.

\begin{figure}
\centering
\includegraphics[scale=1,angle=90]{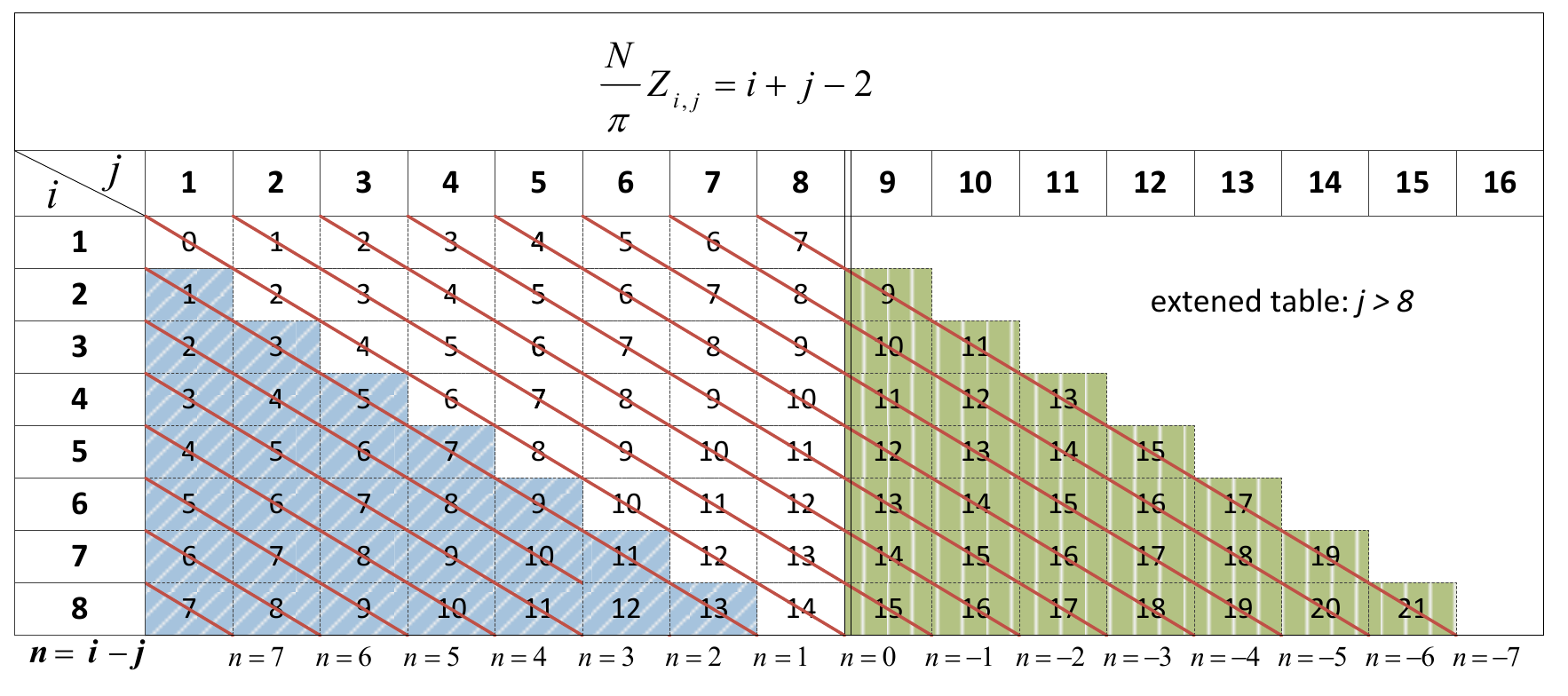}
\caption{Table for $Z_{i,j}$}
\label{fig:appdx_bessel_table_UCA}
\end{figure}

For convenience, let $n=i-j$.
Then, $W_n=W_{i,j}=2\sin(\frac{n}{N}\pi)$.
The terms $A_{0,i,j}$ on the red diagonal lines in the table have the same $W_n$.
In the table, $\frac{N}{\pi}Z_{i,j}$ is allocated according to their indice $i$ and $j$. 
Given $n=i-j$, it can be derived that
\begin{align}
	Z_{n,i}=Z_{i,j}=\frac{i+j-2}{N}\pi=\frac{2i-n-2}{N}\pi.
\end{align}
Thus, it can be derived that
\begin{align}
	A_{0,n,i}=A_{0,i,j}=\frac{2\pi}{N} J_0(kRW_n) e^{jkRW_n\sin(\theta_B-Z_{n,i})}.
\end{align}

$A_0$ is the summation of all elements in the original table (i.e., $i,j=1,...,8$).
Because $A_{0,i,j+N}=A_{0,i,j}$, the calculation of $A_0$ can be executed by replacing the lower triangle in the original table (i.e., $i>j$) with the lower triangle in the extended table (i.e., $i>j-N$).
In the new formation of $A_0$, which is a parallelogram table, the summation can be carried out along the diagonal lines from $n=0$ to $n=-(N-1)$.
For any $n$, the summation of $A_{0,n,i}$ includes $N$ terms with $Z_{n,i}$. 
Thus, (\ref{eq:appdx_bessel_vbnewuioa}) can be converted into
\begin{align}\label{eq:appdx_bessel_reqtq}
	A_0&=\sum_{i=1}^N\sum_{j=1}^N A_{0,i,j} \nonumber \\
	&=\sum_{n=0}^{-(N-1)} \sum_{i=1}^N A_{0,n,i} \nonumber \\
	&=\sum_{n=0}^{-(N-1)} \sum_{i=1}^N \frac{2\pi}{N} J_0(kRW_n) e^{jkRW_n\sin(\theta_B-Z_{n,i})} \nonumber \\
	&=\frac{2\pi}{N}\sum_{n=0}^{-(N-1)} J_0(kRW_n) \sum_{i=1}^N e^{jkRW_n\sin(\theta_B-Z_{n,i})}.
\end{align}
According to Jacobi-Anger expansion, $e^{j\alpha\sin\gamma}=\sum_{m=-\infty}^{\infty}J_m(\alpha)e^{jm\gamma}$, (\ref{eq:appdx_bessel_reqtq}) can be further derived by
\begin{align}
	A_0&=\frac{2\pi}{N}\sum_{n=0}^{-(N-1)}J_0(kRW_n)\sum_{i=1}^N \sum_{m=-\infty}^{\infty} J_m(kRW_n)e^{jm(\theta_B-Z_{n,i})} \nonumber \\
	&= \frac{2\pi}{N}\sum_{n=0}^{-(N-1)}J_0(kRW_n) \sum_{m=-\infty}^{\infty} J_m(kRW_n) e^{jm\theta_B} \sum_{i=1}^N e^{-jmZ_{n,i}} \nonumber \\
	&= \frac{2\pi}{N}\sum_{n=0}^{-(N-1)}J_0(kRW_n) \sum_{m=-\infty}^{\infty} J_m(kRW_n) e^{jm\theta_B} e^{j\pi\frac{m}{N}(n+2)}\sum_{i=1}^N e^{-j2\pi\frac{m}{N}i}.
\end{align}
When $m=lN$, $l\in\mathbb{Z}$, $e^{j\pi\frac{m}{N}(n+2)}=e^{jln\pi}e^{j2\pi l}=e^{jln\pi}$ and 
\begin{align}
	\sum_{i=1}^N e^{-j2\pi\frac{m}{N}i}=\sum_{i=1}^N e^{-j2\pi li}=N.
\end{align}
When $m\neq lN$, 
\begin{align}
	\sum_{i=1}^N e^{-j2\pi\frac{m}{N}i}=e^{-j2\pi\frac{m}{N}}\frac{1-e^{-j2\pi\frac{m}{N}N}}{1-e^{-j2\pi\frac{m}{N}}}=0.
\end{align}
Thus, it can be derived that
\begin{align}
	A_0&=\frac{2\pi}{N}\sum_{n=0}^{-(N-1)}J_0(kRW_n) \sum_{l=-\infty}^{\infty} J_{lN}(kRW_n) e^{jlN\theta_B} e^{jln\pi} N \nonumber \\
	&=2\pi\sum_{n=0}^{-(N-1)}J_0(kRW_n) \sum_{l=-\infty}^{\infty} J_{lN}(kRW_n) e^{jlN\theta_B} (-1)^{ln} \nonumber \\
	&=2\pi\sum_{n=0}^{N-1}J_0(-kRW_n) \sum_{l=-\infty}^{\infty}(-1)^{-ln} J_{lN}(-kRW_n) e^{jlN\theta_B} \nonumber \\
	&=2\pi\sum_{n=0}^{N-1}J_0(kRW_n) \sum_{l=-\infty}^{\infty}(-1)^{ln+lN} J_{lN}(kRW_n) e^{jlN\theta_B}. 
\end{align}
When $n=0$, $kRW_0=0$. 
Thus, $J_0(kRW_0)=1$ and $J_{lN}(kRW_0)=0$, $l\neq 0$.
So 
\begin{align}\label{eq:appdx_bessel_A0C}
	A_0=2\pi+2\pi\sum_{n=1}^{N-1}J_0(kRW_n) \sum_{l=-\infty}^{\infty}(-1)^{ln+lN} J_{lN}(kRW_n) e^{jlN\theta_B}. 
\end{align}

\section{Further Derivation for Pattern Area of UCA}
\label{appdx:bessel:bxv}

The expression of $A_{0,C}$ for UCA in (\ref{eq:appdx_bessel_A0C}) has exponential term $e^{jlN\theta_B}$.
Unlike the expression of $A_{0,L}$ in (\ref{eq:appdx_bessel_ieur}), it is hard to tell $A_{0,C}$ is real or complex.
Since $A_0$ is the pattern area, which should be a real number.
This section further derives the expression of $A_{0,C}$ to get a real-number expression.

The summation term in (\ref{eq:appdx_bessel_A0C}), denoted by $A_{0,C,n}$, can be written by
\begin{align}
	A_{0,C,n}&=2\pi J_0(kRW_n) \sum_{l=-\infty}^{\infty}(-1)^{ln+lN} J_{lN}(kRW_n) e^{jlN\theta_B},  
\end{align}
for $n=1,...,N-1$.
To simply $A_{0,C,n}$, the complex component $e^{jlN\theta_B}$ should be converted into some other form.
The simplest one to exploit is $e^{jx}+e^{-jx}=2\cos x$.
To this end, imagine to fold the $l$ axis at the middle point $l=0$.
One problem could be that for even or odd $N$, the term $(-1)^{lN}$ makes a big difference.
Therefore, to begin with, the simplification of $A_{0,C,n}$ for even and odd $N$ is discussed separately.

When $N$ is even, $(-1)^{lN}=1$ for any $l$.
Thus, it can be derived that
\begin{align}
	A_{0,C,n}=2\pi J_0(kRW_n) \sum_{l=-\infty}^{\infty}(-1)^{ln} J_{lN}(kRW_n) e^{jlN\theta_B}.
\end{align}
Then, combine the summation terms regarding to $+l$ and $-l$,
\begin{align}\label{eq:appdx_bessel_A0C_n_even}
	A_{0,C,n}&=2\pi J_0(kRW_n)[J_0(kRW_n)+\sum_{l=-\infty}^{-1}(-1)^{ln} J_{lN}(kRW_n) e^{jlN\theta_B} \nonumber \\
	&\qquad\qquad\qquad\qquad\qquad+\sum_{l=1}^{\infty}(-1)^{ln} J_{lN}(kRW_n) e^{jlN\theta_B}] \nonumber \\
  &=2\pi J_0(kRW_n)[J_0(kRW_n)+\sum_{l=1}^{\infty}(-1)^{-ln} J_{-lN}(kRW_n) e^{-jlN\theta_B} \nonumber \\
	&\qquad\qquad\qquad\qquad\qquad+\sum_{l=1}^{\infty}(-1)^{ln} J_{lN}(kRW_n) e^{jlN\theta_B}] \nonumber \\
	&=2\pi J_0(kRW_n)[J_0(kRW_n)+\sum_{l=1}^{\infty}(-1)^{ln} (-1)^{lN}J_{lN}(kRW_n) e^{-jlN\theta_B} \nonumber \\
	&\qquad\qquad\qquad\qquad\qquad+\sum_{l=1}^{\infty}(-1)^{ln} J_{lN}(kRW_n) e^{jlN\theta_B}] \nonumber \\
	&=2\pi J_0(kRW_n)[J_0(kRW_n)+\sum_{l=1}^{\infty}(-1)^{ln} J_{lN}(kRW_n) (e^{jlN\theta_B}+e^{-jlN\theta_B})] \nonumber \\
	&=2\pi J_0(kRW_n)[J_0(kRW_n)+2\sum_{l=1}^{\infty}(-1)^{ln} J_{lN}(kRW_n)\cos(lN\theta_B)] \nonumber \\
	&=2\pi J_0^2(kRW_n)+4\pi J_0(kRW_n)\sum_{l=1}^{\infty}(-1)^{ln} J_{lN}(kRW_n)\cos(lN\theta_B).
\end{align}
Thus, 
\begin{align}\label{eq:appdx_bessel_A0C_even}
	A_{0,C}=2\pi+2\pi \sum_{n=1}^{N-1} J_0^2(kRW_n)+4\pi\sum_{n=1}^{N-1} J_0(kRW_n)\sum_{l=1}^{\infty}(-1)^{ln} J_{lN}(kRW_n)\cos(lN\theta_B).
\end{align}

When $N$ is odd, the simplification of $A_{0,C,n}$ is more complex.
The property that $W_n=W_{N-n}$ is needed,
\begin{align}
	A_{0,C,N-n}&=2\pi J_0(kRW_{N-n}) \sum_{l=-\infty}^{\infty}(-1)^{l(N-n)+lN} J_{lN}(kRW_{N-n}) e^{jlN\theta_B} \nonumber \\
		         &=2\pi J_0(kRW_n) \sum_{l=-\infty}^{\infty}(-1)^{-ln+2lN} J_{lN}(kRW_n) e^{jlN\theta_B} \nonumber \\
						 &=2\pi J_0(kRW_n) \sum_{l=-\infty}^{\infty}(-1)^{ln} J_{lN}(kRW_n) e^{jlN\theta_B}.
\end{align}
Thus, it can be derived that
\begin{align}
	A_{0,C,n}+A_{0,C,N-n}=2\pi J_0(kRW_n) \sum_{l=-\infty}^{\infty}(-1)^{ln}[1+(-1)^{lN}] J_{lN}(kRW_n) e^{jlN\theta_B}.
\end{align}
Because $N$ is odd, when $l$ is odd, $1+(-1)^{lN}=0$; when $l$ is even, $1+(-1)^{lN}=2$. 
Then, it can be derived that
\begin{align}
	A_{0,C,n}+A_{0,C,N-n}&=4\pi J_0(kRW_n) \sum_{l=-\infty}^{\infty}(-1)^{2ln} J_{2lN}(kRW_n) e^{j2lN\theta_B} \nonumber \\
	                     &=4\pi J_0(kRW_n) \sum_{l=-\infty}^{\infty}J_{2lN}(kRW_n) e^{j2lN\theta_B} \nonumber \\
	                     &=4\pi J_0(kRW_n)[J_0(kRW_n)+\sum_{l=-\infty}^{-1} J_{2lN}(kRW_n) e^{j2lN\theta_B} \nonumber \\
											 &\qquad\qquad\qquad\qquad\qquad   +\sum_{l=1}^{\infty} J_{2lN}(kRW_n) e^{j2lN\theta_B}]  \nonumber \\
											 &=4\pi J_0(kRW_n)[J_0(kRW_n)+\sum_{l=1}^{\infty} J_{-2lN}(kRW_n) e^{-j2lN\theta_B} \nonumber \\
											 &\qquad\qquad\qquad\qquad\qquad   +\sum_{l=1}^{\infty} J_{2lN}(kRW_n) e^{j2lN\theta_B}]  \nonumber \\
											 &=4\pi J_0(kRW_n)[J_0(kRW_n)+2\sum_{l=1}^{\infty} J_{2lN}(kRW_n) \cos(2lN\theta_B)] \nonumber \\
											 &=4\pi J_0^2(kRW_n)+8\pi J_0(kRW_n)\sum_{l=1}^{\infty} J_{2lN}(kRW_n) \cos(2lN\theta_B).
\end{align}
Thus, 
\begin{align}\label{eq:appdx_bessel_A0C_odd}
	A_{0,C}&=2\pi+\sum_{n=1}^{N-1}A_{0,C,n}
	=2\pi+\sum_{n=1}^{\frac{N-1}{2}}A_{0,C,n}+\sum_{\frac{N+1}{2}}^{N-1}A_{0,C,n} \nonumber \\
	&=2\pi+4\pi\sum_{n=1}^{\frac{N-1}{2}} J_0^2(kRW_n)+8\pi\sum_{n=1}^{\frac{N-1}{2}} J_0(kRW_n)\sum_{l=1}^{\infty} J_{2lN}(kRW_n) \cos(2lN\theta_B).
\end{align}
To compare with $A_{0,C}$ in (\ref{eq:appdx_bessel_A0C_even}) when $N$ is even, (\ref{eq:appdx_bessel_A0C_odd}) can be also written by
\begin{align}\label{eq:appdx_bessel_A0C_odd2}
	A_{0,C}=2\pi+2\pi\sum_{n=1}^{N-1} J_0^2(kRW_n)+4\pi\sum_{n=1}^{N-1} J_0(kRW_n)\sum_{l=1}^{\infty} J_{2lN}(kRW_n) \cos(2lN\theta_B).
\end{align}

\section{Proof of Proposition\,\ref{th:chp4_HPBW_C}}
\label{appdx:bessel:yter}

In this section, the subscript $_C$ is omitted for convenience.
The array factor $G(\theta,\theta_B)$ for UCA is given,
\begin{align}
	G(\theta,\theta_B)=\frac{1}{\sqrt{N}}\sum_{i=1}^N e^{jkR[\cos(\theta_B-\psi_i)-\cos(\theta-\psi_i)]},
\end{align}
where $k=\frac{2\pi}{\lambda}$ and $\psi_i=2\pi(i-1)/N$.
Then, $G(\theta,\theta_B)$ is derived in form of summation of Bessel functions.

According to $\cos a-\cos b=-2\sin(\frac{a+b}{2})\sin(\frac{a-b}{2})$,
\begin{align}
	&\cos(\theta_B-\psi_i)-\cos(\theta-\psi_i) \nonumber \\
	=&-2\sin(\frac{\theta_B-\psi_i+\theta-\psi_i}{2})\sin(\frac{\theta_B-\psi_i-\theta+\psi_i}{2}) \nonumber \\
	=&2\sin(\frac{\theta_B+\theta}{2}-\psi_i)\sin(\frac{\theta-\theta_B}{2}). \nonumber \\
\end{align}
Then, according to the Jacobi-Anger expansion, $e^{j\alpha\sin\gamma}=\sum_{n=-\infty}^{\infty}J_n(\alpha)e^{jn\gamma}$, $G(\theta,\theta_B)$ can be written as
\begin{align}
	G(\theta,\theta_B)&=\frac{1}{\sqrt{N}}\sum_{i=1}^N e^{j2kR\sin(\frac{\theta_B+\theta}{2}-\psi_i)\sin(\frac{\theta-\theta_B}{2})} \nonumber \\
	&=\frac{1}{\sqrt{N}}\sum_{i=1}^N 
	\sum_{n=-\infty}^{\infty}J_n(2kR\sin\frac{\theta-\theta_B}{2})e^{jn(\frac{\theta_B+\theta}{2}-\psi_i)} \nonumber \\
	&=\frac{1}{\sqrt{N}} 
	\sum_{n=-\infty}^{\infty}J_n(2kR\sin\frac{\theta-\theta_B}{2}) e^{jn\frac{\theta_B+\theta}{2}} \sum_{i=1}^N e^{-jn\psi_i} \nonumber \\
	&=\frac{1}{\sqrt{N}} 
	\sum_{n=-\infty}^{\infty}J_n(2kR\sin\frac{\theta-\theta_B}{2}) e^{jn\frac{\theta_B+\theta}{2}} \sum_{i=0}^{N-1} e^{-j2\pi\frac{n}{N}i}.
\end{align}
When $n=lN$, $l\in\mathbb{Z}$, $\sum_{i=0}^{N-1} e^{-j2\pi\frac{n}{N}i}=N$.
When $n\neq lN$, $\sum_{i=0}^{N-1} e^{-j2\pi\frac{n}{N}i}=0$.
Thus, it can be derived that
\begin{align}\label{eq:appdx_bessel_AF_C}
	G(\theta,\theta_B)&=\sqrt{N}	\sum_{l=-\infty}^{\infty}J_{lN}(2kR\sin\frac{\theta-\theta_B}{2}) e^{jlN\frac{\theta_B+\theta}{2}}.
\end{align}
Although the above derivation to get (\ref{eq:appdx_bessel_AF_C}) is similar to \cite{balanis2005antenna}, (\ref{eq:appdx_bessel_AF_C}) is simpler than the expression given in \cite{balanis2005antenna}.

Next, a (loose) constraint will be derived for $N$ that satisfies that the mainbeam of the pattern is dominated by the $l=0$ term, i.e., 
\begin{align}
	G(\theta,\theta_B)\approx\sqrt{N}J_0(2kR\sin\frac{\theta-\theta_B}{2}), \qquad\theta_{FN,-}<\theta<\theta_{FN,+},
\end{align}
where $\theta_{FN,\pm}$ are the first nulls of the pattern.

Let $x=2kR\sin\frac{\theta-\theta_B}{2}$.
Then $G(\theta,\theta_B)$ is 
\begin{align}
	G(\theta,\theta_B)&=\sqrt{N}\sum_{l=-\infty}^{\infty}J_{lN}(x) e^{jlN\frac{\theta_B+\theta}{2}}.
\end{align}
Except $J_0(x)$, $J_N(x)$ is the dominant term among the rest terms.
It is known that the first zero of $J_0(x)$ is 2.4048.
Therefore, if $J_N(2.4048)$ is small enough, in the range $x\in[0,2.4048]$, then $G(\theta,\theta_B)\approx\sqrt{N}J_0(2kR\sin\frac{\theta-\theta_B}{2})$ is valid.

Examples of $J_N(x)$ is shown in Fig.\,\ref{fig:appdx_bessel_functions}.
It can be seen that $J_4(2.4048)$ is already much less than 1.
It can be deduced that for most $N\ge4$, $J_N(2.4048)$ is negligible.

\begin{figure}
\centering
\includegraphics[scale=0.9]{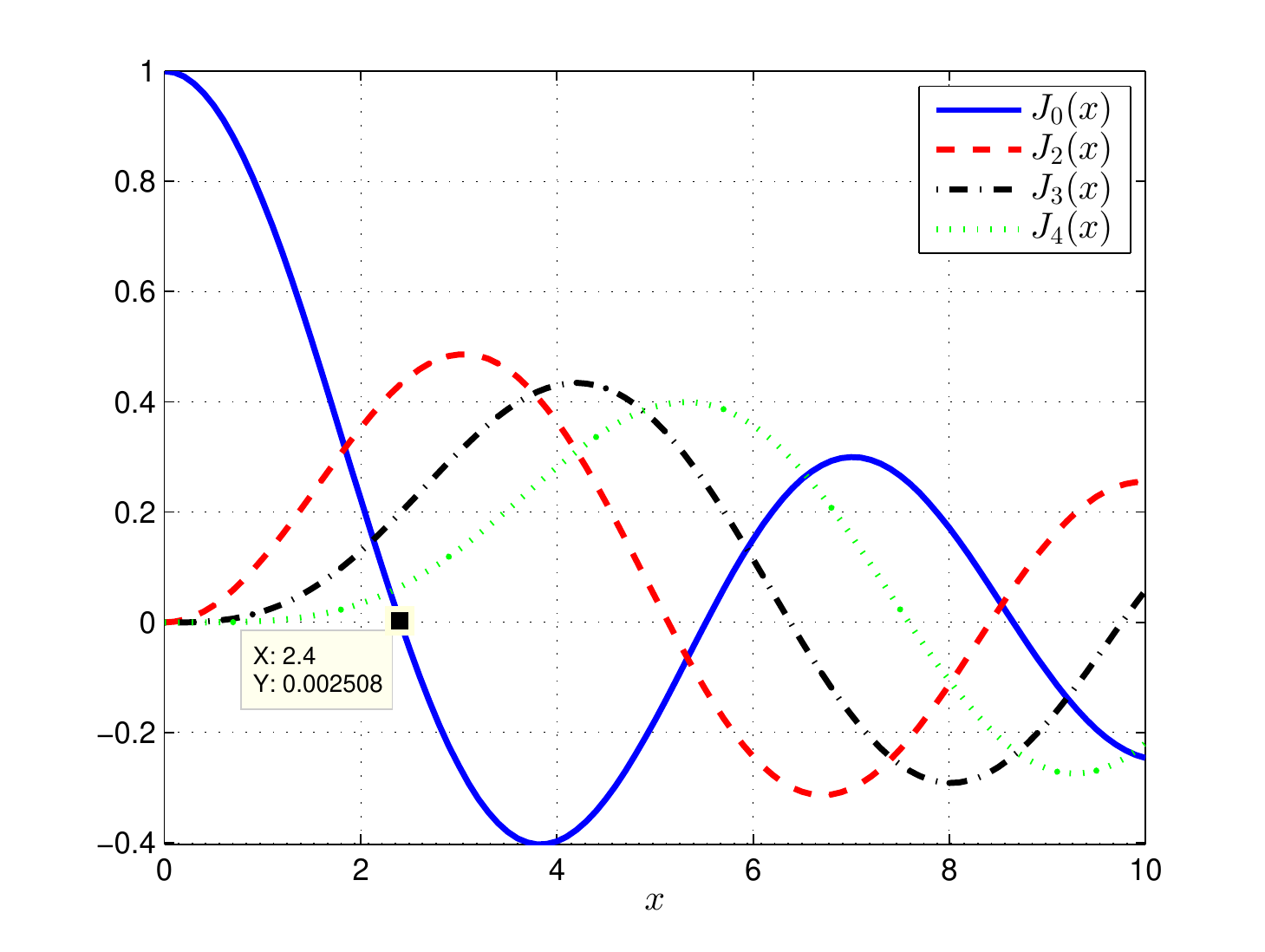}
\caption{$J_{N}(x)$, $N=0,2,4$.}
\label{fig:appdx_bessel_functions}
\end{figure}

When $N\ge4$, the first nulls $\theta_{FN,\pm}$ can be approximately calculated by
\begin{align}
	|2kR\sin\frac{\theta_{FN,\pm}-\theta_B}{2}|=2.4048.
\end{align}
Thus, 
\begin{align}
	\theta_{FN,\pm}=\theta_B\pm2\arcsin\frac{2.4048}{2kR}.
\end{align}

It is now easy to calculate the half-power points $\theta_{\pm}$, because half-power points are within the range $\theta\in[\theta_{FN,-},\theta_{FN,+}]$.
Thus, the values of $\theta_{\pm}$ can be calculated from 
\begin{align}
	J_0(|2kR\sin\frac{\theta_{\pm}-\theta_B}{2}|)=1/\sqrt{2}.
\end{align}
It is known that $J_0(1.1264)=1/\sqrt{2}$.
Therefore, $\theta_{\pm}=\theta_B\pm2\arcsin\frac{1.1264}{2kR}$ and 
\begin{align}\label{eq:appdx_bessel_irue}
	\Delta \theta_{HP,C}=\theta_{+}-\theta_{-}=4\arcsin\frac{1.1264}{2kR}.
\end{align}

\section{Proof of Lemma\,\ref{le:chp4_arcsin}}
\label{appdx:bessel:nvoeor}

First, obtain the inverse function of $y=\arcsin(x)$, i.e., $x=\sin(y)$.
It is well known that when $y\to 0$, $x=\sin(y)=y$.
Similarly, $\sin(y/2)=y/2=x/2$.
Thus, it can be derived that $x=2\sin(y/2)$ and $x=\sin(y)$.  
Now invert the previous equations, $y=2\arcsin(x/2)=\arcsin(x)$.
The comparison between these two functions are shown in Fig.\,\ref{fig:appdx_bessel_arcsin}.
It can be seen that for $|x|<0.5$, $2\arcsin(x/2)\approx\arcsin(x)$.
Thus, for $|x|\ll 1$, $2\arcsin(x/2)\approx\arcsin(x)$.

\begin{figure}
\centering
\includegraphics[scale=0.9]{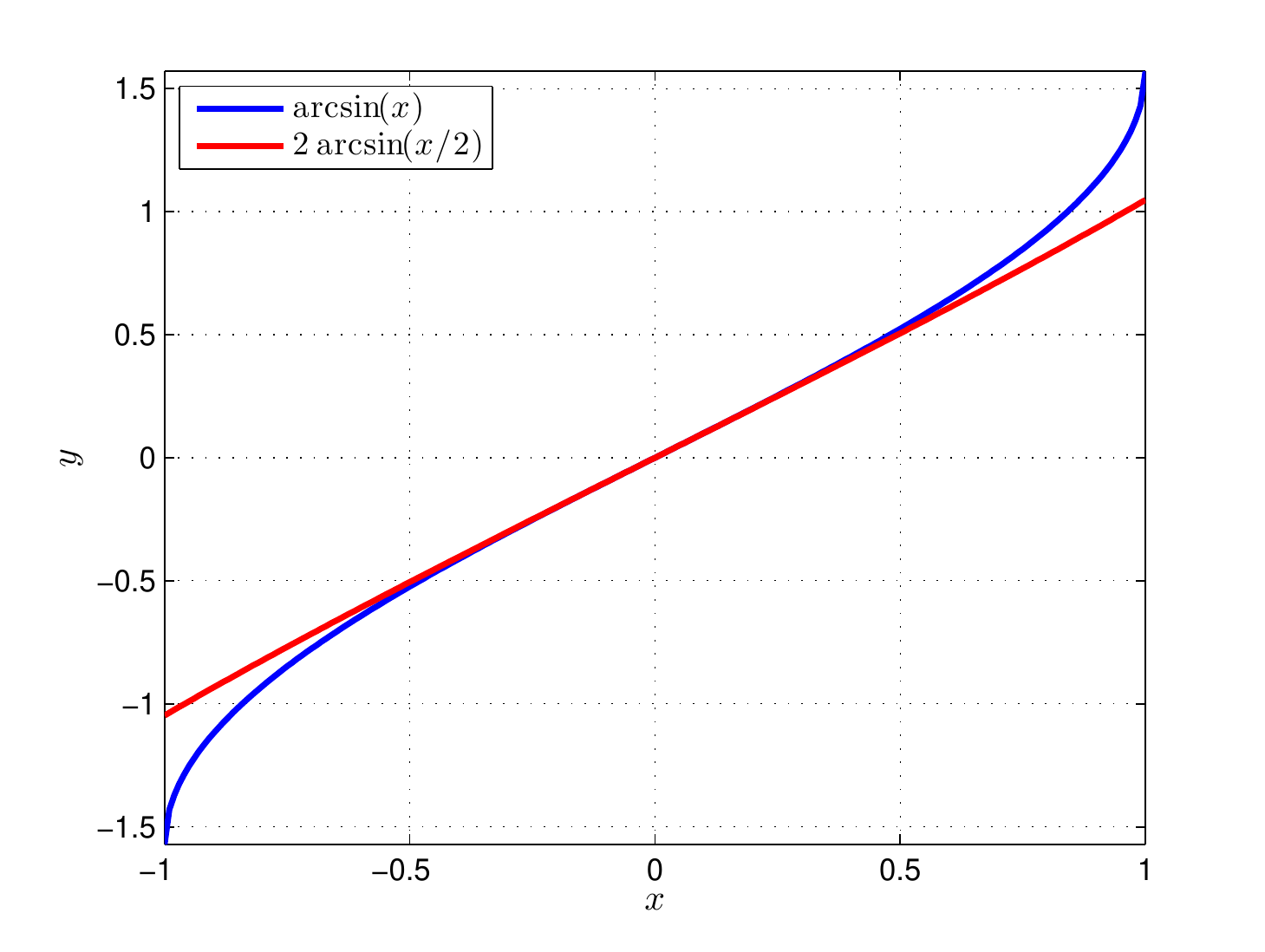}
\caption{$\arcsin(x)$ and $2\arcsin(x/2)$.}
\label{fig:appdx_bessel_arcsin}
\end{figure}

\section{One-Dimension MMSE Solution}
\label{appdx:bessel:mnkeorie}

This section solves the problem in (\ref{eq:chp5_mmse1}),
\begin{align}
	R_{opt}=\arg\min_R err(R),
\end{align}
where 
\begin{align}
	err(R)=\mathbb{E}_{\theta_B}[(\bar{p}_C-\bar{p}_{C,\text{min}})^2].
\end{align}
Because $\bar{p}_C-\bar{p}_{C,\text{min}}\geq 0$, a simplified, first-order error is used to replace $err(R)$, which is denoted by $err_1(R)$ and is given by
\begin{align}
	err_1(R)=\mathbb{E}_{\theta_B}[\bar{p}_C-\bar{p}_{C,\text{min}}].
\end{align}
Thus, the problem of finding the optimum $R$ can be written as 
\begin{align}\label{eq:appdx_bessel_mmse1}
	R_{opt}=\arg\min_R err_1(R).
\end{align}

Because $\theta_B\sim \mathcal{U}(0,2\pi)$, $err_1(R)$ can be calculated by
\begin{align}
	err_1(R)=\frac{1}{2\pi}\int_{0}^{2\pi}(\bar{p}_C-\bar{p}_{C,\text{min}})\,\mathrm{d}\theta_B.
\end{align}

To solve (\ref{eq:appdx_bessel_mmse1}), the zeros of the partial derivative of $err_1(R)$ with respect to $R$ are calculated,
\begin{align}
	&\frac{\partial}{\partial R}err_1(R)=0 \\
	\Rightarrow &\frac{\partial}{\partial R} \frac{1}{2\pi}\int_{0}^{2\pi}(\bar{p}_C-\bar{p}_{C,\text{min}})\,\mathrm{d}\theta_B=0 \\
	\Rightarrow & \frac{1}{2\pi}\int_{0}^{2\pi}(\frac{\partial}{\partial R}\bar{p}_C-\frac{\partial}{\partial R}\bar{p}_{C,\text{min}})\,\mathrm{d}\theta_B=0.
\end{align}
Because $\bar{p}_{C,\text{min}}$ is a fixed value for certain $\theta_B$ and only depends on $\theta_B$, the partial derivative $\frac{\partial}{\partial R}\bar{p}_{C,\text{min}}=0$.
Thus, it can be derived that
\begin{align}
	&\frac{\partial}{\partial R}err_1(R)=0 \\
	\Rightarrow  &\frac{1}{2\pi}\int_{0}^{2\pi}\frac{\partial}{\partial R}\bar{p}_C\,\mathrm{d}\theta_B=0 \\
  \Rightarrow  &\frac{\partial}{\partial R}\frac{1}{2\pi}\int_{0}^{2\pi}\bar{p}_C\,\mathrm{d}\theta_B=0 \\
	\Rightarrow  &\frac{\partial}{\partial R}\bar{\bar{p}}_C=0, 
\end{align}
where $\bar{\bar{p}}_C$ is the averaged SSOP over Bob's angle and is defined by
\begin{align}
	\bar{\bar{p}}_C=\frac{1}{2\pi}\int_{0}^{2\pi}\bar{p}_C\,\mathrm{d}\theta_B.
\end{align}
Thus, (\ref{eq:appdx_bessel_mmse1}) can be converted into
\begin{align}\label{eq:appdx_bessel_mmse2}
	R_{opt}=\arg\min_R \bar{\bar{p}}_C.
\end{align}

\section{Two-Dimension MMSE Solution}
\label{appdx:bessel:owgvow}

This section solves the problem in (\ref{eq:chp5_mmse3}),
\begin{align}\label{eq:appdx_bessel_mmse3}
	R_{opt}=\arg\min_R err_2(R),
\end{align}
where 
\begin{align}\label{eq:appdx_bessel_meanErr}
	err_2(R)=\frac{1}{S}\int_{0}^{2\pi}\int_{0}^{d_{max}}d_B(\bar{p}_C-\bar{p}_{C,\text{min}})^2\,\mathrm{d}d_B\,\mathrm{d}\theta_B.
\end{align}
Because $\bar{p}_C-\bar{p}_{C,\text{min}}\geq 0$, a simplified, first-order error is used to replace $err_2(R)$, which is denoted by $err_3(R)$ and is given by
\begin{align}
	err_3(R)=\frac{1}{S}\int_{0}^{2\pi}\int_{0}^{d_{max}}d_B(\bar{p}_C-\bar{p}_{C,\text{min}})\,\mathrm{d}d_B\,\mathrm{d}\theta_B.
\end{align}
Thus, the problem of finding the optimum $R$ can be written as 
\begin{align}\label{eq:appdx_bessel_mmse4}
	R_{opt}=\arg\min_R err_3(R).
\end{align}

In this case, $\bar{p}_{C,\text{min}}$ is the minimum $\bar{p}_C$ in zone $k$, $k=1,2,...$, at angle $\theta_B$.
For convenience of mathematical derivation, the example of the coverage zones shown in Fig.\,\ref{fig:chp5_N_zone} is used.

To solve (\ref{eq:appdx_bessel_mmse4}), the zeros of the partial derivative of $err_3(R)$ are calculated.
Using $\frac{\partial}{\partial R}\bar{p}_{C,\text{min}}=0$, it can be derived that
\begin{align}
	&\frac{\partial}{\partial R}err_3(R)=0 \\
	\Rightarrow  &\frac{\partial}{\partial R}\frac{1}{S}\int_{0}^{2\pi}\int_{0}^{d_{max}}d_B\bar{p}_C\,\mathrm{d}d_B\,\mathrm{d}\theta_B=0. 
\end{align}
For any $d_B$, $\bar{p}_C$ is the same for the same zone. 
Let $\bar{p}_C^{(k)}$ denote the value of $\bar{p}_C$ in zone $k$, $k=1,2,3$.
For certain $\theta_B$, $\bar{p}_C^{(k)}$ is a constant in zone $k$.
Thus, it can be derived that
\begin{align}
	            &\frac{\partial}{\partial R}\frac{1}{S}\int_{0}^{2\pi}\int_{0}^{d_{max}}d_B\bar{p}_C\,\mathrm{d}d_B\,\mathrm{d}\theta_B=0 \\
	\Rightarrow &\frac{\partial}{\partial R}\frac{1}{S}\int_{0}^{2\pi}\Big(     
					                                                                       \int_{0}^{d_{th,2}}d_B\bar{p}_C^{(1)}\,\mathrm{d}d_B+
																																								 \int_{d_{th,2}}^{d_{th,4}}d_B\bar{p}_C^{(2)}\,\mathrm{d}d_B+
																																								 \int_{d_{th,4}}^{d_{max}}d_B\bar{p}_C^{(3)}\,\mathrm{d}d_B
															                                                                \Big)\,\mathrm{d}\theta_B=0 \label{eq:appdx_bessel_ierueqqq} \\
  \Rightarrow &\frac{\partial}{\partial R}\frac{1}{S}\int_{0}^{2\pi}\Big(     
					                                                                       \bar{p}_C^{(1)}\int_{0}^{d_{th,2}}d_B\,\mathrm{d}d_B+
																																								 \bar{p}_C^{(2)}\int_{d_{th,2}}^{d_{th,4}}d_B\,\mathrm{d}d_B+
																																								 \bar{p}_C^{(3)}\int_{d_{th,4}}^{d_{max}}d_B\,\mathrm{d}d_B
															                                                                \Big)\,\mathrm{d}\theta_B=0 \\
	\Rightarrow &\frac{\partial}{\partial R}\frac{1}{S}\int_{0}^{2\pi}\Big(     
					                                                                       \bar{p}_C^{(1)}\frac{d^2_{th,2}}{2}+
																																								 \bar{p}_C^{(2)}\frac{d^4_{th,2}-d^2_{th,2}}{2}+
																																								 \bar{p}_C^{(3)}\frac{d^4_{max}-d^2_{th,2}}{2}
															                                                                \Big)\,\mathrm{d}\theta_B=0. \label{eq:appdx_bessel_yert}
\end{align}
In addition, let $\bar{\bar{p}}_C^{(k)}$ denote the average value of $\bar{p}_C^{(k)}$ over $\theta_B\in[0,2\pi]$,
\begin{align}
	\bar{\bar{p}}_C^{(k)}=\frac{1}{2\pi}\int_{0}^{2\pi}\bar{p}_C^{(k)}\,\mathrm{d}\theta_B.
\end{align}
$d_{th,2}$, $d_{th,4}$ and $d_{max}$ are regardless of $\theta_B$. According to (\ref{eq:chp5_zone1_area}) to (\ref{eq:chp5_zone_prob}), (\ref{eq:appdx_bessel_yert}) can be converted into
\begin{align}
&\frac{\partial}{\partial R}\Big(\frac{1}{S}2\pi\frac{d^2_{th,2}}{2} \bar{\bar{p}}_C^{(1)}+
					                                    2\pi\frac{d^4_{th,2}-d^2_{th,2}}{2}\bar{\bar{p}}_C^{(2)}+                                   
    																				  2\pi\frac{d^4_{max}-d^2_{th,2}}{2}\bar{\bar{p}}_C^{(3)}\Big)=0 \\
&\Rightarrow \frac{\partial}{\partial R}\Big(q^{(1)}\bar{\bar{p}}_C^{(1)}+q^{(2)}\bar{\bar{p}}_C^{(2)}+q^{(3)}\bar{\bar{p}}_C^{(3)}\Big)=0. \label{eq:appdx_bessel_iweu}
\end{align}
Thus, (\ref{eq:appdx_bessel_mmse4}) can be converted into
\begin{align}\label{eq:appdx_bessel_mmse5}
	R_{opt}=\arg\min_R \Big(q^{(1)}\bar{\bar{p}}_C^{(1)}+q^{(2)}\bar{\bar{p}}_C^{(2)}+q^{(3)}\bar{\bar{p}}_C^{(3)}\Big).
\end{align}

\section{Proof of Theorem\,\ref{th:chp5_p2bar_up}}
\label{appdx:bessel:oieor}

There are two methods to obtain the upper bound $\bar{\bar{p}}_{up,C}$.
Both methods exploit Jensen's inequalities in (\ref{eq:chp3_JI_1}) and (\ref{eq:chp3_JI_2}).
In this section, the subscript $_C$ is omitted for convenience.


For the convenience of analyzing the tightness of the upper bound, $\bar{\bar{p}}_{up}$ can be derived based on $\bar{p}\leq \bar{p}_{up}$.
Using (\ref{eq:chp3_A_0}), it can be derived that
\begin{align}
	\bar{\bar{p}}&=\mathbb{E}_{\theta_B}[\bar{p}] \leq \mathbb{E}_{\theta_B}[\bar{p}_{up}] \nonumber \\
	&= 1-\mathbb{E}_{\theta_B}\Big[\text{exp}\Big\{-\lambda_e\pi\Big[\frac{c_0K}{2\pi(K+1)}A_0+\frac{c_0}{K+1}\Big]^{\frac{2}{\beta}}\Big\}\Big].
\end{align}
Using (\ref{eq:chp3_JI_1}) and (\ref{eq:chp3_JI_2}), it can be derived that
\begin{align}
	&1-\mathbb{E}_{\theta_B}\Big[\text{exp}\Big\{-\lambda_e\pi\Big[\frac{c_0K}{2\pi(K+1)}A_0+\frac{c_0}{K+1}\Big]^{\frac{2}{\beta}}\Big\}\Big] \label{eq:appdx_bessel_p2bar_derive1}\\
	<& 1-\text{exp}\Big\{-\lambda_e\pi\mathbb{E}_{\theta_B}\Big[\Big[\frac{c_0K}{2\pi(K+1)}A_0+\frac{c_0}{K+1}\Big]^{\frac{2}{\beta}}\Big]\Big\} \label{eq:appdx_bessel_p2bar_derive2} \\	
	\leq & 1-\text{exp}\Big\{-\lambda_e\pi\Big[\frac{c_0K}{2\pi(K+1)}\mathbb{E}_{\theta_B}[A_0]+\frac{c_0}{K+1}\Big]^{\frac{2}{\beta}}\Big\}.	\label{eq:appdx_bessel_p2bar_derive3}
\end{align}
The equality in (\ref{eq:appdx_bessel_p2bar_derive2}) does not hold because $\theta_B$ is random in this case.
Then, $\bar{\bar{p}}_{up}$ can be obtained by
\begin{align}\label{eq:appdx_bessel_p2bar_up}
	\bar{\bar{p}}_{up} = 1-\text{exp}\Big\{-\lambda_e \pi \Big[ \frac{c_0K\bar{A}_0}{2\pi(K+1)}+\frac{c_0}{K+1}  \Big]^{\frac{2}{\beta}}  \Big\},
\end{align}
where $\bar{A}_0$ is the expectation of $A_0$ over $\theta_B$ and is given by
\begin{align}\label{eq:appdx_bessel_meanA_0}
	\bar{A}_0=\mathbb{E}_{\theta_B}[A_0]=\frac{1}{2\pi}\int_0^{2\pi} A_0 \,\mathrm{d}\theta_B.
\end{align}
In this method, $A_0$ is treated throughout the derivation.
Based on the conclusion in Appendix\,\ref{appdx:bessel:uirtnvb}, the less random $A_0$ is, the tighter the bound is.
Because for the UCA, $A_0$ is much more constant with respect to $\theta_B$ than that of the ULA, as discussed in Section\,\ref{chp4:sec3:opbwec541}, the upper bound is tight.

The second method is shown in the following.
Substituting (\ref{eq:chp3_meanSSOP_Ri_0}) into (\ref{eq:chp5_p2bar_1}), $\bar{\bar{p}}$ is written by
\begin{align}\label{eq:appdx_bessel_p2bar}
	\bar{\bar{p}}=\mathbb{E}_{\theta_B}[1-\mathbb{E}_{|\tilde{h}|}[e^{-\lambda_eA}]]=1-\mathbb{E}_{\theta_B}[\mathbb{E}_{|\tilde{h}|}[e^{-\lambda_eA}]].
\end{align}
In fact, due to the similarity between the expressions of $\bar{\bar{p}}$ in (\ref{eq:appdx_bessel_p2bar}) and $\bar{p}$ in (\ref{eq:chp3_meanSSOP_Ri_0}), the same method in Section\,\ref{chp3:metric:bounds} can be reused here to obtain the upper bound for $\bar{\bar{p}}$.

Substituting (\ref{eq:chp3_meanSSOP_up_inequality}) into (\ref{eq:appdx_bessel_p2bar}), it can be derived that
\begin{align}\label{eq:appdx_bessel_p2bar_ineq1}
	\bar{\bar{p}}&=\mathbb{E}_{\theta_B}[1-\mathbb{E}_{|\tilde{h}|}[e^{-\lambda_eA}]] \nonumber \\
	&\leq  \mathbb{E}_{\theta_B}[1-e^{-\lambda_e\mathbb{E}_{|\tilde{h}|}[A]}] \nonumber \\
	&=1-\mathbb{E}_{\theta_B}[e^{-\lambda_e\mathbb{E}_{|\tilde{h}|}[A]}]. 
\end{align}
The equality holds only for the deterministic channel.
Using (\ref{eq:chp3_JI_1}), it can be further derived as
\begin{align}\label{eq:appdx_bessel_p2bar_ineq2}
 1-\mathbb{E}_{\theta_B}[e^{-\lambda_e\mathbb{E}_{|\tilde{h}|}[A]}]
< 1-\text{exp}\{-\lambda_e\mathbb{E}_{\theta_B}[\mathbb{E}_{|\tilde{h}|}[A]]\}.
\end{align}
The equality does not hold because $\theta_B$ is regarded random in this case.
Substituting $\mathbb{E}_{|\tilde{h}|}[|\tilde{h}|^2]=\frac{KG^2(\theta,\theta_B)+1}{K+1}$ into (\ref{eq:chp3_meanA_up}), it can be derived that
\begin{align}\label{eq:appdx_bessel_meanA}
	\mathbb{E}_{|\tilde{h}|}[A] \leq \pi\Big[
	\frac{c_0K}{2\pi(K+1)}\int_0^{2\pi}G^2(\theta,\theta_B)\,\mathrm{d}\theta
	+\frac{c_0}{K+1} \Big]^{\frac{2}{\beta}}.
\end{align}
When $\beta=2$, the equality holds.
For convenience, $A_0$ is used to replace $\int_0^{2\pi}G^2(\theta,\theta_B)\,\mathrm{d}\theta$.
Then, it can be derived that
\begin{align}\label{eq:appdx_bessel_p2bar_ineq3}
	\mathbb{E}_{\theta_B}[\mathbb{E}_{|\tilde{h}|}[A]] 
	\leq \pi\mathbb{E}_{\theta_B}\Big[	\Big[ \frac{c_0KA_0}{2\pi(K+1)}+\frac{c_0}{K+1} \Big]^{\frac{2}{\beta}} \Big]. 
\end{align}
Using (\ref{eq:chp3_JI_2}), it can be derived that
\begin{align}\label{eq:appdx_bessel_p2bar_ineq4}
	\pi\mathbb{E}_{\theta_B}\Big[	\Big[ \frac{c_0KA_0}{2\pi(K+1)}+\frac{c_0}{K+1} \Big]^{\frac{2}{\beta}} \Big] 
	\leq \pi \Big[\mathbb{E}_{\theta_B}	\Big[ \frac{c_0KA_0}{2\pi(K+1)}+\frac{c_0}{K+1}  \Big] \Big]^{\frac{2}{\beta}}.
\end{align}
When $\beta=2$, the equality holds.
Combining (\ref{eq:appdx_bessel_p2bar_ineq1}), (\ref{eq:appdx_bessel_p2bar_ineq2}), (\ref{eq:appdx_bessel_p2bar_ineq3}) and (\ref{eq:appdx_bessel_p2bar_ineq4}), it can be derived that
\begin{align}\label{eq:appdx_bessel_p2bar_ineq5}
	\bar{\bar{p}}	&\leq 1-\mathbb{E}_{\theta_B}[e^{-\lambda_e\mathbb{E}_{|\tilde{h}|}[A]}] \nonumber \\
  &< 1-\text{exp}\{-\lambda_e\mathbb{E}_{\theta_B}[\mathbb{E}_{|\tilde{h}|}[A]]\} \nonumber \\
	&\leq 1-\text{exp}\Big\{ -\lambda_e \pi\mathbb{E}_{\theta_B}\Big[	\Big[ \frac{c_0KA_0}{2\pi(K+1)}+\frac{c_0}{K+1} \Big]^{\frac{2}{\beta}} \Big] \Big\} \nonumber \\
	&\leq 1-\text{exp}\Big\{-\lambda_e \pi \Big[\mathbb{E}_{\theta_B}	\Big[ \frac{c_0KA_0}{2\pi(K+1)}+\frac{c_0}{K+1}  \Big] \Big]^{\frac{2}{\beta}}  \Big\}.
\end{align}

Then, (\ref{eq:appdx_bessel_p2bar_ineq5}) can be written as
\begin{align}\label{eq:appdx_bessel_p2bar_ineq6}
	\bar{\bar{p}}	< 1-\text{exp}\Big\{-\lambda_e \pi \Big[ \frac{c_0K\bar{A}_0}{2\pi(K+1)}+\frac{c_0}{K+1}  \Big]^{\frac{2}{\beta}}  \Big\}.
\end{align}
Then, the upper bound $\bar{\bar{p}}_{up}$ can be written as
\begin{align}\label{eq:appdx_bessel_p2bar_up_3}
	\bar{\bar{p}}_{up} = 1-\text{exp}\Big\{-\lambda_e \pi \Big[ \frac{c_0K\bar{A}_0}{2\pi(K+1)}+\frac{c_0}{K+1}  \Big]^{\frac{2}{\beta}}  \Big\}.
\end{align}
It can be seen that for deterministic channel, the upper bound $\bar{\bar{p}}_{up}$ is tighter when $\beta=2$ than that when $\beta>2$.

\section{Proof of Proposition\,\ref{prop:chp5_tightness}}
\label{appdx:bessel:uirtnvb}

First, the tightness of (\ref{eq:chp3_JI_1}) is examined.
Let $X$ be a random variable that is $X\sim \mathcal{U}(a,b)$.
Thus, it can be derived that 
\begin{align}
  &\mathbb{E}[e^X]=\frac{1}{b-a}\int_a^b e^x\,\mathrm{d}x=\frac{e^b-e^a}{b-a} \\
	&e^{\mathbb{E}[X]}=e^{\frac{a+b}{2}}.
\end{align}
Let the ratio $\eta_1$ define by the tightness of this inequality.
\begin{align}
	\eta_1=\frac{\mathbb{E}[e^X]}{e^{\mathbb{E}[X]}}=\frac{e^b-e^a}{b-a}e^{-\frac{a+b}{2}}=\frac{e^{\frac{b-a}{2}}-e^{-\frac{b-a}{2}}}{b-a}.
\end{align}
For convenience, let $y=b-a>0$.
\begin{align}
	\eta_1=\frac{e^{\frac{y}{2}}-e^{-\frac{y}{2}}}{y}.
\end{align}
According to the quotient rule of the derivative, the derivative of $\eta_1$ with respect to $y$ is 
\begin{align}\label{eq:appdx_bessel_eta1}
	\frac{\mathrm{d}\eta_1}{\mathrm{d}y}=\frac{(-1+\frac{y}{2})e^{\frac{y}{2}}-(-1-\frac{y}{2})e^{-\frac{y}{2}}}{y^2}.
\end{align}
Let $g(y)=(-1+\frac{y}{2})e^{\frac{y}{2}}$.
It is obvious that $g(y)$ is a monotonic increasing function for $y\in\mathbb{R}$.
Thus, for $y>0$, $g(y)-g(-y)>0$.
As a results, $\frac{\mathrm{d}\eta_1}{\mathrm{d}y}>0$.

It can be concluded that the larger $y=b-a$ is, the larger $\eta_1$ is.
In other words, the larger the range of $X$ is, the less tighter the inequality is.

Next, the tightness of (\ref{eq:chp3_JI_2}) is examined.
For convenience, let $\kappa=\frac{2}{\beta}$.
For $\beta\geq 2$, $0<\kappa\leq 1$.
\begin{align}
	&(\mathbb{E}[X])^{\kappa} = (\frac{a+b}{2})^{\kappa}  \\
	&\mathbb{E}[X^{\kappa}]=\frac{1}{b-a}\int_a^b x^{\kappa}\,\mathrm{d}x=\frac{1}{\kappa+1}\frac{b^{\kappa+1}-a^{\kappa+1}}{b-a}.
\end{align}
Let the ratio $\eta_2$ define by the tightness of this inequality,
\begin{align}
	\eta_2=\frac{(\mathbb{E}[X])^{\kappa}}{\mathbb{E}[X^{\kappa}]}=\frac{\kappa+1}{2^{\kappa}}\frac{(b-a)(b+a)^{\kappa}}{b^{\kappa+1}-a^{\kappa+1}}.
\end{align}
For convenience, let $y=b-a>0$, then $b=y+a$ and $a+b=y+2a$.
\begin{align}
	\eta_2=\frac{\kappa+1}{2^{\kappa}}
	       \frac{y(y+2a)^{\kappa}}{(y+a)^{\kappa+1}-a^{\kappa+1}}.
\end{align}
Then, the partial derivative of $\eta_2$ with respect to $y$ is examined.
According to the product rule of the derivative, it can be derived that
\begin{align}
	\frac{\partial}{\partial y}[y(y+2a)^{\kappa}]&=(y+2a)^{\kappa-1}[(\kappa+1)y+2a].
\end{align}
Then, it can be derived that
\begin{align}
	\frac{\partial}{\partial y}[(y+a)^{\kappa+1}-a^{\kappa+1}]=(\kappa+1)(y+a)^{\kappa}.
\end{align}
According to the quotient rule of the derivative,
\begin{align}
	\frac{\partial\eta_2}{\partial y} &=\frac{\kappa+1}{2^{\kappa}}\frac{1}{[(y+a)^{\kappa+1}-a^{\kappa+1}]^2} \nonumber \\
	&\{
	(y+2a)^{\kappa-1}[(\kappa+1)y+2a][(y+a)^{\kappa+1}-a^{\kappa+1}]
	-y(\kappa+1)(y+a)^{\kappa}(y+2a)^{\kappa}
	\} \nonumber \\
	&=\frac{\kappa+1}{2^{\kappa}}\frac{a(y+2a)^{\kappa-1}}{[(y+a)^{\kappa+1}-a^{\kappa+1}]^2} \nonumber \\
	&\{
	y[(1-\kappa)(y+a)^{\kappa}-(\kappa+1)a^{\kappa}]+2a[(y+a)^{\kappa}-a^{\kappa}]
	\}\nonumber \\
	&=\frac{\kappa+1}{2^{\kappa}}\frac{a(y+2a)^{\kappa-1}}{[(y+a)^{\kappa+1}-a^{\kappa+1}]^2}\{	g(y)+2a[(y+a)^{\kappa}-a^{\kappa}]\},
\end{align}
where $g(y)$ is
\begin{align}
	g(y)=y[(1-\kappa)(y+a)^{\kappa}-(\kappa+1)a^{\kappa}].
\end{align}
Because $y$ and $(y+a)^{\kappa}$ are both monotonic increasing functions for $y\in\mathbb{R}$ and $0<\kappa\leq 1$, $g(y)$ is also a monotonic increasing functions.
Because $g(0)=0$, $g(y)>0$, for $y>0$.
For $y>0$, $(y+a)^{\kappa}-a^{\kappa}>0$.
Therefore, for $y>0$, $\frac{\partial\eta_2}{\partial y}>0$.

It can be concluded that $\eta_2$ is a monotonic increasing function with respect to $y$. The larger $y=b-a$ is, the larger $\eta_2$ is.
In other words, the larger the range of $X$ is, the less tighter the inequality is.

\section{Proof of Theorem\,\ref{th:chp5_meanA0}}
\label{appdx:bessel:twqwrq}

In this section, the aim is to obtain the analytic expression for the following expression,
\begin{align}
	\bar{A}_{0,C}=\frac{1}{2\pi}\int_0^{2\pi}A_{0,C}\,\mathrm{d}\theta_B.
\end{align}

There are two different methods that lead to the same result.
First, let's look at the first method.
According to (\ref{eq:appdx_bessel_vbnewuioa}), $A_{0,C}$ is
\begin{align}
	A_{0,C}=\frac{2\pi}{N}\sum_{i,j} J_0(kRW_{i,j}) e^{jkRW_{i,j}\sin(\theta_B-Z_{i,j})},
\end{align}
where $k=2\pi/\lambda$ and $W_{i,j}=2\sin(\frac{i-j}{N}\pi)$, $i,j=1,...,N$.
Applying the integral representation of Bessel function of the first kind, $J_n(x)=\frac{1}{2\pi}\int_{-\pi}^{\pi}e^{j(n\tau-x\sin\tau)}\mathrm{d}\tau$ and $J_n(-x)=(-1)^nJ_n(x)$, it can be derived that
\begin{align}\label{eq:appdx_bessel_averageAC}
	\bar{A}_{0,C}&=\frac{1}{2\pi}\int_0^{2\pi}A_{0,C}\,\mathrm{d}\theta_B \nonumber \\
	&=\frac{1}{2\pi}\int_0^{2\pi}\frac{2\pi}{N}\sum_{i,j} J_0(kRW_{i,j}) e^{jkRW_{i,j}\sin(\theta_B-Z_{i,j})}\,\mathrm{d}\theta_B \nonumber \\
	&=\frac{1}{N}\sum_{i,j} J_0(kRW_{i,j}) \int_0^{2\pi}e^{jkRW_{i,j
	}\sin(\theta_B-Z_{i,j})}\,\mathrm{d}\theta_B \nonumber \\
	&=\frac{2\pi}{N}\sum_{i,j} J_0(kRW_{i,j}) J_0(-kRW_{i,j}) \nonumber \\
	&=\frac{2\pi}{N}\sum_{i,j} J_0^2(kRW_{i,j}).
\end{align}
For (\ref{eq:appdx_bessel_averageAC}), let $n=i-j$. 
The same method of deriving (\ref{eq:appdx_bessel_ieur}) from (\ref{eq:appdx_bessel_euirue}) is adopted, and the simpler expressions for $\bar{A}_{0,C}$ is obtained,
\begin{align}\label{eq:appdx_bessel_A0C_ave}
	\bar{A}_{0,C}=2\pi+\frac{4\pi}{N}\sum_{n=1}^{N-1}(N-n)J_0^2(2kR\sin(\frac{n}{N}\pi)).
\end{align}


For easy of analysis, $\bar{A}_{0,C}$ in (\ref{eq:appdx_bessel_A0C_ave}) can be further simplified. 
Let $\bar{A}_{0,C,n}$ be the $n$-th summation term, $n=1,...,N-1$,
\begin{align}
	\bar{A}_{0,C,n}=4\pi\frac{N-n}{N}J_0^2(2kR\sin(\frac{n}{N}\pi)).
\end{align}
Then, it can be derived that
\begin{align}
	\bar{A}_{0,C}=2\pi+\sum_{n=1}^{N-1}\bar{A}_{0,C,n}.
\end{align}
Because $\sin(\frac{n}{N}\pi)=\sin(\frac{N-n}{N}\pi)$, it can be derived that
\begin{align}
	\bar{A}_{0,C,N-n}&=4\pi\frac{N-(N-n)}{N}J_0^2(2kR\sin(\frac{N-n}{N}\pi))\nonumber \\
	&=4\pi\frac{n}{N}J_0^2(2kR\sin(\frac{n}{N}\pi)).
\end{align}
Thus, combine $\bar{A}_{0,C,n}$ and $\bar{A}_{0,C,N-n}$,
\begin{align}
	\bar{A}_{0,C,n}+\bar{A}_{0,C,N-n}=4\pi J_0^2(2kR\sin(\frac{n}{N}\pi)). 
\end{align}
When $N$ is even, $N-1$ is odd.
\begin{align}\label{eq:appdx_bessel_A0C_ave_even}
	\bar{A}_{0,C}&=2\pi+\sum_{n=1}^{\frac{N}{2}-1}\bar{A}_{0,C,n}+\bar{A}_{0,C,\frac{N}{2}}+\sum_{n=\frac{N}{2}+1}^{N-1}\bar{A}_{0,C,n} \nonumber \\
	&=2\pi+2\pi J_0^2(2kR)+4\pi\sum_{n=1}^{\frac{N}{2}-1} J_0^2(2kR\sin(\frac{n}{N}\pi)). 
\end{align}
When $N$ is odd, $N-1$ is even.
\begin{align}\label{eq:appdx_bessel_A0C_ave_odd}
	\bar{A}_{0,C}&=2\pi+\sum_{n=1}^{\frac{N-1}{2}}\bar{A}_{0,C,n}+\sum_{n=\frac{N+1}{2}}^{N-1}\bar{A}_{0,C,n} \nonumber \\
	&=2\pi+4\pi\sum_{n=1}^{\frac{N-1}{2}} J_0^2(2kR\sin(\frac{n}{N}\pi)). 
\end{align}
Observe (\ref{eq:appdx_bessel_A0C_ave_even}) and (\ref{eq:appdx_bessel_A0C_ave_odd}), these two equations are of similar form.
Because $\sin(\frac{n}{N}\pi)=\sin(\frac{N-n}{N}\pi)$, the two equations can be uniformly written by
\begin{align}\label{eq:appdx_bessel_A0C_ave2}
	\bar{A}_{0,C}=2\pi+2\pi\sum_{n=1}^{N-1} J_0^2(2kR\sin(\frac{n}{N}\pi)),
\end{align}
for $N$ is even or odd.


As a matter of fact, the above equation can be calculated from (\ref{eq:appdx_bessel_A0C_even}) and (\ref{eq:appdx_bessel_A0C_odd2}) by directly solving the integral, because
\begin{align}
	&\int_0^{2\pi}\cos(lN\theta_B)\,\mathrm{d}\theta_B=0, \\
	&\int_0^{2\pi}\cos(2lN\theta_B)\,\mathrm{d}\theta_B=0.
\end{align}

\chapter{Supporting Figures and Graphs}
\label{appdx:fig}

\section{SSOP and Its Upper Bound for UCA}
\label{appdx:fig:cmvs}

Comparing the curves for $K\to\infty$ and $\beta=3$ in Fig.\,\ref{fig:appdx_fig_p_and_bounds_DoE_beta_3_C} with the plot in Fig.\,\ref{fig:chp4_p_DoE_De} when $\beta=2$, it can be seen that the two curves have very similar fluctuating behavior with respect to $\theta_B$.
Furthermore, comparing the curves with different $K$ (except for $K=0$) in Fig.\,\ref{fig:appdx_fig_p_and_bounds_DoE_beta_3_C}, it can be seen that all curves have very similar behavior with respect to $\theta_B$, which verifies that the properties of $\bar{p}_C$ and $\bar{p}_{up,C}$ with respect to $N$, $R$ and $\theta_B$ is similar to those of $A_{0,C}$.

\begin{figure}
\centering
\includegraphics[scale=0.9]{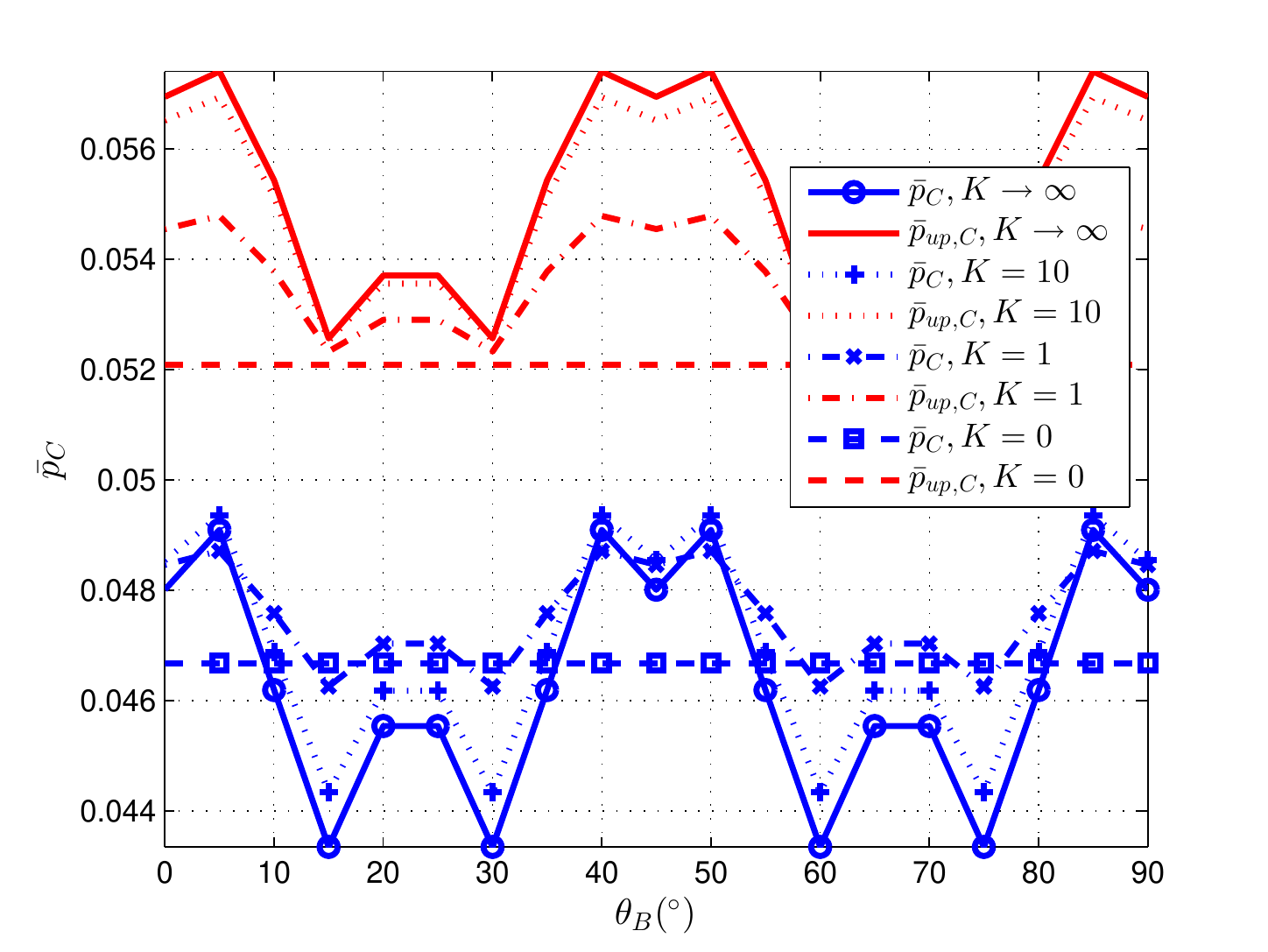}
\caption{$\bar{p}_C$ and $\bar{p}_{up,C}$ versus $\theta_B$ for different $K$. $\beta=3$, $N=8$, $R=1.75\lambda$. $P_t/\sigma_n^2=40$\,dB, $R_B=3.4594$\,bps/Hz, $R_s=1$\,bps/Hz, $\lambda_e=1\times10^{-4}$}
\label{fig:appdx_fig_p_and_bounds_DoE_beta_3_C}
\end{figure}

Similar conclusions can be drawn by comparing the curves for $K\to\infty$ and $\beta=3$ in Fig.\,\ref{fig:appdx_fig_p_and_bounds_N_beta_3_C}  with the upper plot in Fig.\,\ref{fig:chp4_p_N_and_R_De} when $\beta=2$ as well as by comparing the curves with different $K$ (except for $K=0$) in Fig.\,\ref{fig:appdx_fig_p_and_bounds_N_beta_3_C}.

\begin{figure}
\centering
\includegraphics[scale=0.9]{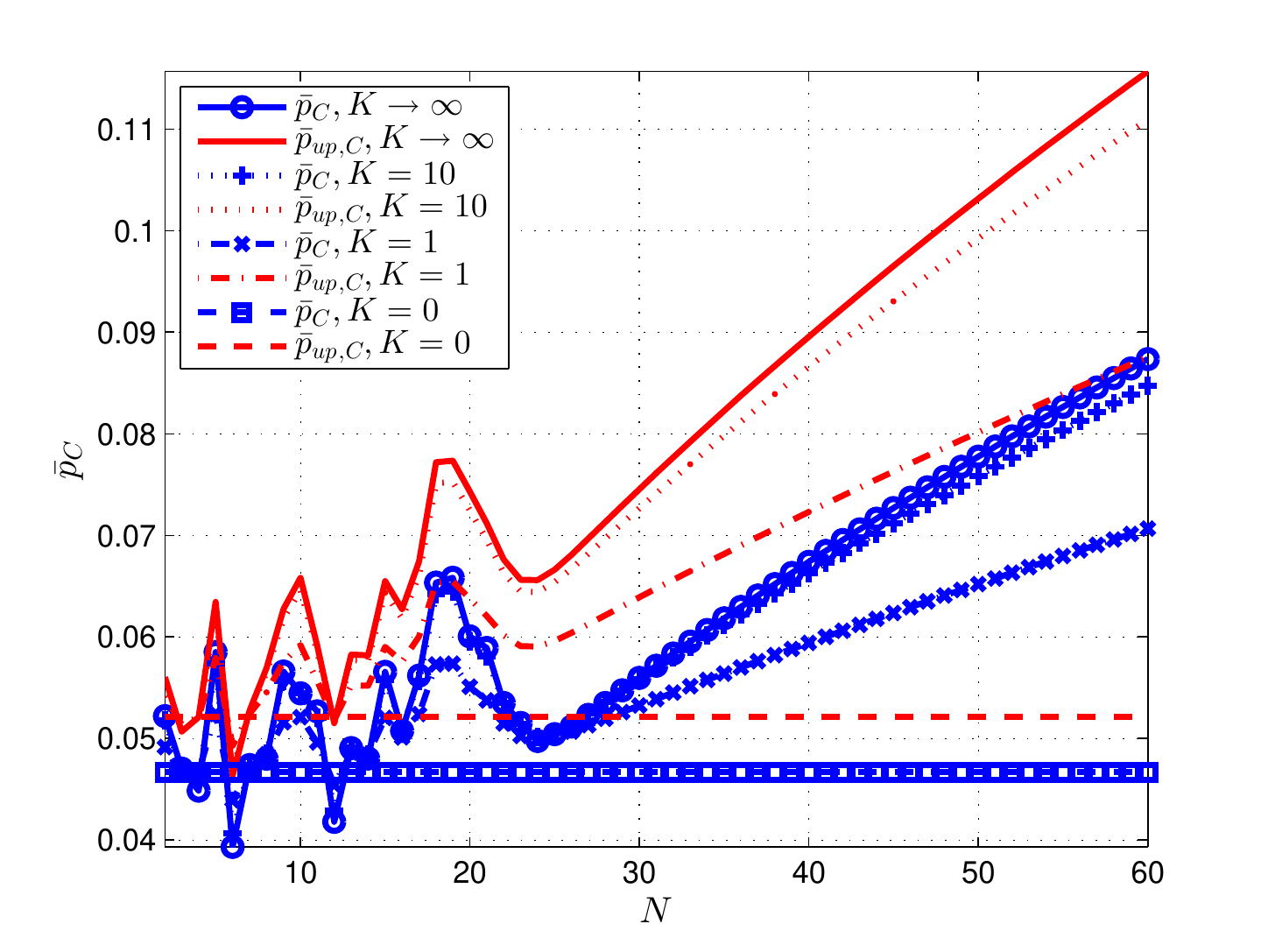}
\caption{$\bar{p}_C$ and $\bar{p}_{up,C}$ versus $N$ for different $K$. $\beta=3$, $R=1.75\lambda$, $\theta_B=0^{\circ}$. $P_t/\sigma_n^2=40$\,dB, $R_B=3.4594$\,bps/Hz, $R_s=1$\,bps/Hz, $\lambda_e=1\times10^{-4}$}
\label{fig:appdx_fig_p_and_bounds_N_beta_3_C}
\end{figure}

\section{Tightness of Upper Bound for UCA}
\label{appdx:fig:snoa}

In Fig.\,\ref{fig:appdx_fig_eta_DoE_bounds_UCA}, as $\theta_B$ changes, there are some fluctuations for $\eta_C$. 
However, the variation is in a relatively small compared to the absolute value of $\eta_C$.
Thus, it can be said that $\eta_C$ does not change much with $\theta_B$.

\begin{figure}
\centering
\includegraphics[scale=0.9]{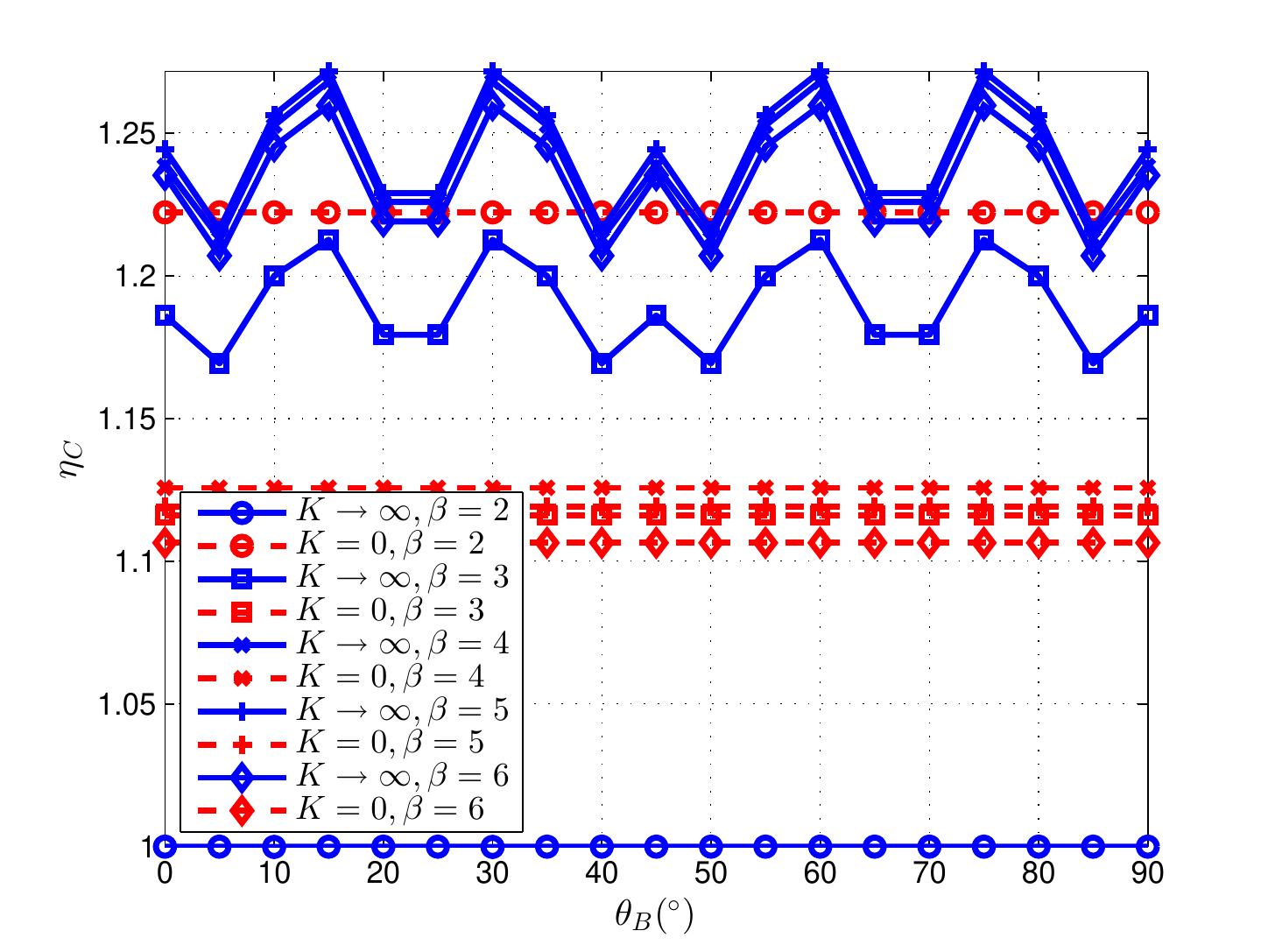}
\caption{$\eta_C$ versus $\theta_B$ for the deterministic for all $\beta$, $N=8$, $R=1.75\lambda$}
\label{fig:appdx_fig_eta_DoE_bounds_UCA}
\end{figure}

In Fig.\,\ref{fig:appdx_fig_eta_DoE_bounds_UCA}, as $N$ changes, $\eta_C$ first increase with some fluctuations, then approaches to certain values. 
The asymptotic behavior with $N$ converges instead of linearly increasing, which means that the upper bound works well even for large $N$.
Overall, the value of $\eta_C$ is smaller than 1.4.

\begin{figure}
\centering
\includegraphics[scale=0.9]{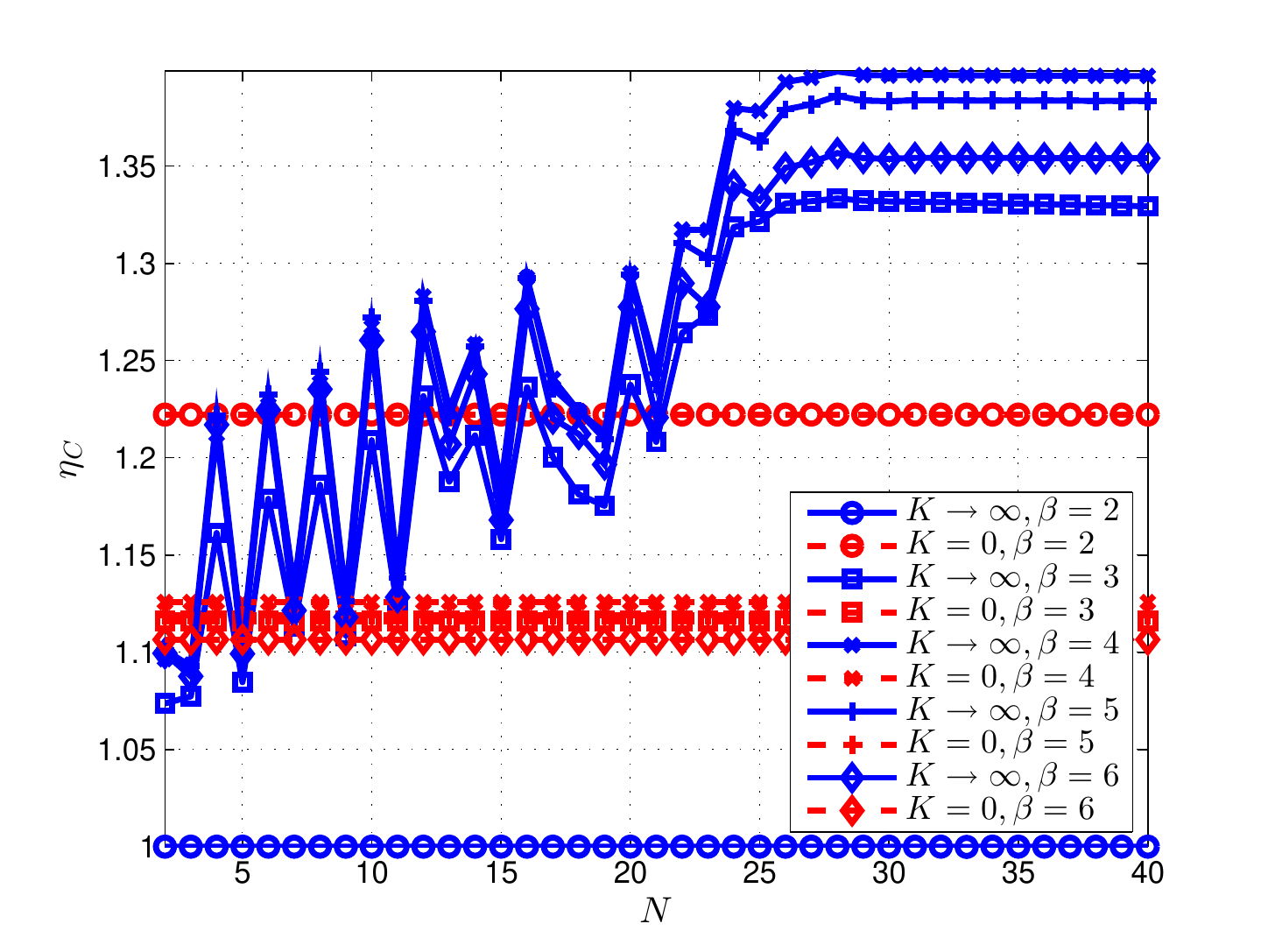}
\caption{$\eta_C$ versus $N$ for the deterministic for all $\beta$, $R=1.75\lambda$, $\theta_B=0^{\circ}$}
\label{fig:appdx_fig_eta_N_bounds_UCA}
\end{figure}

\section{Averaged SSOP over Bob's Angles for UCA}
\label{appdx:fig:ldfsjo}

As analyzed in Section\,\ref{chp5:analysis:vneiwenve}, $\bar{\bar{p}}_{up,C}$ in general decreases with some fluctuations as $R$ increases, the behavior of which is consistent with $\bar{A}_{0,C}$.
Because the upper bound $\bar{\bar{p}}_{up,C}$ is tight to $\bar{\bar{p}}_C$, the behavior of $\bar{\bar{p}}_C$ with respect to $R$ is consistent with that of $\bar{A}_{0,C}$.
In this section, the numerical results of $\bar{\bar{p}}_C$ are given for the generalized Rician channel.

An example of $\bar{A}_{0,C}$ versus $R$ is shown in Fig.\,\ref{fig:appdx_fig_A0_ave_R}, where $N=8$.
It can be seen that the curve for $\bar{A}_{0,C}$ fluctuates as $R$ increases.
In addition, $\bar{A}_{0,C}$ decreases in general.

\begin{figure}
\centering
\includegraphics[scale=0.9]{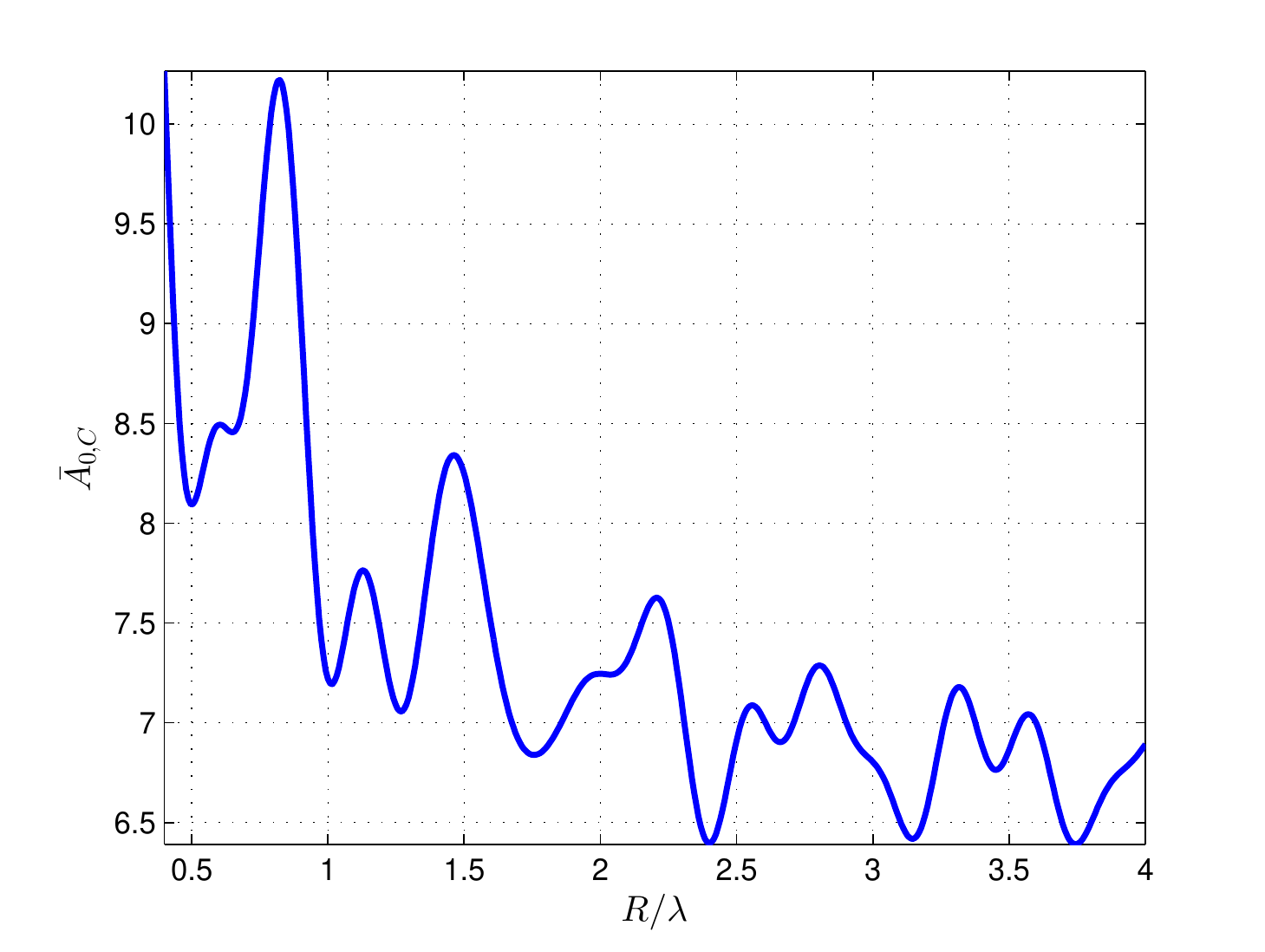}
\caption{$\bar{A}_{0,C}$. $N=8$, $\Delta d=0.5\lambda$.}
\label{fig:appdx_fig_A0_ave_R}
\end{figure}

Examples of $\bar{\bar{p}}_C$ and $\bar{\bar{p}}_{up,C}$ versus $R$ for different $K$ and $\beta=3$ are shown in Fig.\,\ref{fig:appdx_fig_p2bar_and_bounds_R_1}.
It can be seen that the behavior of $\bar{\bar{p}}_C$ and $\bar{\bar{p}}_{up,C}$ are similar and consistent with that of $\bar{A}_{0,C}$, which is decreasing with $R$ with some fluctuations.
For the Rayleigh channel, $\bar{\bar{p}}_C$ and $\bar{\bar{p}}_{up,C}$ are constant.
The curves for $K=10$ are closer to those for $K=\infty$;
The curves for $K=1$ are closer to those for $K=10$.

\begin{figure}
\centering
\includegraphics[scale=0.9]{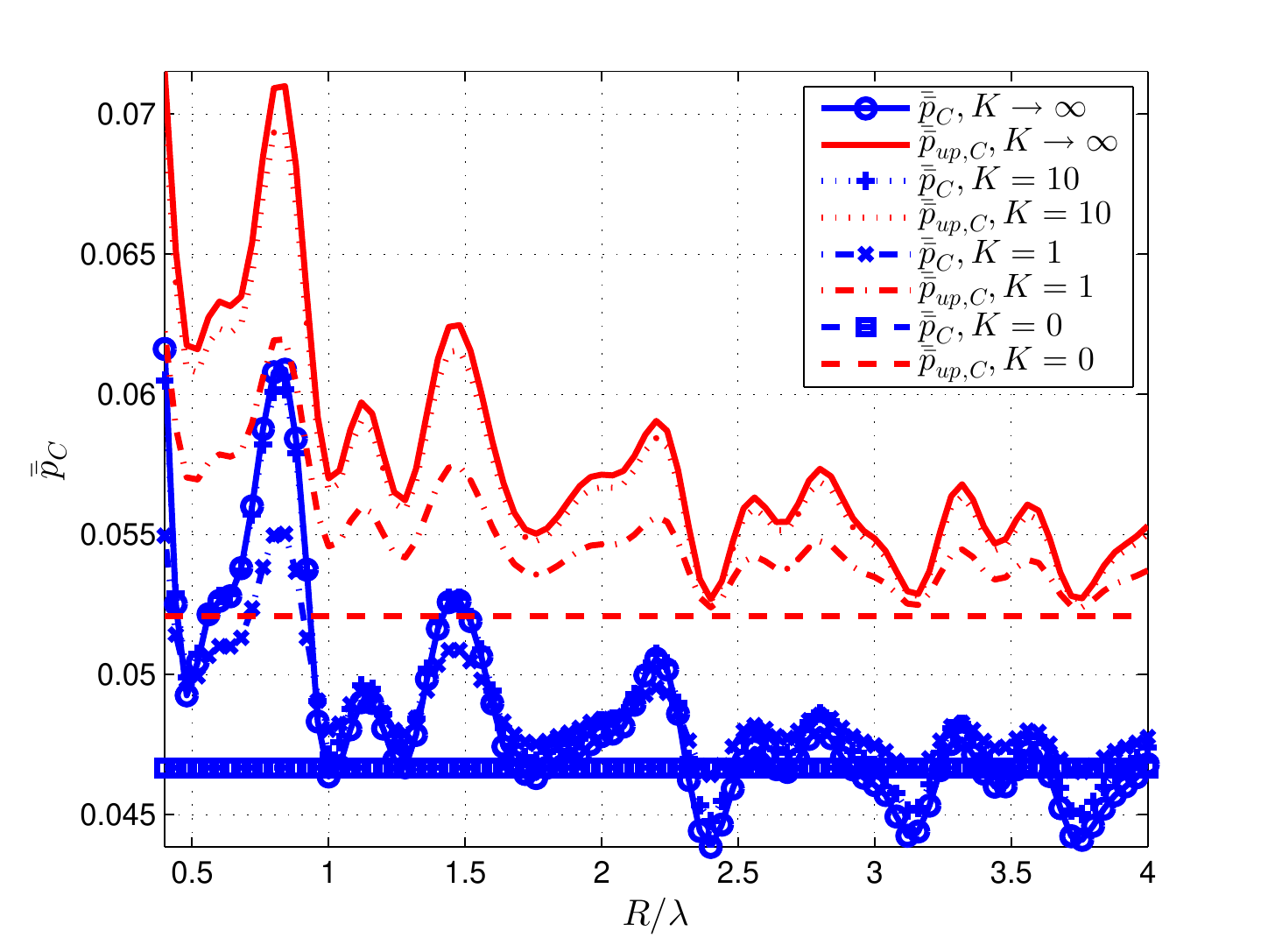}
\caption{$\bar{\bar{p}}_C$ and $\bar{\bar{p}}_{up,C}$ versus $R$, $N=8$. $P_t/\sigma_n^2=40$\,dB, $R_B=3.4594$\,bps/Hz, $R_s=1$\,bps/Hz, $\lambda_e=1\times10^{-4}$}
\label{fig:appdx_fig_p2bar_and_bounds_R_1}
\end{figure}

Combining Fig.\,\ref{fig:appdx_fig_A0_ave_R} and\,\ref{fig:appdx_fig_p2bar_and_bounds_R_1}, it can be seen that the upper bound $\bar{\bar{p}}_{up}$ has similar behavior to $\bar{A}_0$.

\chapter{Beamformer Calibrations on WARP}
\label{appdx:warp}

CFO is common is all electronic devices and is caused by various factors, e.g., oscillator frequency fluctuation and temperature changes\,\cite{murphyPhD}.
Because of the randomness of CFO, different radio interfaces experience different CFO.
In the experiments, this problem is overcome by sharing the same reference and sampling frequencies between all radio interfaces.

The second problem is more unique in this case, which is the random initial phase in each radio interface.
Even though all radio interfaces share the same reference frequency, they have separate transceiver chips to generate the carrier signal.
Because each transceiver's PLL locks at $f_0$ with some random initial phase when first turned on power, each carrier signal has different initial phase, which causes extra phase difference other than $\mathbf{w}$ in $\mathbf{u}$, resulting in a random distortion to the designed beam pattern.
As shown in Fig.\,\ref{fig:chp4_WARP_commsyst}, the relative phase offset of the $(i+1)$-th carrier signal ($i=1,...,N-1$) compared to the 1st carrier is denoted by $\varphi_i$. 

\nomenclature{CFO}{carrier frequency offset}
\nomenclature{$\varphi$}{relative phase offset between neighbor RF interfaces}

In order to form beam patterns at the transmitter, the random initial phase $\varphi_i$, $i=1,...,N-1$, needs to be calibrated. 
First, the calibration for the phase offset between two RF interfaces, e.g., RF 1 and RF 3, is introduced.


To measure the phase offset $\varphi_2$ between RF 3 and RF 1, the phase sweeping method is used.
Let $L_p$ be the total length of the packet.
A sinusoid with frequency $f$ on RF 1 is transmitted, and on RF 3, the same sinusoid is transmitted, however, in addition with an incremental phase from 0 to $2\pi$,
\begin{align}
	u_{1i}&=\sqrt{P_t}x_i=\sqrt{P_t}e^{j2\pi\frac{f}{f_s}i},\\
	u_{3i}&= \sqrt{P_t}x_i e^{j2\pi\frac{i-1}{L_p}}=\sqrt{P_t}e^{j2\pi\frac{f}{f_s}i+j2\pi\frac{i-1}{L_p}},
\end{align}
where $x_i=e^{j2\pi f i/f_s}$ $i=1,...,L_p$, where $f_s$ is the sampling frequency;
in WARPLab, $f_s$ is normally $40$\,MHz;
$e^{j2\pi\frac{i-1}{L_p}}$ is the phase increment from $0$ to $2\pi$.

\nomenclature{$f_s$}{sampling frequency}
\nomenclature{$L_p$}{packet length in WARP }

Due to the phase offset on RF 3, the received signal on RF 2 is
\begin{align}\label{eq:chp4_WARP_mspe}
	r_{2i}=h_{12}u_{1i}+h_{32}u_{3i}e^{j\varphi_2}=\sqrt{P_t}e^{j2\pi\frac{f}{f_s}i}\big(h_{12}+h_{32}e^{j2\pi\frac{i-1}{L_p}}e^{j\varphi_2}\big),
\end{align}
where $h_{12}$ and $h_{32}$ are the complex channel gains between RF 1 and RF 2 and between RF 3 and RF 2.
To be able to detect $\varphi_2$, the receiver RF interface needs to be put the same distance away from RF 1 and RF 3, so that there is no extra phase offset introduced by $h_{12}$ and $h_{32}$ in (\ref{eq:chp4_WARP_mspe}).
It is preferable to exploit the array geometry to find the middle antenna as the receive antenna.
This is the reason why RF 2 in Fig.\,\ref{fig:chp2_ULA} is chosen to estimate $\varphi_2$.

$r_{2i}$ in (\ref{eq:chp4_WARP_mspe}) can be written as
\begin{align}
	r_{2i}= \sqrt{P_t}h_{12}e^{j2\pi\frac{f}{f_s}i}\big(1+\frac{h_{32}}{h_{12}}e^{j2\pi\frac{i-1}{L_p}}e^{j\varphi_2}\big).
\end{align}
The received power is then 
\begin{align}\label{eq:chp4_WARP_oweij}
	r_{2i}^2=P_t|h_{12}|^2\big(1+\frac{h_{32}}{h_{12}}e^{j2\pi\frac{i-1}{L_p}}e^{j\varphi_2}\big)^2.
\end{align}
Because $h_{12}$ and $h_{32}$ have the same phase, $\frac{h_{32}}{h_{12}}$ in (\ref{eq:chp4_WARP_oweij}) has  a positive real value.
When $2\pi\frac{i-1}{L_p}+\varphi_2=\pi$, $r_{2i}^2$ takes its minimum value $P_t|h_{12}|^2(1-\frac{h_{32}}{h_{12}})^2$.
Thus, the index $i$ that gives $r_{2i,min}^2$ can be used to calculate $\varphi_2$,
\begin{align}
	\varphi_2=\pi-2\pi\frac{i-1}{L_p}.
\end{align}

\nomenclature{$_{\text{min}}$}{minimum value}

While it is straightforward when there is a middle antenna in the array, it requires more complex processing when there is none, e.g., estimation of $\varphi_3$ between RF 4 and RF 1.
In this case, RF 2 and RF 3 are needed as receive antennas.
On RF 1, a normal sinusoid signal is transmitted, while on RF 4, an additional phase increment is added.
\begin{align}
	u_{1i}&=\sqrt{P_t}e^{j2\pi\frac{f}{f_s}i},\\
	u_{4i}&=\sqrt{P_t}e^{j2\pi\frac{f}{f_s}i+j2\pi\frac{i-1}{L_p-1}}.
\end{align}

Due to the phase offset on RF 4, the received signals on RF 2 and RF 3 are,
\begin{align}
	r_{2i}&= \sqrt{P_t}h_{12}e^{j2\pi\frac{f}{f_s}i}\big(1+\frac{h_{42}}{h_{12}}e^{j2\pi\frac{i-1}{L_p}}e^{j\varphi_3}\big), \label{eq:chp4_WARP_ivo}\\
	r_{3i}&= \sqrt{P_t}h_{13}e^{j2\pi\frac{f}{f_s}i}\big(1+\frac{h_{43}}{h_{13}}e^{j2\pi\frac{i-1}{L_p}}e^{j\varphi_3}\big), \label{eq:chp4_WARP_owvne}
\end{align}
where $h_{12}$ and $h_{42}$ are the complex channel gains between RF 1 and RF 2 and between RF 4 and RF 2;
$h_{13}$ and $h_{43}$ are the complex channel gains between RF 1 and RF 3 and between RF 4 and RF 3.
As shown in Fig.\,\ref{fig:chp2_ULA}, RF 1 is closer to RF 2 than RF 4 by $\Delta d$; similarly, RF 1 is farther to RF 3 than RF 4 by $\Delta d$.
Thus, $h_{12}$ has the phase advance of $2\pi\frac{\Delta d}{\lambda}$ than $h_{42}$; but, $h_{13}$ has the phase advance of $-2\pi\frac{\Delta d}{\lambda}$ than $h_{43}$.
Let $\Delta\varphi=e^{-j2\pi\frac{\Delta d}{\lambda}}$.
(\ref{eq:chp4_WARP_ivo}) and (\ref{eq:chp4_WARP_owvne}) can be written as
 \begin{align}
	r_{2i}&= \sqrt{P_t}h_{12}e^{j2\pi\frac{f}{f_s}i}\big[1+\frac{|h_{42}|}{|h_{12}|}e^{j(2\pi\frac{i-1}{L_p}+\varphi_3+\Delta\varphi)}\big],  \label{eq:chp4_WARP_vonwefownv}\\
	r_{3i}&	= \sqrt{P_t}h_{13}e^{j2\pi\frac{f}{f_s}i}\big[1+\frac{|h_{43}|}{|h_{13}|}e^{j(2\pi\frac{i-1}{L_p}+\varphi_3-\Delta\varphi)}\big]. \label{eq:chp4_WARP_boen}
\end{align}

$\varphi_3+\Delta\varphi$ can be estimated from $r_{2i}$ in (\ref{eq:chp4_WARP_vonwefownv}) and $\varphi_3-\Delta\varphi$ can be estimated from $r_{3i}$ in (\ref{eq:chp4_WARP_boen}).
The method is the same as estimating $\varphi_2$.
Then $\varphi_3$ can also be estimated.

\nomenclature{$\Delta\varphi$}{addition phase offset caused by unequal propagation delay between neighbor RF interfaces}

Take the 8-element UCA in Fig.\,\ref{fig:chp2_UCA} as an example to demonstrate how to estimate $\varphi_1,...,\varphi_7$.
The 8 elements are divided into two groups, i.e., $(1,3,5,7)$ and $(2,4,6,8)$. 
For any two antenna in each group, it is easy to find a middle antenna to be the receive antenna, which allows the calibration inside each group.
To calibrate between two groups, RF 1 and RF 4 are chosen.
Then the two groups, i.e., all RF interfaces, are calibrated.
Now the beamformer is correctly set-up.

\chapter{Optimization Algorithms}
\label{appdx:opt}

\section{Algorithm 1}
\label{appdx:opt:one}

The input of the algorithm requires the first three groups of parameters that are mentioned in Section\,\ref{chp5:prob_form} as well as the wavelength $\lambda$.
In addition, the ranges and the resolutions of $R$, $\theta_B$ and $\theta$ are required as well as the limit value $Q$.
The limit of the integral is set to $Q=\pm3$ according to Section\,\ref{chp3:metric:ssop}.
The output is $R_{opt}$.

There are two main sections of the algorithm.
After the numbers of iterations regarding to $R$, $\theta_B$ and $Q$ and $\theta$ are calculated from line\,\ref{alg1:line1} to\,\ref{alg1:line5}, the first main section is from line\,\ref{alg1:line6} to\,\ref{alg1:line27}, where $\bar{\bar{p}}$ for different $R$ is calculated.
The second main section is from line\,\ref{alg1:line28} to\,\ref{alg1:line35}, where the minimum value in the vector $\bar{\bar{p}}$ is searched to find $R_{opt}$.

\begin{algorithmic}[1]
\INPUT $R_B$, $R_s$, $\sigma_n^2$, $\lambda_e$ 
\INPUT $\beta$, $K$ 
\INPUT  $\lambda$, $P_t$, $N$, $R_1$, $R_2$, $\Delta R$
\INPUT $\theta_{B1}=0$, $\Delta\theta_B$  
\INPUT $\theta_1=0$, $\theta_2=2\pi$, $\Delta\theta$ 
\INPUT $Q$, $\Delta Q$
\OUTPUT $R_{opt}$
\State{$N_R=\lfloor\frac{R_2-R_1}{\Delta R}\rfloor+1$} \label{alg1:line1}
\State{$\theta_{B2}=\frac{\pi}{N}$} \label{alg1:line2}
\State{$N_{\theta_B}=\lfloor\frac{\theta_{B2}-\theta_{B1}}{\Delta\theta_B}\rfloor+1$} \label{alg1:line3}
\State{$N_Q=\lfloor\frac{2Q}{\Delta Q}\rfloor+1$} \label{alg1:line4}
\State{$N_{\theta}=\lfloor\frac{\theta_2-\theta_1}{\Delta\theta}\rfloor+1$} \label{alg1:line5}
\State{$\bar{\bar{p}}=zeros(1,N_R)$} \Comment{preallocate memory for $\bar{\bar{p}}$} \label{alg1:line6}
\For{$idx_R \gets 1\; to\; N_R$} \Comment{loop over $R\in[R_1,R_2]$} \label{alg1:line7}
	\State{$R=R_1+(idx_R-1)\Delta R$} \Comment{current radius} \label{alg1:line8}
	\For{$idx_{\theta_B} \gets 1\; to\; N_{\theta_B}$} \Comment{loop over $\theta_B\in[\theta_{B1},\theta_{B2}]$}	 \label{alg1:line9}	
		\State{$\theta_B=\theta_{B1}+(idx_{\theta_B}-1)\Delta\theta_B$} \Comment{current Bob's angle}  \label{alg1:line10}
		\State{$S_1=0$} \label{alg1:line11}
		\For{$m \gets 1\; to\; N_Q$}   \Comment{loop over $x$}		 \label{alg1:line12} 
		  \State{$x=-Q+(m-1)\Delta Q$}	\Comment{current $x$}		\label{alg1:line13}
			\For{$n \gets 1\; to\; N_Q$} \Comment{loop over $y$}\label{alg1:line14}
				\State{$y=-Q+(n-1)\Delta Q$}	\Comment{current $y$}	\label{alg1:line15}
				\State{$S_2=0$}\label{alg1:line16}
				\For{$idx_\theta \gets 1\; to\; N_{\theta}$} \Comment{loop over $\theta\in[\theta_1,\theta_2]$}\label{alg1:line17}
					\State{$\theta=\theta_1+(idx_\theta-1)\Delta\theta$}	\Comment{current angle $\theta$}	\label{alg1:line18}
					\State{$S_2=S_2+\Big[\frac{KG_C^2(\theta,\theta_B)}{K+1}+\frac{x^2+y^2}{K+1}+\frac{2\sqrt{K}G_C(\theta,\theta_B)}{K+1}x\Big]^{\frac{2}{\beta}}\Delta\theta$} \label{alg1:line19}
				\EndFor \label{alg1:line20}
				\State{$S_1=S_1+\text{exp}\{-\frac{\lambda_e}{2}[\frac{P_t}{\sigma_n^2(2^{R_B-R_s}-1)}]^{\frac{2}{\beta}}S_2\}\frac{e^{-(x^2+y^2)}}{\pi}\Delta Q^2$} \label{alg1:line21}
			\EndFor \label{alg1:line22}
		\EndFor		\label{alg1:line23}
		\State{$\bar{\bar{p}}(idx_R)=\bar{\bar{p}}(idx_R)+1-S_1$}	\label{alg1:line24}
	\EndFor \label{alg1:line25}
	\State{$\bar{\bar{p}}(idx_R)=\bar{\bar{p}}(idx_R)/N_{\theta_B}$}		\Comment{average over all Bob's angle}	\label{alg1:line26}
\EndFor \label{alg1:line27}
\State{$R_{opt}=R_1$} \label{alg1:line28}
\State{$\bar{\bar{p}}_{min}=\bar{\bar{p}}(1)$} \label{alg1:line29}
\For{$idx_R \gets 2\; to\; N_R$} \Comment{loop over $R\in[R_1,R_2]$} \label{alg1:line30}
	\If{$\bar{\bar{p}}_{min}>\bar{\bar{p}}(idx_R)$} \label{alg1:line31}
		\State{$R_{opt}=R_1+(idx_R-1)\Delta R$} \label{alg1:line32}
		\State{$\bar{\bar{p}}_{min}=\bar{\bar{p}}(idx_R)$} \label{alg1:line33}
	\EndIf \label{alg1:line34}
\EndFor \label{alg1:line35}
\end{algorithmic}

Notice that in line\,\ref{alg1:line2}, $\bar{\bar{p}}$ is calculated in the range $\theta_B\in[0,\pi/N]$ instead of $[0,2\pi]$, as mentioned in Proposition\,\ref{prop:chp4_theta_B_range}, in order to reduce the computational complexity.
It is also worth noticing that in line\,\ref{alg1:line19},\,\ref{alg1:line21},\,\ref{alg1:line24},\,\ref{alg1:line26}, the integral in (\ref{eq:chp5_p2bar_3}) is calculated via summation.

\section{Algorithm 2}
\label{appdx:opt:two}

The first three groups of parameters that are mentioned in Section\,\ref{chp5:prob_form} as well as $\lambda$ are required as input; while the output is the look-up tables $T^{(k)}$.

There are two main sections of Algorithm\,2.
After the numbers of iterations regarding to $\theta_B$, $\theta$ and $Q$ are calculated from line\,\ref{alg2:line1} to\,\ref{alg2:line3}, the first main section is from line\,\ref{alg2:line4} to\,\ref{alg2:line33},
where $\bar{p}_C$ is calculated for all $M_{ij}$ and $\theta_B$, regardless of Bob's zone.
The second main section is from line\,\ref{alg2:line34} to\,\ref{alg2:line51}, where the look-up table $T^{(k)}$ for zone $k$, $k=1,...,K$, is calculated.

In line\,\ref{alg2:line6}, $|N\_vec|$ is the cardinality of the set $N\_vec$, which is $K$.
Similarly, $|\{M_{ij}\}|$ is the number of array modes in $\{M_{ij}\}$. 
From line\,\ref{alg2:line16} to\,\ref{alg2:line29}, $\bar{p}_C$ is calculated according to (\ref{eq:chp5_meanSSOP_Ri}), for which the integral is calculated via summation.
Notice that in line\,\ref{alg2:line15}, for the $j$-th array mode in $\{M_i\}$, there is a rotation in the DoE angle.
For example, in Fig.\,\ref{fig:chp5_arraymode}, when $\theta_B=5^{\circ}$, $\theta_{\text{doe}}=5^{\circ}$ for $M_{21}$; but $\theta_{\text{doe}}=-40^{\circ}$ for $M_{22}$.

From line\,\ref{alg2:line36} to\,\ref{alg2:line39}, the number of available array modes for zone $k$ is decided.
From line\,\ref{alg2:line40} to\,\ref{alg2:line50}, $T^{(k)}$ is calculated for all possible Bob's angle $\theta_B\in[0,2\pi]$.
Notice that the method to find the minimum value in a vector is the same as Algorithm\,1.

\begin{algorithmic}[1]
\INPUT $R_B$, $R_s$, $\sigma_n^2$, $\lambda_e$
\INPUT $\beta$, $K$
\INPUT $\lambda$, $P_t$, $N_{\text{max}}$, $R$
\INPUT $\theta_{B1}=0$, $\theta_{B2}=2\pi$, $\Delta\theta_B$
\INPUT $\theta_1=0$, $\theta_2=2\pi$, $\Delta\theta$
\INPUT $Q$, $\Delta Q$ 
\OUTPUT $T^{(k)}$
\State{$N_{\theta_B}=\lfloor\frac{\theta_{B2}-\theta_{B1}}{\Delta\theta_B}\rfloor+1$}  \label{alg2:line1}
\State{$N_{\theta}=\lfloor\frac{\theta_2-\theta_1}{\Delta\theta}\rfloor+1$}  \label{alg2:line2}
\State{$N_Q=\lfloor\frac{2Q}{\Delta Q}\rfloor+1$}  \label{alg2:line3}
\State{$N\_vec=[N_1>N_2>\cdots>N_K]$} \Comment{$N\_vec$ contains all possible numbers of active elements, $N_1=N_{\text{max}}$}  \label{alg2:line4}
\State{$N_{M_{ij}}=0$}  \Comment{total number of array modes}  \label{alg2:line5}
\For{$i \gets 1\; to\; |N\_vec|$} \Comment{loop over $N\_vec$}  \label{alg2:line6}
	\State{$N_{M_{ij}}=N_{M_{ij}}+|\{M_{i}\}|$}  \label{alg2:line7}
\EndFor    \label{alg2:line8}
\State{$\bar{p}=zeros(N_{M_{ij}},N_{\theta_B})$} \Comment{pre-allocate memory}  \label{alg2:line9}
\State{$idx=1$} \Comment{row index for matrix $\bar{p}$} \label{alg2:line10}
\For{$i \gets 1\; to\; |N\_vec|$}  \Comment{loop over $N\_vec$}  \label{alg2:line11}
	\State{$N=N\_vec(i)$} \Comment{current number of elements} \label{alg2:line12}
	\For{$j \gets 1\; to\; |\{M_{i}\}|$} \Comment{loop over $|\{M_{i}\}|$}	  \label{alg2:line13}
		\For{$idx_{\theta_B} \gets 1\; to\; N_{\theta_B}$} \Comment{loop over $\theta_B\in[\theta_{B1},\theta_{B2}]$}  \label{alg2:line14}
			\State{$\theta_{\text{doe}}=\theta_{B1}+(idx_{\theta_B}-1)\Delta\theta_B-(j-1)\frac{\pi}{N_{\text{max}}}$} \Comment{current DoE angle} \label{alg2:line15}
			\State{$S_1=0$}  \label{alg2:line16}
			\For{$m \gets 1\; to\; N_Q$}   \Comment{loop over $x$}  \label{alg2:line17}
				\State{$x=-Q+(m-1)\Delta Q$}			  \label{alg2:line18}
				\For{$n \gets 1\; to\; N_Q$} \Comment{loop over $y$}  \label{alg2:line19}
					\State{$y=-Q+(n-1)\Delta Q$}	  \label{alg2:line20}
					\State{$S_2=0$}  \label{alg2:line21}
					\For{$idx_\theta \gets 1\; to\; N_{\theta}$} \Comment{loop over $\theta\in[\theta_1,\theta_2]$}  \label{alg2:line22}
						\State{$\theta=\theta_1+(idx_\theta-1)\Delta\theta$}  \label{alg2:line23}
						\State{$S_2=S_2+\Big[\frac{KG_C^2(\theta,\theta_{\text{doe}})}{K+1}+\frac{x^2+y^2}{K+1}+\frac{2\sqrt{K}G_C(\theta,\theta_{\text{doe}})}{K+1}x\Big]^{\frac{2}{\beta}}\Delta\theta$}   \label{alg2:line24}
					\EndFor 					  \label{alg2:line25}
					\State{$S_1=S_1+\text{exp}\{-\frac{\lambda_e}{2}[\frac{P_t}{\sigma_n^2(2^{R_B-R_s}-1)}]^{\frac{2}{\beta}}S_2\}\frac{e^{-(x^2+y^2)}}{\pi}\Delta Q^2$}  \label{alg2:line26}
				\EndFor   \label{alg2:line27}
			\EndFor		  \label{alg2:line28}
			\State{$\bar{p}(idx,idx_{\theta_B})=1-S_1$}   \label{alg2:line29}
		\EndFor  \label{alg2:line30}
		\State{$idx=idx+1$}   \label{alg2:line31}
	\EndFor  \label{alg2:line32}
\EndFor  \label{alg2:line33}
\For{$k \gets 1\; to\; K$}  \Comment{zone 1 to zone $K$}  \label{alg2:line34}
	\State{$T^{(k)}=zeros(1,N_{\theta_B})$}   \label{alg2:line35}
	\State{$N_{temp}=0$}  \Comment{number of available array modes}  \label{alg2:line36}
	\For{$i \gets 1\; to\; K-k+1$}   \label{alg2:line37}
		\State{$N_{temp}=N_{temp}+|\{M_{i}\}|$}   \label{alg2:line38}
	\EndFor   \label{alg2:line39}
	\For{$idx_{\theta_B} \gets 1\; to\; N_{\theta_B}$}  \label{alg2:line40}
		\State{$idx_{min}=1$}  \label{alg2:line41}
		\State{$\bar{p}_{min}=\bar{p}(1,idx_{\theta_B})$}  \label{alg2:line42}
		\For{$idx \gets 2\; to\; N_{temp}$}  \label{alg2:line43}
			\If{$\bar{p}_{min}>\bar{p}(idx,idx_{\theta_B})$}  \label{alg2:line44}
				\State{$idx_{min}=idx$}  \label{alg2:line45}
				\State{$\bar{p}_{min}=\bar{p}(idx,idx_{\theta_B})$}   \label{alg2:line46}
			\EndIf  \label{alg2:line47}
		\EndFor  \label{alg2:line48}
		\State{$T^{(k)}(idx_{\theta_B})=idx_{min}$}   \label{alg2:line49}
	\EndFor   \label{alg2:line50}
\EndFor  \label{alg2:line51}
\end{algorithmic}

\end{appendices}

\printthesisindex 

\end{document}